\documentclass[11pt,letterpaper]{amsart}
\usepackage[foot]{amsaddr}
\usepackage{ifxetex}
\ifxetex
  \usepackage[no-math]{fontspec}
\else
\fi
\usepackage{tikz}
\usetikzlibrary{arrows.meta, positioning, shapes, calc}
\usepackage{amsmath}
\usepackage{amsfonts}
\usepackage{amssymb}
\usepackage{amsthm}

\usepackage[margin=1in]{geometry}
\usepackage{microtype}

\ifxetex 
  \usepackage[libertine]{newtxmath}
\else
  \usepackage{newtxmath}
\fi
\usepackage[tt=false]{libertine} 
\usepackage{caption}
\usepackage{tcolorbox}
\usepackage{xspace}
\usepackage{bbm}
\usepackage{hyperref}
\hypersetup{colorlinks=true,allcolors=blue}
\usepackage[linesnumbered,boxed,ruled,vlined]{algorithm2e}
\usepackage{bm}
\usepackage{bbm}
\usepackage[numbers]{natbib}
\usepackage{soul}
\soulregister\Cref7
\soulregister\eqref7

\usepackage{enumerate} 
\usepackage{enumitem}
\usepackage{ragged2e}
\usepackage{tabularx}
\usepackage{array}
\usepackage{longtable}
\usepackage{hhline}
\usepackage{multirow}
\usepackage{booktabs}
\usepackage{aliascnt}
\usepackage{braket}
\usepackage[nameinlink]{cleveref}

\makeatletter

\renewcommand{\@tocpagenum}[1]{\hss{\footnotesize\mdseries #1}}
\def\l@section{\@tocline{1}{.55em plus .1em}{1pc}{}{\bfseries}}
\def\l@subsection{\@tocline{2}{.08em}{2pc}{4pc}{\small}}
\makeatother

\newcolumntype{L}[1]{>{\raggedright\arraybackslash}p{#1}}
\newcolumntype{C}[1]{>{\centering\arraybackslash}m{#1}}
\newcolumntype{R}[1]{>{\raggedleft\arraybackslash}p{#1}}
\renewcommand\arraystretch{1.5}  

\usepackage{makecell}

\usepackage{footnote}
\usepackage{float}
\makesavenoteenv{tabular}

\renewcommand{\epsilon}{\varepsilon}

\newtheorem{theorem}{Theorem}[section]
\newaliascnt{observation}{theorem}

\aliascntresetthe{observation}
\newaliascnt{claim}{theorem}

\aliascntresetthe{claim}
\newtheorem*{claim*}{Claim}
\newaliascnt{condition}{theorem}

\aliascntresetthe{condition}
\newaliascnt{example}{theorem}

\aliascntresetthe{example}
\newaliascnt{fact}{theorem}

\aliascntresetthe{fact}
\newaliascnt{lemma}{theorem}
\newtheorem{lemma}[lemma]{Lemma}
\aliascntresetthe{lemma}
\newaliascnt{proposition}{theorem}

\aliascntresetthe{proposition}
\newaliascnt{corollary}{theorem}
\newtheorem{corollary}[corollary]{Corollary}
\aliascntresetthe{corollary}
\newcommand{\restatedtheoremnumber}{}
\newtheorem*{restatedlemma}{Lemma \restatedtheoremnumber}

\theoremstyle{definition}

\newaliascnt{definition}{theorem}
\newtheorem{definition}[definition]{Definition}
\aliascntresetthe{definition}
\newaliascnt{remark}{theorem}

\aliascntresetthe{remark}
\newtheorem*{remark*}{Remark}

\crefname{theorem}{Theorem}{Theorems}
\Crefname{theorem}{Theorem}{Theorems}
\crefname{observation}{observation}{observations}
\Crefname{observation}{Observation}{Observations}
\crefname{claim}{claim}{claims}
\Crefname{claim}{Claim}{Claims}
\crefname{condition}{condition}{conditions}
\Crefname{condition}{Condition}{Conditions}
\crefname{example}{example}{examples}
\Crefname{example}{Example}{Examples}
\crefname{fact}{fact}{facts}
\Crefname{fact}{Fact}{Facts}
\crefname{lemma}{lemma}{lemmas}
\Crefname{lemma}{Lemma}{Lemmas}
\crefname{proposition}{proposition}{propositions}
\Crefname{proposition}{Proposition}{Propositions}
\crefname{corollary}{corollary}{corollaries}
\Crefname{corollary}{Corollary}{Corollaries}
\crefname{definition}{definition}{definitions}
\Crefname{definition}{Definition}{Definitions}
\crefname{remark}{remark}{remarks}
\Crefname{remark}{Remark}{Remarks}
\crefname{assumption}{assumption}{assumptions}
\Crefname{assumption}{Assumption}{Assumptions}

\newcommand{\R}{\mathbb{R}}

\DeclareMathOperator*{\oPr}{\mathbb{P}}
\renewcommand{\Pr}[2][]{\ifthenelse{\isempty{#1}}
    {\oPr\left[#2\right]}
    {\oPr_{#1}\left[#2\right]}}

\DeclareMathOperator*{\oE}{\mathbb{E}}
\newcommand{\E}[2][]{\ifthenelse{\isempty{#1}}
    {\oE\left[#2\right]}
    {\oE_{#1}\left[#2\right]}}

\DeclareMathOperator*{\oInf}{\mathrm{Inf}}
\newcommand{\Inf}[2][]{\ifthenelse{\isempty{#1}}
    {\oInf\left[#2\right]}
    {\oInf_{#1}\left[#2\right]}}

\newcommand{\e}{\mathrm{e}}

\newcommand{\norm}[1]{\left\Vert#1\right\Vert}
\newcommand{\abs}[1]{\left\vert#1\right\vert}
\newcommand{\tuple}[1]{\left(#1\right)} 
\newcommand{\inbr}[1]{\left[#1\right]} 
\newcommand{\ceil}[1]{\lceil #1\rceil}

\newcommand{\tp}{\tuple}

\newcommand{\defeq}{\triangleq}

\renewcommand{\d}{\,\-d}

\newcommand{\ones}{{\bm 1}}

\newcommand{\DTV}[2]{d_{\mathrm{TV}}\left({#1},{#2}\right)}

\def\*#1{\bm{#1}} 
\def\+#1{\mathcal{#1}} 
\def\-#1{\mathrm{#1}} 
\def\=#1{\mathbb{#1}} 
\def\`#1{\mathfrak{#1}}

\usepackage[textsize=tiny]{todonotes}

\usepackage{xifthen}

\makeatletter
\def\prob#1#2#3{\goodbreak\begin{list}{}{\labelwidth\z@ \itemindent-\leftmargin
                        \itemsep\z@  \topsep6\p@\@plus6\p@
                        \let\makelabel\descriptionlabel}
                \item[\it Name]#1
               \item[\it Instance]                #2
                \item[\it Output]#3
                \end{list}}
\makeatother

\usepackage{enumitem}

\newcommand{\inner}[2]{\langle #1,#2\rangle}
\newcommand{\Diri}[2]{\ensuremath{\+E_{#1}\left(#2\right)}}

\DeclareMathOperator*{\oVar}{Var}
\newcommand{\Var}[2][]{\ifthenelse{\isempty{#1}}
    {\oVar\left[#2\right]}
    {\oVar_{#1}\left[#2\right]}}

\newcommand{\1}{\mathbf{1}}

\makeatletter
\renewcommand{\author@andify}{\nxandlist{\unskip, }{\unskip, }{\unskip, }}
\makeatother

\title[Simulated annealing and Weak Poincar\'e inequalities for the SK model]{Simulated annealing and Weak Poincar\'e inequalities with inverse-polynomial accuracy for the Sherrington--Kirkpatrick model}
\author[Zhe Hou]{Zhe Hou}
\author[Jingcheng Liu]{Jingcheng Liu}
\author[Yixiao Yu]{Yixiao Yu}
\address{State Key Laboratory for Novel Software Technology, Nanjing University}
\email{houzhe@smail.nju.edu.cn}
\email{liu@nju.edu.cn}
\email{yixiaoyu@smail.nju.edu.cn}

\begin{document}

\begin{abstract}

We study sampling from the Gibbs distribution of the Sherrington--Kirkpatrick (SK) model with Glauber dynamics.
For every fixed inverse temperature $0\leq\beta<1$ and every fixed $M,D>0$, we give a polynomial-time simulated annealing algorithm whose output distribution is within $n^{-M}$ total-variation distance of the Gibbs distribution, with probability at least $1-n^{-D}$ over the interaction matrix.
The algorithm starts from uniform product spins and uses Glauber dynamics along an increasing inverse-temperature schedule.
The same approach also yields partition-function estimates with relative error $n^{-M}$.
The key ingredient is a quantitative weak Poincar\'e inequality.

Building on the stochastic-localization approach to weak Poincar\'e inequalities~\cite{DaviesLeeSandhuShiWPI26}, we use Gaussian isoperimetry to directly compare exact stochastic localization paths.
Our comparison tolerates failures of local Lipschitz continuity of the posterior mean by controlling their accumulated cost.

The local Lipschitzness comes from approximating the posterior means by stationary points of the Thouless--Anderson--Palmer (TAP) free energy in locally strongly convex regions located by approximate message passing, a paradigm introduced by algorithmic stochastic localization~\cite{ElAlaouiMontanariSellke25,Celentano24}.
The approximation errors are roughly characterized by the validity of the TAP gradient equations, but in a different coordinate.
To show that the errors rarely accumulate too much, we need a strong concentration control.
A direct argument would require solving a variant of an open problem by Talagrand~\cite[Research Problem~1.7.9]{Talagrand11}.
We bypass this open problem with an intermediate conditioning step and transfer the control back by establishing the stability of the TAP gradient under deletion of coordinates.

\end{abstract}
\maketitle
\begingroup
\setcounter{tocdepth}{1}
\tableofcontents
\endgroup
\clearpage

\section{Introduction}
\label{sec:intro}

Disordered systems provide a setting for understanding how random interactions produce collective behavior.
A central example is the Sherrington--Kirkpatrick (SK) model, which consists of $n$ binary spins with random interactions between every pair of spins \cite{SherringtonKirkpatrick75}.
The interactions may favor either agreement or disagreement between spins, creating competing preferences across the system.
Despite its simple definition, the model exhibits a rich equilibrium structure and has played a central role in the development of replica symmetry breaking and the cavity method \cite{Talagrand11,Panchenko13}, as well as iterative methods for solving the Thouless--Anderson--Palmer (TAP) equations \cite{ThoulessAndersonPalmer77,Bolthausen14}.

To define the SK model, let $A$ be a random symmetric $n\times n$ matrix with zeros on the diagonal, whose entries above the diagonal are independent $N(0,1/n)$ random variables.
We call $A$ the interaction matrix and write $\mathbb P_{0,n}$ for its law.
At inverse temperature $\beta\geq0$, the model assigns to each spin configuration $x\in\{\pm1\}^n$ the Gibbs probability
\begin{equation*}
 \mu_A(x)=\frac{\exp(\beta\langle x,Ax\rangle/2)}{\sum_{z\in\{\pm1\}^n}\exp(\beta\langle z,Az\rangle/2)}.
\end{equation*}
When the inverse temperature varies, we write $\mu_{\beta,A}$ for this distribution.

The Parisi theory describes the limiting free energy and the organization of the Gibbs measure through the overlaps between spin configurations \cite{Talagrand11,Panchenko13}.
In the absence of an external field, the SK model undergoes a spin-glass transition at inverse temperature $\beta=1$: the high-temperature phase $0\leq\beta<1$ is replica symmetric, while the low-temperature phase exhibits replica symmetry breaking \cite{Talagrand11,Panchenko13}.
This equilibrium transition also suggests a computational question: does the high-temperature phase coincide with a regime in which the Gibbs distribution can be sampled efficiently?

Sampling makes the connection between equilibrium structure and computation particularly concrete.
An optimization algorithm may return a single low-energy configuration, whereas a sampling algorithm must reproduce the probabilities assigned by the Gibbs distribution.
For the SK model, a natural candidate is Glauber dynamics.
One step updates a configuration $x\in\{\pm1\}^n$ as follows:
\begin{enumerate}
 \item Choose a coordinate $i\in[n]$ uniformly at random.
 \item Resample $x_i$ from its conditional distribution under $\mu_A$ given $x_{-i}$, leaving all other coordinates unchanged.
\end{enumerate}
Understanding whether such local updates suffice throughout $\beta<1$ has been a central challenge.

Recent work has made progress on sampling from the SK Gibbs distribution.
El Alaoui, Montanari, and Sellke combined algorithmic stochastic localization with approximate message passing and a correction based on the TAP free energy; together with Celentano's local convexity estimates, this gives polynomial-time sampling throughout $\beta<1$ in normalized $2$-Wasserstein distance \cite{ElAlaouiMontanariSellke25,Celentano24}.
For single-site Glauber dynamics, a long line of work has progressively sharpened functional inequalities and mixing bounds for the SK model and related spin glasses \cite{BauerschmidtBodineau19,EldanKoehlerZeitouni22,AnariJainKoehlerPhamVuong22,AdhikariBrenneckeXuYau24,AnariJainKoehlerPhamVuong24,AnariKoehlerVuong24,ChenEldan25}.
{The known range for polynomial mixing from every initial configuration has advanced}
to $\beta<1/2$ \cite{Wang26} and, most recently, to $\beta<1/2+\varepsilon$ for some absolute constant $\varepsilon>0$ \cite{BobanLiOveisGharan26}.

We show that combining Glauber dynamics with simulated annealing gives polynomial-time sampling throughout the replica-symmetric phase, with inverse-polynomial total-variation error.

\subsection{Our results}
\label{subsec:intro-results}

Our first result establishes a weak Poincar\'e inequality (WPI) for Glauber dynamics.
Its Dirichlet form and the oscillation of a function $f:\{\pm1\}^n\to\R$ are
\begin{equation*}
 \Diri{\mu_A}{f,f}=\frac1n\sum_{i=1}^n\E[X\sim\mu_A]{\Var[X\sim\mu_A]{f(X)\mid X_{-i}}},\qquad \operatorname{osc}(f)=\max f-\min f.
\end{equation*}

\begin{theorem}[Polynomial-scale weak Poincar\'e inequality]
\label{thm:app-wpi-poly}
Fix $0\leq\beta<1$ and $D,M_0>0$.
There are constants $\Gamma>0$, depending only on $\beta$, and $C,B,N>0$, depending only on $(\beta,D,M_0)$, such that for $n\geq N$, with $\mathbb P_{0,n}$-probability at least $1-n^{-D}$,
\begin{equation*}
 \Var[\mu_A]{f}\leq Cn^{7+B\Gamma}\epsilon^{-\Gamma}\Diri{\mu_A}{f,f}+\epsilon\operatorname{osc}(f)^2
\end{equation*}
simultaneously for every $\epsilon\in[n^{-M_0},1/2]$ and every $f:\{\pm1\}^n\to\R$.
\end{theorem}

Following the annealing framework of Huang, Mohanty, Rajaraman, and Wu \cite{HuangMohantyRajaramanWu25}, we obtain the following sampling guarantee.

\begin{theorem}[Sampling by simulated annealing]
\label{thm:app-wpi-annealing}
Fix $0\leq\beta<1$ and $M,D>0$.
There is a randomized algorithm that starts from uniform product spins and uses Glauber dynamics along a deterministic increasing inverse-temperature schedule.
For all sufficiently large $n$, it runs in $n^{O_{\beta,M,D}(1)}$ time and, with $\mathbb P_{0,n}$-probability at least $1-n^{-D}$, produces a sample whose law is within $n^{-M}$ total variation of $\mu_A$.
\end{theorem}

By a standard connection between counting and sampling, we also give an algorithm to approximate the partition function $Z_A(0)=\sum_{x\in\{\pm1\}^n}\exp(\beta\langle x,Ax\rangle/2)$~\cite{JerrumValiantVazirani86,StefankovicVempalaVigoda09}.

\begin{corollary}[Approximate counting]
\label{cor:intro-counting}
Fix $0\leq\beta<1$ and $M,D>0$.
There is a randomized algorithm that, given $A$ and $\eta\in(0,1/2]$, outputs an estimate $\widehat Z$ of $Z_A(0)$.
For all sufficiently large $n$, it runs in $n^{O_{\beta,M,D}(1)}\log(1/\eta)$ time and, with $\mathbb P_{0,n}$-probability at least $1-n^{-D}$, satisfies $\Pr{\abs{\widehat Z/Z_A(0)-1}>n^{-M}\mid A}\leq\eta$ simultaneously for every $\eta\in(0,1/2]$.
\end{corollary}

\subsection{Related work}
\label{subsec:related-work}

\subsubsection{Algorithmic stochastic localization}

El Alaoui, Montanari, and Sellke developed a polynomial-time sampler based on algorithmic stochastic localization (ASL), approximating posterior means with approximate message passing (AMP) and a correction based on the TAP free energy \cite{ElAlaouiMontanariSellke25}.
Celentano subsequently strengthened the analysis of the TAP correction by proving local strong convexity around AMP iterates using a Sudakov--Fernique comparison adapted to AMP \cite{Celentano24}.
Combined with the ASL sampler, these estimates yield vanishing normalized $2$-Wasserstein error throughout $\beta<1$, which roughly measures the minimum expected fraction of disagreeing coordinates under a coupling. In comparison, our sampling result concerns total variation distance, which measures the minimum probability of any disagreement under a coupling.

The TAP free energy also underlies potential Hessian ascent (PHA), which Jekel, Sandhu, and Shi introduced to find near-optimal SK configurations through randomized updates guided by a TAP-corrected Hessian, under a full replica-symmetry-breaking assumption \cite{JekelSandhuShi25}.
Davies, Lee, Sandhu, and Shi later adapted this optimization method to sampling, approximating posterior covariances with inverse TAP Hessians to implement ASL \cite{DaviesLeeSandhuShi26}.
They combined this construction with Jarzynski reweighting and rejection sampling to obtain polynomial-time sampling with vanishing total-variation (TV) error for $\beta<1/2$.

\subsubsection{Rapid mixing of Glauber dynamics}

For Glauber dynamics, early results for the SK model covered $\beta<1/4$: Bauerschmidt and Bodineau established a log-Sobolev inequality \cite{BauerschmidtBodineau19}, and Eldan, Koehler, and Zeitouni proved a spectral-gap bound yielding polynomial mixing from every initial configuration \cite{EldanKoehlerZeitouni22}.
Within this range, entropic independence and localization subsequently improved the bounds to optimal-order mixing \cite{AnariJainKoehlerPhamVuong22,ChenEldan25}.
Anari, Koehler, and Vuong then used trickledown to extend the temperature range to a threshold of approximately $0.295$ \cite{AnariKoehlerVuong24}.
More recently, Wang extended optimal-order mixing to every fixed $\beta<1/2$ \cite{Wang26}, and Boban, Li, and Oveis Gharan pushed polynomial mixing beyond this threshold to $\beta\leq1/2+c$ for an absolute constant $c>0$ \cite{BobanLiOveisGharan26}.

Beyond this high-temperature regime, Bandeira, El Alaoui, and R\"odder proved polynomial mixing at any fixed inverse temperature provided the SK model has a sufficiently strong constant external field, whose strength is independent of $n$ \cite{BandeiraElAlaouiRodder26}.

\subsubsection{Glauber dynamics from a warm start}

Huang, Mohanty, Rajaraman, and Wu developed a framework for analyzing simulated annealing through weak Poincar\'e inequalities (WPI) and warm starts, and applied it to spherical spin glasses \cite{HuangMohantyRajaramanWu25}.
For the SK model, Davies, Lee, Sandhu, and Shi subsequently derived a WPI for $\beta<1/2$ from approximate stochastic localization and conductance, and used it to analyze Glauber dynamics from a warm start \cite{DaviesLeeSandhuShiWPI26}. Notably, their WPI is not strong enough to support a simple simulated annealing.

\subsubsection{Zero-free methods}

For approximate counting, Bencs, Huang, Lee, Liu, and Regts combined zero-freeness with a truncated Taylor expansion of the log partition function \cite{BencsHuangLeeLiuRegts26}.
Their method yields deterministic quasipolynomial-time approximation for the SK model throughout $\beta<1$, with high probability over the interaction matrix.

Together, these works provide several routes to sampling and counting in the SK model, with different temperature ranges and approximation guarantees.
We compare these guarantees with those of our annealing algorithm in \Cref{tab:related-work}.

\begin{table}[H]
\centering
\footnotesize
\setlength{\tabcolsep}{4pt}
\renewcommand\arraystretch{1.05}
\captionsetup{font=footnotesize,skip=5pt}
\begin{tabular}{@{}L{0.18\textwidth}L{0.13\textwidth}L{0.30\textwidth}L{0.30\textwidth}@{}}
\toprule
Work & Range & Method & Guarantee\\
\midrule
\cite{ElAlaouiMontanariSellke25,Celentano24} & $\beta<1$ & ASL with AMP and TAP free energy & vanishing normalized $2$-Wasserstein error\\
\cite{DaviesLeeSandhuShi26} & $\beta<1/2$ & potential Hessian ascent and ASL & vanishing TV error\\
\midrule
\cite{EldanKoehlerZeitouni22} & $\beta<1/4$ & Glauber dynamics & polynomial mixing from every start\\
\cite{Wang26} & $\beta<1/2$ & Glauber dynamics & optimal-order mixing from every start\\
\cite{BobanLiOveisGharan26} & $\beta\leq1/2+c$ & Glauber dynamics & polynomial mixing from every start\\
\cite{BandeiraElAlaouiRodder26} & strong field & Glauber dynamics & polynomial mixing\\
\midrule
\cite{DaviesLeeSandhuShiWPI26} & $\beta<1/2$ & WPI and warm-start Glauber dynamics & WPI coefficient $n\e^{C/\epsilon}$\\
This paper & $\beta<1$ & WPI and heat-bath simulated annealing & TV and relative error $n^{-M}$ in polynomial time\\
\midrule
\cite{BencsHuangLeeLiuRegts26} & $\beta<1$ & zero-free methods & deterministic quasipolynomial counting\\
\bottomrule
\end{tabular}
\caption{Selected guarantees for a typical SK interaction matrix, grouped by the four approaches above; the strong-field result includes an external field.
Here $c>0$ is an absolute constant, $\epsilon$ is the WPI defect, and ``every start'' means every initial configuration.}
\label{tab:related-work}
\end{table}

\subsection{Concurrent work}
\label{subsec:concurrent-work}

Very recently, Lee, Sandhu, and Shi obtained polynomial-time sampling throughout $\beta<1$ with $o_n(1)$ TV error \cite{LeeSandhuShiPHAIV26} by potential Hessian ascent (PHA).
Their algorithm combines algorithmic stochastic localization with rejection sampling on path space via Jarzynski's equality, using local convexity of the TAP free energy to control the approximations.
Lee, Sandhu, and Shi also announce a forthcoming WPI result for the SK model throughout $\beta<1$ together with an algorithmic construction of a warm start \cite{LeeSandhuShiPHAIV26}.

In comparison, we establish a WPI by controlling the gradient residual of the TAP free energy at the posterior mean.
This inequality yields a sampling algorithm based on simulated annealing with single-site Glauber dynamics.
For every fixed $0\leq\beta<1$ and $M>0$, it achieves TV error $n^{-M}$ in polynomial time with high probability over the interaction matrix.

{
In two concurrent preprints, OpenAI reports a Poincar\'e inequality with spectral gap $\Omega_\beta(n^{-1})$ for Glauber dynamics \cite{OpenAISKSpectralGap26}, and worst-case total-variation cutoff on the $n\log n$ scale for the same dynamics \cite{OpenAISKCutoff26}  throughout $0\le \beta < 1$.
At criticality $\beta=1$, two further preprints report a lower bound of $n^{5/3-o(1)}$ on the worst-case mixing time \cite{OpenAISKCriticalSlowing26} and a matching mixing time of $n^{5/3+o(1)}$ \cite{OpenAISKCriticalMixing26}.
\par}

\subsection{AI Disclosure}
We used AI throughout various stages of the project, both in explorations of different proof strategies and in carrying out a few heavy calculations. We started with GPT 5.5 Pro in trying to explore and understand what it takes to combine~\cite{ElAlaouiMontanariSellke25,Celentano24} with~\cite{DaviesLeeSandhuShi26} to get sampling algorithms in TV distance. GPT assisted with
strengthening the bounds in \cite{ElAlaouiMontanariSellke25,Celentano24} in a number of ways (such as deriving $\E[]{\|m-m^*\|_2^2} = O(1)$), and it is the starting point of the current framework of analyzing {the} TAP residual. These bounds were enough to support the application of Girsanov comparison and {the path-space correction post processing} in~\cite{DaviesLeeSandhuShi26} up to $\beta<1$.  Then we found out the two new advances in the SK model in~\cite{BandeiraElAlaouiRodder26,DaviesLeeSandhuShiWPI26}, and {when} GPT 5.6 Pro was released, we used it to extend our analysis to a weak Poincar\'e inequality throughout $\beta<1$, as well as to the current inverse polynomial scale to support simulated annealing. We supplied high-level guidance throughout the proof developments. GPT 5.6 Pro and Ultra {were} also used extensively in exploring alternative proof strategies. In particular, the square-root interpolation idea in Lemma 4.2, as well as the strategy that sidesteps sub-Gaussian concentration of the TAP residual, which underlies the better concentration control in Section 4.2 and Section 4.3, were selected among a few alternatives found by GPT 5.6 Pro. GPT 5.6 Pro also claimed a false proof of {an} FPRAS extension based on our simulated annealing approach, which GPT 6 Astra Pro managed to {fix}. GPT 6 Astra Pro also found a separate deterministic FPTAS for the SK model, which we had to defer to a separate manuscript in view of {OpenAI's recent announcements}~\cite{OpenAISKSpectralGap26,OpenAISKCutoff26,OpenAISKCriticalSlowing26,OpenAISKCriticalMixing26}. %
We also used Codex and GPT 5.5 xHigh, GPT 5.6 Sol Ultra, and GPT 6 Astra Ultra through various stages of preparing our manuscript, especially in checking for typos and writing. 
We verified all the proofs produced by AI, and we take full responsibility in their correctness.

\subsection{Acknowledgements}
We thank Zongchen Chen and Heng Guo for their helpful discussions during early stages of the work.
We acknowledge the support of the National Natural Science Foundation of China under Grant No.~62472212.

\section{Technical overview}
Our simulated annealing algorithm follows directly from~\cite{HuangMohantyRajaramanWu25} and our weak Poincar\'e inequality (WPI),
thanks to its better dependency on the \emph{defect} $\epsilon$. We give a high level overview around WPI next.

\subsection{Variance decomposition}
The starting point of our weak Poincar\'e inequality is to use the algorithmic stochastic localization~\cite{ElAlaouiMontanariSellke25} as an analytical proxy, similar to~\cite{DaviesLeeSandhuShiWPI26}.
{We fix the interaction matrix $A$.}
{Following the Gaussian-observation representation of stochastic localization~\cite{Eldan13,ChenEldan25,ElAlaouiMontanari22}, we draw $X\sim\mu_A$ and an independent standard Brownian motion $(B_t)_{t\geq0}$ in $\R^n$ with $B_0=0$, and set $Y_t=tX+B_t$.}
{For an external field $y\in\R^n$, we write $\mu_{A,y}(x)\propto\mu_A(x)\exp(\inner{y}{x})$ for the Gibbs distribution tilted with this field.}
By {Bayes' rule}, the posterior distribution of $X$ given $Y_t$ is the tilted distribution $\mu_{A,Y_t}$. As such, observing $Y_t$ at a later time $t$ induces a large external field for $X$.
Davies {et al.}~\cite{DaviesLeeSandhuShiWPI26} then {decompose} the variance by the law of total variance at a sufficiently large fixed time $T$:
\begin{equation}
\label{eq:law-of-total-var}
 \Var[\mu_A]{f}
 =\E[Y_T]{\Var[\mu_{A,Y_T}]{f}}+\Var[Y_T]{\E[\mu_{A,Y_T}]{f}}.
\end{equation}
To establish the weak Poincar\'e inequality, one needs to upper bound the RHS by comparing to the Dirichlet form $\+E_{\mu_A}(f, f)$.
{The two terms  are controlled separately.}

To control the first term, note that $\mu_{A, Y_T}$ is the SK model with the external field $Y_T$.
Recent work~\cite{BandeiraElAlaouiRodder26} analyzed the mixing time of the Glauber dynamics for the SK model under a strong enough external field, which provides a Poincar\'e inequality and can be used to upper bound the first term.

To control the second term of \cref{eq:law-of-total-var}, Davies et al.~\cite{DaviesLeeSandhuShiWPI26} introduced an approximate stochastic localization $\hat{Y_t}$ which is close to $Y_t$ in KL divergence, such that $\hat{Y_t}$ has {a} ``Lipschitz'' property and is easier to control. However, transferring probability masses between $\hat{Y_t}$ and $Y_t$ assuming only control in KL divergence inevitably introduces an exponential loss.
{In comparison, we analyze $Y_t$ through its \emph{innovation-process representation}, whose probability law is governed by Brownian motion. This allows us to control this term directly through Gaussian isoperimetry. }

\subsection{The idealized Lipschitz case}
\label{subsection:idealized}
For any $y\in \mathbb{R}^n$, let $m(A, y) = \E[x\sim \mu_{A, y}]{x}$, {which we call the \emph{posterior mean}}.
We first demonstrate our method by considering a simple case: $m(A, \cdot)$ is globally $L$-Lipschitz, where {$L$ is independent of $n$}.
{{We remark that} \cite{DaviesLeeSandhuShiWPI26} has no loss in this idealized case, but their argument through {a} log-Sobolev inequality~\cite[Lemma~2.12]{DaviesLeeSandhuShiWPI26} does not in itself handle rare failures in Lipschitzness, and their separate argument for transferring small-set probability estimates between probability laws through KL introduces exponential loss~\cite[Lemma~4.3 and Theorem~4.4]{DaviesLeeSandhuShiWPI26} even in the case of $\beta<\frac12$.
Our method closely follows the architecture of~\cite{DaviesLeeSandhuShiWPI26} in this case, except that we sidestep their LSI argument through quantile comparisons.}

Our idea is to upper bound the variance $\Var[Y_T]{\E[\mu_{A,Y_T}]{f}}$ {using quantiles}.
Let $Q_f$ be the quantile function (inverse CDF) of $\E[\mu_{A,Y_T}]{f}$, where the randomness is over $Y_T$.
\begingroup
For any square-integrable real random variable $Z$, we take $q=Q_Z(1/2)$ to be the median and write
\begin{align*}
 \Var{Z}\leq\E{(Z-q)^2}
 =&~\int_0^{1/2}\bigl((Q_Z(v)-q)^2+(Q_Z(1-v)-q)^2\bigr)\,\mathrm d v\\
 \leq&~\int_0^{1/2}[Q_Z(1-v)-Q_Z(v)]^2\,\mathrm d v.
\end{align*}
The last inequality uses $Q_Z(v)\leq q\leq Q_Z(1-v)$ for $0<v<1/2$.
\endgroup
Hence, we have that
\[
\Var[Y_T]{\E[\mu_{A,Y_T}]{f}} \le \int_{0}^{1/2} [Q_f(1-v)-Q_f(v)]^2 \mathrm d v.
\]
Thus, it suffices to upper bound $[Q_f(1-v)-Q_f(v)]^2$.

{We write $C_0([0,T];\R^n)$ for the space of continuous paths $\omega$ with $\omega_0=0$.
The \emph{Wiener measure} on this space is the law of a standard Brownian motion in $\R^n$ starting at zero.}
{Next, observe that $Y_t$ can be realized by Brownian motions (see~\cite[Section~III]{ElAlaouiMontanari22} and \Cref{lem:app-wpi-observation-law}):}
for any $0\le t\le T$, the solution of $Y_t^{\omega} = \omega_t + \int_{0}^t m(A, Y_s^{\omega}) \mathrm d s$, when {$\omega$ is sampled from Wiener measure}, has the same distribution as $Y_t$.
Hence, given $A$, the only randomness needed to study $Y_t$ is the Brownian motion $\omega_t$.
Let $F(\omega):=\E[\mu_{A, Y^{\omega}_T}]{f}$. Since $Y^{\omega}_T\overset{d}{=} Y_T$ under the Wiener measure, we also have $F(\omega)\overset{d}{=} \E[\mu_{A, Y_T}]{f}$ under the Wiener measure.
This is where Gaussian isoperimetry enters the picture.
For illustration, it is helpful to consider a one-dimensional Gaussian.
Two measurable sets with probability $v\le \frac12$ can have distance at most $2\Phi^{-1}(1-v)\leq2\sqrt{2\log(1/v)}$, where $\Phi(\cdot)$ is the CDF of a standard Gaussian distribution.
Borell~\cite{Borell75} and independently Sudakov and Tsirel'son~\cite{SudakovTsirelson78} showed that this Gaussian isoperimetry survives for {an} arbitrary Gauss space.
{We apply it on $C_0([0,T];\R^n)$ with Wiener measure, measuring distances through shifts in the \emph{Cameron--Martin space} $\mathbb H_T$.
This space consists of absolutely continuous paths $h$ with $h_0=0$ and square-integrable derivative, with norm
$ \|h\|_{\mathbb H_T}^2=\int_0^T\left\|\frac{\mathrm d h_t}{\mathrm d t}\right\|_2^2\,\mathrm d t$.
Thus the Brownian paths belong to $C_0([0,T];\R^n)$, whereas the shift between the two paths chosen below belongs to $\mathbb H_T$; see \Cref{subsec:prelim-wiener-space} for the formal definitions.}

We consider the following two sets of ``Brownian paths'', each with Wiener measure at least $v$:
\[
\begin{aligned}
 L(v)&\mathrel{\defeq}\{\omega \in {C_0([0,T];\R^n)}: F(\omega) \le Q_{f}(v)\},\\
 U(v)&\mathrel{\defeq}\{\omega' \in {C_0([0,T];\R^n)}: F(\omega') \ge Q_{f}(1-v)\}.
\end{aligned}
\]
By Markov's inequality, the following set has Wiener measure at least $1-\frac v2$:
\[
G(v)\mathrel{\defeq} \left\{
\omega \in {C_0([0,T];\R^n)}:
\Var[\mu_{A, Y_T^{\omega }}]{f} \le \frac{2}{v}\E[\omega]{\Var[\mu_{A, Y_T^{\omega }}]{f}}
\right\}.
\]
Then, the sets $L(v) \cap G(v)$ and $U(v) \cap G(v)$ each {have} Wiener measure at least $\frac v2$, and Gaussian isoperimetry tells us that they cannot be farther apart than $O(\sqrt{\log(1/v)})$ in the {Cameron--Martin distance} (see \Cref{lem:app-wpi-two-set-gaussian}).
Put differently, there are two paths $\omega \in L(v)\cap G(v)$ and $\omega+h \in U(v)\cap G(v)$ with $h\in\mathbb H_T$ and $\|h\|_{\mathbb H_T}=O(\sqrt{\log(1/v)})$.
Then, by the definition of $L(v)$ and $U(v)$, we have 
\[
[Q_f(1-v)-Q_f(v)]^2  \le \inbr{ \E[\mu_{A, Y^{\omega + h}_T}]{f} - \E[\mu_{A, Y^{\omega}_T}]{f} }^2.
\]
Thus, Gaussian isoperimetry reduces the quantile comparison to a pair of posteriors generated by nearby Brownian paths.

{{We note that}
\[
Y_t^{\omega+h}-Y_t^\omega
=
h_t+
\int_0^t
\bigl[m(A,Y_s^{\omega+h})-m(A,Y_s^\omega)\bigr]ds.
\]
{By the assumed $L$-Lipschitz property of $m(A,\cdot)$, Gronwall's inequality (\Cref{lem:gronwall-path-stability}) shows that the fields $Y_t$ under nearby Brownian paths are also close:}}
\begin{equation}
\label{eq:Lipschitz-1}
\|Y_T^{\omega+h}-Y_T^\omega\|_2\le\sqrt T\,\exp(LT)\|h\|_{\mathbb H_T}.
\end{equation}
The same Cameron--Martin stability estimate is used in~\cite[Lemma~6.4]{DaviesLeeSandhuShiWPI26}, but for the approximate rather than the exact stochastic localization process.
{This suggests comparing the two posterior expectations through the distance between the fields $Y_T$.}

{We formalize this comparison in \Cref{lem:app-wpi-posterior-comparison}: for any two fields $y,z\in\R^n$, we have}
\begin{align*}
 \inbr{\E[\mu_{A, y}]{f}-\E[\mu_{A, z}]{f}}^2
 \leq&~2\bigl(\exp([D(\mu_{A, y}\Vert\mu_{A, z})
                         +D(\mu_{A, z}\Vert\mu_{A, y})]/2)-1\bigr)
             \tp{\Var[\mu_{A, y}]{f}+\Var[\mu_{A, z}]{f}}.
\end{align*}

By the assumed $L$-Lipschitzness of $m(A,\cdot)$, we have that
\begin{equation}
\label{eq:Lipschitz-2}
 D(\mu_{A,y}\Vert\mu_{A,z})+D(\mu_{A,z}\Vert\mu_{A,y})
 =\inner{y-z}{m(A,y)-m(A,z)}
 \le L\|y-z\|_2^2.
\end{equation}
{Taking $y=Y_T^\omega$ and $z=Y_T^{\omega+h}$, we bound $\|y-z\|_2$ using \cref{eq:Lipschitz-1} and the Cameron--Martin norm of $h$.}
Combined with the fact that $\omega$ and $\omega+h$ are both in $G(v)$, we have that for some  $\Gamma=\Gamma(L,T)>0$,
\[
[Q_f(1-v)-Q_f(v)]^2
\le
\frac{C_{L,T}}{v^{1+\Gamma}}\E[\omega]{\Var[\mu_{A, Y_T^{\omega }}]{f}}.\]

Additionally, we also have the trivial bound $Q_f(1-v)-Q_f(v)\le \operatorname{osc}(f)$ for $v$ that are too small.
Finally, by truncating the integration at $\varepsilon/2$, we have that 
\[
\Var[Y_T]{\E[\mu_{A,Y_T}]{f}} \le \frac{\varepsilon}{2} \operatorname{osc}(f)^2 + 2C_{L,T} \E[\omega]{\Var[\mu_{A, Y_T^{\omega }}]{f}}  \int_{\varepsilon/2}^{1/2} \frac{1}{{v}^{1+\Gamma}} \mathrm{d} v \le \frac{\varepsilon}{2} \operatorname{osc}(f)^2 + \frac{C'_{L,T}}{\varepsilon^{\Gamma}} \E[\omega]{\Var[\mu_{A, Y_T^{\omega }}]{f}}  .
\]

Finally,  $\E[\omega]{\Var[\mu_{A, Y_T^{\omega }}]{f}}$ can be converted back to the original Dirichlet form by a Poincar\'e inequality under sufficiently large external fields~\cite{BandeiraElAlaouiRodder26}, concluding the WPI with defect $\varepsilon$.

\subsection{A locally Lipschitz approximation to the posterior mean}
A global Lipschitz bound for the posterior mean $m(A,\cdot)$, with a constant independent of $n$, is not known to hold with high probability throughout the full range $\beta<1$.
Instead, we approximate $m$ by a {locally} Lipschitz $m^*$ and prove that the cumulative failure of Lipschitzness is rarely destructive to our argument.
We remark that this is conceptually similar to~\cite{DaviesLeeSandhuShiWPI26}, but crucially, our Lipschitz argument is much more robust than their LSI-based argument in handling the rare failures of Lipschitzness.
At a high level, the Lipschitz property only enters our argument exactly twice:
firstly in \cref{eq:Lipschitz-1}, which is to translate the closeness in ``Brownian paths'' to closeness in the fields $Y_t$; and secondly in \cref{eq:Lipschitz-2}, which is to translate the closeness in the fields $Y_t$ into closeness in symmetric KL divergence.
These two uses do not require \emph{global} Lipschitzness of $m${;} they only concern a local Lipschitzness at the chosen pair of ``Brownian paths'' $\omega$ and $\omega+h$.
It is also helpful to recall how we chose a ``good set of Brownian paths'' $G(v)$ in the last subsection: we can safely exclude a bad set of paths of Wiener measure smaller than $v$, and our Lipschitz argument is robust to that. In essence, this suggests that we only have to show {that violations} of the Lipschitzness of $m$ that are cumulatively severe are rare in the Wiener measure.

Our idea is to relax the Lipschitz property by a cost function  $\overline{\mathfrak e}_{n,\beta} \approx \norm{m-m^*}_2^2$. 
 This is formalized in the next lemma: part~(i) allows errors measured by the cost function, and part~(ii) bounds the cumulative costs of violating Lipschitzness along the path and at the endpoint $T$.

\begin{lemma}[Lipschitz property with cost functions]
\label{lem:app-wpi-branches}
Let $\R^{n\times n}_{\mathrm{sym},0}$ be the space of real symmetric zero-diagonal matrices.  
Fix $0<\beta_0<\beta_1<1$ and a constant $T>0$.
There is a constant $L$ such that, for every $n$ and every
$\beta\in[\beta_0,\beta_1]$, there is a jointly Borel function
\begin{equation*}
 \overline{\mathfrak e}_{n,\beta}:
 \R^{n\times n}_{\mathrm{sym},0}\times[0,T]
 \times\R^n\longrightarrow[0,4n]
\end{equation*}
with the following properties.
\begin{enumerate}[label=\textup{(\roman*)}]
\item For every $A,t$ and $y_1,y_2\in\R^n$,
\begin{align*}
 \|m(A,y_1)-m(A,y_2)\|_2
 \le{}&L\|y_1-y_2\|_2
 +\sqrt{\overline{\mathfrak e}_{n,\beta}(A,t;y_1)}
 +\sqrt{\overline{\mathfrak e}_{n,\beta}(A,t;y_2)}.
\end{align*}

\item
{For every fixed $R>0$, there are $C_R,n_R>0$ such that,
for every $n\ge n_R$,
\begin{align*}
 \sup_{\beta\in[\beta_0,\beta_1]}
 \mathbb P_{0,n,\beta}^{\mathrm{aug}}
 \left\{
 \int_0^T
 \overline{\mathfrak e}_{n,\beta}
       (A,t;Y_t)\,\mathrm d t
 >C_R\log n
 \right\}
 &\le n^{-R},
\end{align*}
and
\begin{align*}
 \sup_{\beta\in[\beta_0,\beta_1]}
 \mathbb P_{0,n,\beta}^{\mathrm{aug}}
 \left\{
 \overline{\mathfrak e}_{n,\beta}
       (A,T;Y_T)
 >C_R\log n
 \right\}
 &\le n^{-R}.
\end{align*}
Here $\mathbb P_{0,n,\beta}^{\mathrm{aug}}$ is the joint law where one draws the matrix $A\sim \mathbb P_{0,n}$, then draws $X \sim \mu_A$, then {sets} $Y_t=tX+B_t$ for an independent Brownian motion $B_t, t\geq0$.}
\end{enumerate}
\end{lemma}

Next, we explain how this lemma replaces the Lipschitzness assumption used in our arguments.  Part~\textup{(ii)} of the lemma is what defines a good set of Brownian paths. For a fixed $A$, denote the two errors by
\begin{equation*}
 U_A(\omega)
 =\int_0^T\overline{\mathfrak e}_{n,\beta}(A,t;Y_t^\omega)\,\mathrm d t , \qquad \qquad E_A(\omega)
 =\overline{\mathfrak e}_{n,\beta}(A,T;Y_T^\omega).
\end{equation*}
Then we can define the set of ``good Brownian paths'' by
\[
 H_R=\{\omega:U_A(\omega)\le C_R\log n,\ E_A(\omega)\le C_R\log n\}.
\]
Part~\textup{(ii)} of~\Cref{lem:app-wpi-branches} tells us that, for all but a negligible fraction of matrices $A \sim \mathbb P_{0,n}$, the set $H_R$ will have Wiener measure at least $1-n^{-r}$ for every $r$ (by taking $R$ sufficiently large). 

Then,  we follow the same quantile comparison as before, but applying Gaussian isoperimetry to the sets $L(v)\cap G(v)\cap H_R$ and $U(v)\cap G(v)\cap H_R$, replacing $L(v)\cap G(v)$ and $U(v)\cap G(v)$. Assuming $v \ge n^{-M}$ for some $M$, the same argument would still give two nearby ``Brownian paths'' $\omega$ and $\omega+h$ such that  $\|h\|_{\mathbb H_T}=O(\sqrt{\log(1/v)})$, but with the additional property that they are both in $H_R$.

The first use of Lipschitzness is replaced by the path cost.
Subtracting the two localization equations and applying
part~\textup{(i)} to $Y_t^\omega$ and $Y_t^{\omega+h}$,
we obtain the following bound by Cauchy--Schwarz and Gronwall's inequality:
\begin{equation*}
 \|Y_T^{\omega+h}-Y_T^\omega\|_2
 \le \sqrt T\,\exp(LT)
 \bigl(\|h\|_{\mathbb H_T}
       +\sqrt{U_A(\omega+h)}+\sqrt{U_A(\omega)}\bigr) = O\bigl(\sqrt{\log(1/v)+\log n}\bigr),
\end{equation*}
where the {last estimate} uses the fact that both $\omega$ and $\omega+h$ are in $H_R$.
This is what replaces \cref{eq:Lipschitz-1}.

The second use of Lipschitzness is replaced by the cost at the endpoint $T$.
The symmetric KL identity part is unchanged:
\begin{equation*}
 D(\mu_{A,y}\Vert\mu_{A,z})+D(\mu_{A,z}\Vert\mu_{A,y})
 =\inner{y-z}{m(A,y)-m(A,z)}.
\end{equation*}
Taking $y=Y_T^\omega$ and $z=Y_T^{\omega+h}$ and applying part~\textup{(i)} of \Cref{lem:app-wpi-branches} at time $T$ gives
\begin{align}
\label{eq:overview-cost-kl}
 D(\mu_{A,y}\Vert\mu_{A,z})+D(\mu_{A,z}\Vert\mu_{A,y})
 \le{}&L\|y-z\|_2^2
 +\bigl(\sqrt{E_A(\omega)}+\sqrt{E_A(\omega+h)}\bigr)\|y-z\|_2.
\end{align}
Consequently, on the paths $\omega$ and $\omega+h$,
\begin{equation*}
 D(\mu_{A,y}\Vert\mu_{A,z})+D(\mu_{A,z}\Vert\mu_{A,y})
 =O\bigl(\log n+\log(1/v)\bigr).
\end{equation*}
This replaces \cref{eq:Lipschitz-2}, and it gets exponentiated in the idealized argument, which loses only by a factor $n^{O(1)}v^{-O(1)}$.
Combining with the bound on variances from $G(v)$ yields that, for constants $B,\Gamma>0$,
\begin{equation*}
 [Q_f(1-v)-Q_f(v)]^2
 \le Cn^Bv^{-1-\Gamma}
 \E[\omega]{\Var[\mu_{A,Y_T^\omega}]{f}}.
\end{equation*}
Integrating over $v\ge\varepsilon/2$ and using the trivial bound below $\varepsilon/2$ gives the same polynomial-in-$1/\varepsilon$ WPI as in the idealized case, now with additional polynomial factors in $n$.

This is the reason why we need a strong tail bound on $U_A$ and $E_A$, as the right hand side of \cref{eq:Lipschitz-2} will get exponentiated. {The existence of a locally Lipschitz $m^*$ alone is insufficient; we also need} to control its cumulative errors along the path and at the endpoint $T$.

\subsection{TAP stationary points are locally Lipschitz}
We construct the locally Lipschitz $m^*$ and the cost function using the approximate message passing (AMP) and {Thouless--Anderson--Palmer} (TAP) geometry developed in~\cite{ElAlaouiMontanariSellke25,Celentano24}.
Given {$Y_t$}, a fixed number of AMP iterations (see formal definition in \Cref{subsec:prelim-tap-amp}) outputs a rough approximation $\hat m$ of the posterior mean.
With high probability over $A$ and the observed field $Y_t$, a sufficiently late AMP iterate lies in a neighborhood on which the TAP free energy is uniformly strongly convex and contains a unique local stationary point, which we denote by $m^*(Y_t)$ as a function of the field $Y_t$ for a fixed $A$.
\cite{ElAlaouiMontanariSellke25,Celentano24} used this stationary point to approximate the localization scheme, giving rise to the \emph{algorithmic stochastic localization}.
What we need is the strong convexity part of the stationary point, as it implies that $z\mapsto m^*(z)$ is locally Lipschitz: Let $F$ be differentiable and $\kappa$-strongly convex on a convex set $K$, and $m^\star\in K$ be a stationary point {of} $F$. Then for any $u\in K$,
 it holds that $\kappa\norm{u-m^\star}_2^2\leq\langle\nabla F(u)-\nabla F(m^\star),u-m^\star\rangle\leq\norm{\nabla F(u)}_2\norm{u-m^\star}_2$, so
$ \norm{u-m^\star}_2\leq\kappa^{-1}\norm{\nabla F(u)}_2$.
{Choosing $F$ to be the TAP free energy and applying the inequality} at two stationary points from two different fields $y$ and $z$ gives $\norm{m^\star(y)-m^\star(z)}_2\leq\kappa^{-1}\norm{y-z}_2$, as long as both stationary points lie in a common strongly convex region. Therefore, the local Lipschitzness of $m^*$ essentially comes from the local strong convexity established in~\cite{ElAlaouiMontanariSellke25,Celentano24}.
The key missing piece is quantitative: 

\begin{center}
\emph{How well {does} {the} TAP stationary point $m^\star$ approximate the true posterior mean $m$?}
\end{center}

{The results of \cite{ElAlaouiMontanariSellke25,Celentano24} suggest} that $\E[]{\|m(A,Y_t)-m^*(Y_t)\|_2^2}=o(n)$,
while we need that $\|m(A,Y_t)-m^*(Y_t)\|_2^2 =O(\log n)$ with high probability.
{We also need the constants in this fixed-time probability bound to be uniform in the temperature $\beta$ and the localization time $t$.} Uniformity in the temperature is needed as the simulated annealing algorithm runs across different {temperatures}, and uniformity in the time $t$ is used to integrate along the stochastic localization path and control the cumulative error.

We give a high level sketch of the argument in~\cite{ElAlaouiMontanariSellke25,Celentano24}. They run AMP for a fixed number of iterations to locate the local stationary point $m^*$ of TAP free energy, and then use it to approximate the posterior mean $m$. Let $\widehat{m}$ be the output of AMP.
A good event $\mathcal{G}$ consists of three geometric properties: (1) the residual at $\widehat{m}$ is small, (2) there is a local convex basin for the TAP free energy around $\widehat{m}$, and (3) $\widehat{m}$ is close to the posterior mean $m$. They show that $\mathcal{G}$ holds with probability at least $1 - o_n(1)$.

Next, to show that this stationary point is close to the posterior mean conditioned on $\mathcal{G}$, they compare both with the AMP output $\widehat{m}$.
Let $F_{\mathrm{TAP}}$ be the TAP free energy {and let $\kappa$ be the local curvature lower bound. Their argument is}
\begin{equation*}
 \|m-m^*\|_2
 \leq \|m-\widehat m\|_2+\|\widehat m-m^*\|_2
 \leq \|m-\widehat m\|_2+\kappa^{-1}\|\nabla F_{\mathrm{TAP}}(\widehat m)\|_2.
\end{equation*}
The first term is the AMP approximation error; the second is controlled by the small TAP gradient $\|\nabla F_{\mathrm{TAP}}(\widehat m)\|_2$ at the AMP iterate.
For any fixed $\rho>0$, they show that  there is a sufficiently large fixed AMP depth such that $\|m-m^*\|_2^2\leq\rho n$, with the required geometric properties $\mathcal{G}$ holding with probability $1-o_n(1)$.

One may be tempted to run both AMP and the natural gradient descent longer to reduce the error, but existing AMP error analysis {relies} on state evolution with the number of iterations fixed as $n\to\infty$. {This is because the Lipschitz bound of AMP may get worse} as one increases the number of iterations.

Here we take a different approach: we bound the TAP gradient (also known as TAP residual) directly at the posterior mean.
Intuitively, $\nabla F_{\mathrm{TAP}}(\widehat m)$ measures how well AMP solves the TAP optimization problem, while we are more concerned about how well $m$ approximates the TAP stationary point $m^*$ which is measured by $\nabla F_{\mathrm{TAP}}(m)$.
Once $m$ and $m^*$ lie in a common $\kappa$-strongly convex region, we upper bound the $\|m-m^*\|_2^2$ by the TAP gradient at $m$ directly,
\begin{equation*}
 \|m-m^*\|_2^2\leq\kappa^{-2}\|\nabla F_{\mathrm{TAP}}(m)\|_2^2.
\end{equation*}

We only use AMP to locate the locally strongly convex region, and the residual bound at the posterior mean $m$ implies our $O(\log n)$ approximation error. The next subsections will explain the two components we need: the residual bound at the posterior mean and the inverse polynomial probability bound for the geometric conditions $\mathcal{G}$. An important note is that their constants are independent of time so that we can integrate the resulting cost estimates as above.

We close this subsection by explaining how controlling $\|m-m^*\|_2$ yields \Cref{lem:app-wpi-branches}.
Part~\textup{(i)} comes from a triangle inequality: by inserting $m^*(y_1)$ and $m^*(y_2)$ between the two posterior means, we have that 
\[
\|m(A,y_1)-m(A, y_2)\|_2 \le \norm{m(A, y_1) - m^*(y_1)}_2 + \norm{m^* (y_1) - m^*(y_2)}_2 + \norm{m^* (y_2) - m(A, y_2)}_2.
\]
{When both fields satisfy the geometric conditions, nearby AMP outputs place the stationary points in a common strongly convex region.
If instead the AMP outputs are far apart, the Lipschitz property of fixed-depth AMP forces the fields to be separated by order $\sqrt n$, so the trivial bound $\|m^*(y_1)-m^*(y_2)\|_2\le 2\sqrt n$ suffices.
Thus, for a constant $L$, we obtain}
\[
\|m(A,y_1)-m(A, y_2)\|_2 \le L \norm{y_1 - y_2}_2  + \norm{m(A, y_1) - m^*(y_1)}_2 + \norm{m^* (y_2) - m(A, y_2)}_2.
\]
We therefore construct $\overline{\mathfrak e}_{n,\beta}(A,t;y)$ from the squared approximation error $\|m(A,y)-m^*(y)\|_2^2$ at the same field $y$ when the geometric conditions $\mathcal{G}$ at time $t$ hold, and use the trivial bound $4n$ otherwise.

For part~\textup{(ii)} of \Cref{lem:app-wpi-branches}, we first bound the cost at each fixed time.
After the averaging in the cost construction and the transfer from the planted law to the original augmented law by the contiguity (see the proof of \Cref{lem:scalar-reference-path-error} for details), the two guarantees imply that, for every $H>0$, there are constants $C_H,n_H$ such that, for $n\ge n_H$,
\begin{equation*}
 \sup_{\beta\in[\beta_0,\beta_1]}\sup_{t\in[0,T]}
 \mathbb P_{0,n,\beta}^{\mathrm{aug}}
 \left\{\overline{\mathfrak e}_{n,\beta}(A,t;Y_t)>C_H\log n\right\}
 \le C_Hn^{-H}.
\end{equation*}
To control the time integral, we separate the times when the cost is at most $C_H\log n$ from the remaining times.
The former contribute at most $TC_H\log n$.
For the remaining contribution, we use the trivial bound $4n$ and apply Fubini-Tonelli's theorem to obtain
\begin{equation*}
 \E[\mathbb P_{0,n,\beta}^{\mathrm{aug}}]{
 \int_0^T\overline{\mathfrak e}_{n,\beta}(A,t;Y_t)
 \mathbf 1_{\{\overline{\mathfrak e}_{n,\beta}(A,t;Y_t)>C_H\log n\}}\,\mathrm d t}
 \le 4TC_Hn^{1-H}.
\end{equation*}
By Markov's inequality, this integral is at most $\log n$ except with probability at most $4TC_Hn^{1-H}/\log n$.
On the {event that this integral is at most $\log n$}, the cumulative path error is at most $(TC_H+1)\log n$.
Taking $H>R+2$ proves the path-error bound for all sufficiently large $n$.
For the endpoint error $\overline{\mathfrak e}_{n,\beta}
       (A,T;Y_T)$, we apply the same fixed-time estimate at $t=T$.
Neither step requires the approximation bound to hold simultaneously at every time.

The formal construction and proof of \Cref{lem:app-wpi-branches} are in \Cref{sec:app-inverse-poly-geometry-proof}.
We  explain the TAP gradient (TAP residual) bound next, and then the probability boosting for the geometric properties.

\subsection{From post-AMP geometry to the TAP residual bound}
\label{subsec:technical-overview-wpi}

\begingroup

{
As explained earlier,  the crucial missing piece is to control $\|m-m^*\|_2$.
By running AMP, we can already locate the TAP stationary point $m^*$ in a locally strongly convex basin. Once $m$ and $m^*$ belong to a common region of $\kappa$-strong convexity, we use strong monotonicity of the gradient, leading to $\|m-m^*\|_2\leq\kappa^{-1}\|\nabla F_{\mathrm{TAP}}(m)\|_2$.
Thus we focus on analyzing the TAP residual, $\nabla F_{\mathrm{TAP}}(m)$, using the exact Bayesian structure of the posterior mean to prove that it nearly satisfies the same TAP equation.
\par}

{
\subsubsection{The TAP residual}
We prove the residual bound under the planted model.
By contiguity, one can transfer properties back to the standard SK model.
We first define the planted model.
Let $X\sim\operatorname{Unif}(\{\pm1\}^n)$ and let $W$ be a GOE matrix, with $W_{ij}=W_{ji}\sim N(0,1/n)$ for $i<j$ and $W_{ii}\sim N(0,2/n)$.
We take $X$, $W$, and a standard Brownian motion $(B_t)_{t\geq0}$ in $\mathbb R^n$, with $B_0=0$, to be mutually independent.
We set
\begin{equation*}
 A^{\mathrm{full}}=W+\frac{\beta}{n}XX^\top,
 \qquad
 Y_t=tX+B_t,
\end{equation*}
and let $A$ be the matrix obtained from  $A^{\mathrm{full}}$ by replacing the diagonal entries with zeros, and write $d$ for the diagonal vector of $A^{\mathrm{full}}$.
And we use $\mathbb{P}_{n,\beta}^{\mathrm{pl,aug}}$ to denote the joint distribution of $(X,A,d,Y)$.
We write $m=m(A,Y_t)$ for the exact posterior mean.
For a scalar overlap parameter $q$, we define the TAP residual by
\begin{equation*}
 R_{\mathrm{TAP}}(m;A^{\mathrm{full}},Y_t,q)
 =\operatorname{arctanh}(m)-Y_t-\beta A^{\mathrm{full}}m+\beta^2(1-q)m.
\end{equation*}
We remark that $\nabla F_{\mathrm{TAP}} = R_{\mathrm{TAP}}$, and the deterministic $q_\beta(t)$ is the scalar-channel fixed point from \Cref{subsec:prelim-tap-amp}.
Our target is the following:
\par}

\begin{theorem}[TAP residual bound]
\label{thm:polynomial-tap-residual}
Fix $0<\beta_*<1$ and $T>0$.
For every $H>0$, there are $C_H,n_H>0$, depending only on $(\beta_*,T,H)$, such that, for $n\ge n_H$,
\begin{equation*}
 \sup_{\beta\in[0,\beta_*]}\sup_{t\in[0,T]}
 \mathbb P_{n,\beta}^{\mathrm{pl,aug}}
 \left\{
 \left\|R_{\mathrm{TAP}}\bigl(m(A,Y_t);A^{\mathrm{full}},Y_t,q_\beta(t)\bigr)\right\|_2^2
 >C_H\log n
 \right\}
 \le n^{-H}.
\end{equation*}
\end{theorem}

The bound holds for each fixed $(\beta,t)$ with uniform constants.
The $O(\log n)$ threshold gives the logarithmic comparison cost needed for the WPI.

Talagrand's classical TAP estimates \cite{Talagrand11}, combined with his overlap concentration bound \cite[Theorem~13.7.1]{Talagrand11VolumeII}, also show that the exact magnetization approximately satisfies the TAP equations, by bounding fixed-order coordinate moments in a different coordinate in the original SK model.
To state his estimate in our notation, we take $A\sim\mathbb P_{0,n}$, a fixed homogeneous field $y=h\mathbf1$, and $m=m(A,y)$ in the original SK model.
Here $q=\E[Z\sim N(0,1)]{\tanh^2(\beta\sqrt q\,Z+h)}$ is the scalar-channel fixed point for this homogeneous field.
These estimates give the following moment bound for the \emph{bounded TAP defect}: for every $i\in[n]$ and fixed integer $k\geq1$,
\[
 \E[A\sim\mathbb P_{0,n}]{\left|m_i-\tanh\bigl(y_i+\beta(Am)_i-\beta^2(1-q)m_i\bigr)\right|^{2k}}\leq C_kn^{-k}
 \qquad(\beta<1).
\]
Subsequent dynamical cavity arguments, based on Brownian interpolation and It\^o calculus, also establish bounded TAP defect for TAP equations with an empirical-overlap correction \cite[Corollary~1.4]{AdhikariBrenneckeVonSoostenYau2021}.

However, our theorem requires control of the TAP gradient residual in the planted model, before applying $\tanh$.
Moreover, since $(\operatorname{arctanh})'(u)=(1-u^2)^{-1}$ is unbounded near $\pm1$, the bounded-defect estimate does not directly yield this gradient-residual bound.

Even for the gradient residual, $L^p$ bounds of the form $C_p/\sqrt n$ give only inverse powers of $\log n$ via Markov's inequality at threshold $C_H\log n$.
A direct moment argument for failure probability $n^{-H}$ would require quantitative control of how $C_p$ grows with $p$.

To obtain this tail bound, we give a self-contained cavity analysis directly on the TAP residual.
We now return to the planted model and fix $0\le\beta\le\beta_*<1$ and $t\in[0,T]$, writing $y=Y_t$, $m=m(A,y)$, and $\ell=\beta^2/n$.
We first analyze $\mathbf r=R_{\mathrm{TAP}}(m;A^{\mathrm{full}},y,Q_n(m))$, with the \emph{empirical overlap} $Q_n(m)=\|m\|_2^2/n$ in the Onsager coefficient.
The \emph{prescribed residual} $R_{\mathrm{TAP}}(m;A^{\mathrm{full}},y,q_\beta(t))$ differs from $\mathbf r$ by $\beta^2(Q_n(m)-q_\beta(t))m$.
The overlap exponential moment in \Cref{lem:tap-overlap-mgf} controls this correction and transfers the logarithmic tail bound for $\mathbf r$ to the prescribed residual.

We first control one coordinate $r_i$ of the residual $\*r$.
Let $\mathcal F^{\setminus i}$ be the sigma-field generated by $(A^{\mathrm{full}}_{-i,-i},y_{-i})$, the observations after deleting spin $i$.
Informally, we first fix this randomness and then average over the remaining randomness.
We bound the conditional moment $\E{|r_i|^p\mid\mathcal F^{\setminus i}}$ using the strengthened overlap concentration (\Cref{lem:tap-complete-row,cor:tap-covariance}).
For every fixed $p\ge2$, averaging the conditional estimate and using $\ell=\beta^2/n$, we obtain
\[
\|r_i\|_{L^p}\le C_pn^{-1/2}.
\]

Here expectations and $L^p$ norms refer to $\mathbb P_{n,\beta}^{\mathrm{pl,aug}}$, $C_p$ depends only on $\beta_*$ and $p$.
Taking $p=2$ and summing over the coordinates, we obtain $\E{\norm{\*r}_{2}^2}=O(1)$.
However, it is insufficient to give an inverse-polynomial tail bound for $\norm{\*r}_{2}^2$.
Now, we give an overview of how we upgrade the above coordinate estimates to an $O(\log n)$ bound for $\norm{\*r}_2^2$ with failure probability $n^{-H}$ for every fixed $H>0$.

We also use the conditional estimate at order two and set $e_i\defeq\E{|r_i|^2\mid\mathcal F^{\setminus i}}$.
Taking $L^p$ norms and applying \Cref{cor:tap-covariance}-\textup{(ii)} to the deleted system, for every $p\ge1$ we obtain
\begin{equation*}
 \|e_i\|_{L^p}\le\frac{C\cdot p}{n}.
\end{equation*}

We set $\mathcal E\defeq\sum_{i=1}^ne_i$.
By Minkowski's inequality, we have that $\|\mathcal E\|_{L^p}\le C\cdot p$.
Since $\mathcal E\ge0$, a standard argument of Taylor series give the next exponent moment bound of $\+E$:
\begin{equation*}
 \E{\exp(\kappa\mathcal E)}
 =1+\sum_{p=1}^{\infty}\frac{\kappa^p\E{\mathcal E^p}}{p!}
 \le1+\sum_{p=1}^{\infty}\frac{(\kappa C\cdot p)^p}{p!}
 \le1+\sum_{p=1}^{\infty}(\kappa C\mathrm e)^p
 \le2,
\end{equation*}
where in the second inequality, we use $p!\ge(p/\mathrm e)^p$ and choose $0<\kappa\le(2C\mathrm e)^{-1}$.

Indeed, each $e_i$ averages the deleted-row fluctuation in $r_i^2$, and the tower property gives $\E{\mathcal E}=\E{\|\mathbf r\|_2^2}$, but this equality only concerns first moments.
We first consider an ideal case.
If the event $\{\|\mathbf r\|_2\ge s\}$ were $\mathcal F^{\setminus i}$-measurable for every $i$, then we can prove that $\E{\|\mathbf r\|_2^2\1_{\{\|\mathbf r\|_2\ge s\}}}$ is equal to $\E{\mathcal E\1_{\{\|\mathbf r\|_2\ge s\}}}$.
We could then bound the latter expectation using the exponential moment of $\mathcal E$.
Then we can lower bound $\|\mathbf r\|_2^2$ by $s^2$ and take $s^2=C_H\log n$ to obtain the desired polynomial tail bound.

However, the event $\{\|\mathbf r\|_2\ge s\}$ is not generally $\mathcal F^{\setminus i}$-measurable, so we instead approximate it using the residual of a model removing the spin $i$.
For each $i$, let $\mathbf r^{\setminus i}$ be the empirical residual in this deleted system.
We show that for every fixed integer $p\ge1$,
\begin{equation}
\label{eq:overview-stability}
 \max_i\left\|\,\|\mathbf r\|_2-\|\mathbf r^{\setminus i}\|_2\,\right\|_{L^{4p}}
 \le C_pn^{-1/4},
\end{equation}
and later we provide the intuition of this stability.
We first use this stability to give a tail bound of $\norm{\*r}_2^2$.

Set $\eta=n^{-1/8}$.
Let $E_i^{\mathrm{del}}=\{\bigl|\,\|\mathbf r\|_2-\|\mathbf r^{\setminus i}\|_2\,\bigr|>\eta\}$ be the bad event that the stability fails.
Also note that the above stability gives a tail bound of the above bad event by Markov's inequality.
Outside the bad event $E_i^{\mathrm{del}}$, $\{\|\mathbf r\|_2\ge s\}$ implies $\{\|\mathbf r^{\setminus i}\|_2\ge s-\eta\}$, and the latter event is $\mathcal F^{\setminus i}$-measurable.
We may therefore condition on $\mathcal F^{\setminus i}$ to replace $r_i^2$ by $e_i$ on this event; summing over $i$ gives that
\begin{equation*}
 s^2\Pr{\|\mathbf r\|_2\ge s}
 \le \E{\mathcal E\1_{\{\|\mathbf r\|_2\ge s-2\eta\}}}+\epsilon_n.
\end{equation*}
Here $\epsilon_n\defeq\sum_i\E{(r_i^2+e_i)\1_{E_i^{\mathrm{del}}}}$.
The error $\epsilon_n$ involves both the probabilities of deletion failures and the sizes of $r_i^2+e_i$ on those events.
The moment bounds above control these weights; combining them with the deletion-stability estimate \cref{eq:overview-stability} at a sufficiently large fixed moment order gives $\epsilon_n\le n^{-H-3}$ for sufficiently large $n$.
We now cut off $\mathcal E$ at $a=C_H\log n$.
On $\{\mathcal E\le a\}$ the expectation above is at most $a\Pr{\|\mathbf r\|_2\ge s-2\eta}$, while the exponential moment of $\mathcal E$ makes the contribution from $\{\mathcal E>a\}$ at most $n^{-H-3}$.
We therefore obtain and iterate the recursion
\begin{equation*}
 s^2\Pr{\|\mathbf r\|_2\ge s}
 \le a\Pr{\|\mathbf r\|_2\ge s-2\eta}+O(n^{-H-3}).
\end{equation*}
For $s^2\ge2a$, each iteration reduces the preceding tail probability by a factor at most $1/2$, while $O_H(\log n)$ iterations increase the threshold by only $O_H(n^{-1/8}\log n)=o(1)$.
It follows that
\begin{equation*}
 \Pr{\|\mathbf r\|_2^2>C_H\log n}\le n^{-H}.
\end{equation*}
Recall that by the overlap concentration, we can replace the empirical overlap by the prescribed overlap and give a similar tail bound of the prescribed residual.
\endgroup

{\subsubsection{Stability under deletion of one coordinate}

It remains to justify \cref{eq:overview-stability}.
The intuition is that we decompose the residual $\*r$ by pinning the spin $i$ as $+1$ and $-1$ respectively.
Let $\*r^{b}_{\mathrm{pin}}$ be the residual condition the spin of $i$ is $b\in \set{\pm 1}$.
Then $\*r$ can be expressed as a linear combination of $\*r^{+1}_{\mathrm{pin}}$ and $\*r^{-1}_{\mathrm{pin}}$ together with some small remnants.
It suffices to upper bound the difference between $\*r^{+1}_{\mathrm{pin}}$ and $\*r^{\setminus i}$ and $\*r^{-1}_{\mathrm{pin}}$ can be handled similarly.

Once spin $i$ is pinned, its interaction with the other spins can be regarded as an extra external field.
We can therefore remove spin $i$ and work with the resulting $(n-1)$-spin model.
Hence, one can imagine an interpolation process starting with the external field $y_t$, then we gradually add new external fields, and the terminal distribution is equivalent to the model with the pinning on $i$.
And it suffices to control the gradient of the residual in order to upper bound the difference of two endpoint residuals.

We formalize this intuition by designing an SDE and show that this random process is identical to the algorithmic stochastic localization process in~\Cref{lem:tap-residual-sde}.
It suffices to upper bound the $L^p$ norm of $\partial \*r / \partial y$ along the algorithmic stochastic localization process. 
We first consider the full $n$-spin system.
We write $m=m(A,y)$ and let $\Sigma$ be the covariance matrix of $\mu_{A,y}$.
Keeping $A^{\mathrm{full}}$ and $\ell$ fixed, the empirical TAP Hessian is
\begin{equation*}
 H_{\mathrm{TAP}}^{(Q)}=\operatorname{diag}((1-m_j^2)^{-1})-\beta A^{\mathrm{full}}+\ell(n-\|m\|_2^2)I-2\ell mm^\top.
\end{equation*}
Since $\partial m/\partial y=\Sigma$, differentiating the empirical TAP residual gives
$ \frac{\partial\mathbf r}{\partial y}=H_{\mathrm{TAP}}^{(Q)}\Sigma-I$.
The same formula applies to the deleted system, using its own mean and covariance, replacing $A^{\mathrm{full}}$ by $A^{\mathrm{full}}_{-i,-i}$ and $n$ by $n-1$, while keeping $\beta$ and $\ell=\beta^2/n$ unchanged.

\cite[Theorem~4.10]{DaviesLeeSandhuShi26} show that $\E{\|H_{\mathrm{TAP}}^{(Q)}\Sigma-I\|_{\mathrm F}^2}=O(1)$ when $\beta<1/2$.
Their sharp second-moment estimate was then used to approximate covariance in a sampling process\footnote{In their subsequent work, which is concurrent to ours, Potential Hessian Ascent IV \cite[Theorem~B.3]{LeeSandhuShiPHAIV26} extends this dimension-independent second-moment bound to every fixed $0<\beta<1$.}.
In our case, we need to control higher moments for better tail bounds.
For our concentration bounds, however, a looser upper bound on $\|\partial\*r/\partial y\|_{L^p}$ suffices.
By Gaussian smoothing, we show that for any $p$, we have $\left\|\frac{\partial\mathbf r}{\partial y}\right\|_{L^p}
 \leq C_p n^{1/4}$ in \Cref{lem:tap-susceptibility}.
{
Restoring one deleted row corresponds to a Gaussian observation of duration $\ell=\beta^2/n$.
Thus the Burkholder--Davis--Gundy inequality bounds the martingale increment in $L^p$ by
$ C_p\sqrt\ell\,n^{1/4}=O(n^{-1/4})$.
The integrated drift and the errors from recombining the two pinned residuals are $O(n^{-1/2})$ in $L^p$.
Together these estimates prove \cref{eq:overview-stability}.
}

}

\subsubsection{Probability boosting for geometric properties}
Recall that we need to boost success probabilities for the following geometric properties: (1) the residual at $\widehat{m}$ is small, (2) there is a local convex basin around $\widehat{m}$, and (3) $\widehat{m}$ is close to the posterior mean $m$.
\cite{ElAlaouiMontanariSellke25,Celentano24} showed that they hold with probability $1-o_n(1)$.

{For the first two properties, we boost their probability through stability and Gaussian concentration.}
We first extract the Gaussian randomness part.
We fix $\beta$ and $t$ and condition on $X=x$.
Recall that the planted model is given by $A^{\mathrm{full}}=W+\beta xx^\top/n$ and $Y_t=tx+\sqrt t\,G$.
Hence the randomness here contains $W$ which is GOE, and $G\sim N(0,I_n)$ which is independent of $W$.
Let $\xi$ collect the independent standard Gaussian variables generating these data.
\begingroup

\begin{lemma}[{\cite[Theorem~3.1]{Borell75}}]
\label{lem:gaussian-enlargement-boosting}
Let $N\geq1$ be an integer, let $L>0$, and let $\xi\sim N(0,I_N)$.
Let $F:\R^N\to\R$ be an $L$-Lipschitz function, and let $q$ be a median of $F(\xi)$.
Then, for every $u\geq0$,
\begin{equation*}
 \Pr{|F(\xi)-q|>u}\leq2\Phi(-u/L),
\end{equation*}
where $\Phi$ is the standard Gaussian distribution function.
Consequently, for every integer $p\geq1$, there is a constant $C_p>0$ depending only on $p$ such that
\begin{equation*}
 \E{|F(\xi)-\E{F(\xi)}|^{2p}}\leq C_pL^{2p}.
\end{equation*}
\end{lemma}

In our application, let $\mathcal G_x$ consist of all Gaussian parameters for which the first two geometric properties hold, and we take the $1$-Lipschitz function $F_x(z)=\operatorname{dist}(z,\mathcal G_x)$.
For all sufficiently large $n$, it can be verified through standard results in~\cite{ElAlaouiMontanariSellke25,Celentano24} and symmetry, that $\Pr{\xi\in\mathcal G_x}\ge1/2$ for every $x$, so zero is a median of $F_x(\xi)$.
The stability of the residual at $\widehat m$ and the curvature around $\widehat m$ shows that the two properties with slightly weaker margins hold whenever $F_x(\xi)\leq\rho\sqrt n$ for some constant $\rho>0$ independent of $x$.
Then, for each $x$, applying \Cref{lem:gaussian-enlargement-boosting} to $F_x$ yields an exponentially small failure probability.
\endgroup

\begingroup

For the third property that $\norm{m - \widehat{m}}_2$ is small, we boost the probability through overlap concentration and Gaussian concentration.
{We let $m$ be the posterior mean and $\widehat m$ be the output of AMP, and set $D_n=\|m-\widehat m\|_2^2/n$.}
We have that
\begin{equation*}
 \frac{\|m-\widehat m\|_2^2}{n}
 =\frac{\|m\|_2^2}{n}+\frac{\|\widehat m\|_2^2}{n}
 -\frac2n\E{\inner{X}{\widehat m}\mid A^{\mathrm{full}},Y_t}.
\end{equation*}
For the first term, we use overlap concentration. 
We remark that the Lipschitz property of $m$ is hard to establish. 
Also recall that by~\Cref{subsection:idealized}, if $m$ is Lipschitz, then it directly leads to a weak Poincar\'e inequality.
\begingroup

For the other two terms, we again use the Lipschitz property of $\widehat{m}$. 
{We condition on $X$ and use stability, symmetry, and \Cref{lem:gaussian-enlargement-boosting} to control these scalar terms; after averaging over $X$ and applying Jensen's inequality, we obtain $\E{|D_n-\E{D_n}|^p}\leq C_pn^{-p/2}$ for every fixed $p\geq1$.}
\endgroup

It remains to bound the expectation of $D_n$.
{By a standard result of AMP, we can choose a fixed depth $K$ so that $\E{D_n}\leq r^2/2$ for all sufficiently large $n$, where $r\sqrt n$ is the radius of the convex basin (see \Cref{prop:fixed-time-amp-tap-input}-\textup{(ii)}).}
We then apply Markov's inequality to obtain
\begin{equation*}
 \Pr{D_n>r^2}
 \le \left(\frac{2}{r^2}\right)^p\E{|D_n-\E{D_n}|^p}
 \le C_{p,r}n^{-p/2}.
\end{equation*}
We formalize this argument in~\Cref{lem:inverse-poly-positive-posterior}.

\subsection{From WPI to simulated annealing}
Given WPI, it is standard to use simulated annealing or a warm start to perform approximate sampling~\cite{HuangMohantyRajaramanWu25}.
We remark that our WPI is strong enough to perform simulated annealing.
In contrast, the WPI of Davies et al.\ for $\beta<1/2$ is not strong enough to give a polynomial-time guarantee for simulated annealing, so they instead construct an approximate warm start for the original Gibbs distribution~\cite{DaviesLeeSandhuShiWPI26}.
They first use Langevin dynamics with a TAP-based drift to generate an auxiliary field $y$, then run Glauber dynamics for the localized distribution $\mu_{A,y}$ to obtain this warm start.

\endgroup

\section{Preliminaries}
\label{sec:preliminaries}

Throughout this paper, we use $\log$ for the logarithm with base $\mathrm e$ and indicate other bases explicitly.

\subsection{The planted model}
\label{subsec:prelim-planted-model}

For inverse temperature $0\leq\beta<1$, a symmetric matrix $A\in\R^{n\times n}$ with zeros on the diagonal, and a field $y\in\R^n$, we define the Gibbs distribution $\mu_{A,y}$ and its partition function $Z_A(y)$ by
\begin{equation*}
 Z_A(y)=\sum_{x\in\{\pm1\}^n}
 \exp(\beta\inner{x}{Ax}/2+\inner{y}{x}),\qquad
 \mu_{A,y}(x)=Z_A(y)^{-1}
 \exp(\beta\inner{x}{Ax}/2+\inner{y}{x}).
\end{equation*}
We suppress $\beta$ in this notation, write $\mu_A=\mu_{A,0}$, and define the mean and covariance of $\mu_{A,y}$ by
\begin{equation*}
 m(A,y)=\E[\sigma\sim\mu_{A,y}]{\sigma},\qquad
 \Sigma(A,y)=\E[\sigma\sim\mu_{A,y}]{(\sigma-m(A,y))(\sigma-m(A,y))^\top}.
\end{equation*}

Next we discuss the distribution of the interaction matrix $A$ and the external field $y$ in our models.
In the two definitions below, $(B_t)_{t\geq0}$ denotes a standard Brownian motion in $\R^n$ with $B_0=0$, drawn independently of $(X,A,d)$ in each construction.

\begin{definition}[Original model]
\label{def:original-matrix-law}
For $n\geq1$ and $0\leq\beta<1$, let $\mathbb P_{0,n}$ be the law of a symmetric zero-diagonal matrix $A\in\R^{n\times n}$ with independent entries $A_{ij}\sim N(0,1/n)$ for $i<j$.
Let $\mathbb P_{n,\beta}^{\mathrm{diag}}$ be the product distribution on $\R^n$ with independent coordinates $N(\beta/n,2/n)$.
We draw $A\sim\mathbb P_{0,n}$, $X\sim\mu_A$
and draw $d\sim\mathbb P_{n,\beta}^{\mathrm{diag}}$ independently of $(A,X)$ and set $Y_t=tX+B_t$ for $t\geq0$.

We denote the joint law of $(X,A,d,Y)$ by $\mathbb P_{0,n,\beta}^{\mathrm{aug}}$, on the product space of spins $X\in\{\pm1\}^n$, zero-diagonal symmetric matrices $A\in\R^{n\times n}$, diagonals $d\in\R^n$, and paths $Y\in C([0,\infty);\R^n)$.
Its $A$-marginal is $\mathbb P_{0,n}$.
\end{definition}

\begin{definition}[Planted model]
\label{def:tap-planted-model}
For $n\geq1$ and $0\leq\beta<1$, we draw $X\sim\operatorname{Unif}(\{\pm1\}^n)$ and a GOE matrix $W$, independently.
The entries $W_{ij}=W_{ji}\sim N(0,1/n)$ for $i<j$ and $W_{ii}\sim N(0,2/n)$ for $i\in[n]$ are independent over $i\leq j$.
Set $A^{\mathrm{full}}=W+\beta XX^\top/n$, let $A$ be its zero-diagonal part, write $d$ for its diagonal, and set $Y_t=tX+B_t$ for $t\geq0$.

We denote the joint law of $(X,A,d,Y)$ on the same product space as in \Cref{def:original-matrix-law} by $\mathbb P_{n,\beta}^{\mathrm{pl,aug}}$.
We denote its $A$-marginal by $\mathbb P_{n,\beta}^{\mathrm{pl}}$.
\end{definition}

For every $x\in\{\pm1\}^n$, the diagonal contributes the same term $\sum_i d_i/2$ to $\inner{x}{A^{\mathrm{full}}x}/2$, so it does not change the Gibbs distribution.
Moreover, in both models, the diagonal vector $d$ is independent of $(X,A,Y)$ and has i.i.d.\ $N(\beta/n,2/n)$ coordinates; we denote its distribution by $\mathbb P_{n,\beta}^{\mathrm{diag}}$.

By Bayes' rule in the planted model and by construction in the original model, we have $\operatorname{Law}(X\mid A)=\mu_A$ under both augmented laws.
Given $(X,A,d)$, both models use the same Brownian observation process $Y_t=tX+B_t$, $t\geq0$.
Thus the two augmented laws have the same conditional distribution of $(X,d,Y)$ given $A$, and their likelihood ratio depends only on $A$:
\begin{equation}
  \frac{\mathrm d\mathbb P_{n,\beta}^{\mathrm{pl,aug}}}
       {\mathrm d\mathbb P_{0,n,\beta}^{\mathrm{aug}}}(X,A,d,Y)
  =\frac{\mathrm d\mathbb P_{n,\beta}^{\mathrm{pl}}}
       {\mathrm d\mathbb P_{0,n}}(A)=:L_n(A).
  \label{eq:planting-density-ratio}
\end{equation}
The following lemma gives quantitative probability comparisons between the two augmented laws.

\begin{lemma}
\label{lem:planting-likelihood}
For $0\leq\beta<1$, we have
\begin{equation*}
  \sup_n\E[A\sim\mathbb P_{0,n}]{L_n\tp{A}^2}
  \leq\tp{1-\beta^2}^{-1/2},
  \qquad
  \sup_n\E[A\sim\mathbb P_{0,n}]{L_n\tp{A}^{-1}}
  \leq\tp{1-\beta^2}^{-1/2}.
\end{equation*}
For every $n\geq1$ and every event $E$ determined by $(X,A,d,Y)$, we have
\begin{equation*}
  \mathbb P_{n,\beta}^{\mathrm{pl,aug}}\tp{E}
  \leq(1-\beta^2)^{-1/4}\mathbb P_{0,n,\beta}^{\mathrm{aug}}\tp{E}^{1/2},
\qquad
  \mathbb P_{0,n,\beta}^{\mathrm{aug}}\tp{E}
  \leq(1-\beta^2)^{-1/4}\mathbb P_{n,\beta}^{\mathrm{pl,aug}}\tp{E}^{1/2}.
\end{equation*}
\end{lemma}

We prove \Cref{lem:planting-likelihood} in \Cref{sec:planted-contiguity}.

\subsection{Wiener space and Gaussian estimates}
\label{subsec:prelim-wiener-space}

We collect the path comparison and Gaussian estimates used in the localization argument and the TAP estimates.
Fix $T>0$ and let $C_0([0,T];\R^n)$ be the space of continuous paths $\omega$ with $\omega_0=0$, equipped with the uniform norm $\|\omega\|_\infty=\sup_{0\leq t\leq T}\|\omega_t\|_2$.
Wiener measure $\gamma$ is the law on this space of a standard Brownian motion in $\R^n$ starting at zero.
The Cameron--Martin space $\mathbb H_T$ consists of absolutely continuous paths $h$ with $h_0=0$ and square-integrable derivative, with norm
\begin{equation*}
 \|h\|_{\mathbb H_T}^2
 =\int_0^T\left\|\frac{\mathrm d h_t}{\mathrm d t}\right\|_2^2\,\mathrm d t.
\end{equation*}
The Cameron--Martin distance between $\omega$ and $\omega'$ is $\|\omega'-\omega\|_{\mathbb H_T}$ if $\omega'-\omega\in\mathbb H_T$, and is $+\infty$ otherwise.
By Cauchy--Schwarz, every $h\in\mathbb H_T$ satisfies
\begin{equation}
 \|h\|_\infty\leq\sqrt T\|h\|_{\mathbb H_T}.
 \label{eq:cameron-martin-uniform-bound}
\end{equation}

We use Gronwall's inequality to transfer a bound on the driving paths to a bound on solutions of an integral equation.
\begin{lemma}
\label{lem:gronwall-path-stability}
Let $a,L\geq0$ and let $d:[0,T]\to[0,\infty)$ be continuous.
If $d(t)\leq a+L\int_0^t d(s)\,\mathrm d s$ for every $t\in[0,T]$, then $d(t)\leq a\exp(Lt)$.
Consequently, for a globally $L$-Lipschitz function $b:\R^n\to\R^n$, the solutions of
\begin{equation*}
 Y_t^\omega=\omega_t+\int_0^t b(Y_s^\omega)\,\mathrm d s
\end{equation*}
satisfy, for all $\omega,\omega'\in C_0([0,T];\R^n)$ and $h\in\mathbb H_T$,
\begin{equation*}
 \|Y^{\omega'}-Y^\omega\|_\infty
 \leq\exp(LT)\|\omega'-\omega\|_\infty,
 \qquad
 \|Y^{\omega+h}-Y^\omega\|_\infty
 \leq\sqrt T\exp(LT)\|h\|_{\mathbb H_T}.
\end{equation*}
\end{lemma}

\begin{proof}
Set $g(t)=a+L\int_0^t d(s)\,\mathrm d s$.
Since $g'(t)=Ld(t)\leq Lg(t)$ and $g(0)=a$, integration gives $d(t)\leq g(t)\leq a\exp(Lt)$.
For the path comparison, subtract the two integral equations and apply this bound to $d(t)=\sup_{s\leq t}\|Y_s^{\omega'}-Y_s^\omega\|_2$ with $a=\|\omega'-\omega\|_\infty$.
The second estimate follows from \cref{eq:cameron-martin-uniform-bound}.
\end{proof}

For a real random variable $Z$, we define its quantile function by $Q_Z(v)=\inf\{a\in\R:\Pr{Z\leq a}\geq v\}$ for $0<v<1$.
The events $\{Z\leq Q_Z(v)\}$ and $\{Z\geq Q_Z(1-v)\}$ each have probability at least $v$.

The following consequence of Gaussian isoperimetry bounds the distance between two sets of positive Wiener measure.
\begin{lemma}[{\cite[Theorem~3.1]{Borell75}}]
\label{lem:app-wpi-two-set-gaussian}
Let $0<v\leq1/2$ and let $A_-,A_+$ be Borel subsets of $C_0([0,T];\R^n)$ with $\gamma(A_-),\gamma(A_+)>v/4$.
There are $\omega\in A_-$ and $h\in\mathbb H_T$ such that $\omega+h\in A_+$ and $\|h\|_{\mathbb H_T}\leq2\sqrt{2\log(4/v)}$.
In particular, for any Borel function $F:C_0([0,T];\R^n)\to\R$, writing $Q_F$ for its quantile function under $\gamma$, there are $\omega$ and $h\in\mathbb H_T$ such that
\begin{equation*}
 F(\omega)\leq Q_F(v),\qquad
 F(\omega+h)\geq Q_F(1-v),\qquad
 \|h\|_{\mathbb H_T}\leq2\sqrt{2\log(4/v)}.
\end{equation*}
\end{lemma}

\begin{proof}
Let $\Phi$ be the standard Gaussian distribution function and set $r=\sqrt{2\log(4/v)}$, so that $\Phi(-r)<v/4$.
By \cite[Theorem~3.1]{Borell75}, the set $\{\omega+h:\omega\in A_-,\ h\in\mathbb H_T,\ \|h\|_{\mathbb H_T}\leq2r\}$ has inner Wiener measure at least $\Phi(r)>1-v/4$.
It therefore intersects $A_+$, which gives the first assertion.
For the quantile assertion, we take $A_-=\{F\leq Q_F(v)\}$ and $A_+=\{F\geq Q_F(1-v)\}$, each of Wiener measure at least $v$.
\end{proof}

Finally, we use the Burkholder--Davis--Gundy inequality to control stochastic integrals in the cavity residual estimates.
\begin{lemma}[{\cite[Theorem~1.1]{MarinelliRockner16}}]
\label{lem:bdg-stochastic-integral}
Let $p\geq2$, let $(B_t)_{0\leq t\leq T}$ be a standard Brownian motion in $\R^n$ with respect to a filtration $(\mathcal F_t)$, and let $G_t\in\R^{k\times n}$ be predictable with respect to this filtration.
If $\E{(\int_0^T\|G_t\|_{\mathrm F}^2\,\mathrm d t)^{p/2}}<\infty$, then
\begin{equation*}
 \left(\E{\left\|\int_0^T G_t\,\mathrm d B_t\right\|_2^p}\right)^{1/p}
 \leq C_p\left(\E{\left(\int_0^T\|G_t\|_{\mathrm F}^2\,\mathrm d t\right)^{p/2}}\right)^{1/p},
\end{equation*}
where $C_p$ depends only on $p$.
\end{lemma}

\begin{proof}
We apply the upper bound in \cite[Theorem~1.1]{MarinelliRockner16} to the $\R^k$-valued continuous martingale $M_t=\int_0^tG_s\,\mathrm dB_s$.
Its scalar quadratic variation is $[M,M]_T=\int_0^T\|G_t\|_{\mathrm F}^2\,\mathrm d t$, and $\|M_T\|_2\leq\sup_{t\leq T}\|M_t\|_2$, giving the stated bound with a constant depending only on $p$.
\end{proof}

\subsection{Stochastic localization}
\label{subsec:prelim-stochastic-localization}

For a fixed $A$ and a fixed constant $T>0$, we use the Wiener space and measure $\gamma$ from \Cref{subsec:prelim-wiener-space}.
For every $\omega\in C_0([0,T];\R^n)$, let $Y^\omega$ be defined by
\begin{equation}
 Y_t^\omega=\omega_t+\int_0^t m(A,Y_s^\omega)\,\mathrm d s,
 \qquad 0\leq t\leq T.
 \label{eq:localization-path}
\end{equation}

Differentiating the Gibbs weights with $A$ fixed gives $\partial m(A,y)/\partial y=\Sigma(A,y)$.
Since $\Sigma(A,y)$ is positive semidefinite, we have
$
 \|\Sigma(A,y)\|_{\mathrm{op}}
 \leq \operatorname{Tr}\Sigma(A,y)
 =\sum_{i=1}^n\bigl(1-m_i(A,y)^2\bigr)\leq n$.
Thus $m(A,\cdot)$ is globally $n$-Lipschitz, so \cref{eq:localization-path} has a unique solution for every continuous driving path $\omega$.
Applying \Cref{lem:gronwall-path-stability} with $L=n$ shows that $\omega\mapsto Y^\omega$ is continuous in the uniform norm, hence Borel measurable.

\begin{lemma}[{\cite[Section~III]{ElAlaouiMontanari22}}]
\label{lem:app-wpi-observation-law}
For a fixed $A$, let $X\sim\mu_A$, and let $(B_t)_{0\leq t\leq T}$ be a standard Brownian motion in $\R^n$ with $B_0=0$, independent of $X$; set $Y_t=tX+B_t$.
Let $(\mathcal F_t)$ be the usual augmentation of the filtration $\sigma((Y_s)_{0\leq s\leq t})$.
\begin{enumerate}[label=\textup{(\roman*)}]
\item For every $t\in[0,T]$, we have
\begin{equation*}
 \operatorname{Law}(X\mid\mathcal F_t)
 =\operatorname{Law}(X\mid Y_t)=\mu_{A,Y_t}.
\end{equation*}
In particular, $m(A,Y_t)=\E{X\mid\mathcal F_t}$ is a bounded martingale.
\item The innovation process $\widetilde B_t=Y_t-\int_0^t m(A,Y_s)\,\mathrm d s$ is a standard Brownian motion with respect to $(\mathcal F_t)$, and the field satisfies
\begin{equation*}
 \mathrm dY_t=m(A,Y_t)\,\mathrm d t+\mathrm d\widetilde B_t.
\end{equation*}
The posterior mean satisfies
\begin{equation*}
 \mathrm d m(A,Y_t)=\Sigma(A,Y_t)\,\mathrm d\widetilde B_t.
\end{equation*}
\end{enumerate}
Consequently, $Y$ has the same path distribution as $Y^\omega$ under $\omega\sim\gamma$.
All conclusions also hold for any full-support distribution $\mu$ on $\{\pm1\}^n$, replacing $\mu_A$ by $\mu$, $\mu_{A,y}$ by the normalized exponential tilt $\mu_y(x)\propto\mu(x)\exp(\inner{y}{x})$, and $m(A,y)$ and $\Sigma(A,y)$ by the mean and covariance of $\mu_y$, with the same replacement in the drift of \cref{eq:localization-path}.
\end{lemma}

\begin{proof}
We apply \cite[Section~III]{ElAlaouiMontanari22} with $Q=I_n$: given $X=x$, the observed-path likelihood relative to Wiener measure is $\exp(\inner{x}{Y_t}-t\|x\|_2^2/2)$, and the innovation representation makes $\widetilde B$ Brownian in $(\mathcal F_t)$.
Since $\|x\|_2^2=n$, Bayes' rule proves part~\textup{(i)}, and the definition of $\widetilde B$ gives the field SDE.
Summing the posterior-weight SDE in the same reference against $x$ gives the mean SDE in part~\textup{(ii)}.
Finally, uniqueness in \cref{eq:localization-path} gives $Y=Y^{\widetilde B}$, so the Brownian distribution of $\widetilde B$ proves the equality of path distributions.
The argument and the covariance bound above apply to every full-support base distribution on $\{\pm1\}^n$.
\end{proof}

Under either augmented law, we have $A^{\mathrm{full}}=A+\operatorname{diag}(d)$, where $d$ is independent of $(X,A,Y)$.
Thus observing the diagonal does not change the posterior distribution of $X$, and by \Cref{lem:app-wpi-observation-law} we obtain, for every fixed $t\geq0$,
\begin{equation}
 \operatorname{Law}(X\mid A^{\mathrm{full}},Y_t)
 =\operatorname{Law}(X\mid A,Y_t)=\mu_{A,Y_t}.
 \label{eq:planted-posterior}
\end{equation}

\subsection{Scalar overlap, TAP free energy, and AMP}
\label{subsec:prelim-tap-amp}

Under $\mathbb P_{n,\beta}^{\mathrm{pl,aug}}$, for a fixed $t\geq0$, we draw $\sigma^1,\sigma^2$ independently from $\mu_{A,Y_t}$ given $(A,Y_t)$ and define their overlap $R_{12}=n^{-1}\inner{\sigma^1}{\sigma^2}$.

For $t\geq0$, we define $q_\beta(t)$ as the unique solution in $[0,1]$ of the following fixed-point equation:
\begin{equation}
 q=f(\beta^2q+t),\qquad
 f(a)=\E[G\sim N(0,1)]{\tanh^2(a+\sqrt a\,G)}.
 \label{eq:compact-rs-fixed-point}
\end{equation}
{For $a\geq0$, the scalar Gaussian channel} $Y=aX+\sqrt a\,G$, with independent $X\sim\operatorname{Unif}(\{\pm1\})$, $G\sim N(0,1)$, has posterior mean $\tanh Y$.
Hence $f(a)=\E{\tanh^2(Y)}$ is the mean overlap of two conditionally independent posterior replicas; in \cref{eq:compact-rs-fixed-point}, we set $a=\beta^2q+t$.

\begin{lemma}[{\cite[Lemmas~5.5--5.6]{ElAlaouiMontanariSellke25}}]
\label{lem:q-beta-regularity}
For every $0\leq\beta<1$ and $t\geq0$, the function $q\mapsto f(\beta^2q+t)$ takes $[0,1]$ into itself and is $\beta^2$-Lipschitz.
Hence \cref{eq:compact-rs-fixed-point} has a unique solution $q_\beta(t)\in[0,1]$.
Moreover, $q_\beta(0)=0$ and, for every $s,t\geq0$,
\begin{equation*}
 \abs{q_\beta(t)-q_\beta(s)}\leq\frac{\abs{t-s}}{1-\beta^2}.
\end{equation*}
In particular, $q_\beta(t)\leq t/(1-\beta^2)$.
\end{lemma}

\begin{proof}
For the SK choice $\xi(q)=q^2/2$, \cite[Remark~2.3]{ElAlaouiMontanariSellke25} gives $\beta_1=1$, and its scalar function $\psi$ equals $f$.
Thus \cite[Lemmas~5.5--5.6]{ElAlaouiMontanariSellke25} gives the contraction and uniqueness statements, $q_\beta(0)=0$, and the $1$-Lipschitz property of $f$.
Subtracting the fixed-point equations at $s$ and $t$, we obtain
\begin{equation*}
 \abs{q_\beta(t)-q_\beta(s)}\leq\beta^2\abs{q_\beta(t)-q_\beta(s)}+\abs{t-s}.
\end{equation*}
We rearrange this inequality to obtain the Lipschitz bound, then take $s=0$ to obtain the final bound.
\end{proof}

{For $u\in[-1,1]^n$, we call $Q_n(u)\defeq\|u\|_2^2/n$ the \emph{empirical overlap}; at the posterior mean $u=m(A,Y_t)$, it equals $\langle R_{12}\rangle$.
In the TAP expressions below, a \emph{prescribed overlap} is a scalar $q\in[0,1]$ held fixed as $u$ varies, and we take $q=q_\beta(t)$ at observation time $t$.}

For $u\in(-1,1)^n\subset\mathbb R^n$, with $\operatorname{arctanh}(u)$ applied coordinatewise, define the prescribed-overlap TAP residual
\begin{equation*}
 R_{\mathrm{TAP}}(u;A^{\mathrm{full}},y,q)
 =\operatorname{arctanh}(u)-y-\beta A^{\mathrm{full}}u
   +\beta^2(1-q)u.
\end{equation*}
The iterates start from $\widehat m_{-1}=0$, $z_0=0$ and satisfy
\begin{equation*}
 \widehat m_k=\tanh(z_k),\qquad
 z_{k+1}=\beta A^{\mathrm{full}}\widehat m_k+y
  -\beta^2\left(1-\|\widehat m_k\|_2^2/n\right)\widehat m_{k-1}.
\end{equation*}
For $u\in[-1,1]$, define $h_{\mathrm{bin}}(u)=-(1+u)\log((1+u)/2)/2-(1-u)\log((1-u)/2)/2$ with the convention $0\log0=0$.
The prescribed-overlap TAP free energy, with $q$ held fixed, is the following function on $[-1,1]^n$:
\begin{equation*}
 F_{\mathrm{TAP}}(u;A^{\mathrm{full}},y,q)
 =-\sum_i h_{\mathrm{bin}}(u_i)
  -\frac\beta2\inner{u}{A^{\mathrm{full}}u}
  -\inner{y}{u}+\frac{\beta^2}{2}(1-q)\|u\|_2^2.
\end{equation*}
For $u\in(-1,1)^n$, we have $\nabla_uF_{\mathrm{TAP}}=R_{\mathrm{TAP}}$.

\begingroup

\section{TAP residual bounds}
\label{sec:tap-residual-bounds}

In this section, we prove \Cref{thm:polynomial-tap-residual}, which gives an unnormalized $O(\log n)$ squared-residual bound with failure probability at most $n^{-H}$ for every fixed $H>0$.
Throughout this section and \Cref{sec:omitted-tap-residual-proofs}, we fix $0\leq\beta_*<1$ as a common upper bound on inverse temperature.
We fix $n$ and $0\leq\beta\leq\beta_*$, and set $\ell=\beta^2/n$.
At a fixed time under the planted model of \Cref{def:tap-planted-model}, we set $J=\beta A^{\mathrm{full}}$ and $y=Y_t$.
Since the diagonal of $A^{\mathrm{full}}$ does not affect the Gibbs distribution, we use $\mu_{J,y}$ as shorthand for $\mu_{A,y}$.
We abbreviate $\mu=\mu_{J,y}$, $m=m(A,y)$, and $\Sigma=\Sigma(A,y)$.

Recall that, for a scalar $q$,
$ R_{\mathrm{TAP}}(u;A^{\mathrm{full}},y,q)
 =\operatorname{arctanh}(u)-y-\beta A^{\mathrm{full}}u+\beta^2(1-q)u$.
Set $Q_n(u)=\norm u_2^2/n$.
{We recall that $q_\beta(t)$ is the deterministic fixed point in \cref{eq:compact-rs-fixed-point}, while $Q_n(m)=\langle R_{12}\rangle$ depends on the observations $(A,Y_t)$.}
We call $Q_n(m)$ the empirical overlap and $q_\beta(t)$ the prescribed overlap, and define
\begin{equation}
 \mathbf r_t^{(Q)}=R_{\mathrm{TAP}}(m;A^{\mathrm{full}},y,Q_n(m)),\qquad
 \mathbf r_t^{(q)}=R_{\mathrm{TAP}}(m;A^{\mathrm{full}},y,q_\beta(t)).
 \label{eq:tap-prescribed-residual}
\end{equation}
{We call $\mathbf r_t^{(Q)}$ the empirical residual and $\mathbf r_t^{(q)}$ the prescribed residual, according to the overlap used in their Onsager coefficients.}
At fixed time we write $\mathbf r=\mathbf r_t^{(Q)}$.
Although \Cref{thm:polynomial-tap-residual} concerns $\mathbf r_t^{(q)}$, we first prove the tail bound for $\mathbf r_t^{(Q)}$
since the two residuals are related by the exact relation
$ \mathbf r_t^{(q)}-\mathbf r_t^{(Q)} =\beta^2\bigl(Q_n(m)-q_\beta(t)\bigr)m$
and the overlap concentration in \Cref{lem:tap-overlap-mgf} shows that this difference is small with high probability.

Unless otherwise indicated, $\E{\cdot}$ and $L^p$ refer to the outer planted law of \Cref{def:tap-planted-model}, while $\langle\cdot\rangle$ averages posterior replicas given the observations; when an integrand contains $X$, the outer law also averages over the plant.
We reserve $C$ for constants independent of the moment order and $C_p$ for constants that may depend on the fixed order $p$.
We use
\[
 \|V\|_{L^p}=\bigl(\E{\norm V_2^p}\bigr)^{1/p},\qquad
 \|K\|_{L^p}=\bigl(\E{\norm K_{\mathrm F}^p}\bigr)^{1/p}
\]
for vector and matrix random variables, respectively, unless another matrix norm is specified.

We next introduce the cavity model and its interpolation.

\begin{definition}[Cavity interpolation model]
\label{def:tap-cavity-pinning}
For any symmetric $J\in\mathbb R^{n\times n}$, external field $y\in\mathbb R^n$,
let $\mu=\mu_{J,y}$ have Gibbs weights proportional to $\exp(\inner{\sigma}{J\sigma}/2+\inner{y}{\sigma})$, and write $m,\Sigma$ for its mean and covariance.
For coordinate $i\in[n]$, write $-i=[n]\setminus\{i\}$ for the surviving
coordinate set, $J^{\setminus i}=J_{-i,-i}$, and $a_i=J_{-i,i}$.
This deletion preserves $\ell$, and the surviving $M=n-1$ spins have inverse temperature $\beta_M=\sqrt{M\ell}=\beta\sqrt{M/n}$.

The cavity model is obtained by deleting coordinate $i$.  Its interaction
and external field are $J^{\setminus i}$ and $y_{-i}$, its Gibbs distribution is
denoted by $\mu^{\setminus i}$, and its mean and covariance are
$m^{\setminus i}$ and $\Sigma^{\setminus i}$.
In the planted model, we set
$
 \mathcal F=\sigma(J,y),\
 \mathcal F^{\setminus i}=\sigma(J^{\setminus i},y_{-i})
$,
and call $\mathcal F^{\setminus i}$ the cavity sigma-field.
Deleting coordinate $i$ discards $a_i$, $J_{ii}$, and $y_i$ without conditioning on the plant.
Under the planted law of \Cref{def:tap-planted-model}, given $\mathcal F^{\setminus i}$, the variables $X_{-i}\sim\mu^{\setminus i}$ and $X_i\sim\operatorname{Unif}(\{\pm1\})$ are independent.
The removed interaction vector satisfies
\begin{equation*}
 a_i=\ell X_iX_{-i}+\sqrt\ell\,G_i,
 \qquad G_i\sim N(0,I_{n-1}),
\end{equation*}
where $G_i$ is independent of $(\mathcal F^{\setminus i},X)$.

We connect the cavity distribution to the conditional laws obtained by fixing $\sigma_i=\pm1$ through the exponential tilt $\exp(\gamma \inner{a_i}{\sigma_{-i}})$.
For the chosen $i$ and $\gamma\in[-1,1]$, define the cavity interpolation by
\begin{equation*}
 Z_\gamma
 \defeq\sum_{\tau\in\{\pm1\}^{n-1}}\mu^{\setminus i}(\tau)
       \exp(\gamma\inner{a_i}\tau),
 \qquad
 \mu_\gamma(\tau)
 \defeq\frac{\mu^{\setminus i}(\tau)\exp(\gamma\inner{a_i}\tau)}
 {Z_\gamma},
 \qquad \tau\in\{\pm1\}^{n-1}.
\end{equation*}
Thus the interaction remains $J^{\setminus i}$, while the external field at
parameter $\gamma$ is $y_{-i}+\gamma a_i$.
We write $m_\gamma$ and $\Sigma_\gamma$ for the mean and covariance under $\mu_\gamma$, and $\langle\cdot\rangle_\gamma$ for expectation under this distribution.
At $\gamma=0$, the interpolated law is the cavity law:
$Z_0=1$, $\mu_0=\mu^{\setminus i}$, $m_0=m^{\setminus i}$, and $\Sigma_0=\Sigma^{\setminus i}$.
For $b\in\{\pm1\}$, the endpoint $\gamma=b$ is the pinned law,
\begin{equation*}
 \mu_b(\tau)
 =\mu(\sigma_{-i}=\tau\mid\sigma_i=b),
 \qquad \tau\in\{\pm1\}^{n-1}.
\end{equation*}
To recover the marginal of $\mu$ on the surviving coordinates, we mix the pinned laws with weights $\pi_b=\mu(\sigma_i=b)=\frac{1+b(m)_i}{2}$.
If $m_{-i}$ denotes the restriction of the full mean $m$ to the
surviving coordinates, then
\[
 m_{-i}=\sum_{b=\pm1}\pi_b m_b,
 \qquad m^{\setminus i}\ne m_{-i}\quad\text{in general}.
\]
Given $i\in[n]$, external field $z\in\mathbb R^{n-1}$, and $u\in(-1,1)^{n-1}$, the empirical residual of the surviving system is
\begin{equation}
 R_{n-1}^{(Q)}(u;J^{\setminus i},z)
 =\operatorname{arctanh}(u)-z-J^{\setminus i}u+\ell\bigl(n-1-\norm u_2^2\bigr)u.
 \label{eq:tap-empirical-residual}
\end{equation}
{For $n\ge2$, the empirical overlap here is $Q_{n-1}(u)=\|u\|_2^2/(n-1)$, and the Onsager coefficient equals $(n-1)\ell(1-Q_{n-1}(u))$.}
Here $\ell$ is held fixed after deletion, and the diagonal of $J^{\setminus i}$ is retained in the residual even though it cancels from the Gibbs distribution.
The cavity and pinned residuals are
\[
 \mathbf r^{\setminus i}
 =R_{n-1}^{(Q)}(m^{\setminus i};J^{\setminus i},y_{-i}),\qquad
 \mathbf r_{\mathrm{pin}}^{i,b}
 =R_{n-1}^{(Q)}(m_b;J^{\setminus i},y_{-i}+ba_i).
\]
For fixed $i$, we abbreviate $\mathbf r_{\mathrm{pin}}^b=\mathbf r_{\mathrm{pin}}^{i,b}$.
\end{definition}

At a fixed time, $r_i$ is coordinate $i$ of $\mathbf r=\mathbf r_t^{(Q)}$, and we set
\[
 \begin{aligned}
 e_i&=\E{r_i^2\mid\mathcal F^{\setminus i}},&
 \mathcal E&=\sum_{i=1}^n e_i.
 \end{aligned}
\]
We view $e_i$ as a smoothed version of $r_i^2$: it averages over the deleted data and, like $\mathbf r^{\setminus i}$, depends only on the surviving observations $(J^{\setminus i},y_{-i})$.
Averaging over all observations, we obtain $\E{\mathcal E}=\sum_i\E{r_i^2}=\E{\norm{\mathbf r}_2^2}$, although $\mathcal E$ need not equal $\norm{\mathbf r}_2^2$ pointwise.

With this notation, we prove \Cref{thm:polynomial-tap-residual} by combining three estimates:
\begin{equation*}
 \sup_i\|r_i\|_{L^p}\leq C_pn^{-1/2},\qquad \E{\exp(\kappa\mathcal E)}\leq 2,\qquad \max_i\bigl\|\,\|\mathbf r\|_2-\|\mathbf r^{\setminus i}\|_2\,\bigr\|_{L^{4p}}\leq C_pn^{-1/4}.
\end{equation*}
For every fixed integer $p\geq2$, the moment bounds hold with $C_p$ independent of $n$, while the second bound holds for some $\kappa>0$ depending only on $\beta_*$.
The first two estimates, \Cref{lem:tap-complete-row}-\textup{(iii)} and \Cref{cor:tap-energy-mgf}, are proved below.
The second estimate yields an exponential tail for $\mathcal E$, whereas the coordinate moment bounds alone do not give the desired residual tail, since we do not control how $C_p$ grows with $p$.
Using the third estimate, \Cref{lem:tap-parent-child}, proved in \Cref{subsec:tap-stability-proofs}, we approximate tail events by events determined by the surviving observations and then replace $r_i^2$ by $e_i$ through conditioning.
In the proof of \Cref{thm:polynomial-tap-residual}, we transfer this tail control to $\norm{\mathbf r}_2^2$, choosing $p$ depending on $H$ but not on $n$.

\subsection{Overlap concentration and covariance bounds}

By \Cref{lem:q-beta-regularity}, for every $0\leq\beta<1$ and $t\geq0$, the equation $q=f(\beta^2q+t)$ has the unique solution $q_\beta(t)$.
To prove the conditional residual estimates in \Cref{lem:tap-complete-row}, we first show that the replica overlaps $R_{ab}=n^{-1}\inner{\sigma^a}{\sigma^b}$ concentrate around this value.
We use the following exponential-moment estimate to derive the required posterior covariance bounds in \Cref{cor:tap-covariance}.

\begin{lemma}[Finite-dimensional overlap MGF and moments]
\label{lem:tap-overlap-mgf}
There are constants $\tau_0,K_0>0$, depending only on $\beta_*$, such that
the following statements hold for every $n\ge1$, $0\le \beta\le \beta_*$, and $t\ge0$.
\begin{enumerate}[label=\textup{(\roman*)}]
\item The squared overlap satisfies
\begin{equation*}
 \E[\mathbb P_{n,\beta}^{\mathrm{pl,aug}}]{\left\langle \exp\left(\tau_0n(R_{12}-q_{\beta}(t))^2\right)\right\rangle}
 \le K_0.
\end{equation*}
The same bound holds with $R_{12}$ replaced by $n^{-1}\sigma^1\cdot X$ or
$Q_n(m)=n^{-1}\norm{m}_2^2$.

\item For every $p\geq1$,
\begin{equation*}
 \E[\mathbb P_{n,\beta}^{\mathrm{pl,aug}}]{
   \left\langle|R_{12}-q_\beta(t)|^{2p}\right\rangle}
 \leq K_0\left(\frac p{\mathrm e\tau_0 n}\right)^p,
\end{equation*}
and the same moment bound holds for the replica--plant overlap and
$Q_n(m)$.
\end{enumerate}
\end{lemma}

\begin{proof}
We use Gaussian interpolation with a decreasing quadratic coupling, as in Lata{\l}a's argument \cite[Corollary~1.4.7 and Proposition~1.4.8]{Talagrand11}.
In the planted model, we take a square root before averaging over the observations to cancel the posterior normalizing constant.
To compare the original planted model with an independent scalar-channel
model, we interpolate over $s\in[0,\beta^2]$.  Fix any $q\in[0,1]$ and
$\lambda_f\in(0,1-\beta^2)$, and define
\[
 t_s=t+(\beta^2-s)q,\qquad a = t_0,\qquad \lambda_s=\lambda_f+\beta^2-s,
 \qquad 0\le s\le \beta^2.
\]
Notice that
\[
 (s,t_s)=(\beta^2,t)\quad\text{at }s=\beta^2,
 \qquad (s,t_s)=(0,a)\quad\text{at }s=0,
 \qquad sq+t_s=a.
\]
For each $s$, let $\mathbb P_s$ be the law of the observations $(w,z)$
generated by first drawing $X\sim\operatorname{Unif}(\{\pm1\}^n)$, and then setting
\[
 w_{jk}=\sqrt{s/n}\,X_jX_k+G_{jk}\quad(j<k),\qquad
 z_j=\sqrt{t_s}\,X_j+g_j\quad(j\le n),
\]
where $(G_{jk})_{j<k}$ and $(g_j)_{j\le n}$ are mutually independent standard Gaussians, independent of $X$.
Let $\mu_s(\cdot\mid w,z)$ be the posterior distribution of $X$, and let
$\langle\cdot\rangle_s$ denote expectation over independent replicas from
$\mu_s(\cdot\mid w,z)$. Thus $\E[\mathbb P_s]{\cdot}$ averages the observations, while $\langle\cdot\rangle_s$ averages posterior replicas given those observations.

At $s=\beta^2$, we have $w_{jk}=\sqrt n\,A_{jk}$ and $\sqrt t\,z=Y_t$. So the interpolating model coincides with the fixed-time planted model and $\mu_{\beta^2}(\cdot\mid w,z)=\mu_{A,Y_t}$.
Thus, for every fixed-time observable \(F\) of finitely many replicas, \(\langle F\rangle\) is determined by \((A,Y_t)\) alone.
Averaging over $d$ and the other times of $Y$ therefore gives
$\E[\mathbb P_{\beta^2}]{\langle F\rangle_{\beta^2}}=\E[\mathbb P_{n,\beta}^{\mathrm{pl,aug}}]{\langle F\rangle}$.

We now set
\[
 \Phi(s)
 \defeq\E[\mathbb P_s]{
 \left\langle\exp\left(\frac{\lambda_sn}{2}(R_{12}-q)^2\right)
 \right\rangle_s^{1/2}}.
\]
Our goal is to prove that $\Phi(s)$ is nonincreasing for every $q\in[0,1]$, thereby bounding the
original endpoint by the independent endpoint.

\emph{Step 1: cancel the posterior normalizer.}
To cancel the posterior
normalizer, we compare $\mathbb P_s$ with the fixed reference law
$\mathbb Q$, under which all coordinates of $(w,z)$ are independent
standard Gaussians.

For each possible value $x\in\{\pm1\}^n$ of the planted spin $X$, the
standardized observation has, relative to $\mathbb Q$, mean-shift coordinates
\[
 \sqrt{s/n}\,x_jx_k\quad(j<k),\qquad
 \sqrt{t_s}\,x_j\quad(j\le n).
\]
The Gaussian shift formula gives the likelihood factor
\[
 L_{s,x}(w,z)
 \defeq\exp(H_s(x)-v_s/2),
\]
where
\[
 H_s(x)=\sqrt{\frac sn}\sum_{j<k}w_{jk}x_jx_k
             +\sqrt{t_s}\sum_jz_jx_j.
\]
The configuration $x$ is
fixed in this likelihood calculation.  In particular,
$\E[\mathbb Q]{\cdot}$ below averages only
over the reference observation $(w,z)$ and does not average over either
$X$ or $x$.  By independence and centering, all cross terms have expectation zero,
so
\[
 v_s
 \defeq\E[\mathbb Q]{H_s(x)^2}
 =\frac{s}{n}\sum_{j<k}(x_jx_k)^2
   +t_s\sum_jx_j^2
 =\frac{s(n-1)}{2}+nt_s.
\]
This is independent of $x$ because every spin satisfies
$x_j^2=1$: all possible plants produce mean-shift vectors with the
same Euclidean norm.
$L_{s,x}$ is a likelihood factor, not a posterior mass.
Averaging it over the uniform prior $\pi(x)\defeq2^{-n}$, $x\in\{\pm1\}^n$, we set $L_s=\sum_x\pi(x)L_{s,x}$.
The planted observation law defined above therefore satisfies
\[
 \frac{\mathrm d\mathbb P_s}{\mathrm d\mathbb Q}(w,z)=L_s(w,z).
\]
Bayes' formula also gives the explicit posterior distribution
\[
 \mu_s(\sigma\mid w,z)
 =\frac{\pi(\sigma)L_{s,\sigma}(w,z)}{L_s(w,z)}.
\]
We define the unnormalized overlap-tilted two-replica numerator by
\[
 \mathcal Z_s
 \defeq\sum_{\sigma,\tau}\pi(\sigma)\pi(\tau)L_{s,\sigma}L_{s,\tau}
       \exp\left(\lambda_sn(R(\sigma,\tau)-q)^2/2\right).
\]
We define the conditional exponential moment by
$
 \mathcal M_s
 \defeq\left\langle
      \exp\left(\lambda_sn(R_{12}-q)^2/2\right)
    \right\rangle_s
$.
Using the preceding expression for $\mu_s$ in the two-replica sum, we have
\[
 \mathcal M_s
 =\sum_{\sigma,\tau}\mu_s(\sigma\mid w,z)\mu_s(\tau\mid w,z)
      \exp\left(\lambda_sn(R(\sigma,\tau)-q)^2/2\right)
 =\frac{\mathcal Z_s}{L_s^2}.
\]
The square root is chosen specifically to remove the random normalizer:
the outer observation law contributes a factor $L_s$, whereas
$\sqrt{\mathcal M_s}=\sqrt{\mathcal Z_s}/L_s$.  We substitute Bayes'
formula into the change of measure $\mathrm d\mathbb P_s=L_s\mathrm d\mathbb Q$
and cancel $L_s$ to obtain
\begin{equation}
 \E[\mathbb P_s]{\sqrt{\mathcal M_s}}
 =\int L_s\frac{\sqrt{\mathcal Z_s}}{L_s}\mathrm d\mathbb Q
 =\E[\mathbb Q]{\sqrt{\mathcal Z_s}}.
 \label{eq:tap-bayes-cancellation}
\end{equation}

\smallskip
\noindent\emph{Step 2: prove monotonicity of the square-root interpolation.}
We define
$\Psi(s)=\log\E[\mathbb Q]{\sqrt{\mathcal Z_s}}=\log\Phi(s)$ and prove that it is
nonincreasing on $[0,\beta^2]$.  For fixed reference observations, we define
the overlap-tilted replica-pair law
\[
 \rho_s(\sigma,\tau\mid w,z)
 \defeq
 \frac{1}{\mathcal Z_s} \pi(\sigma)\pi(\tau)L_{s,\sigma}L_{s,\tau}
       \exp\left(\lambda_sn(R(\sigma,\tau)-q)^2/2\right).
\]
For a function $F$ of one pair, we set
\[
 \langle F\rangle_{\rho_s}
 \defeq\sum_{\sigma^1,\sigma^2}
       \rho_s(\sigma^1,\sigma^2\mid w,z)F(\sigma^1,\sigma^2).
\]
For a function $F$ of two pairs, we set
\[
 \langle F\rangle_{\rho_s^{\otimes2}}
 \defeq\sum_{\sigma^1,\ldots,\sigma^4}
       \rho_s(\sigma^1,\sigma^2\mid w,z)
       \rho_s(\sigma^3,\sigma^4\mid w,z)
       F(\sigma^1,\sigma^2,\sigma^3,\sigma^4).
\]
Thus the two overlap-tilted pairs are conditionally independent given $(w,z)$.
We include the outer reweighting by defining
\[
 \mathopen{[\![}F\mathclose{]\!]}_s
 \defeq\frac{\E[\mathbb Q]{\sqrt{\mathcal Z_s}
       \langle F\rangle_{\rho_s^{\otimes2}}}}
      {\E[\mathbb Q]{\sqrt{\mathcal Z_s}}}.
\]
Under $\mathbb Q$, the collection $H_s=(H_s(\eta))_{\eta\in\{\pm1\}^n}$
is a centered $2^n$-dimensional Gaussian vector with covariance kernel
$\mathcal K_s(\sigma,\tau)=\E[\mathbb Q]{H_s(\sigma)H_s(\tau)}$.
For $h=(h_\eta)_{\eta\in\{\pm1\}^n}$, we define the deterministic function
\[
 \mathcal Z_s(h)
 \defeq
 \sum_{\sigma,\tau}\pi(\sigma)\pi(\tau)
 \exp\left(h_\sigma+h_\tau-v_s
       +\frac{\lambda_sn}{2}(R(\sigma,\tau)-q)^2\right)
 .
\]
The random variable defined above is its evaluation
$\mathcal Z_s=\mathcal Z_s(H_s)$.
We write $\partial_\eta$ for differentiation with respect to $h_\eta$.
In every coordinate derivative below, we differentiate the corresponding
function of $h$ and then evaluate at $h=H_s$; we suppress this final
substitution from the notation.
Keeping $q$ and $a$ fixed, we apply the standard Gaussian interpolation
differentiation formula to the centered Gaussian vector $H_s$ and the
deterministic function $h\mapsto \sqrt{\mathcal Z_s(h)}$:
\[
 \frac{\mathrm d}{\mathrm d s}\E[\mathbb Q]{\sqrt{\mathcal Z_s}}
 =\E[\mathbb Q]{\frac{\partial}{\partial s}\sqrt{\mathcal Z_s(h)}\big|_{h=H_s}}
  +\frac12\sum_{\eta,\xi}\frac{\mathrm d}{\mathrm d s}\mathcal K_s(\eta,\xi)
       \E[\mathbb Q]{\partial_\eta\partial_\xi\sqrt{\mathcal Z_s}}.
\]
The partial derivative holds $h$ fixed and differentiates only
$v_s$ and $\lambda_s$.  Its contribution is
\[
\frac{1}{\sqrt{\mathcal Z_s}}
 \left.\frac{\partial}{\partial s}
      \sqrt{\mathcal Z_s}\right|_{h=H_s}
 =
   -\frac12\frac{\mathrm d v_s}{\mathrm d s}
   +\frac n4\frac{\mathrm d\lambda_s}{\mathrm d s}
      \langle(R_{12}-q)^2\rangle_{\rho_s}.
\]

For one overlap-tilted pair, we set
$n_\eta=\1_{\{\sigma^1=\eta\}}+\1_{\{\sigma^2=\eta\}}$.
{We differentiate the finite sum defining $\mathcal Z_s$ and its normalizing factor to obtain}
\begin{gather*}
 \partial_\eta\log\mathcal Z_s
 =\langle n_\eta\rangle_{\rho_s},\\
 \partial_\eta\partial_\xi\log\mathcal Z_s
 =\langle n_\eta n_\xi\rangle_{\rho_s}
   -\langle n_\eta\rangle_{\rho_s}\langle n_\xi\rangle_{\rho_s},\\
 \frac{\partial_\eta\partial_\xi\sqrt{\mathcal Z_s}}
      {\sqrt{\mathcal Z_s}}
 =\frac12\langle n_\eta n_\xi\rangle_{\rho_s}
   -\frac14\langle n_\eta\rangle_{\rho_s}
              \langle n_\xi\rangle_{\rho_s}.
\end{gather*}
The last identity follows from
$\partial_\eta\partial_\xi\sqrt{\mathcal Z_s}/\sqrt{\mathcal Z_s}
=\tfrac12\partial_\eta\partial_\xi\log\mathcal Z_s
+\tfrac14(\partial_\eta\log\mathcal Z_s)
            (\partial_\xi\log\mathcal Z_s)$.

Using $\mathcal K_s(\sigma,\sigma)=v_s$, we contract the two Hessian terms as
\begin{gather*}
 \sum_{\eta,\xi}\frac{\mathrm d}{\mathrm d s}\mathcal K_s(\eta,\xi)
       \langle n_\eta n_\xi\rangle_{\rho_s}
 =\left\langle\sum_{a,b=1}^2\frac{\mathrm d}{\mathrm d s}\mathcal K_s(\sigma^a,\sigma^b)\right\rangle_{\rho_s}
  =2\frac{\mathrm d v_s}{\mathrm d s}+2\langle\frac{\mathrm d}{\mathrm d s}\mathcal K_s(\sigma^1,\sigma^2)\rangle_{\rho_s},\\
 \sum_{\eta,\xi}\frac{\mathrm d}{\mathrm d s}\mathcal K_s(\eta,\xi)
       \langle n_\eta\rangle_{\rho_s}\langle n_\xi\rangle_{\rho_s}
 =\left\langle\sum_{a=1}^2\sum_{b=3}^4
       \frac{\mathrm d}{\mathrm d s}\mathcal K_s(\sigma^a,\sigma^b)\right\rangle_{\rho_s^{\otimes2}}
  =4\left\langle
       \frac{\mathrm d}{\mathrm d s}\mathcal K_s(\sigma^1,\sigma^3)\right\rangle_{\rho_s^{\otimes2}}.
\end{gather*}
For each fixed $(w,z)$, $\rho_s(\sigma,\tau\mid w,z)$ is symmetric in
$\sigma,\tau$, which gives the last equality by exchanging the two
coordinates within either pair.
Multiplying by $\sqrt{\mathcal Z_s}$ and averaging under $\mathbb Q$ preserves
this identity.
Thus the second contraction, after outer reweighting, is
$4\mathopen{[\![}\frac{\mathrm d}{\mathrm d s}\mathcal K_s(\sigma^1,\sigma^3)\mathclose{]\!]}_s$.
Dividing by $\E[\mathbb Q]{\sqrt{\mathcal Z_s}}$ and substituting the preceding
identities, we obtain
\begin{equation}
\begin{aligned}
 \frac{\mathrm d\Psi(s)}{\mathrm d s}
 =&~-\frac12\frac{\mathrm d v_s}{\mathrm d s}
   +\frac n4\frac{\mathrm d\lambda_s}{\mathrm d s}
      \mathopen{[\![}(R_{12}-q)^2\mathclose{]\!]}_s\\
 &~+\frac14\left(2\frac{\mathrm d v_s}{\mathrm d s}
       +2\mathopen{[\![}\frac{\mathrm d}{\mathrm d s}\mathcal K_s(\sigma^1,\sigma^2)\mathclose{]\!]}_s\right)
       -\frac18\left(4
       \mathopen{[\![}\frac{\mathrm d}{\mathrm d s}\mathcal K_s(\sigma^1,\sigma^3)\mathclose{]\!]}_s\right)\\
 =&~\frac12\mathopen{[\![}\frac{\mathrm d}{\mathrm d s}\mathcal K_s(\sigma^1,\sigma^2)\mathclose{]\!]}_s
   -\frac12\mathopen{[\![}\frac{\mathrm d}{\mathrm d s}\mathcal K_s(\sigma^1,\sigma^3)\mathclose{]\!]}_s
   +\frac n4\frac{\mathrm d\lambda_s}{\mathrm d s}
      \mathopen{[\![}(R_{12}-q)^2\mathclose{]\!]}_s.
\end{aligned}
 \label{eq:tap-gaussian-pair-derivative}
\end{equation}

We compute the covariance derivative by setting $u_j=\sigma_j\tau_j$.  Then
\[
 \sum_ju_j=nR(\sigma,\tau),\qquad
 \sum_{j<k}u_ju_k
 =\frac12\left(\left(\sum_ju_j\right)^2-n\right).
\]
Since $t_s=a-sq$, we have
\begin{align*}
 \mathcal K_s(\sigma,\tau)
 =&~\frac{s}{n}\sum_{j<k}u_ju_k+t_s\sum_ju_j
 =\frac{sn}{2}R(\sigma,\tau)^2-\frac s2
   +n(a-sq)R(\sigma,\tau),
\end{align*} 
\begin{align*}  
 \frac{\mathrm d}{\mathrm d s}\mathcal K_s(\sigma,\tau)
 =&~\frac n2R(\sigma,\tau)^2-\frac12-nqR(\sigma,\tau)
 =\frac n2(R(\sigma,\tau)-q)^2-\frac{1+nq^2}{2}.
\end{align*}
Since $\frac{\mathrm d\lambda_s}{\mathrm d s}=-1$, we substitute the covariance derivative
into \cref{eq:tap-gaussian-pair-derivative}:
\begin{equation*}
\begin{aligned}
 \Psi'(s)
 =&~\frac12\mathopen{[\![}\tfrac n2(R_{12}-q)^2-\tfrac{1+nq^2}{2}\mathclose{]\!]}_s
   -\frac12\mathopen{[\![}\tfrac n2(R_{13}-q)^2-\tfrac{1+nq^2}{2}\mathclose{]\!]}_s
   -\frac n4\mathopen{[\![}(R_{12}-q)^2\mathclose{]\!]}_s\\
 =&~-\frac n4\mathopen{[\![}(R_{13}-q)^2\mathclose{]\!]}_s\le0.
\end{aligned}
\end{equation*}
All derivatives above are justified on compact subintervals where $s>0$
and $t_s>0$: the sums over spins are finite, and their derivatives are
bounded by a polynomial in $(w,z)$ times $C\exp(C\|(w,z)\|_2)$,
which is integrable under $\mathbb Q$.  If $t_s$ is identically zero, the scalar
Gaussian term is absent.  At the endpoints, the same integrable bound
and continuity of the square-root coefficients justify passage to the
limit.  Thus $\Psi$ is nonincreasing on the entire interpolation interval
for every $q\in[0,1]$.

\smallskip
\noindent\emph{Step 3: use the independent-coordinate endpoint.}
At $s=0$, the coefficient of the matrix term is zero, so the posterior distribution factors over coordinates and only independent scalar channels of strength $a$ remain.
For a scalar observation $z_j=\sqrt a\,X_j+g_j$, the posterior mean is
$\tanh(\sqrt a\,z_j)$.  Thus the conditional mean of the product of two
independent posterior replicas is $\tanh^2(\sqrt a\,z_j)$.  By symmetry,
$X_j\sqrt a\,z_j=a+\sqrt a\,X_jg_j$ has the same distribution as
$a+\sqrt a\,G$ for $G\sim N(0,1)$, and hence
\[
 \E[\mathbb P_0]{\langle\sigma_j^1\sigma_j^2\rangle_0}
 =\E[G\sim N(0,1)]{\tanh^2(a+\sqrt a\,G)}=f(a).
\]
We choose $q=f(a)$.  Together with $a=\beta^2q+t$, this gives
$q=f(\beta^2q+t)$, and hence $q=q_\beta(t)$ by \Cref{lem:q-beta-regularity}.
Since both the observation channel and the posterior distribution factor over coordinates, the products $\sigma_j^1\sigma_j^2$ are independent signs with mean $q$ after the outer planted average.
The variables $\sigma_j^1\sigma_j^2-q$ are independent, centered,
and each lies in an interval of length two.  We factor their exponential
moment and apply Hoeffding's lemma with parameter $z/\sqrt n$ to each factor:
\begin{align*}
 \E[\mathbb P_0]{\left\langle
   \exp\left(\frac z{\sqrt n}\sum_j(\sigma_j^1\sigma_j^2-q)\right)
   \right\rangle_0}
=\prod_{j=1}^n\E[\mathbb P_0]{
   \left\langle\exp\left(\frac z{\sqrt n}(\sigma_j^1\sigma_j^2-q)\right)\right\rangle_0}
 \leq\prod_{j=1}^n\exp\left(\frac{z^2}{2n}\right)
 =\exp\left(\frac{z^2}{2}\right).
\end{align*}
Since $\lambda_0=\lambda_f+\beta^2<1$,
we set $W=\sqrt n(R_{12}-q)$ and introduce an independent
$G_0\sim N(0,1)$.  For every fixed value of $W$, we evaluate the standard
Gaussian moment-generating function (MGF) and obtain
\[
 \E[G_0\sim N(0,1)]{\exp\left(\sqrt{\lambda_0}G_0W\right)}
 =\exp\left(\lambda_0W^2/2\right).
\]
Here the expectation over $W$ averages both the planted observations and the posterior replicas.
Since the integrands are nonnegative, we may average first over $W$ or first over
$G_0$. We condition on $G_0$ and apply the preceding sub-Gaussian estimate $\E[\mathbb P_0]{\langle\exp(zW)\rangle_0}\leq\exp(z^2/2)$ with $z=\sqrt{\lambda_0}G_0$; therefore
\begin{align*}
 \E[\mathbb P_0]{\mathcal M_0}
 =&~\E[\mathbb P_0]{\left\langle
       \exp\left(\lambda_0W^2/2\right)\right\rangle_0}
 =\E[G_0\sim N(0,1)]{
      \E[\mathbb P_0]{\left\langle
       \exp\left(\sqrt{\lambda_0}G_0W\right)\right\rangle_0}}
 \leq~\E[G_0\sim N(0,1)]{\exp\left(\lambda_0G_0^2/2\right)}.
\end{align*}
Finally, we evaluate the Gaussian integral using $0<\lambda_0<1$ and obtain
\begin{equation*}
 \E[G_0\sim N(0,1)]{\exp\left(\lambda_0G_0^2/2\right)}
 =\frac1{\sqrt{2\pi}}\int_{\mathbb R}
      \exp\left(-(1-\lambda_0)x^2/2\right)\mathrm d x
 =(1-\lambda_0)^{-1/2}.
\end{equation*}

At the other endpoint, $\lambda_{\beta^2}=\lambda_f$.
We now transport the estimate from $s=0$ to $s=\beta^2$.  We combine monotonicity
of $\Psi$ with \cref{eq:tap-bayes-cancellation} and obtain
$
 \E[\mathbb P_{\beta^2}]{\sqrt{\mathcal M_{\beta^2}}}
 \le \E[\mathbb P_0]{\sqrt{\mathcal M_0}}
$.
We apply Cauchy--Schwarz under $\mathbb P_0$ and the endpoint bound:
\begin{equation*}
 \E[\mathbb P_0]{\sqrt{\mathcal M_0}}
 \leq\E[\mathbb P_0]{\mathcal M_0}^{1/2}
 \leq(1-\beta^2-\lambda_f)^{-1/4}.
\end{equation*}
On the other hand, for fixed observations we apply
Cauchy--Schwarz under the posterior distribution and obtain
\[
 \left\langle \exp\left(\lambda_fn(R_{12}-q)^2/4\right)\right\rangle_{\beta^2}
 \le\sqrt{\mathcal M_{\beta^2}}.
\]
We combine these estimates and obtain
\[
 \E[\mathbb P_{\beta^2}]{\left\langle \exp\left(\lambda_fn(R_{12}-q)^2/4\right)\right\rangle_{\beta^2}}
 \le \E[\mathbb P_{\beta^2}]{\sqrt{\mathcal M_{\beta^2}}}
 \le (1-\beta^2-\lambda_f)^{-1/4}.
\]
We choose $\lambda_f=(1-\beta^2)/2$.  Then
$\lambda_f+\beta^2=(1+\beta^2)/2<1$, the exponent
$\lambda_f/4=(1-\beta^2)/8$ is at least
$\tau_0=(1-\beta_*^2)/8$, and the right side is bounded by $K_0$.
Because the exponent contains a square, decreasing it from
$\lambda_f/4$ to $\tau_0$ can only decrease the integrand.  Zero channel
strengths follow by continuity.

For completeness, given $(A,Y_t)$, the true spin $X$ and a fresh posterior replica have the same law $\mu_{A,Y_t}$.
We draw $\sigma^1$ independently of $X$ given these observations.  Then
\begin{equation*}
 \E[\mathbb P_{n,\beta}^{\mathrm{pl,aug}}]{\langle F(\sigma^1,X)\rangle\mid A,Y_t}
 =\sum_{\sigma,\tau}\mu_{A,Y_t}(\sigma)\mu_{A,Y_t}(\tau)F(\sigma,\tau)
 =\langle F(\sigma^1,\sigma^2)\rangle.
\end{equation*}
We average over $(A,Y_t)$ to obtain the Nishimori identity
\[
 \E[\mathbb P_{n,\beta}^{\mathrm{pl,aug}}]{\langle F(\sigma^1,X)\rangle}
 =\E[\mathbb P_{n,\beta}^{\mathrm{pl,aug}}]{\langle F(\sigma^1,\sigma^2)\rangle}.
\]
This proves the replica--plant version.  Finally, we use the independence of two replicas given $(A,Y_t)$ to compute
$Q_n(m)=\frac1n\sum_jm_j^2=\langle R_{12}\rangle$.
For fixed $(A,Y_t)$, we apply Jensen's inequality to the convex function
$x\mapsto \exp(\tau_0nx^2)$ and obtain
\[
 \exp\left(\tau_0n(Q_n(m)-q)^2\right)
 \le\left\langle \exp\left(\tau_0n(R_{12}-q)^2\right)\right\rangle,
\]
which proves the posterior-mean version.  The independent auxiliary
diagonal changes none of these posterior identities.

Finally, for $z>0$,
\[
 \frac{\mathrm d}{\mathrm dz}\bigl(z^p\exp(-\tau_0z)\bigr)
 =z^{p-1}\exp(-\tau_0z)(p-\tau_0z).
\]
Thus $z^p\exp(-\tau_0z)$ is maximized at $z=p/\tau_0$, where its
value is $(p/(\mathrm e\tau_0))^p$.  The same bound is trivial at $z=0$,
and hence, for every $z\geq0$,
$
 z^p\leq\left(\frac{p}{\mathrm e\tau_0}\right)^p
 \exp(\tau_0z).
$
We apply this with $z=n(R_{12}-q_\beta(t))^2$, use the exponential bound,
and divide by $n^p$.  The replica--plant and posterior-mean moment bounds
follow from their corresponding exponential bounds in the same way.
\end{proof}

\begin{corollary}[Posterior covariance moment bounds]
\label{cor:tap-covariance}
Consider the covariance $\Sigma$ in the fixed-time representation above, with $n\ell\leq\beta_*^2$ and $t\geq0$.
Given $(J,y)$, let $\sigma^1,\ldots,\sigma^4$ be independent samples from $\mu_{J,y}$ and write $\langle\cdot\rangle$ for their expectation.
We define the four-replica quantity $\mathcal D_{1234}$ and state its relation to the covariance as
\begin{equation*}
 \mathcal D_{1234}=\frac14
 \bigl((\sigma^1-\sigma^2)\cdot(\sigma^3-\sigma^4)\bigr)^2,
 \qquad
 \left\langle\mathcal D_{1234}\right\rangle=\operatorname{Tr}\Sigma^2.
\end{equation*}
There is a constant $K>0$, depending only on $\beta_*$, such that the following statements hold, where all $L^p$ norms and outer expectations are taken over $(J,y)$ under the same planted model.
\begin{enumerate}[label=\textup{(\roman*)}]
\item The four-replica moments satisfy
\begin{equation*}
 \E{\langle\mathcal D_{1234}^{p/2}\rangle}
 \leq(K\sqrt{np})^p\qquad(p\geq2).
\end{equation*}

\item Consequently, the covariance trace moments satisfy
\begin{equation*}
 \|\operatorname{Tr} \Sigma^2\|_{L^p}\le Knp\quad(p\ge1).
\end{equation*}
\end{enumerate}
\end{corollary}
\begin{proof}
For fixed observations, we have
$\Sigma_{jk}=\langle(\sigma_j-m_j)(\sigma_k-m_k)\rangle$.  Two independent
replicas satisfy
\[
 \left\langle(\sigma_j^1-\sigma_j^2)
                 (\sigma_k^1-\sigma_k^2)\right\rangle
 =2\langle\sigma_j\sigma_k\rangle-2m_jm_k=2\Sigma_{jk}.
\]
We use the independence of the two replica pairs given the observations to compute
\begin{equation}
 \langle\mathcal D_{1234}\rangle
 =\frac14\sum_{j,k}
 \langle(\sigma_j^1-\sigma_j^2)(\sigma_k^1-\sigma_k^2)\rangle
 \langle(\sigma_j^3-\sigma_j^4)(\sigma_k^3-\sigma_k^4)\rangle
 =\sum_{j,k}\Sigma_{jk}^2=\operatorname{Tr}\Sigma^2.
\label{eq:tap-covariance-rectangle}
\end{equation}
Since $R_{ab}=n^{-1}\sigma^a\cdot\sigma^b$,
$\mathcal D_{1234}=\frac{n^2}{4}
(R_{13}-R_{14}-R_{23}+R_{24})^2$.
We write $q=q_{\beta}(t)$.
The common center $q$ cancels in this alternating sum, so taking the $r$th power gives the first identity below.
We then combine the elementary inequality $|x_1+\cdots+x_4|^r\le4^{r-1}\sum_{\alpha=1}^4|x_\alpha|^r$ and the overlap-moment conclusion of \Cref{lem:tap-overlap-mgf}-\textup{(ii)} to obtain, for $r\ge1$,
\begin{align*}
 \E{\langle\mathcal D_{1234}^{r}\rangle}
 =&~\frac{n^{2r}}{4^r}\E{
 \langle|(R_{13}-q)-(R_{14}-q)
 -(R_{23}-q)+(R_{24}-q)|^{2r}\rangle}\\
 \leq&~\frac{n^{2r}}{4^r}4^{2r-1}
 \sum_{(a,b)\in\{(1,3),(1,4),(2,3),(2,4)\}}
 \E{\langle|R_{ab}-q|^{2r}\rangle}\\
 \leq&~4^rK_0\left(\frac{nr}{\mathrm e\tau_0}\right)^r
 \leq(C n r)^r.
\end{align*}
Taking $r=p/2$ for $p\geq2$ proves part~\textup{(i)}:
\begin{equation*}
 \E{\langle\mathcal D_{1234}^{p/2}\rangle}
 \leq(C\sqrt{np})^p.
\end{equation*}
For part~\textup{(ii)}, we fix the observations, apply Jensen's inequality with exponent $p\geq1$, and then average the observations:
\[
 \E{(\operatorname{Tr}\Sigma^2)^p}
 =\E{\langle\mathcal D_{1234}\rangle^p}
 \le\E{\langle\mathcal D_{1234}^p\rangle},
\]
and the bound above with $r=p$ yields $\|\operatorname{Tr}\Sigma^2\|_{L^p}\le C n p$.
\end{proof}

\subsection{Conditional residual moments}

Recall that for the cavity interpolation of \Cref{def:tap-cavity-pinning}, we fix $i$ and write $a=a_i$, $G=G_i$, and $U=X_iX_{-i}$.
Given $(\mathcal F^{\setminus i},U)$, $\mu_0=\mu^{\setminus i}$ is fixed and the only randomness in $a=\ell U+\sqrt\ell G$ comes from $G$.
The next lemma controls the cancellation between the row contribution and the Onsager correction.
The Gaussian cancellation after conditioning on the couplings outside one row also appears in Chatterjee's TAP analysis \cite[Lemma~1.4 and equation~(44)]{Chatterjee2010}; here we bound the row residual in $L^p$ under the planted model using the Gaussian inequalities in \Cref{lem:tap-gaussian-poincare-p}.

\begin{lemma}
\label{lem:tap-gaussian-cancellation}
For the cavity interpolation of \Cref{def:tap-cavity-pinning}, assume $M=n-1\geq1$ and $0\leq\ell M\leq1$.
We condition on $(\mathcal F^{\setminus i},U)$ and let $x$ be any $[-1,1]$-valued random variable that may depend on $G$ and additional randomness.
Under this conditioning, $\E[\mathrm{row}]{\cdot}$ and $L^p(\mathrm{row})$ average over $G$ and the additional randomness used to construct $x$, while $\E[G\sim N(0,I_M)]{\cdot}$ and $L^p(G)$ average $G$ alone.
We define
\[
 \overline m_x=\tfrac{1+x}{2}m_{+1}+\tfrac{1-x}{2}m_{-1},
\]
\[
r_{\mathrm{row}}
 =\tfrac12(\log Z_{+1}-\log Z_{-1})-a\cdot \overline m_x
       +\ell(M-\norm{\overline m_x}_2^2)x.
\]
Then, for every $p\geq2$, there is a constant $C_p$, depending only on $p$, such that
\begin{equation*}
 \|r_{\mathrm{row}}\|_{L^p(\mathrm{row})}\leq C_p\ell\sup_{|\gamma|\leq1}\norm{\Sigma_\gamma}_{L^p(G)}.
\end{equation*}
\end{lemma}

We restate and prove this lemma in \Cref{subsec:tap-row-proofs}.
Once $r_{\mathrm{row}}$ is controlled, we obtain conditional moment bounds for $r_i$ given $\mathcal F^{\setminus i}$ by bounding the remaining cubic and diagonal terms and averaging over $U$.

\begin{lemma}
\label{lem:tap-complete-row}
For $\beta=\sqrt{n\ell}\leq\beta_*$ and $t\geq0$, the following statements hold.
\begin{enumerate}[label=\textup{(\roman*)}]
\item The conditional second moment of $r_i$ satisfies, almost surely,
\begin{equation*}
 e_i
 \le C\ell^2\bigl(1+\operatorname{Tr}[(\Sigma^{\setminus i})^2]\bigr)+4\ell.
\end{equation*}

\item For every $p\ge1$,
\begin{equation*}
 \sup_i\|e_i\|_{L^p}\le\frac{C\cdot p}{n}.
\end{equation*}

\item For every fixed $p\ge2$,
\begin{equation*}
 \sup_i\|r_i\|_{L^p}\le C_pn^{-1/2}.
\end{equation*}
\end{enumerate}
Here $C$ depends only on $\beta_*$, and $C_p$ depends only on $p,\beta_*$.
\end{lemma}

\begin{proof}
We write $J_{ii}=\ell+\zeta_i$, where $\zeta_i\sim N(0,2\ell)$ is independent of the plant and all other observations.
When $n=1$, we have
$m_1=\tanh y_1$ and $r_1=-\ell m_1^3-\zeta_1m_1$; independence
of $\zeta_1$ gives all three bounds directly.  We now assume $n\geq2$.
We set $M=n-1$, $x=m_i$, and $a=a_i$, and use the cavity interpolation of \Cref{def:tap-cavity-pinning}.
For $b\in\{\pm1\}$, we obtain
\[
 \tfrac{1+bx}{2}=\mu(\sigma_i=b)\propto\exp(by_i)Z_b.
\]
Their ratio is $(1+x)/(1-x)=\exp(2y_i)Z_{+1}/Z_{-1}$; taking half its logarithm, we obtain
\[
 \operatorname{arctanh}\tp{x}-y_i
 =\tfrac12(\log Z_{+1}-\log Z_{-1}).
\]
By the mixture identity in \Cref{def:tap-cavity-pinning}, $\overline m_x=m_{-i}$, with $\overline m_x$ as in \Cref{lem:tap-gaussian-cancellation}.
Substituting $n=M+1$, $\norm{m}_2^2=x^2+\norm{\overline m_x}_2^2$, and $J_{ii}=\ell+\zeta_i$ into the {empirical residual (with $q=Q_n(m)=\|m\|_2^2/n$)} in \cref{eq:tap-prescribed-residual}, we obtain
\[
 r_i
 =\tfrac12(\log Z_{+1}-\log Z_{-1})
   -a\cdot\overline m_x-(\ell+\zeta_i)x
   +\ell(M+1-x^2-\norm{\overline m_x}_2^2)x.
\]
Using $r_{\mathrm{row}}$ from \Cref{lem:tap-gaussian-cancellation}, we obtain
\begin{equation}
 r_i=r_{\mathrm{row}}-\ell x^3-\zeta_i x.
 \label{eq:tap-complete-row-identity}
\end{equation}

For every fixed $p\geq2$, we condition on $(\mathcal F^{\setminus i},U)$, apply \Cref{lem:tap-gaussian-cancellation} with $x=m_i$, and then apply \Cref{lem:tap-weak-tilt} to its covariance supremum.
The dependence of $x$ on $G$ and the removed scalar observation is allowed by \Cref{lem:tap-gaussian-cancellation}; since $\ell M\leq1$ and $\mu_0=\mu^{\setminus i}$, we obtain, almost surely,
\[
 \E[\mathbb P_{n,\beta}^{\mathrm{pl,aug}}]{
 |r_{\mathrm{row}}|^p\mid\mathcal F^{\setminus i},U}
 \leq C_p\ell^p
 \left\langle\mathcal D_{1234}^{p/2}\right\rangle_{(\mu^{\setminus i})^{\otimes4}}.
\]
We first use this estimate at $p=2$, when $\langle\mathcal D_{1234}\rangle_{(\mu^{\setminus i})^{\otimes4}}=\operatorname{Tr}[(\Sigma^{\setminus i})^2]$.

The auxiliary diagonal contributes only a constant to the Hamiltonian, so $x$ and $r_{\mathrm{row}}$ do not depend on $\zeta_i$.
By independence, the cross term involving $\zeta_i$ in the square of \cref{eq:tap-complete-row-identity} has conditional expectation zero.
Using $|x|\le1$ and $\E[\zeta_i\sim N(0,2\ell)]{\zeta_i^2}=2\ell$, we obtain
\begin{align*}
 \E[\mathbb P_{n,\beta}^{\mathrm{pl,aug}}]{r_i^2\mid\mathcal F^{\setminus i},U}
 =&~\E[\mathbb P_{n,\beta}^{\mathrm{pl,aug}}]{(r_{\mathrm{row}}-\ell x^3)^2
                 \mid\mathcal F^{\setminus i},U}
   +\E[\mathbb P_{n,\beta}^{\mathrm{pl,aug}}]{\zeta_i^2x^2\mid\mathcal F^{\setminus i},U}\\
 \le&~2\E[\mathbb P_{n,\beta}^{\mathrm{pl,aug}}]{r_{\mathrm{row}}^2
                 \mid\mathcal F^{\setminus i},U}
      +2\E[\mathbb P_{n,\beta}^{\mathrm{pl,aug}}]{\ell^2 x^6\mid\mathcal F^{\setminus i},U}
      +2\ell\\
 \le&~C\ell^2
       \operatorname{Tr}[(\Sigma^{\setminus i})^2]+2\ell^2+2\ell.
\end{align*}
The upper bound is $\mathcal F^{\setminus i}$-measurable and independent of $U$, so we average over $U$ and use $\ell\leq1$ to obtain
\[
 e_i = \E[\mathbb P_{n,\beta}^{\mathrm{pl,aug}}]{r_i^2\mid\mathcal F^{\setminus i}}\leq C\ell^2\operatorname{Tr}[(\Sigma^{\setminus i})^2]+2\ell^2+2\ell
 \leq C\ell^2\bigl(1+\operatorname{Tr}[(\Sigma^{\setminus i})^2]\bigr)+4\ell.
\]
By \Cref{def:tap-cavity-pinning}, the cavity system has dimension $M$, inverse temperature $\sqrt{M\ell}\leq\beta_*$, and time $t$.
Taking $L^p$ norms, applying \Cref{cor:tap-covariance}-\textup{(ii)} to this system, and using $\ell\leq\beta_*^2/n$, we obtain, for every $p\geq1$,
\[
 \|e_i\|_{L^p}
 \le C\ell^2(1+np)+4\ell
 \le\frac{C\cdot p}{n}.
\]

For fixed $p\geq2$, we average the conditional $r_{\mathrm{row}}$ estimate above and apply \Cref{cor:tap-covariance}-\textup{(i)} to the same cavity system; the factor $\sqrt p$ is absorbed into $C_p$.
We obtain
\[
 \norm{r_{\mathrm{row}}}_{L^p}
 \leq C_p\ell\sqrt n
 \leq C_pn^{-1/2}.
\]
Applying Minkowski to
\cref{eq:tap-complete-row-identity} and using the pointwise bounds
$|x^3|\leq1$ and $|\zeta_i x|\leq|\zeta_i|$, we obtain
\begin{align*}
 \|r_i\|_{L^p}
 &\leq\|r_{\mathrm{row}}\|_{L^p}
       +\ell\|x^3\|_{L^p}+\|\zeta_i x\|_{L^p}\\
 &\leq\|r_{\mathrm{row}}\|_{L^p}+\ell
       +\sqrt{2\ell}\,\E[G\sim N(0,1)]{|G|^p}^{1/p}\\
 &\leq C_p\ell\sqrt n+\ell+C_p\sqrt\ell
 \leq C_pn^{-1/2},
\end{align*}
where the last step uses $\ell\leq\beta_*^2/n$.
\end{proof}

The $L^p$ bound for each $e_i$ in \Cref{lem:tap-complete-row}-\textup{(ii)} grows linearly in $p$; summing these bounds yields an exponential moment for $\mathcal E$.

\begin{corollary}[Exponential moment of the sum of conditional second moments]
\label{cor:tap-energy-mgf}
Under the assumptions of \Cref{lem:tap-complete-row}, there are positive constants $C,\kappa$, depending only on $\beta_*$, such that
\begin{equation*}
 \|\mathcal E\|_{L^p}\le C\cdot p\quad(p\ge1),\qquad
 \E[\mathbb P_{n,\beta}^{\mathrm{pl,aug}}]{\exp(\kappa\mathcal E)}\le 2.
\end{equation*}
\end{corollary}

\begin{proof}
We use the moment bound for each $e_i$ in
\Cref{lem:tap-complete-row}-\textup{(ii)} and Minkowski's inequality to obtain
\[
 \|\mathcal E\|_{L^p}
 =\left\|\sum_i e_i\right\|_{L^p}
 \le\sum_i\|e_i\|_{L^p}\le C\cdot p.
\]
We choose $0<\kappa\le(2C\mathrm e)^{-1}$ to convert the linear moment bound into an exponential moment.
Since $\mathcal E\ge0$, its exponential series has nonnegative terms, so we apply Tonelli's theorem and $p!\ge(p/\mathrm e)^p$ to obtain
\[
 \E[\mathbb P_{n,\beta}^{\mathrm{pl,aug}}]{\exp(\kappa\mathcal E)}
 =1+\sum_{p=1}^\infty
        \frac{\kappa^p\E[\mathbb P_{n,\beta}^{\mathrm{pl,aug}}]{\mathcal E^p}}{p!}
 \le 1+\sum_{p=1}^\infty\frac{(\kappa C\cdot p)^p}{p!}
 \le 1+\sum_{p=1}^\infty(\kappa C\mathrm e)^p\le 2.
\]
\end{proof}

\subsection{From conditional moments to residual tails}

By \Cref{cor:tap-energy-mgf}, $\mathcal E=\sum_i e_i$ has an exponential moment, while the desired tail bound concerns $\norm{\mathbf r}_2^2=\sum_i r_i^2$.
For any $\mathcal F^{\setminus i}$-measurable event $A$, the definition $e_i=\E{r_i^2\mid\mathcal F^{\setminus i}}$ gives $\E{r_i^2\1_A}=\E{e_i\1_A}$.
However, we cannot apply this identity with the tail event for $\norm{\mathbf r}_2$, since it also depends on the deleted data.
In contrast, the tail event for $\|\mathbf r^{\setminus i}\|_2$ depends only on the surviving observations, so we use it in the conditional expectation and then return to the full-norm event.
The following lemma bounds the norm difference needed for these event replacements.

\begin{lemma}
\label{lem:tap-parent-child}
Under the planted model of \Cref{def:tap-planted-model}, set $J=\beta A^{\mathrm{full}}$.
If $n\geq2$ and $n\ell\leq\beta_*^2$, let $\mathbf r$ and
$\mathbf r^{\setminus i}$ be the residual before and after
deleting coordinate $i$.  Each $\mathbf r^{\setminus i}$ is
$\mathcal F^{\setminus i}$-measurable.  For every fixed integer $p\geq1$,
there is a constant $C_p$, depending only on $p,\beta_*,T$, such that,
for every $t\in[0,T]$,
\begin{equation*}
 \max_i\bigl\|\,\|\mathbf r\|_2-\|\mathbf r^{\setminus i}\|_2\,\bigr\|_{L^{4p}}\le C_pn^{-1/4}.
\end{equation*}
\end{lemma}

We restate and prove this lemma in \Cref{subsec:tap-stability-proofs}.
We now combine it with the conditional-moment bounds to prove \Cref{thm:polynomial-tap-residual}.

\begin{proof}[Proof of \Cref{thm:polynomial-tap-residual}]
We fix $H>0$.
At $\beta=0$, the posterior distribution factors over coordinates, so $m(A,Y_t)=\tanh(Y_t)$ and both TAP residuals vanish:
\[
 \mathbf r_t^{(q)}=\mathbf r_t^{(Q)}
 =\operatorname{arctanh}(m(A,Y_t))-Y_t=0.
\]
We now fix $0<\beta\leq\beta_*$ and $t\in(0,T]$, and take $\mathbf r=\mathbf r_t^{(Q)}$.
We write $\mathbb P=\mathbb P_{n,\beta}^{\mathrm{pl,aug}}$ for the outer law in this proof.
By \Cref{cor:tap-energy-mgf}, we have $\E[\mathbb P]{\exp(\kappa\mathcal E)}\le 2$ for some $\kappa>0$ depending only on $\beta_*$.
We apply \Cref{lem:tap-complete-row}-\textup{(iii)} with $p=4$ to obtain $\sum_i\|r_i\|_{L^4}^2\le C_4^2$.
By \Cref{lem:tap-parent-child}, for every fixed integer $p\geq1$, the $L^{4p}$ norm of $\|\mathbf r\|_2-\|\mathbf r^{\setminus i}\|_2$ is at most $C_pn^{-1/4}$ for every $i$, with $C_p$ depending only on $p,\beta_*,T$.
Here $\mathbf r^{\setminus i}$ and $e_i$ use the same sigma-field $\mathcal F^{\setminus i}$.
We first prove the tail bound for $\mathbf r_t^{(Q)}$ with failure probability at most $n^{-(H+1)}$; the extra power of $n^{-1}$ leaves room for the final union bound with the overlap error.

\emph{Step 1: transfer the tail conditionally.}
We set $\eta_n=n^{-1/8}$ and define the events
\[
 E_i^{\mathrm{del}}=\{\bigl|\,\|\mathbf r\|_2-\|\mathbf r^{\setminus i}\|_2\,\bigr|>\eta_n\}.
\]
Outside $E_i^{\mathrm{del}}$, the two norms differ by at most $\eta_n$.
We set $\epsilon_n\defeq\sum_i\E[\mathbb P]{(r_i^2+e_i)\1_{E_i^{\mathrm{del}}}}$.
For every $s>0$, we obtain
\begin{align}
 s^2\mathbb P(\norm{\mathbf r}_2\ge s)
 \le&~\sum_i\E[\mathbb P]{r_i^2\1_{\{\norm{\mathbf r}_2\ge s\}}}\notag\\
 \le&~\sum_i\E[\mathbb P]{r_i^2
            \1_{\{\|\mathbf r^{\setminus i}\|_2\ge s-\eta_n\}}}
                  +\sum_i\E[\mathbb P]{r_i^2\1_{E_i^{\mathrm{del}}}}\notag\\
 =&~\sum_i\E[\mathbb P]{e_i
            \1_{\{\|\mathbf r^{\setminus i}\|_2\ge s-\eta_n\}}}
                  +\sum_i\E[\mathbb P]{r_i^2\1_{E_i^{\mathrm{del}}}}\notag\\
 \le&~\sum_i\E[\mathbb P]{e_i\1_{\{\norm{\mathbf r}_2\ge s-2\eta_n\}}}
       +\sum_i\E[\mathbb P]{(r_i^2+e_i)\1_{E_i^{\mathrm{del}}}}\notag\\
 =&~\E[\mathbb P]{\mathcal E\,
            \1_{\{\norm{\mathbf r}_2\ge s-2\eta_n\}}}+\epsilon_n.
 \label{eq:tap-tail-transfer}
\end{align}
For the first equality in \cref{eq:tap-tail-transfer}, $\{\|\mathbf r^{\setminus i}\|_2\ge s-\eta_n\}\in\mathcal F^{\setminus i}$ and $\E[\mathbb P]{r_i^2\mid\mathcal F^{\setminus i}}=e_i$; thus conditioning leaves the indicator unchanged and replaces $r_i^2$ by $e_i$.

\smallskip
\noindent\emph{Step 2: control the deletion error.}
By \Cref{lem:tap-parent-child} and Markov's inequality, we obtain
\begin{equation*}
 \max_i\mathbb P(E_i^{\mathrm{del}})\leq\eta_n^{-4p}C_p^{4p}n^{-p}=C_p'n^{-p/2}.
\end{equation*}
We condition on $\mathcal F^{\setminus i}$ and apply Jensen's inequality to obtain $\E[\mathbb P]{e_i^2}\leq\E[\mathbb P]{r_i^4}$, so $\|e_i\|_{L^2}\leq\|r_i\|_{L^4}^2$.
We apply Cauchy--Schwarz and the fourth-moment bound above to obtain
\begin{align*}
 \epsilon_n \le&~\sum_i\bigl(\|r_i\|_{L^4}^2+\|e_i\|_{L^2}\bigr)
             \mathbb P(E_i^{\mathrm{del}})^{1/2}\\
 \le&~2\left(\sum_i\|r_i\|_{L^4}^2\right)
             \max_i\mathbb P(E_i^{\mathrm{del}})^{1/2}\\
 \le&~2C_4^2\max_i\mathbb P(E_i^{\mathrm{del}})^{1/2}
 \le C_p''n^{-p/4}.
\end{align*}
We choose a fixed integer $p>4(H+5)$ before choosing the lower bound on $n$.
Then $\epsilon_n\le n^{-(H+4)}$ for all sufficiently large $n$.

\smallskip
\noindent\emph{Step 3: truncate and iterate the threshold.}
By \Cref{cor:tap-energy-mgf} and Markov's inequality, we have $\mathbb P(\mathcal E>x)\le 2\exp(-\kappa x)$ for every $x\geq0$.
We integrate this tail bound to obtain
\begin{align}
 \E[\mathbb P]{\mathcal E\,\1_{\{\mathcal E>a\}}}
 =&~a\mathbb P(\mathcal E>a)
       +\int_a^\infty\mathbb P(\mathcal E>x)\d x \le 2(a+\kappa^{-1})\exp(-\kappa a).
 \label{eq:tap-energy-cutoff}
\end{align}
We take $a=(H+5)\kappa^{-1}\log n$.
The last expression is at most $n^{-(H+4)}$ for large $n$.
We split $\mathcal E$ at $a$ to bound the expectation in \cref{eq:tap-tail-transfer} by
\begin{align*}
 \E[\mathbb P]{\mathcal E\1_{\{\norm{\mathbf r}_2\ge s-2\eta_n\}}}
 \le&~ a\mathbb P(\norm{\mathbf r}_2\ge s-2\eta_n)
       +\E[\mathbb P]{\mathcal E\1_{\{\mathcal E>a\}}}.
\end{align*}
We bound $\E[\mathbb P]{\mathcal E\1_{\{\mathcal E>a\}}}$ by
\cref{eq:tap-energy-cutoff} and substitute this estimate and
$\epsilon_n\le n^{-(H+4)}$ into \cref{eq:tap-tail-transfer} to obtain
\begin{equation}
 s^2\mathbb P(\norm{\mathbf r}_2\ge s)
 \le a\mathbb P(\norm{\mathbf r}_2\ge s-2\eta_n)+2n^{-(H+4)}.
 \label{eq:tap-threshold-recursion}
\end{equation}

To iterate \cref{eq:tap-threshold-recursion}, we set
\[
 s_0=\sqrt{2a},\qquad s_j=s_0+2j\eta_n,\qquad
 L=\lceil(H+3)\log_2n\rceil.
\]
For $j\ge1$, we divide \cref{eq:tap-threshold-recursion} by $s_j^2\ge2a$
and use $s_j-2\eta_n=s_{j-1}$ to obtain
\[
 \mathbb P(\norm{\mathbf r}_2\ge s_j)
 \le\tfrac12\mathbb P(\norm{\mathbf r}_2\ge s_{j-1})+n^{-(H+4)}/a.
\]
Starting from $\mathbb P(\norm{\mathbf r}_2\geq s_0)\leq1$, we iterate $L$ times
and sum the geometric error terms to obtain
\[
 \mathbb P(\norm{\mathbf r}_2\ge s_L)
 \le2^{-L}+2n^{-(H+4)}/a\le n^{-(H+1)}
\]
for large $n$.  Since $L\eta_n=O_H(n^{-1/8}\log n)=o(1)$,
$s_L^2\le4a$ eventually, so
\begin{equation*}
 \mathbb P\left\{\norm{\mathbf r_t^{(Q)}}_2^2>4(H+5)\kappa^{-1}\log n\right\}\le n^{-(H+1)}.
\end{equation*}

\smallskip
\noindent\emph{Step 4: restore the prescribed overlap.}
We write $Q=Q_n(m_t)=\langle R_{12}\rangle$.
{We replace this empirical overlap $Q=\|m_t\|_2^2/n$ by the deterministic prescribed overlap $q_\beta(t)$ from \cref{eq:compact-rs-fixed-point}.}
The identity $\mathbf r_t^{(q)}-\mathbf r_t^{(Q)}=\beta^2(Q-q_\beta(t))m_t$ and $\|m_t\|_2^2\le n$ yield
\begin{equation}
 \norm{\mathbf r_t^{(q)}}_2^2
    \le 2\norm{\mathbf r_t^{(Q)}}_2^2
          +2\beta^4n(Q-q_\beta(t))^2.
 \label{eq:tap-restore-q}
\end{equation}
We apply the posterior-mean conclusion of \Cref{lem:tap-overlap-mgf} at dimension $n$, inverse temperature $\beta$, and time $t$ to obtain
\[
 \E[\mathbb P]{\exp(\tau_0n(Q-q_\beta(t))^2)}\le K_0.
\]
By Markov's inequality applied to this exponential, we have
\[
 \mathbb P\left\{n(Q-q_\beta(t))^2>
               \frac{H+2}{\tau_0}\log n\right\}
 \le K_0n^{-(H+2)}\le n^{-(H+1)}
\]
for large $n$.
We combine the two tail bounds with \cref{eq:tap-restore-q} and choose $C_H$ to absorb the two logarithmic thresholds.
The union bound is $2n^{-(H+1)}\le n^{-H}$ for $n\ge2$.

The constants in the overlap and conditional second-moment bounds depend only on $\beta_*$, not on $\beta\in[0,\beta_*]$ or $t\geq0$.
Once $H$ is fixed, all moment orders used to bound the residual coordinates and compare the cavity, pinned, and full residuals depend only on $H$.
The inverse-marginal estimates used in the deletion comparison are evaluated only at $t$ and $t+\ell\leq T+1$, while deleting coordinates keeps $\ell=\beta^2/n$ unchanged.
Thus the remaining stability constants depend only on $\beta_*$ and $T$.
We can therefore choose $C_H,n_H$ independently of $\beta,t$ and take both suprema in the theorem.
At $t=0$, we have $m_t=0$ by spin-flip symmetry, so the residual vanishes.
\end{proof}

\endgroup

\section{Lipschitz property with cost functions}
\label{sec:app-inverse-poly-geometry-proof}

In this section, we prove \Cref{lem:app-wpi-branches}.
{For a fixed constant $T>0$, we use the original augmented law $\mathbb P_{0,n,\beta}^{\mathrm{aug}}$ from \Cref{def:original-matrix-law} and consider the restriction of $Y$ to $[0,T]$.
We write $A^{\mathrm{full}}=A+\operatorname{diag}(d)$ and $m_t=m(A,Y_t)$.}
Intuitively, to prove the above lemma, we need to approximate the posterior mean $m(A, Y_t)$.
{To do this, we use approximate message passing (AMP) and the TAP free energy.}
We follow the use of AMP to approximate posterior means in sampling
algorithms for the SK model~\cite{ElAlaouiMontanariSellke25},
together with the local convexity analysis of the TAP free energy
around AMP iterates~\cite{Celentano24}.
{AMP locates a region where the TAP free energy is strongly convex.}
{We show that both $m(A,Y_t)$ and a local TAP stationary point lie in this region; the curvature lower bound then controls their distance by the TAP residual at $m(A,Y_t)$.}
{For each fixed $(\beta,t)$ with $t>0$ and every $H>0$, we establish these geometric conditions except with probability at most $C_Hn^{-H}$.}

We first state the AMP good events and show how they imply
\Cref{lem:app-wpi-branches}.  Their probability estimates are proved
separately in \Cref{sec:amp-good-events-proof}.

\begin{definition}[{AMP and TAP conditions}]
\label{def:amp-tap-geometry}
Fix $0<\beta<1$, $t>0$, and an integer $K\geq0$.
Let $A^{\mathrm{full}}$ be a real symmetric matrix.
{Let $A$ be its zero-diagonal part.}
{We use the depth-$K$ AMP output $\widehat m_K(y)$, computed with matrix $A^{\mathrm{full}}$, external field $y$, and inverse temperature $\beta$, as the center of the balls below.}
{We use the following conditions on $(A^{\mathrm{full}},y)$.}
\begin{enumerate}[label=\textup{(\roman*)}]
\item For $\epsilon>0$, \emph{$\epsilon$-small TAP residual} means
\begin{equation*}
 \norm{R_{\mathrm{TAP}}\tp{\widehat m_K(y);A^{\mathrm{full}},y,q_\beta(t)}}_2
 \leq\epsilon\sqrt n.
\end{equation*}
\item For $r,\kappa>0$, \emph{$(r,\kappa)$-local strong convexity}
means that, for every $u\in B(\widehat m_K(y),r\sqrt n)\cap(-1,1)^n$,
\begin{equation*}
 \nabla_u^2F_{\mathrm{TAP}}\tp{u;A^{\mathrm{full}},y,q_\beta(t)}
 \succeq\kappa I_n.
\end{equation*}
\item For $r>0$, \emph{$r$-posterior-mean proximity} means
\begin{equation*}
 \norm{m(A,y)-\widehat m_K(y)}_2\leq r\sqrt n.
\end{equation*}
\end{enumerate}
\end{definition}

The next lemma collects the probability estimates needed to prove
\Cref{lem:app-wpi-branches}.  At small times we use strong convexity
on the whole cube.  At large times we use the three properties in
\Cref{def:amp-tap-geometry}, with the same AMP depth and geometric
constants throughout each parameter rectangle.

\begin{lemma}[Probability of the AMP good events]
\label{lem:amp-good-events}
{Fix $0<\beta_0<\beta_1<1$ and a constant $T>0$.}
There exist $t_0\in(0,T]$, $\kappa_0>0$, a finite collection of
closed rectangles $I_1,\ldots,I_{\ell_{\max}}$ covering
$[\beta_0,\beta_1]\times[t_0,T]$, and constants
$K_\ell\in\mathbb Z_{\geq0}$ and $r_\ell,\kappa_\ell>0$.
For every $H>0$, there are $C_H,n_H$ such that, for all $n\geq n_H$,
the following hold.
\begin{enumerate}[label=\textup{(\roman*)}]
\item For every $\beta\in[\beta_0,\beta_1]$, with
$\mathbb P_{n,\beta}^{\mathrm{pl,aug}}$-probability at least
$1-C_Hn^{-H}$, we have $\norm{A^{\mathrm{full}}}_{\mathrm{op}}\leq4$
and
\begin{equation*}
 I_n-\beta A^{\mathrm{full}}+\beta^2(1-q_\beta(t))I_n
 \succeq4\kappa_0I_n
 \qquad\text{for every }t\in[0,t_0].
\end{equation*}
\item For every $\ell$ and every
$(\beta,t)\in I_\ell\cap([\beta_0,\beta_1]\times[t_0,T])$, with
$\mathbb P_{n,\beta}^{\mathrm{pl,aug}}$-probability at least
$1-C_Hn^{-H}$, the pair $(A^{\mathrm{full}},Y_t)$ satisfies
$(\kappa_\ell r_\ell/8)$-small TAP residual,
$(4r_\ell,4\kappa_\ell)$-local strong convexity, and
$r_\ell$-posterior-mean proximity at inverse temperature $\beta$,
time $t$, and AMP depth $K_\ell$.
On the same event,
$\norm{A^{\mathrm{full}}}_{\mathrm{op}}\leq4$ and
$\norm{Y_t}_\infty\leq n^2$.
\end{enumerate}
\end{lemma}

We prove \Cref{lem:amp-good-events} in
\Cref{sec:amp-good-events-proof}.  We next show how it
implies \Cref{lem:app-wpi-branches}, using the residual estimate from
\Cref{thm:polynomial-tap-residual}.

We first bound the change of the AMP center with the field following the proof of \cite[Lemma~5.17]{ElAlaouiMontanariSellke25}, then transfer this bound to the local TAP stationary point.
\begin{lemma}
\label{lem:amp-spatial-lipschitz}
For every fixed integer $K\geq0$, there is a constant $L_K$ such that,
for every $n$, $0<\beta<1$, and symmetric matrix $A^{\mathrm{full}}$
with $\norm{A^{\mathrm{full}}}_{\mathrm{op}}\leq4$,
\begin{equation*}
 \norm{\widehat m_K(y)-\widehat m_K(z)}_2
 \leq L_K\norm{y-z}_2
 \qquad(y,z\in\R^n).
\end{equation*}
\end{lemma}

\begin{proof}
Fix $y,z\in\R^n$.  We compare the two AMP recursions and prove the
Lipschitz property by induction on their depth.
Recall that $\widehat m_k(y)=\tanh(z_k(y))$, with
$\widehat m_{-1}(y)=z_0(y)=0$ and
\begin{equation*}
 z_{k+1}(y)=\beta A^{\mathrm{full}}\widehat m_k(y)+y
 -\beta^2\left(1-\frac{\norm{\widehat m_k(y)}_2^2}{n}\right)
   \widehat m_{k-1}(y).
\end{equation*}
The recursion at $z$ has the same matrix and inverse temperature.
Subtracting the two updates, and adding and subtracting the Onsager
coefficient at $y$ times $\widehat m_{k-1}(z)$, gives
\begin{align*}
 z_{k+1}(y)-z_{k+1}(z)
 =&~\beta A^{\mathrm{full}}(\widehat m_k(y)-\widehat m_k(z))+(y-z)
 ~-\beta^2\left(1-\frac{\norm{\widehat m_k(y)}_2^2}{n}\right)
       (\widehat m_{k-1}(y)-\widehat m_{k-1}(z))\\
 &~+\frac{\beta^2}{n}
       (\norm{\widehat m_k(y)}_2^2-\norm{\widehat m_k(z)}_2^2)
       \widehat m_{k-1}(z).
\end{align*}
We bound the three terms other than $y-z$ separately.
The matrix term satisfies
\begin{equation*}
 \norm{\beta A^{\mathrm{full}}(\widehat m_k(y)-\widehat m_k(z))}_2
 \leq4\beta\norm{\widehat m_k(y)-\widehat m_k(z)}_2.
\end{equation*}
Since every coordinate of $\widehat m_k(y)$ lies in $[-1,1]$,
$0\leq1-\norm{\widehat m_k(y)}_2^2/n\leq1$.  Hence
\begin{align*}
 \left\|\beta^2\left(1-\frac{\norm{\widehat m_k(y)}_2^2}{n}\right)
       (\widehat m_{k-1}(y)-\widehat m_{k-1}(z))\right\|_2
 \leq \beta^2\norm{\widehat m_{k-1}(y)-\widehat m_{k-1}(z)}_2.
\end{align*}
For the last term, the difference of squared norms satisfies
\begin{align*}
 \abs{\norm{\widehat m_k(y)}_2^2-\norm{\widehat m_k(z)}_2^2}
 =&~\abs{\inner{\widehat m_k(y)+\widehat m_k(z)}
                   {\widehat m_k(y)-\widehat m_k(z)}}\\
 \leq&~(\norm{\widehat m_k(y)}_2+\norm{\widehat m_k(z)}_2)
          \norm{\widehat m_k(y)-\widehat m_k(z)}_2\\
 \leq&~2\sqrt n\norm{\widehat m_k(y)-\widehat m_k(z)}_2.
\end{align*}
Together with $\norm{\widehat m_{k-1}(z)}_2\leq\sqrt n$, this gives
\begin{align*}
 \left\|\frac{\beta^2}{n}
       (\norm{\widehat m_k(y)}_2^2-\norm{\widehat m_k(z)}_2^2)
       \widehat m_{k-1}(z)\right\|_2
 \leq \frac{\beta^2}{n}\,2\sqrt n
       \norm{\widehat m_k(y)-\widehat m_k(z)}_2\,\sqrt n
 =2\beta^2\norm{\widehat m_k(y)-\widehat m_k(z)}_2.
\end{align*}
Taking norms in the update difference and using these three bounds,
we obtain
\begin{align*}
 \norm{z_{k+1}(y)-z_{k+1}(z)}_2
 \leq&~(4\beta+2\beta^2)\norm{\widehat m_k(y)-\widehat m_k(z)}_2
       +\beta^2\norm{\widehat m_{k-1}(y)-\widehat m_{k-1}(z)}_2
       +\norm{y-z}_2\\
 \leq&~6\norm{\widehat m_k(y)-\widehat m_k(z)}_2
       +\norm{\widehat m_{k-1}(y)-\widehat m_{k-1}(z)}_2
       +\norm{y-z}_2.
\end{align*}
The last inequality uses $0<\beta<1$.

Set $L_{-1}=L_0=0$ and $L_{k+1}=6L_k+L_{k-1}+1$.
We claim that, for every $k\geq-1$,
$\norm{\widehat m_k(y)-\widehat m_k(z)}_2\leq L_k\norm{y-z}_2$.
For $k=-1,0$, both AMP outputs are zero, so the claim holds.
Suppose it holds at depths $k-1$ and $k$.
Since $\tanh$ is one-Lipschitz, the preceding recursion implies
\begin{align*}
 \norm{\widehat m_{k+1}(y)-\widehat m_{k+1}(z)}_2
 \leq&~\norm{z_{k+1}(y)-z_{k+1}(z)}_2
 \leq 6L_k\norm{y-z}_2+L_{k-1}\norm{y-z}_2+\norm{y-z}_2\\
 =&~L_{k+1}\norm{y-z}_2.
\end{align*}
This proves the claim at depth $k+1$.  Taking $k=K$ gives the
asserted bound, with $L_K$ depending only on $K$.
\end{proof}

{The small-residual and local strong-convexity conditions in \Cref{def:amp-tap-geometry} give a local TAP stationary point.
We combine them with the Lipschitz property of the AMP center to compare stationary points at different fields.}
\begin{lemma}
\label{lem:tap-stationary-point}
Fix $0<\beta<1$, $t>0$, an integer $K\geq0$, a symmetric matrix
$A^{\mathrm{full}}$, and $r,\kappa>0$.  Write
$F_y(u)=F_{\mathrm{TAP}}(u;A^{\mathrm{full}},y,q_\beta(t))$.
We have that
\begin{enumerate}[label=\textup{(\roman*)}]
\item Fix $y\in\R^n$.  Suppose $(A^{\mathrm{full}},y)$ satisfies
$(\kappa r/8)$-small TAP residual and
$(4r,4\kappa)$-local strong convexity
at time $t$ and AMP depth $K$, as defined in
\Cref{def:amp-tap-geometry}.
Then $F_y$ has a unique stationary point $m^\star$ in
$B(\widehat m_K(y),4r\sqrt n)\cap(-1,1)^n$,
with $\norm{m^\star-\widehat m_K(y)}_2\leq r\sqrt n/32$.
For every $u\in B(\widehat m_K(y),4r\sqrt n)\cap(-1,1)^n$,
\begin{equation*}
 \norm{u-m^\star}_2\leq(4\kappa)^{-1}\norm{\nabla F_y(u)}_2.
\end{equation*}
\item Let $D\subset\R^n$ and suppose the AMP center
$y\mapsto\widehat m_K(y)$ is $L$-Lipschitz on $D$.
Suppose the two properties in part~\textup{(i)} hold for every $y\in D$.
Write $m^\star(y)$ for its unique stationary point.
Then, for every $y,z\in D$,
\begin{equation*}
 \norm{m^\star(y)-m^\star(z)}_2
 \leq\max((4\kappa)^{-1},\,2L/(3r))\norm{y-z}_2.
\end{equation*}
\end{enumerate}
\end{lemma}

\begin{proof}
For part~\textup{(i)}, $F_y$ extends continuously to $[-1,1]^n$,
using $0\log0=0$ in $h_{\mathrm{bin}}$.
It attains its minimum on the nonempty compact set
$\overline B(\widehat m_K(y),4r\sqrt n)\cap[-1,1]^n$.
Let $u$ be a minimizer.  {We first show that $u$ lies in the interior of the cube $[-1,1]^n$.}
To the contrary, suppose $u_i=1$ for some $i$.  Decreasing only this coordinate by
$\epsilon>0$ keeps the point in the cube, and
\begin{equation*}
 \norm{u-\epsilon e_i-\widehat m_K(y)}_2^2-\norm{u-\widehat m_K(y)}_2^2
 =-2\epsilon(1-\widehat m_{K,i}(y))+\epsilon^2<0{.}
\end{equation*}
Here $0<\epsilon<\min(2,2(1-\widehat m_{K,i}(y)))$, so the perturbed point is feasible.
Keeping all other coordinates fixed, the derivative with respect
to the $i$th coordinate, at its value $s\in(-1,1)$, is
\begin{equation*}
 \operatorname{arctanh}(s)-y_i
 -\beta\sum_{j\ne i}A^{\mathrm{full}}_{ij}u_j
 +(\beta^2(1-q_\beta(t))-\beta A^{\mathrm{full}}_{ii})s.
\end{equation*}
This tends to $+\infty$ as $s\uparrow1$, so the one-variable
function is strictly increasing sufficiently near $1$.
By continuity at $1$, replacing $u_i=1$ by $1-\epsilon$ strictly
decreases $F_y$ for sufficiently small $\epsilon$, a contradiction.
If $u_i=-1$, increasing it by $\epsilon$ changes the squared distance
to $\widehat m_K(y)$ by $-2\epsilon(1+\widehat m_{K,i}(y))+\epsilon^2<0$ for small $\epsilon$.
The same derivative tends to $-\infty$ as $s\downarrow-1$, so this
feasible move also strictly decreases $F_y$.
Hence no coordinate of a minimizer is on the cube boundary.

Next, we show that $u$ is not on the boundary of $B(\widehat m_K(y),4r\sqrt n)$.
Suppose $\norm{u-\widehat m_K(y)}_2=4r\sqrt n$.
For $0<s<1$, the segment from $\widehat m_K(y)$ to $\widehat m_K(y)+s(u-\widehat m_K(y))$ lies in
$B(\widehat m_K(y),4r\sqrt n)\cap(-1,1)^n$, so the Hessian bound gives
\begin{equation*}
 F_y(\widehat m_K(y)+s(u-\widehat m_K(y)))-F_y(\widehat m_K(y))
 \geq s\inner{\nabla F_y(\widehat m_K(y))}{u-\widehat m_K(y)}
       +2\kappa s^2\norm{u-\widehat m_K(y)}_2^2.
\end{equation*}
Letting $s\uparrow1$ and using continuity, we obtain
\begin{align*}
 F_y(u)-F_y(\widehat m_K(y))
 \geq&~\inner{\nabla F_y(\widehat m_K(y))}{u-\widehat m_K(y)}+2\kappa\norm{u-\widehat m_K(y)}_2^2\\
 \geq&~-\frac{\kappa r\sqrt n}{8}(4r\sqrt n)
          +2\kappa(4r\sqrt n)^2>0.
\end{align*}
Since $\widehat m_K(y)$ is feasible, this contradicts the assumption that $u$ is the minimizer.
Thus a minimizer lies in $B(\widehat m_K(y),4r\sqrt n)\cap(-1,1)^n$.
{We denote it by $m^\star$; then} $\nabla F_y(m^\star)=0$.

For every $u\in B(\widehat m_K(y),4r\sqrt n)\cap(-1,1)^n$, integrating the Hessian bound
along the segment from $m^\star$ to $u$ and using the stationarity of $m^\star$ gives
\begin{equation*}
 4\kappa\norm{u-m^\star}_2^2
 \leq\inner{\nabla F_y(u)-\nabla F_y(m^\star)}{u-m^\star}
 =\inner{\nabla F_y(u)}{u-m^\star}
 \leq\norm{\nabla F_y(u)}_2\norm{u-m^\star}_2.
\end{equation*}
Taking $u=\widehat m_K(y)$ gives
$\norm{m^\star-\widehat m_K(y)}_2\leq(4\kappa)^{-1}\norm{\nabla F_y(\widehat m_K(y))}_2
\leq r\sqrt n/32$.

For part~\textup{(ii)}, fix $y,z\in D$.
If $\norm{\widehat m_K(y)-\widehat m_K(z)}_2\leq3r\sqrt n$,
both stationary points and their connecting segment lie in
$B(\widehat m_K(y),4r\sqrt n)\cap(-1,1)^n$.
{Since $\nabla F_y(m^\star(y))=0$ and $\nabla F_y(m^\star(z))=z-y$, the $4\kappa$-strong monotonicity of $\nabla F_y$ on this ball gives
\begin{equation*}
4\kappa\norm{m^\star(y)-m^\star(z)}_2^2
\leq\inner{y-z}{m^\star(y)-m^\star(z)}
\leq\norm{y-z}_2\norm{m^\star(y)-m^\star(z)}_2.
\end{equation*}
Thus $4\kappa\norm{m^\star(y)-m^\star(z)}_2\leq\norm{y-z}_2$.}

Otherwise, $\norm{\widehat m_K(y)-\widehat m_K(z)}_2>3r\sqrt n$.
Since $y\mapsto\widehat m_K(y)$ is $L$-Lipschitz,
we have
\begin{equation*}
3r\sqrt n<\norm{\widehat m_K(y)-\widehat m_K(z)}_2\leq L\norm{y-z}_2.
\end{equation*}
Both stationary points belong to $(-1,1)^n$, so each coordinate
of their difference has absolute value at most two.  Hence
$\norm{m^\star(y)-m^\star(z)}_2 \leq2\sqrt n \leq\frac{2L}{3r}\norm{y-z}_2$.
Then part~(ii) is proved by combining the above two cases.
\end{proof}

\subsection{Local stationary points and their path and terminal costs}
\label{subsec:certified-local-roots}

We fix $0<\beta_0<\beta_1<1$ and a constant $T>0$.
Take $t_0,\kappa_0$, the rectangles $I_\ell$, and the constants
$K_\ell,r_\ell,\kappa_\ell$ from \Cref{lem:amp-good-events}.  Fix
$\beta\in[\beta_0,\beta_1]$.  For $t>t_0$, let $\ell_t$ be the least
index such that $(\beta,t)\in I_{\ell_t}$, and write
$\widehat m_\ell(A^{\mathrm{full}},y)$ for the depth-$K_\ell$ AMP
center.

For $t>t_0$, let $\mathcal D_t$ be the set of pairs
$(A^{\mathrm{full}},y)$ satisfying
\begin{equation*}
 \|A^{\mathrm{full}}\|_{\mathrm{op}}\le4,
 \qquad \|y\|_\infty\le n^2,
\end{equation*}
together with $(\kappa_{\ell_t}r_{\ell_t}/8)$-small TAP residual and $(4r_{\ell_t},4\kappa_{\ell_t})$-local strong convexity at the depth $K_{\ell_t}$ AMP center.
{For every pair in $\mathcal D_t$, \Cref{lem:tap-stationary-point}-\textup{(i)} supplies a unique stationary point in $B(\widehat m_{\ell_t}(A^{\mathrm{full}},y),4r_{\ell_t}\sqrt n)\cap(-1,1)^n$; we denote it by $a_t(A^{\mathrm{full}},y)$.}
For $t\le t_0$, let $\mathcal D_t$ consist of all pairs $(A^{\mathrm{full}},y)$ for which
\begin{equation*}
 \|A^{\mathrm{full}}\|_{\mathrm{op}}\le4,
 \qquad
 I_n-\beta A^{\mathrm{full}}+\beta^2(1-q_\beta(t))I_n
 \succeq4\kappa_0I_n.
\end{equation*}
The condition is independent of $y$.  On this domain the TAP free energy
is strongly convex on the whole cube, and $a_t(A^{\mathrm{full}},y)$
denotes its unique stationary point.
For fixed $A^{\mathrm{full}}$ and $t$, we call a field $y$
\emph{certified} if $(A^{\mathrm{full}},y)\in\mathcal D_t$.

\begin{lemma}
\label{lem:local-root-measurable-lipschitz}
There is a constant $L$, depending only on $(\beta_0,\beta_1,T)$, with the following properties.
The set $\{(A^{\mathrm{full}},t,y):t\in[0,T], (A^{\mathrm{full}},y)\in\mathcal D_t\}$ is Borel, and the local stationary point $(A^{\mathrm{full}},t,y)\mapsto a_t(A^{\mathrm{full}},y)$ is Borel on this domain.
For every fixed $A^{\mathrm{full}},t$ and every two certified fields $y,z$,
\begin{equation*}
 \|a_t(A^{\mathrm{full}},y)-a_t(A^{\mathrm{full}},z)\|_2
 \le L\|y-z\|_2.
\end{equation*}
\end{lemma}

\begin{proof}
We first prove that the set in the statement is Borel.
For $t\le t_0$, recall that $(A^{\mathrm{full}},y)\in\mathcal D_t$ means
\begin{equation*}
 \|A^{\mathrm{full}}\|_{\mathrm{op}}\le4,\qquad
 I_n-\beta A^{\mathrm{full}}+\beta^2(1-q_\beta(t))I_n\succeq4\kappa_0I_n.
\end{equation*}
Since $q_\beta(t)$, the operator norm, and the least eigenvalue are continuous, these inequalities together with $t\le t_0$ define a Borel set of triples $(A^{\mathrm{full}},t,y)$.

For large times $t>t_0$, the least covering index $\ell_t$ gives a finite Borel partition, and on a part with index $\ell$ the depth-$K_\ell$ AMP center and its TAP residual are continuous in $(A^{\mathrm{full}},t,y)$.
{The local strong-convexity condition is the only condition that must be checked at every point of a ball.
For each fixed $v\in\mathbb Q^n\cap(-1,1)^n$, this condition requires the following implication:}
\begin{equation*}
 \|v-\widehat m_\ell(A^{\mathrm{full}},y)\|_2<4r_\ell\sqrt n
 \quad\Longrightarrow\quad
 \nabla_v^2F_{\mathrm{TAP}}(v;A^{\mathrm{full}},y,q_\beta(t))
 \succeq4\kappa_\ell I_n.
\end{equation*}
For fixed $v$, this implication is a Borel condition on $(A^{\mathrm{full}},t,y)$.
The rational points are dense in the AMP-centered ball intersected with $(-1,1)^n$, and the Hessian is continuous in $v$, so the countable intersection of these conditions is exactly the required local strong convexity.
The matrix and field cutoffs $\|A^{\mathrm{full}}\|_{\mathrm{op}}\le4, \|y\|_\infty\le n^2$ and the AMP residual condition are Borel as well.
So the set is Borel.

We next show that $(A^{\mathrm{full}},t,y)\mapsto a_t(A^{\mathrm{full}},y)$ is Borel on the certified domain.
It suffices to prove continuity relative to each part of the finite Borel partition above.
Fix a large-time part with index $\ell$ and $(A_k^{\mathrm{full}},t_k,y_k)\to(A^{\mathrm{full}},t,y)$ in this part.
We set $c_k=\widehat m_\ell(A_k^{\mathrm{full}},y_k)$, $c=\widehat m_\ell(A^{\mathrm{full}},y)$, $u_k=a_{t_k}(A_k^{\mathrm{full}},y_k)$, and $u=a_t(A^{\mathrm{full}},y)$.
By \Cref{lem:tap-stationary-point}-\textup{(i)}, we have $\|u_k-c_k\|_2,\|u-c\|_2\le r_\ell\sqrt n/32$, while continuity of the AMP center gives $c_k\to c$.
{For all sufficiently large $k$, $\|u-c_k\|_2\le\|u-c\|_2+\|c-c_k\|_2<4r_\ell\sqrt n$, so $u\in B(c_k,4r_\ell\sqrt n)\cap(-1,1)^n$.}

{We apply the residual estimate in \Cref{lem:tap-stationary-point}-\textup{(i)} to $u$ in this ball.}
Since $R_{\mathrm{TAP}}(u;A^{\mathrm{full}},y,q_\beta(t))=0$ and $\|u\|_2\le\sqrt n$, subtracting the two residuals gives
\begin{align*}
 \|u_k-u\|_2
 \leq&~(4\kappa_\ell)^{-1}
   \|R_{\mathrm{TAP}}(u;A_k^{\mathrm{full}},y_k,q_\beta(t_k))\|_2\\
 \leq&~(4\kappa_\ell)^{-1}
   \bigl(\|y_k-y\|_2+\beta\sqrt n\|A_k^{\mathrm{full}}-A^{\mathrm{full}}\|_{\mathrm{op}}
          +\beta^2\sqrt n|q_\beta(t_k)-q_\beta(t)|\bigr)\longrightarrow0.
\end{align*}
{At small times, the TAP free energy is continuous on the compact cube $[-1,1]^n$ and therefore attains a minimum.}
Its derivative in coordinate $i$ tends to $+\infty$ at $u_i=1$ and to $-\infty$ at $u_i=-1$, while the other terms remain bounded, so every minimizer lies in $(-1,1)^n$.
At this interior minimizer, $R_{\mathrm{TAP}}=\nabla F_{\mathrm{TAP}}=0$; the matrix condition defining $\mathcal D_t$ makes this zero unique, and we denote it by $a_t(A^{\mathrm{full}},y)$.
For certified small-time triples $(A_k^{\mathrm{full}},t_k,y_k)\to(A^{\mathrm{full}},t,y)$ with a certified limit, let $u_k=a_{t_k}(A_k^{\mathrm{full}},y_k)$ and $u=a_t(A^{\mathrm{full}},y)$.
{We integrate the Hessian lower bound $4\kappa_0I_n$ along the segment from $u_k$ to $u$ and use $R_{\mathrm{TAP}}(u_k;A_k^{\mathrm{full}},y_k,q_\beta(t_k))=0$ to obtain}
\begin{equation*}
 \|u_k-u\|_2
 \le(4\kappa_0)^{-1}\|R_{\mathrm{TAP}}(u;A_k^{\mathrm{full}},y_k,q_\beta(t_k))\|_2
 \longrightarrow0.
\end{equation*}
{The convergence follows from continuity of the residual at the limiting stationary point $u$.}
Hence the local stationary-point function is continuous relative to each part of the finite Borel partition and is therefore Borel on the certified domain.

We now prove the Lipschitz property.
Fix $A^{\mathrm{full}},t$ and two certified fields $y,z$, and set $u=a_t(A^{\mathrm{full}},y)$ and $v=a_t(A^{\mathrm{full}},z)$.

At small times, stationarity gives $R_{\mathrm{TAP}}(u;A^{\mathrm{full}},y,q_\beta(t))=0$ and $R_{\mathrm{TAP}}(v;A^{\mathrm{full}},y,q_\beta(t))=z-y$.
We apply the $4\kappa_0$-strong monotonicity of the first residual to these points and obtain
\begin{equation*}
 4\kappa_0\|u-v\|_2^2
 \le\inner{y-z}{u-v}
 \le\|y-z\|_2\|u-v\|_2.
\end{equation*}
Thus $\|u-v\|_2\le(4\kappa_0)^{-1}\|y-z\|_2$.

At large times, we set $D=\{w\in\R^n:(A^{\mathrm{full}},w)\in\mathcal D_t\}$ and take $\ell=\ell_t$.
By \Cref{lem:amp-spatial-lipschitz}, the depth-$K_\ell$ AMP center is $L_{K_\ell}$-Lipschitz on $D$, and every $w\in D$ satisfies the hypotheses of \Cref{lem:tap-stationary-point}-\textup{(i)} with $r=r_\ell$ and $\kappa=\kappa_\ell$.
We apply \Cref{lem:tap-stationary-point}-\textup{(ii)} and obtain
\begin{equation*}
 \|u-v\|_2
 \le\max\left((4\kappa_\ell)^{-1},\frac{2L_{K_\ell}}{3r_\ell}\right)\|y-z\|_2.
\end{equation*}
Combined, the lemma follows by setting  $L=\max\left\{(4\kappa_0)^{-1},\max_\ell\max\left((4\kappa_\ell)^{-1},\frac{2L_{K_\ell}}{3r_\ell}\right)\right\}$.
\end{proof}

For $d\in\R^n$, set $A^{\mathrm{full}}=A+\operatorname{diag}(d)$.
For a time $s\in[0,T]$ and an external field $y\in\R^n$, we define the cost function
\begin{equation*}
 \mathfrak e_{\beta,A,d,s}(y)
 =\begin{cases}
    \|m(A,y)-a_s(A^{\mathrm{full}},y)\|_2^2,
       &\text{if }(A^{\mathrm{full}},y)\in\mathcal D_s,\\
    4n,&\text{otherwise}.
   \end{cases}
\end{equation*}
We average this quantity over the auxiliary diagonal:
\begin{equation*}
 \overline{\mathfrak e}_{n,\beta}(A,s;y)
 =\E[d\sim\mathbb P_{n,\beta}^{\mathrm{diag}}]{
       \mathfrak e_{\beta,A,d,s}(y)}.
\end{equation*}

\begin{lemma}
\label{lem:local-root-lipschitz-with-costs}
For every $A\in\R^{n\times n}_{\mathrm{sym},0}$, $s\in[0,T]$, and $y_1,y_2\in\R^n$, we have
\begin{align*}
 \|m(A,y_1)-m(A,y_2)\|_2
 \le{}&L\|y_1-y_2\|_2
 +\sqrt{\overline{\mathfrak e}_{n,\beta}(A,s;y_1)}
 +\sqrt{\overline{\mathfrak e}_{n,\beta}(A,s;y_2)}.
\end{align*}
\end{lemma}

\begin{proof}
We fix $d\in\R^n$ and set $A^{\mathrm{full}}=A+\operatorname{diag}(d)$.
If $(A^{\mathrm{full}},y_1)$ and $(A^{\mathrm{full}},y_2)$ belong to $\mathcal D_s$, then by \Cref{lem:local-root-measurable-lipschitz} we obtain
\begin{align*}
 \|m(A,y_1)-m(A,y_2)\|_2
 \leq&~\|m(A,y_1)-a_s(A^{\mathrm{full}},y_1)\|_2
       +\|a_s(A^{\mathrm{full}},y_1)-a_s(A^{\mathrm{full}},y_2)\|_2\\
 &+\|a_s(A^{\mathrm{full}},y_2)-m(A,y_2)\|_2\\
 \leq&~\sqrt{\mathfrak e_{\beta,A,d,s}(y_1)}
       +L\|y_1-y_2\|_2
       +\sqrt{\mathfrak e_{\beta,A,d,s}(y_2)}.
\end{align*}
If either field $y_i$ is uncertified, its cost is $4n$, and its square root $2\sqrt n$ bounds the left side because both posterior means lie in $[-1,1]^n$.
Thus the inequality holds for every $d$.
Its left side and $L\|y_1-y_2\|_2$ do not depend on $d$.
We take expectation over $d$ and apply Cauchy--Schwarz to obtain
\begin{equation*}
 \E[d\sim\mathbb P_{n,\beta}^{\mathrm{diag}}]{\sqrt{\mathfrak e_{\beta,A,d,s}(y_i)}}
 \le\sqrt{\E[d\sim\mathbb P_{n,\beta}^{\mathrm{diag}}]{\mathfrak e_{\beta,A,d,s}(y_i)}}
 =\sqrt{\overline{\mathfrak e}_{n,\beta}(A,s;y_i)}.\qedhere
\end{equation*}
\end{proof}

{
For the localization path $Y$, we define the path cost $U_{n,\beta,A}(Y)$ and terminal cost $E_{n,\beta,A}(Y)$ by
\begin{equation*}
 U_{n,\beta,A}(Y)
 =\int_0^T \overline{\mathfrak e}_{n,\beta}(A,t;Y_t)\,\mathrm d t,\quad
 E_{n,\beta,A}(Y)
 =\overline{\mathfrak e}_{n,\beta}(A,T;Y_T).
\end{equation*}
For every fixed $R>0$, each cost is at most $C_R\log n$ with failure probability at most $n^{-R}$.

\begin{lemma}
\label{lem:scalar-reference-path-error}
For every $R>0$, there are constants $C_R,n_R>0$ such that, for every $n\ge n_R$,
\begin{equation*}
 \sup_{\beta\in[\beta_0,\beta_1]}
 \mathbb P_{0,n,\beta}^{\mathrm{aug}}
 \{U_{n,\beta,A}(Y)>C_R\log n\}\le n^{-R}.
\end{equation*}
Moreover,
\begin{equation*}
 \sup_{\beta\in[\beta_0,\beta_1]}
 \mathbb P_{0,n,\beta}^{\mathrm{aug}}
 \{E_{n,\beta,A}(Y)>C_R\log n\}\le n^{-R}.
\end{equation*}
\end{lemma}

\begin{proof}
We fix $R>0$ and set $H=2R+8$.
We first prove a fixed-time bound for the diagonal cost.
By \Cref{lem:amp-good-events,thm:polynomial-tap-residual}, after increasing a constant $K_H\ge1$, for every $\beta\in[\beta_0,\beta_1]$ and $t\in[0,T]$, outside an event of $\mathbb P_{n,\beta}^{\mathrm{pl,aug}}$-probability at most $K_Hn^{-H}$, the following properties hold.
If $t\le t_0$, then $(A^{\mathrm{full}},Y_t)\in\mathcal D_t$ and
\begin{equation*}
 \|R_{\mathrm{TAP}}(m(A,Y_t);A^{\mathrm{full}},Y_t,q_\beta(t))\|_2^2\le K_H\log n.
\end{equation*}
If $t>t_0$, then $(A^{\mathrm{full}},Y_t)\in\mathcal D_t$, the posterior-mean proximity condition holds at the index $\ell_t$, and the same residual bound holds.
For $t\le t_0$, strong convexity on the whole cube therefore gives
\begin{equation*}
 \|m(A,Y_t)-a_t(A^{\mathrm{full}},Y_t)\|_2^2
 \le (4\kappa_0)^{-2}\|R_{\mathrm{TAP}}(m(A,Y_t);A^{\mathrm{full}},Y_t,q_\beta(t))\|_2^2.
\end{equation*}
For $t>t_0$, posterior-mean proximity puts $m(A,Y_t)$ in the strongly convex ball from \Cref{lem:tap-stationary-point}-\textup{(i)}, and hence
\begin{equation*}
 \|m(A,Y_t)-a_t(A^{\mathrm{full}},Y_t)\|_2^2
 \le (4\kappa_{\ell_t})^{-2}\|R_{\mathrm{TAP}}(m(A,Y_t);A^{\mathrm{full}},Y_t,q_\beta(t))\|_2^2.
\end{equation*}
Since there are finitely many indices $\ell$, we can choose a constant $C_H$ so that
\begin{equation}
\label{eq:fixed-time-raw-diagonal-cost}
 \sup_{\beta\in[\beta_0,\beta_1]}\sup_{t\in[0,T]}
 \mathbb P_{n,\beta}^{\mathrm{pl,aug}}
 \left\{\mathfrak e_{\beta,A,d,t}(Y_t)>C_H\log n\right\}
 \le C_Hn^{-H}.
\end{equation}
We next average over the auxiliary diagonal.
For fixed $(A,Y)$ and $t\in[0,T]$, set
\begin{equation*}
 p_t(A,Y)
 =\mathbb P_{d\sim\mathbb P_{n,\beta}^{\mathrm{diag}}}
 \left\{\mathfrak e_{\beta,A,d,t}(Y_t)>C_H\log n\right\}.
\end{equation*}
Under $\mathbb P_{n,\beta}^{\mathrm{pl,aug}}$, the diagonal $d$ is independent of $(A,Y)$, so \cref{eq:fixed-time-raw-diagonal-cost} gives
\begin{equation*}
 \sup_{\beta\in[\beta_0,\beta_1]}\sup_{t\in[0,T]}
 \E[(A,d,Y)\sim\mathbb P_{n,\beta}^{\mathrm{pl,aug}}]{p_t(A,Y)}
 \le C_Hn^{-H}.
\end{equation*}
Using $0\le\mathfrak e_{\beta,A,d,t}\le4n$, we split the diagonal expectation according to the event defining $p_t(A,Y)$ and obtain
\begin{equation*}
 \overline{\mathfrak e}_{n,\beta}(A,t;Y_t)
 \le C_H\log n+4n\,p_t(A,Y).
\end{equation*}
We apply Markov's inequality to $p_t(A,Y)$ and obtain, for all sufficiently large $n$,
\begin{equation*}
 \sup_{\beta\in[\beta_0,\beta_1]}\sup_{t\in[0,T]}
 \mathbb P_{n,\beta}^{\mathrm{pl,aug}}
 \left\{\overline{\mathfrak e}_{n,\beta}(A,t;Y_t)>2C_H\log n\right\}
 \le4C_Hn^{-(H-1)}.
\end{equation*}
We apply the reverse probability comparison in \Cref{lem:planting-likelihood} and use $\beta\le\beta_1<1$ to obtain a constant $\widetilde C_H>0$ such that
\begin{equation}
\label{eq:fixed-time-averaged-diagonal-cost}
 \sup_{\beta\in[\beta_0,\beta_1]}\sup_{t\in[0,T]}
 \mathbb P_{0,n,\beta}^{\mathrm{aug}}
 \left\{\overline{\mathfrak e}_{n,\beta}(A,t;Y_t)>\widetilde C_H\log n\right\}
 \le \widetilde C_Hn^{-(H-1)/2}.
\end{equation}

We now integrate this fixed-time estimate.
By \Cref{lem:local-root-measurable-lipschitz} and the definition of the diagonal average, $\overline{\mathfrak e}_{n,\beta}$ is jointly Borel.
Together with the continuity of $Y$, this implies that the function below is measurable.
We define
\begin{equation*}
 \Lambda_H(A,Y)
 =\int_0^T\mathbf 1_{\{\overline{\mathfrak e}_{n,\beta}(A,t;Y_t)>\widetilde C_H\log n\}}\,\mathrm d t.
\end{equation*}
By \cref{eq:fixed-time-averaged-diagonal-cost} and Tonelli's theorem, we obtain
\begin{align*}
 \E[(A,d,Y)\sim\mathbb P_{0,n,\beta}^{\mathrm{aug}}]{\Lambda_H(A,Y)}
 =\int_0^T
 \mathbb P_{0,n,\beta}^{\mathrm{aug}}
 \left\{\overline{\mathfrak e}_{n,\beta}(A,t;Y_t)>\widetilde C_H\log n\right\}\,\mathrm d t
 \le T\widetilde C_Hn^{-(H-1)/2}.
\end{align*}
We apply Markov's inequality to $\Lambda_H(A,Y)$ and obtain
\begin{equation*}
 \mathbb P_{0,n,\beta}^{\mathrm{aug}}
 \left\{\Lambda_H(A,Y)>\frac{\log n}{n}\right\}
 \le\frac{T\widetilde C_H}{\log n}n^{-(H-3)/2}.
\end{equation*}
On the event $\{\Lambda_H(A,Y)\le\log n/n\}$, we use $\overline{\mathfrak e}_{n,\beta}\le4n$ to obtain $U_{n,\beta,A}(Y)\le T\widetilde C_H\log n+4n\Lambda_H(A,Y)\le(T\widetilde C_H+4)\log n$.
Since $H=2R+8$, after increasing $C_R$ and $n_R$ we obtain
\begin{equation*}
 \sup_{\beta\in[\beta_0,\beta_1]}
 \mathbb P_{0,n,\beta}^{\mathrm{aug}}
 \{U_{n,\beta,A}(Y)>C_R\log n\}\le n^{-R}.
\end{equation*}
Finally, we apply \cref{eq:fixed-time-averaged-diagonal-cost} at $t=T$ and increase $C_R,n_R$ if necessary to obtain
\begin{equation*}
 \sup_{\beta\in[\beta_0,\beta_1]}
 \mathbb P_{0,n,\beta}^{\mathrm{aug}}
 \{E_{n,\beta,A}(Y)>C_R\log n\}\le n^{-R}.
\end{equation*}
\end{proof}
}

Now, we are ready to prove \Cref{lem:app-wpi-branches}.
\begin{proof}[Proof of \Cref{lem:app-wpi-branches}]
We take the cost function $\overline{\mathfrak e}_{n,\beta}$ defined above and the constant $L$ from \Cref{lem:local-root-measurable-lipschitz}.
The Lipschitz property with cost functions in part~\textup{(i)} follows from \Cref{lem:local-root-lipschitz-with-costs}.
The continuous-time path-cost and terminal-cost bounds in part~\textup{(ii)} follow from \Cref{lem:scalar-reference-path-error}.
\end{proof}

\section{Weak Poincar\'e inequalities and simulated annealing}
\label{sec:app-wpi-annealing}

We prove the WPI of \Cref{thm:app-wpi-poly}, extend it to a prescribed {inverse-temperature} schedule, and then prove \Cref{thm:app-wpi-annealing,cor:intro-counting}.
We first introduce some useful notations.
For a reversible Markov kernel $P$ with stationary law $\mu$, write
$\inner{f}{g}_\mu=\E[X\sim\mu]{f(X)g(X)}$ and
$\Diri{\mu}{f,f}=\inner{f}{(I-P)f}_\mu$.
For random-scan heat-bath dynamics on $\{\pm1\}^n$,
\begin{equation*}
 \Diri{\mu}{f,f}=\frac1n\sum_{i=1}^n
       \E[X\sim\mu]{\Var[X\sim\mu]{f(X)\mid X_{-i}}}.
\end{equation*}
We also write $\operatorname{osc}(f)=\sup f-\inf f$.
We state the WPI criterion for a full-support distribution $\mu$ on $\{\pm1\}^n$, writing $\mu_y$ for its exponential tilt and $m(y)$ for the mean of $\mu_y$:
\begin{equation*}
 \mu_y(x)\defeq\frac{\mu(x)\exp(\inner{y}{x})}{\E[\sigma\sim\mu]{\exp(\inner{y}{\sigma})}},\qquad
 m(y)\defeq\E[\sigma\sim\mu_y]{\sigma}.
\end{equation*}
For a fixed constant $T>0$, we use the Wiener measure $\gamma$ and Cameron--Martin space $\mathbb H_T$ from \Cref{subsec:prelim-wiener-space}, and define the localization paths $Y^\omega$ by \cref{eq:localization-path} with $m$ in place of $m(A,\cdot)$.
The Gaussian observation representation in \Cref{lem:app-wpi-observation-law} applies to this $\mu$.
\subsection{Establishing {a} weak Poincar\'e inequality}
{The next lemma combines a Lipschitz property with cost functions for the posterior mean and a bound on the averaged terminal posterior variance to obtain a WPI.}

{
\begin{lemma}
      \label{lem:app-wpi-two-path-criterion}
Fix constants $T,L,C_0,r>0$.
For each $n\ge 2$, let $\mu$ be a full-support distribution on
$\{\pm1\}^n$.
Suppose there is a jointly Borel function
$\overline{\mathfrak e}:
    [0,T]\times\mathbb R^n
    \longrightarrow [0,4n]$
such that, for every $t\in[0,T]$ and every
$y_1,y_2\in\mathbb R^n$,
\[
    \|m(y_1)-m(y_2)\|_2
    \le
    L\|y_1-y_2\|_2
    +
    \sqrt{\overline{\mathfrak e}(t;y_1)}
    +
    \sqrt{\overline{\mathfrak e}(t;y_2)}.
\]

For the localization paths defined above from $\mu$, define the
continuous-time path cost and terminal cost by
\[
    U(\omega)
    \defeq
    \int_0^T
    \overline{\mathfrak e}
    \bigl(t;Y_t^\omega\bigr)\,\mathrm d t,
    \qquad
    E(\omega)
    \defeq
    \overline{\mathfrak e}
    \bigl(T;Y_T^\omega\bigr).
\]
Assume that
\[
    \mathbb P_{\omega\sim\gamma}
    \left[
        U(\omega)>C_0\log n
        \ \text{or}\
        E(\omega)>C_0\log n
    \right]
    \le n^{-r}.
\]

Suppose there are numbers $a_n,\delta_n\ge0$ such that, for every
$f:\{\pm1\}^n\to\mathbb R$,
\[
    \mathbb E_{\omega\sim\gamma}
    \left[
        \operatorname{Var}_{\mu_{Y_T^\omega}}[f]
    \right]
    \le
    a_n\Diri{\mu}{f,f}
    +
    \delta_n\operatorname{osc}(f)^2.
\]
Then there are constants $\Gamma>0$, depending only on $T,L$, and
$C,B>0$, depending only on $T,L,C_0$, such that, for every $n\ge2$,
every $f:\{\pm1\}^n\to\mathbb R$, and every $8n^{-r}\le s\le \frac12$,
we have
\[
    \operatorname{Var}_\mu[f]
    \le
    \bigl(1+C n^B s^{-\Gamma}\bigr)
    a_n\Diri{\mu}{f,f}
    +
    \left(
        s+
        \bigl(1+C n^B s^{-\Gamma}\bigr)\delta_n
    \right)
    \operatorname{osc}(f)^2.
\]
\end{lemma}
}

We use \Cref{lem:app-wpi-two-set-gaussian} to find two nearby paths in sets of positive Wiener measure.
The next two estimates control their endpoints and the corresponding posterior expectations.

\begin{lemma}
\label{lem:app-wpi-path-distance}
For the {localization paths} $Y^\omega$ and path cost $U$ defined above, for every $\omega\in C_0([0,T];\R^n)$ and $h\in\mathbb H_T$ we have that
\begin{equation*}
 \|Y_T^{\omega+h}-Y_T^\omega\|_2
 \le\sqrt T\,\exp(LT)
 \bigl(\|h\|_{\mathbb H_T}
       +\sqrt{U(\omega+h)}+\sqrt{U(\omega)}\bigr).
\end{equation*}
\end{lemma}
{
\begin{proof}
For $0\le t\le T$, we subtract the defining equations for $Y^{\omega+h}$ and $Y^\omega$ to obtain
\[
Y_t^{\omega+h}-Y_t^\omega
=
h_t
+
\int_0^t
\bigl(
m(Y_u^{\omega+h})-m(Y_u^\omega)
\bigr)\,\mathrm d u .
\]
By \cref{eq:cameron-martin-uniform-bound},
$\|h\|_\infty \le \sqrt T\,\|h\|_{\mathbb H_T}$.
For each $u\in[0,T]$, we apply the Lipschitz property with cost functions assumed in \Cref{lem:app-wpi-two-path-criterion} at time $u$, taking $(y_1,y_2)=(Y_u^{\omega+h},Y_u^\omega)$.
We obtain
\begin{align*}
\|m(Y_u^{\omega+h})-m(Y_u^\omega)\|_2
\le{}
L\|Y_u^{\omega+h}-Y_u^\omega\|_2 
+
\sqrt{
\overline{\mathfrak e}
\bigl(u;Y_u^{\omega+h}\bigr)
}
+
\sqrt{
\overline{\mathfrak e}
\bigl(u;Y_u^\omega\bigr)
}.
\end{align*}
Thus, we have
\begin{align*}
\sup_{0\le s\le t}
\|Y_s^{\omega+h}-Y_s^\omega\|_2
\le{}&
\sqrt T\,\|h\|_{\mathbb H_T}
+
L\int_0^t  \sup_{0\le s\le u}
\|Y_s^{\omega+h}-Y_s^\omega\|_2\,\mathrm d u
+
\int_0^t
\sqrt{
\overline{\mathfrak e}
\bigl(u;Y_u^{\omega+h}\bigr)
}\,\mathrm d u \\
&+
\int_0^t
\sqrt{
\overline{\mathfrak e}
\bigl(u;Y_u^\omega\bigr)
}\,\mathrm d u .
\end{align*}
By Cauchy--Schwarz and the definition of $U$, we have that 
\[
\int_0^t
\sqrt{
\overline{\mathfrak e}
\bigl(u;Y_u^{\omega+h}\bigr)
}\,\mathrm d u
\le
\sqrt{T\,U(\omega+h)},
\]
and similarly
\[
\int_0^t
\sqrt{
\overline{\mathfrak e}
\bigl(u;Y_u^\omega\bigr)
}\,\mathrm d u
\le
\sqrt{T\,U(\omega)}.
\]
Therefore
\[
\sup_{0\le s\le t}
\|Y_s^{\omega+h}-Y_s^\omega\|_2
\le
\sqrt T
\Bigl(
\|h\|_{\mathbb H_T}
+\sqrt{U(\omega+h)}
+\sqrt{U(\omega)}
\Bigr)
+
L\int_0^t \sup_{0\le s\le u}
\|Y_s^{\omega+h}-Y_s^\omega\|_2\,\mathrm d u .
\]
By \Cref{lem:gronwall-path-stability}, we have that
\[
\|Y_T^{\omega+h}-Y_T^\omega\|_2
\le
\sup_{0\le s\le T}
\|Y_s^{\omega+h}-Y_s^\omega\|_2
\le
\sqrt T\mathrm{e}^{LT}
\Bigl(
\|h\|_{\mathbb H_T}
+\sqrt{U(\omega+h)}
+\sqrt{U(\omega)}
\Bigr).
\]
\end{proof}
}

\begin{lemma}
\label{lem:app-wpi-posterior-comparison}
For the posterior distribution $\mu_y$ and cost function $\overline{\mathfrak e}$ defined above, for every $y,z\in\R^n$ and $f:\{\pm1\}^n\to\R$ we have that
\begin{align*}
 |\E[\mu_y]{f}-\E[\mu_z]{f}|^2
 \leq&~2\bigl(\exp([D(\mu_y\Vert\mu_z)
                         +D(\mu_z\Vert\mu_y)]/2)-1\bigr)
             (\Var[\mu_y]{f}+\Var[\mu_z]{f}),
\end{align*}
\begin{align*}
 D(\mu_y\Vert\mu_z)+D(\mu_z\Vert\mu_y)
 \leq&~L\|y-z\|_2^2+
 \bigl(\sqrt{\overline{\mathfrak e}(T;y)}+\sqrt{\overline{\mathfrak e}(T;z)}\bigr)\|y-z\|_2.
\end{align*}
\end{lemma}

\begin{proof}
We first prove the first part of this lemma.
{We take full-support distributions $P,Q$ on the cube and set $a=\sum_x\sqrt{P(x)Q(x)}>0$, their Hellinger affinity.}
We have that
\begin{equation*}
 a(\E[P]{f}-\E[Q]{f})
 =\sum_x\sqrt{P(x)Q(x)}(f(x)-\E[Q]{f})
  -\sum_x\sqrt{P(x)Q(x)}(f(x)-\E[P]{f}){.}
\end{equation*}
For the first sum, we use
$\sqrt{P(x)Q(x)}=Q(x)\sqrt{P(x)/Q(x)}$ and
$\sum_xQ(x)(f(x)-\E[Q]{f})=0$.  Thus we may subtract
$aQ(x)(f(x)-\E[Q]{f})$ and apply Cauchy--Schwarz:
\begin{equation*}
 \left|\sum_x\sqrt{P(x)Q(x)}(f(x)-\E[Q]{f})\right|^2
 \le\Var[Q]{f}\sum_xQ(x)
       \left(\sqrt{P(x)/Q(x)}-a\right)^2
 =(1-a^2)\Var[Q]{f},
\end{equation*}
where the last equality follows from $\sum_xP(x)=\sum_xQ(x)=1$ and $\sum_x\sqrt{P(x)Q(x)}=a$.
Going through the calculation for the other centered sum with $P,Q$ exchanged and using $(u-v)^2\le2u^2+2v^2$, we obtain
\begin{equation*}
 |\E[P]{f}-\E[Q]{f}|^2
 \le2(a^{-2}-1)(\Var[P]{f}+\Var[Q]{f}).
\end{equation*}
Next, we provide a lower bound on $a$.
By Jensen's inequality, we have that
\begin{equation*}
 \log a=\log\E[P]{\sqrt{Q(X)/P(X)}}
 \ge\E[P]{\log\sqrt{Q(X)/P(X)}}=-D(P\Vert Q)/2,
\end{equation*}
and the same argument under $Q$ gives
$\log a\ge-D(Q\Vert P)/2$.
Taking the mean of these two lower bounds yields
$a^{-2}\le\exp([D(P\Vert Q)+D(Q\Vert P)]/2)$.
By setting $P = \mu_y$ and $Q = \mu_z$, we finish the proof of the first part.

For the second part, note that for
each $x\in\{\pm1\}^n$,
\begin{equation*}
 \log\frac{\mu_y(x)}{\mu_z(x)}
 =\inner{y-z}{x}
  +\log\frac{\E[X\sim\mu]{\exp(\inner zX)}}
                  {\E[X\sim\mu]{\exp(\inner yX)}}.
\end{equation*}
Hence, we have that
\begin{equation*}
 D(\mu_y\Vert\mu_z)
 =\E[X\sim\mu_y]{\log\frac{\mu_y(X)}{\mu_z(X)}}
 =\inner{y-z}{m(y)}
  +\log\frac{\E[X\sim\mu]{\exp(\inner zX)}}
                  {\E[X\sim\mu]{\exp(\inner yX)}}.
\end{equation*}
Thus, we have that
\begin{equation*}
 D(\mu_y\Vert\mu_z)+D(\mu_z\Vert\mu_y)
 =\inner{y-z}{m(y)-m(z)}.
\end{equation*}
We apply the Lipschitz property with cost functions in \Cref{lem:app-wpi-two-path-criterion} at time $T$ to $\|m(y)-m(z)\|_2$ and use Cauchy--Schwarz to obtain the second bound.
\end{proof}
\begin{proof}[Proof of \Cref{lem:app-wpi-two-path-criterion}]
We fix $f$.
By \Cref{lem:app-wpi-observation-law} and the law of total variance, we have that
\begin{equation*}
 \Var[\mu]{f}
 =\E[\omega\sim\gamma]{\Var[\mu_{Y_T^\omega}]{f}}
  +\Var[\gamma]{\E[\mu_{Y_T^\omega}]{f}}.
\end{equation*}
The assumed bound on the averaged posterior variance controls the first term.
For the second term, we will prove that it is at most
\begin{equation*}
 \Var[\gamma]{\E[\mu_{Y_T^\omega}]{f}}\le Cn^Bs^{-\Gamma}\E[\omega\sim\gamma]{\Var[\mu_{Y_T^\omega}]{f}}
 +s\operatorname{osc}(f)^2{.}
\end{equation*}
{Here $\Gamma$ depends only on $T,L$, while $C,B$ may also depend on $C_0$.}

{For $0<v<1$, we define the quantile of $\E[\mu_{Y_T^\omega}]{f}$ under $\omega\sim\gamma$ by
\begin{equation*}
 Q_f(v)=\inf\left\{a\in\R:\Pr[\omega\sim\gamma]{\E[\mu_{Y_T^\omega}]{f}\le a}\ge v\right\}.
\end{equation*}
}
For $s\le v\le1/2$, we claim that
\begin{equation*}
 [Q_f(1-v)-Q_f(v)]^2
 \le 16\cdot4^\Gamma n^Bv^{-1-\Gamma}
      \E[\omega\sim\gamma]{\Var[\mu_{Y_T^\omega}]{f}}.
\end{equation*}
Later, we will integrate this estimate over $v\in[s,1/2]$; the remaining
interval $(0,s)$ contributes at most $s\operatorname{osc}(f)^2$.

To prove the above claim, we use Gaussian isoperimetry to choose a lower-quantile path and an upper-quantile path whose Cameron--Martin distance is small.
The lower event
$\{\E[\mu_{Y_T^\omega}]{f}\le Q_f(v)\}$ and the upper event $\{\E[\mu_{Y_T^\omega}]{f}\ge Q_f(1-v)\}$ each have Wiener measure at least $v$.
{We discard paths for which $U>C_0\log n$ or $E>C_0\log n$.}
{By the cost bound, the discarded paths have Wiener measure at most $n^{-r}\le v/8$.}
{We also discard paths for which}
\begin{equation*}
 \Var[\mu_{Y_T^\omega}]{f}>
 \frac4v\,\E[\omega'\sim\gamma]{\Var[\mu_{Y_T^{\omega'}}]{f}}.
\end{equation*}
{By Markov's inequality, these paths have Wiener measure at most $v/4$.}
Combined, each of the lower and upper quantile sets still has measure at least $v-v/8-v/4=5v/8$.

These events are Borel because $\omega\mapsto Y^\omega$ is continuous, the cost functions are Borel, and the posterior expectations and variances are continuous. By
\Cref{lem:app-wpi-two-set-gaussian}, there are a retained lower path
$\omega$ and a retained upper path $\omega+h$ with
$\|h\|_{\mathbb H_T}\le2\sqrt{2\log(4/v)}$.

Let $y=Y_T^\omega$ and $z=Y_T^{\omega+h}$.
Note that by the above construction, both paths have $U\le C_0\log n$.  
By \Cref{lem:app-wpi-path-distance} and
$\|h\|_{\mathbb H_T}\le2\sqrt{2\log(4/v)}$, we have that
\begin{equation*}
 \|y-z\|_2
 \le\sqrt T\exp(LT)
   \bigl(2\sqrt{2\log(4/v)}+2\sqrt{C_0\log n}\bigr),\quad
 \|y-z\|_2^2
 \le T\exp(2LT)\bigl(16\log(4/v)+8C_0\log n\bigr).
\end{equation*}
\begingroup
We set $\Gamma=8(L+1)T\exp(2LT)$ and $B=(4(L+1)T\exp(2LT)+1/2)C_0$.
By \Cref{lem:app-wpi-posterior-comparison} and $2\sqrt{C_0\log n}\|y-z\|_2\leq C_0\log n+\|y-z\|_2^2$, we obtain
\begin{equation*}
 D(\mu_y\Vert\mu_z)+D(\mu_z\Vert\mu_y)
 \le(L+1)T\exp(2LT)
   \bigl(16\log(4/v)+8C_0\log n\bigr)+C_0\log n
 =2\Gamma\log(4/v)+2B\log n.
\end{equation*}
Thus $\Gamma$ is independent of the cost threshold $C_0$.
\endgroup
Also by the above construction, we have that
\begin{equation*}
 \max(\Var[\mu_y]{f},\Var[\mu_z]{f})
 \le\frac4v\E[\omega'\sim\gamma]{\Var[\mu_{Y_T^{\omega'}}]{f}}.
\end{equation*}
Combined with  \Cref{lem:app-wpi-posterior-comparison}, we have that
\begingroup
\begin{align*}
 [Q_f(1-v)-Q_f(v)]^2
 \leq&~\frac{16}{v}\exp(\Gamma\log(4/v)+B\log n)
      \E[\omega\sim\gamma]{\Var[\mu_{Y_T^\omega}]{f}}\\
 =&~16\cdot4^\Gamma n^Bv^{-1-\Gamma}
      \E[\omega\sim\gamma]{\Var[\mu_{Y_T^\omega}]{f}}{.}
\end{align*}
\endgroup

We set $q\defeq Q_f(1/2)$ and let $V\sim\operatorname{Unif}(0,1)$.
The random variable $Q_f(V)$ has the same law as $\E[\mu_{Y_T^\omega}]{f}$ under $\omega\sim\gamma$, so using $Q_f(u)\leq q\leq Q_f(1-u)$ for $0<u<1/2$, we obtain
\begin{align*}
 \Var[\gamma]{\E[\mu_{Y_T^\omega}]{f}}
 \leq&~ \E[\omega\sim\gamma]{\bigl(\E[\mu_{Y_T^\omega}]{f}-q\bigr)^2}
 = \int_0^{1/2}\bigl((Q_f(u)-q)^2+(Q_f(1-u)-q)^2\bigr)\,\mathrm d u\\
 \leq&~ \int_0^{1/2}[Q_f(1-u)-Q_f(u)]^2\,\mathrm d u.
\end{align*}
The part with $u<s$ is at most $s\operatorname{osc}(f)^2$.
On $[s,1/2]$ we have that 
\begingroup
\begin{equation*}
 \int_s^{1/2}u^{-1-\Gamma}\,\mathrm d u
 =\frac{s^{-\Gamma}-2^\Gamma}{\Gamma}
 \le\frac1\Gamma s^{-\Gamma}.
\end{equation*}
Taking $C=16\cdot4^\Gamma/\Gamma$, we obtain
\endgroup
\begin{equation*}
 \Var[\gamma]{\E[\mu_{Y_T^\omega}]{f}}
 \le Cn^Bs^{-\Gamma}
       \E[\omega\sim\gamma]{\Var[\mu_{Y_T^\omega}]{f}}
       +s\operatorname{osc}(f)^2.
\end{equation*}
Hence, we have that
\begin{align*}
 \Var[\mu]{f}
 =&~\E[\omega\sim\gamma]{\Var[\mu_{Y_T^\omega}]{f}}
    +\Var[\gamma]{\E[\mu_{Y_T^\omega}]{f}}\\
 \leq&~(1+Cn^Bs^{-\Gamma})
          \E[\omega\sim\gamma]{\Var[\mu_{Y_T^\omega}]{f}}
          +s\operatorname{osc}(f)^2\\
 \leq&~(1+Cn^Bs^{-\Gamma})a_n\Diri{\mu}{f,f}
    +(s+(1+Cn^Bs^{-\Gamma})\delta_n)
          \operatorname{osc}(f)^2.
\end{align*}
\end{proof}

To apply \Cref{lem:app-wpi-two-path-criterion}, we need to bound the averaged terminal posterior variance.
We prove the following lemma in \Cref{app:terminal-posterior-variance} using the strong-field estimates of Bandeira, El Alaoui, and R\"odder~\cite{BandeiraElAlaouiRodder26}.
\begin{lemma}
\label{lem:app-wpi-two-path-terminal}
Fix $0<\beta_0<\beta_1<1$.
There are constants $T_0,C,c,N>0$, depending only on $(\beta_0,\beta_1)$, such that, for every $n\geq N$, there is an event $\Omega_n^{\mathrm{end}}$ of $\mathbb P_{0,n}$-probability at least $1-C\exp(-cn)$ with the following property.
For every $A\in\Omega_n^{\mathrm{end}}$ and every $\beta\in[\beta_0,\beta_1]$, let $X\sim\mu_{\beta,A}$, let $(B_t)_{0\leq t\leq T_0}$ be an independent standard Brownian motion in $\R^n$ started at $0$, and set $Y_t=tX+B_t$.
Then, for every $f:\{\pm1\}^n\to\R$,
\begin{equation*}
 \E[X\sim\mu_{\beta,A},\,B]{\Var[\mu_{A,Y_{T_0}}]{f}}
 \leq Cn^7\Diri{\mu_{\beta,A}}{f,f}
       +\tfrac14\exp(-cn)\operatorname{osc}(f)^2,
\end{equation*}
where $\mu_{A,y}$ is the Gibbs distribution at inverse temperature $\beta$ and external field $y$.
\end{lemma}

\begin{proof}[Proof of \Cref{thm:app-wpi-poly}]
At $\beta=0$, tensorization for product distributions gives $\Var[\mu_A]{f}\le n\Diri{\mu_A}{f,f}$.

{We fix $0<\beta<1$ and $D,M_0>0$, and choose $0<\beta_0<\beta<\beta_1<1$ depending only on $\beta$.}
We take $T=T_0$ from \Cref{lem:app-wpi-two-path-terminal} on $[\beta_0,\beta_1]$.
We will apply \Cref{lem:app-wpi-two-path-criterion} on a high-probability event of interaction matrices.
We next verify each assumption of this lemma.

By \Cref{lem:app-wpi-branches} with $T$ and
$R=M_0+D+6$, we take the cost function
$\overline{\mathfrak e}_{n,\beta}$.
For fixed $A$, define
\[
    \overline{\mathfrak e}(t;y)
    \defeq
    \overline{\mathfrak e}_{n,\beta}(A,t;y),
    \qquad
    t\in[0,T],\quad y\in\R^n.
\]
By part~\textup{(i)} of \Cref{lem:app-wpi-branches},
$\overline{\mathfrak e}$ satisfies the Lipschitz property with cost functions
assumed in \Cref{lem:app-wpi-two-path-criterion}.
Let
\[
    U_{\beta,A}(\omega)
    \defeq
    \int_0^T
    \overline{\mathfrak e}_{n,\beta}
    \bigl(A,t;Y_t^\omega\bigr)\,\mathrm d t, \quad
    E_{\beta,A}(\omega)
    \defeq
    \overline{\mathfrak e}_{n,\beta}
    \bigl(A,T;Y_T^\omega\bigr).
\]

{Next, we bound the probability that either $U_{\beta,A}$ or $E_{\beta,A}$ exceeds $C_R\log n$, as required by \Cref{lem:app-wpi-two-path-criterion}.}
By \Cref{lem:app-wpi-branches}, we have
\begin{equation*}
 \E[A\sim\mathbb P_{0,n}]{
  \Pr[\omega\sim\gamma]{U_{\beta,A}(\omega)>C_R\log n
                 \text{ or }E_{\beta,A}(\omega)>C_R\log n}}
 \le2n^{-R}.
\end{equation*}
Taking $r=M_0+3$, by Markov's inequality, we have that
\begin{equation*}
 \mathbb P_{0,n}\left\{A:
  \Pr[\omega\sim\gamma]{U_{\beta,A}(\omega)>C_R\log n
                 \text{ or }E_{\beta,A}(\omega)>C_R\log n}>n^{-r}\right\}
 \le2n^{r-R}.
\end{equation*}
Finally, on the event $\Omega_n^{\mathrm{end}}$, we apply \Cref{lem:app-wpi-two-path-terminal,lem:app-wpi-observation-law} to obtain the averaged terminal posterior variance bound in \Cref{lem:app-wpi-two-path-criterion} with $a_n=Cn^7$ and $\delta_n=\exp(-cn)/4$.
For every $\epsilon\in[n^{-M_0},1/2]$, we take $s=\epsilon/2$.
Since $r=M_0+3$, for all sufficiently large $n$,
$8n^{-r}\le\epsilon/2\le1/4$.
\begingroup
By \Cref{lem:app-wpi-two-path-criterion} with $a_n=Cn^7$ and $\delta_n=\exp(-cn)/4$, there are constants $C,b,\Gamma>0$ such that
\begin{equation*}
 \Var[\mu_A]{f}
 \le Cn^7(1+Cn^b(\epsilon/2)^{-\Gamma})
       \Diri{\mu_A}{f,f}
 +\left(\frac\epsilon2
       +\frac{1+Cn^b(\epsilon/2)^{-\Gamma}}4
         \exp(-cn)\right)\operatorname{osc}(f)^2.
\end{equation*}
Here $\Gamma=8(L+1)T\exp(2LT)$ depends only on $\beta$, since $T,L$ depend only on the chosen interval $[\beta_0,\beta_1]$.
Since $\epsilon\ge n^{-M_0}$, for sufficiently large $n$, we have that 
\begin{equation*}
 \frac{1+C2^\Gamma n^{b+M_0\Gamma}}4
       \exp(-cn)
 \le\frac12n^{-M_0}\le\frac\epsilon2.
\end{equation*}
Also, by $n\ge2$ and $\epsilon\le1/2$, we have $Cn^7(1+Cn^b(\epsilon/2)^{-\Gamma})\leq C'n^{7+b}\epsilon^{-\Gamma}$.
Taking $B=b/\Gamma$ gives the coefficient in \Cref{thm:app-wpi-poly}, with $C,B,N$ depending only on $(\beta,D,M_0)$.
\endgroup
Finally, the union bound gives a failure probability of at most $2n^{r-R}+C\exp(-cn)=2n^{-D-3}+C\exp(-cn)\le n^{-D}$ for all sufficiently large $n$.

\end{proof}

\subsection{Simulated annealing}

{For simulated annealing, we need the WPI at every inverse temperature in the schedule. We formalize {this in} the next corollary.}

\begin{corollary}
\label{cor:app-wpi-schedule}
Fix $0\leq\beta_*<1$ and $D,M_0,C_{\mathrm{sch}}>0$.
{There are constants $\Gamma>0$, depending only on $\beta_*$, and $C,a,N>0$, depending only on $(\beta_*,D,M_0,C_{\mathrm{sch}})$, such that every deterministic}
$\mathcal B_n\subseteq [0,\beta_*]$ with
$|\mathcal B_n|\leq n^{C_{\mathrm{sch}}}$ has an event of
$\mathbb P_{0,n}$-probability at least $1-n^{-D}$ on which
\begin{equation*}
 \Var[\mu_{\beta,A}]{f}
 \leq Cn^a\epsilon^{-\Gamma}\Diri{\mu_{\beta,A}}{f,f}
       +\epsilon\operatorname{osc}(f)^2
\end{equation*}
for all $\beta\in\mathcal B_n$, all
$\epsilon\in[n^{-M_0},1/2]$, and all $f:\{\pm1\}^n\to\R$, when $n\geq N$.
\end{corollary}

\begin{proof}
We fix $0\leq\beta_*<1$ and a deterministic schedule $\mathcal B_n\subseteq[0,\beta_*]$ with $|\mathcal B_n|\leq n^{C_{\mathrm{sch}}}$.
We take $r=M_0+3$ and choose $R>D+C_{\mathrm{sch}}+r+3$.
We prove this corollary by considering two cases, small $\beta$ and large $\beta$.

We first handle the large $\beta$.
Suppose $\beta_*>1/15$.
By \Cref{lem:app-wpi-two-path-terminal} with $(\beta_0,\beta_1)=(1/15,\beta_*)$, we choose one terminal time $T_0$ and {the matrix event $\Omega_n^{\mathrm{end}}$} for every $\beta\in[1/15,\beta_*]$.
By \Cref{lem:app-wpi-branches} on this interval with $T=T_0$ and the chosen $R$, we use the same Lipschitz constant $L$, error threshold $C_R\log n$, and lower bound on $n$ for every such $\beta$.

{In \Cref{lem:app-wpi-two-path-criterion}, the exponent $\Gamma$ depends only on $T_0,L$, hence only on $\beta_*$, while $C,B$ may also depend on $C_R$.}
The constants in the averaged terminal posterior variance bound of \Cref{lem:app-wpi-two-path-terminal} are also common to the interval.
Going through the proof of \Cref{thm:app-wpi-poly} with these choices therefore gives constants $C,a,\Gamma,N$, independent of $\beta\in[1/15,\beta_*]$, such that for each fixed $\beta$ the claimed WPI holds for all $f$ and all $\epsilon\in[n^{-M_0},1/2]$ outside a matrix event of probability at most $2n^{r-R}+C\exp(-cn)$.

For $0\leq\beta\leq1/15$, the event $\|A\|_{\mathrm{op}}\leq3$ gives $0\preceq\beta A+3\beta I_n\preceq(2/5)I_n$.
The diagonal shift leaves the Gibbs distribution unchanged, so \cite[Theorem~1]{EldanKoehlerZeitouni22} gives $\Var[\mu_{\beta,A}]{f}\leq(5n/3)\Diri{\mu_{\beta,A}}{f,f}$ for every $f$.
If $\beta_*\leq1/15$, we take $\Gamma=1$, $a=1$, and $C=5/3$.
By \Cref{lem:terminal-disorder-tail} with $\beta_1=1/2$, the probability that $\|A\|_{\mathrm{op}}>3$ is at most $C\exp(-cn)$.
After increasing $C,a,N$ if necessary, the same WPI constants and the same failure bound $2n^{r-R}+C\exp(-cn)$ thus apply to each fixed $\beta\in[0,\beta_*]$.

We now take the union over the prescribed schedule:
\begin{equation*}
 \mathbb P_{0,n}\left\{\text{the WPI fails for some }\beta\in\mathcal B_n\right\}
 \leq |\mathcal B_n|(2n^{r-R}+C\exp(-cn))
 \leq n^{-D}
\end{equation*}
for all sufficiently large $n$.
\end{proof}

We give discrete-time, finite-state versions of the warm-start and simulated-annealing bounds of Huang, Mohanty, Rajaraman, and Wu \cite[Theorems~4.6 and~4.13]{HuangMohantyRajaramanWu25}.
A law $\mu^{\mathrm{init}}$ is
$B$-warm with respect to $\mu$ if $\mu^{\mathrm{init}}(S)\leq B\mu(S)$ for every measurable set $S$, equivalently $\d\mu^{\mathrm{init}}/\d\mu\leq B$ almost everywhere.
{The following bound applies to reversible Markov kernels on a finite space with spectrum in $[0,1]$.}

\begin{lemma}[{\cite[Theorem~4.6]{HuangMohantyRajaramanWu25}}]
\label{lem:app-wpi-warm-decay}
Let $P$ be a reversible Markov kernel on a finite space, with
stationary law $\mu$ and spectrum in $[0,1]$.
Suppose $\Var[\mu]{f}\leq A_0\Diri{\mu}{f,f}
+\delta\operatorname{osc}(f)^2$, with $A_0\geq1$.
For every $B_0$-warm initial law $\mu^{\mathrm{init}}$ and integer $k\geq0$,
\begin{equation*}
 \DTV{\mu^{\mathrm{init}} P^k}{\mu}
 \leq\frac12\left(\e^{-k/A_0}(B_0-1)+\delta B_0^2\right)^{1/2}.
\end{equation*}
\end{lemma}

\begin{proof}
We set $f_k=\d(\mu^{\mathrm{init}}P^k)/\d\mu$ and
$v_k=\Var[\mu]{f_k}=\E[\mu]{(f_k-1)^2}$.
For every state $y$ with $\mu(y)>0$, by detailed balance we have that
\begin{equation*}
 f_{k+1}(y)
 =\frac{\sum_x\mu(x)f_k(x)P(x,y)}{\mu(y)}
 =\sum_xP(y,x)f_k(x)=(Pf_k)(y).
\end{equation*}
Thus $f_k=P^kf_0$.  Since a Markov kernel preserves the interval $[0,B_0]$, we have $0\leq f_k\leq B_0$ and
$\E[\mu]{f_k}=1$.  
In particular, $v_0=\E[\mu]{f_0^2}-1\leq B_0\E[\mu]{f_0}-1=B_0-1$.

For every spectral value $\lambda\in[0,1]$,
$1-\lambda^2\geq1-\lambda$.  By self-adjointness of $P$,
$I-P^2\succeq I-P$ on $L^2(\mu)$, and hence
\begin{align*}
 v_k-v_{k+1}
 =&~\inner{f_k}{(I-P^2)f_k}_\mu
 \geq\Diri{\mu}{f_k,f_k}
 \geq A_0^{-1}(v_k-\delta B_0^2),
\end{align*}
where the last inequality holds due to the WPI and $\operatorname{osc}(f_k)\leq B_0$.
Writing the resulting recurrence as
$v_{k+1}\leq(1-A_0^{-1})v_k+A_0^{-1}\delta B_0^2$ and iterating,
we obtain
\begin{equation*}
 v_k\leq(1-A_0^{-1})^kv_0
       +\delta B_0^2(1-(1-A_0^{-1})^k)
 \leq\e^{-k/A_0}(B_0-1)+\delta B_0^2.
\end{equation*}
Finally, by Cauchy--Schwarz,
$\DTV{\mu^{\mathrm{init}}P^k}{\mu}
=\frac12\E[\mu]{|f_k-1|}\leq\frac12\sqrt{v_k}${,} which finishes the proof.
\end{proof}

Now, we verify the assumption of the above lemma.
For random-scan heat-bath Glauber dynamics, the update at site $i$ is
$P_i f=\E[X\sim\mu]{f(X)\mid X_{-i}}$.
{By the tower property of conditional expectation}, for any $f,g$ we have
\begin{equation*}
 \inner f{P_i g}_\mu
 =\E[\mu]{\E[\mu]{f\mid X_{-i}}\E[\mu]{g\mid X_{-i}}}
 =\inner{P_i f}g_\mu,\qquad P_i^2f=P_i f.
\end{equation*}
In particular,
$\inner f{P_i f}_\mu=\E[\mu]{(\E[\mu]{f\mid X_{-i}})^2}$ lies
between zero and $\E[\mu]{f^2}$: by Jensen's inequality,
$\E[\mu]{(\E[\mu]{f\mid X_{-i}})^2}
\leq\E[\mu]{\E[\mu]{f^2\mid X_{-i}}}=\E[\mu]{f^2}$.
Averaging these quadratic-form bounds over $i$, we obtain
$0\preceq n^{-1}\sum_iP_i\preceq I$.

\begin{lemma}[{\cite[Theorem~4.13]{HuangMohantyRajaramanWu25}}]
\label{lem:app-wpi-annealing-composition}
Let $\mu_0,\ldots,\mu_L$ be probability measures on the same finite
space.  For $1\leq j\leq L$, let $P_j$ be reversible with stationary
law $\mu_j$ and nonnegative spectrum, and suppose
\begin{equation*}
 \mu_{j-1}\leq B_j\mu_j,\qquad
 \Var[\mu_j]{f}\leq A_j\inner{f}{(I-P_j)f}_{\mu_j}
          +\delta_j\operatorname{osc}(f)^2,\qquad A_j\geq1.
\end{equation*}
Then, for nonnegative integers $k_1,\ldots,k_L$,
\begin{equation*}
 \DTV{\mu_0P_1^{k_1}\cdots P_L^{k_L}}{\mu_L}
 \leq\frac12\sum_{j=1}^L
      \left(\e^{-k_j/A_j}(B_j-1)+\delta_jB_j^2\right)^{1/2}.
\end{equation*}
\end{lemma}

\begin{proof}
{We write $\widehat\mu_{j-1}=\mu_0P_1^{k_1}\cdots P_{j-1}^{k_{j-1}}$ for the distribution at the start of stage $j$.}
{We compare $\widehat\mu_{j-1}$ with $\mu_{j-1}$, which is $B_j$-warm with respect to $\mu_j$.}
By total-variation contraction under $P_j^{k_j}$ and the triangle inequality we have
\begin{align*}
 \DTV{\widehat\mu_{j-1}P_j^{k_j}}{\mu_j}
 \leq&~\DTV{\widehat\mu_{j-1}P_j^{k_j}}{\mu_{j-1}P_j^{k_j}}
       +\DTV{\mu_{j-1}P_j^{k_j}}{\mu_j}\\
 \leq&~\DTV{\widehat\mu_{j-1}}{\mu_{j-1}}
       +\DTV{\mu_{j-1}P_j^{k_j}}{\mu_j}.
\end{align*}
Starting from $\widehat\mu_0=\mu_0$, we iterate this inequality and obtain
\begin{equation*}
 \DTV{\mu_0P_1^{k_1}\cdots P_L^{k_L}}{\mu_L}
 \leq\sum_{j=1}^L\DTV{\mu_{j-1}P_j^{k_j}}{\mu_j}.
\end{equation*}
{We apply \Cref{lem:app-wpi-warm-decay} to each term to finish the proof.}
\end{proof}

For $\beta>0$, the sampler is \Cref{alg:app-wpi-annealing}, with a fixed integer exponent $p=p(\beta,M,D)$ whose existence our proof establishes without giving an explicit value.
We count arithmetic operations, evaluations of $\exp$, and uniform random draws as unit-cost operations.

\begin{algorithm}[H]
\caption{Simulated annealing from {the} uniform product distribution}
\label{alg:app-wpi-annealing}
\KwIn{A zero-diagonal matrix $A$, target $0<\beta<1$, and a fixed exponent $p$.}
\KwOut{A spin configuration.}
Set $L=\ceil{\beta n}$ and $\beta_j=j\beta/L$ for $0\leq j\leq L$\;
Draw $x\in\{\pm1\}^n \sim \mathrm{Unif}(\set{\pm 1}^n)$\;
\For{$j=1,\ldots,L$}{
  \For{$k=1,\ldots,\ceil{n^p}$}{
    Choose $i$ uniformly from $[n]$\;
    Set $h_i=\beta_j\sum_{\ell\ne i}A_{i\ell}x_\ell$\;
    Replace $x_i$ by $+1$ with probability $(1+\exp(-2h_i))^{-1}$,
    and by $-1$ otherwise\;
  }
}
\Return{$x$}\;
\end{algorithm}

We now analyze \Cref{alg:app-wpi-annealing}, verifying warmness between consecutive Gibbs distributions and choosing the update exponent.
\begin{proof}[Proof of \Cref{thm:app-wpi-annealing}]
If $\beta=0$, we directly sample from the uniform product distribution. 
We assume that $\beta > 0$ and  we use \Cref{alg:app-wpi-annealing} and \Cref{lem:app-wpi-annealing-composition} to prove this theorem.
There are $L\leq n$ stages in \Cref{alg:app-wpi-annealing}.

We first verify warmness between consecutive Gibbs {distributions}.
On $\norm A_{\mathrm{op}}\leq3$, the energy $H_A(x)=\inner{x}{Ax}/2$ satisfies $|H_A(x)|\leq\|A\|_{\mathrm{op}}\|x\|_2^2/2\leq3n/2$ for every spin $x$.  
Therefore $\sup_{x} H_A(x)-\inf_{x} H_A(x)\leq3n$.
For every spin $x$, the normalized Gibbs probabilities satisfy
\begin{equation*}
 \frac{\mu_{\beta_{j-1},A}(x)}{\mu_{\beta_j,A}(x)}
 =\e^{-(\beta_j-\beta_{j-1})H_A(x)}
   \E[X\sim\mu_{\beta_{j-1},A}]{
        \e^{(\beta_j-\beta_{j-1})H_A(X)}}.
\end{equation*}
We bound the first factor using $\inf H_A$ and the second factor using $\sup H_A$.  
Since $\beta_j-\beta_{j-1}\leq1/n$, we obtain
\begin{equation*}
 \left\|\frac{\d\mu_{\beta_{j-1},A}}
                  {\d\mu_{\beta_j,A}}\right\|_\infty
 \leq\exp(3n(\beta_j-\beta_{j-1}))\leq B_0\defeq\e^3.
\end{equation*}

Next, we apply \Cref{cor:app-wpi-schedule} to $\mathcal B_n=\{\beta_j:0\leq j\leq L\}$, which has at most $n+1\leq n^2$ points for large $n$, with $C_{\mathrm{sch}}=2$,  $D+1$, and $M_0=2M+5$.
Let $C\geq1$ and $a,\Gamma\geq0$ be constants in \Cref{cor:app-wpi-schedule}.
Set $\varepsilon=\e^{-6}n^{-(2M+4)}$ and $p=\bigl\lceil a+(2M+4)\Gamma\bigr\rceil+2$.
For all sufficiently large $n$, $\varepsilon/n^{-M_0}=\e^{-6}n\geq1$ and $\varepsilon\leq1/2$.
{We write $\mu_j=\mu_{\beta_j,A}$, and let $P_j$ be its random-scan heat-bath Glauber transition matrix.}
Combined with \Cref{lem:app-wpi-annealing-composition}, we have that
\begin{equation*}
 \DTV{\mu_0P_1^{\lceil n^p \rceil}\cdots P_L^{\lceil n^p \rceil}}{\mu_L}
 \leq\frac12\sum_{j=1}^L
      \left(\e^{- \lceil n^p \rceil / (C n^a \varepsilon^{-\Gamma})}(\e^3-1)+\varepsilon \e^6\right)^{1/2}.
\end{equation*}
Since $C n^a\varepsilon^{-\Gamma}=C\e^{6\Gamma}n^{a+(2M+4)\Gamma}\geq1$ and $p\geq a+(2M+4)\Gamma+2$, the exponent ratio in the first term is at least $n^2/(C\e^{6\Gamma})$.
Thus, for all sufficiently large $n$, both terms inside the parentheses are at most $n^{-2(M+2)}$, and the total variation distance is at most $L n^{-(M+2)}\leq n^{-(M+1)}\leq n^{-M}$.

By a union bound, the common WPI event and $\{A:\|A\|_{\mathrm{op}}\leq3\}$ hold simultaneously with probability at least $1-n^{-(D+1)}-C\e^{-cn}\geq1-n^{-D}$ for all sufficiently large $n$.
Each of the $L\ceil{n^p}$ updates costs $O(n)$ exact-real operations, so including the input matrix and initial spins, the running time is $O(n^2+nL\ceil{n^p})=n^{O_{\beta,M,D}(1)}$.

\end{proof}

\begin{proof}[Proof of \Cref{cor:intro-counting}]
If $\beta=0$, we output $\widehat Z=2^n$.
We assume that $\beta>0$ and estimate the ratios of consecutive partition functions along the schedule in \Cref{alg:app-wpi-annealing}.
We recall that $L=\ceil{\beta n}\leq n$, $\beta_j=j\beta/L$, and $H_A(x)=\inner{x}{Ax}/2$.
We write $Z_j=\sum_{x\in\{\pm1\}^n}\e^{\beta_jH_A(x)}$ and set $g_j(x)=\e^{(\beta_j-\beta_{j-1})H_A(x)}$ for $1\leq j\leq L$.
Then $Z_0=2^n$, $Z_L=Z_A(0)$, and
\begin{equation*}
 r_j\defeq\frac{Z_j}{Z_{j-1}}
 =\E[X\sim\mu_{\beta_{j-1},A}]{g_j(X)}.
\end{equation*}
Thus we can estimate $Z_A(0)=2^n\prod_{j=1}^Lr_j$ by estimating these expectations.

We first choose one sampler that works at every temperature in the schedule.
We use the proof of \Cref{thm:app-wpi-annealing} with $M+3$ in place of $M$ and the common event from \Cref{cor:app-wpi-schedule}.
We obtain a fixed update exponent $p$ and an event of $\mathbb P_{0,n}$-probability at least $1-n^{-D}$ on which $\|A\|_{\mathrm{op}}\leq3$ and $\DTV{\nu_j}{\mu_{\beta_j,A}}\leq n^{-(M+3)}$ for every $0\leq j\leq L$.
Here $\nu_j$ is the output law of \Cref{alg:app-wpi-annealing} stopped after stage $j$; by \Cref{lem:app-wpi-annealing-composition}, its error is bounded by the sum of the errors of the first $j$ stages.
We fix $A$ in this event for the rest of the proof.
Since $|H_A(x)|\leq3n/2$ and $\beta_j-\beta_{j-1}\leq1/n$, we have $\e^{-3/2}\leq g_j(x)\leq\e^{3/2}$ and hence $r_j\geq\e^{-3/2}$.

We set $\varepsilon=n^{-M}/(4L)$ and $s=\ceil{2\e^6\varepsilon^{-2}\log(2L/\eta)}$.
For each $j$, we use independent runs of the sampler stopped after stage $j-1$ to draw $X_{j,1},\ldots,X_{j,s}$ with law $\nu_{j-1}$, and set $\widehat r_j=s^{-1}\sum_{k=1}^s g_j(X_{j,k})$.
We output $\widehat Z=2^n\prod_{j=1}^L\widehat r_j$.
By the total-variation bound and the bounds on $g_j$ and $r_j$, for all sufficiently large $n$ we have
\begin{equation*}
 \left|\E{\widehat r_j\mid A}-r_j\right|
 \leq\e^{3/2}n^{-(M+3)}
 \leq\frac{\varepsilon r_j}{2}.
\end{equation*}
Applying Hoeffding's inequality to the independent variables $g_j(X_{j,k})$, we obtain
\begin{equation*}
 \Pr{\left|\widehat r_j/r_j-1\right|>\varepsilon\mid A}
 \leq2\exp\left(-\frac{s\varepsilon^2}{2\e^6}\right)
 \leq\frac{\eta}{L}.
\end{equation*}
By a union bound, with probability at least $1-\eta$, all $L$ relative errors are at most $\varepsilon$.
Since $\varepsilon\leq1/4$, on this event we have
\begin{equation*}
 \left|\log\frac{\widehat Z}{Z_A(0)}\right|
 \leq\sum_{j=1}^L\left|\log\frac{\widehat r_j}{r_j}\right|
 \leq2L\varepsilon=\frac{n^{-M}}{2},
\end{equation*}
so $|\widehat Z/Z_A(0)-1|\leq\exp(n^{-M}/2)-1\leq n^{-M}$.

Finally, $L\leq n$, $s=O(n^{2M+2}\log(2n/\eta))$, and every sample takes $n^{O_{\beta,M,D}(1)}$ time.
Combined, the total running time is $n^{O_{\beta,M,D}(1)}\log(1/\eta)$.
\end{proof}

\bibliographystyle{alpha}
\bibliography{references}

\appendix
\crefalias{section}{appendix}
\crefalias{subsection}{subappendix}
\crefalias{subsubsection}{subsubappendix}

\section{Contiguity of the planted model}
\label{sec:planted-contiguity}

\begin{proof}[Proof of \Cref{lem:planting-likelihood}]
We compute $L_n(A)$ by averaging the conditional Gaussian density of $A$ over the uniform spin in the planted model:
\begin{align*}
  L_n(A)
  =&~2^{-n}\sum_{x\in\{\pm1\}^n}
  \exp\left(\beta\sum_{i<j}A_{ij}x_ix_j-\frac{\beta^2(n-1)}4\right)\\
  =&~2^{-n}Z_A(0)\exp(-\beta^2(n-1)/4).
\end{align*}
Thus $L_n(A)=\exp(\log Z_A(0)-n\log2-\beta^2(n-1)/4)$.
This is the Rademacher specialization of \cite[Eq.~(3)]{ElAlaouiKrzakalaJordan20}, with $\lambda=\beta^2${, $\sigma=+\infty$,} and $Y_{ij}=\sqrt n A_{ij}$.
The corresponding asymptotic mutual contiguity for $\beta<1$ follows from \cite[Theorem~2 and Corollary~3]{ElAlaouiKrzakalaJordan20}.
Since $L_n$ is the density of $\mathbb P_{n,\beta}^{\mathrm{pl}}$ relative to $\mathbb P_{0,n}$, we have $\E[A\sim\mathbb P_{0,n}]{L_n(A)}=1$.

To bound the second moment, we take independent uniform spins $x,x'$ and set $\xi_i=x_ix_i'$.
Integrating over $A$ and then applying the Hubbard--Stratonovich identity with $G\sim N(0,1)$, we obtain
\begin{align*}
  \E[A\sim\mathbb P_{0,n}]{L_n(A)^2}
  =&~\E[\xi\sim\operatorname{Unif}(\{\pm1\}^n)]{\exp\left(\frac{\beta^2}{2n}\left((\sum_i\xi_i)^2-n\right)\right)}\\
  =&~\e^{-\beta^2/2}\E[G\sim N(0,1)]{\cosh(\beta G/\sqrt n)^n}\\
  \leq&~\e^{-\beta^2/2}(1-\beta^2)^{-1/2}
  \leq(1-\beta^2)^{-1/2}.
\end{align*}
Here we used $\log\cosh z\leq z^2/2$ and $\beta<1$ in the last line.

To control the inverse moment, let $g=(\sqrt n A_{ij})_{i<j}$ and $F(g)=\log L_n(A)$.
This finite log-sum-exp of affine functions of the standard Gaussian vector $g$ is convex and Lipschitz with constant $\beta\sqrt{(n-1)/2}$, so all its exponential moments are finite.
Set $\Theta(\lambda)=\lambda^{-1}\log\E[A\sim\mathbb P_{0,n}]{L_n(A)^\lambda}$ for $\lambda\ne0$, with $\Theta(0)=\E[A\sim\mathbb P_{0,n}]{\log L_n(A)}$.
By \cite[Theorem~1]{Chen23}, $\Theta$ is convex on $\R$.
Since $\Theta(1)=0$, its successive secant slopes satisfy $\Theta(0)-\Theta(-1)\leq\Theta(1)-\Theta(0)\leq\Theta(2)-\Theta(1)$, and therefore
\begin{equation*}
  \log\E[A\sim\mathbb P_{0,n}]{L_n(A)^{-1}}
  =-\Theta(-1)\leq2(\Theta(2)-\Theta(1))
  =\log\E[A\sim\mathbb P_{0,n}]{L_n(A)^2}.
\end{equation*}
The second-moment bound now also bounds the inverse moment.

For every event $E$ determined by $(X,A,d,Y)$, we apply Cauchy--Schwarz to \cref{eq:planting-density-ratio} and use the second-moment bound to obtain
\begin{align*}
  \mathbb P_{n,\beta}^{\mathrm{pl,aug}}\tp{E}
  =&~\E[(X,A,d,Y)\sim\mathbb P_{0,n,\beta}^{\mathrm{aug}}]{L_n\tp{A}\1_E}\\
  \leq&~\E[A\sim\mathbb P_{0,n}]{L_n\tp{A}^2}^{1/2}
    \mathbb P_{0,n,\beta}^{\mathrm{aug}}\tp{E}^{1/2}\\
  \leq&~(1-\beta^2)^{-1/4}\mathbb P_{0,n,\beta}^{\mathrm{aug}}\tp{E}^{1/2}.
\end{align*}
For the reverse direction, we first use the density ratio and the inverse-moment bound to compute
\begin{equation*}
  \E[A\sim\mathbb P_{n,\beta}^{\mathrm{pl}}]{L_n\tp{A}^{-2}}
  =\E[A\sim\mathbb P_{0,n}]{L_n\tp{A}^{-1}}
  \leq(1-\beta^2)^{-1/2}.
\end{equation*}
We then apply Cauchy--Schwarz to the reciprocal density and obtain
\begin{align*}
  \mathbb P_{0,n,\beta}^{\mathrm{aug}}\tp{E}
  =&~\E[(X,A,d,Y)\sim\mathbb P_{n,\beta}^{\mathrm{pl,aug}}]{L_n\tp{A}^{-1}\1_E}\\
  \leq&~\E[A\sim\mathbb P_{n,\beta}^{\mathrm{pl}}]{L_n\tp{A}^{-2}}^{1/2}
    \mathbb P_{n,\beta}^{\mathrm{pl,aug}}\tp{E}^{1/2}\\
  \leq&~(1-\beta^2)^{-1/4}\mathbb P_{n,\beta}^{\mathrm{pl,aug}}\tp{E}^{1/2}.
\end{align*}
\end{proof}

We use the forward probability comparison in \Cref{lem:planting-likelihood} to transfer the matrix-norm tail bound from the original augmented law to the planted law in \Cref{lem:inverse-poly-planted-norm-tail}, and the reverse comparison to transfer the planted fixed-time cost bound to the original augmented law in \Cref{lem:scalar-reference-path-error}.

\begingroup

\section{Omitted proofs of TAP residual bounds}
\label{sec:omitted-tap-residual-proofs}

This appendix restates the main-section results whose proofs are deferred here and supplies the auxiliary estimates used in those proofs.
We use the fixed $0\leq\beta_*<1$ from \Cref{sec:tap-residual-bounds}; for an $M$-spin planted model, $M\ell\leq\beta_*^2$ means its inverse temperature $\sqrt{M\ell}$ is at most $\beta_*$.
The first subsection proves \Cref{lem:tap-gaussian-cancellation}, and the second proves \Cref{lem:tap-parent-child}.

\subsection{Gaussian estimates for cavity interpolation}
\label{subsec:tap-row-proofs}

We first prove the supporting Gaussian estimates used in the proof of \Cref{lem:tap-gaussian-cancellation}.
For a smooth vector-valued function $F:\mathbb R^M\to\mathbb R^k$, we write
$D_GF(G)$ for its Jacobian with respect to the Gaussian coordinates, whose
$(j,k)$ entry is $\partial_{G_k}F_j(G)$.
When $F:\mathbb R^M\to\mathbb R^M$, we define its Gaussian divergence by
\[
 \delta_GF\defeq\inner G{F(G)}-\operatorname{Tr}(D_GF(G)).
\]
We use part~\textup{(i)} below to bound fluctuations of the cavity means.
Part~\textup{(ii)} controls their contraction with the same Gaussian row after subtracting the derivative trace, which accounts for the dependence of the means on that row.

\begin{lemma}[{\cite[Theorem~4]{LedouxOleszkiewicz07},
\cite[Theorem~3.20]{BoucheronLugosiMassart13},
\cite[Theorem~2.4]{Nualart19}}]
\label{lem:tap-gaussian-poincare-p}
Let $M,k\geq1$, $p\geq2$, and $G\sim N(0,I_M)$.
All $L^p(G)$ norms are taken under this Gaussian law, using the Euclidean
norm for vectors and the Frobenius norm for matrices.
There is a constant $C_p$, depending only on $p$, for which the following
statements hold.
\begin{enumerate}[label=\textup{(\roman*)}]
\item For every smooth $F:\mathbb R^M\to\mathbb R^k$ such that $F(G)$ and
$D_GF(G)$ belong to $L^p(G)$,
\begin{equation*}
 \|F(G)-\E[G\sim N(0,I_M)]{F(G)}\|_{L^p(G)}
 \leq C_p\norm{D_GF(G)}_{L^p(G)}.
\end{equation*}
For $p=2$, we may take $C_2=1$.

\item For every smooth $F:\mathbb R^M\to\mathbb R^M$ such that $F(G)$ and
$D_GF(G)$ belong to $L^p(G)$,
\begin{equation*}
 \|\delta_GF\|_{L^p(G)}
 \leq C_p\left(
  \norm{F(G)}_{L^p(G)}
  +\norm{D_GF(G)}_{L^p(G)}\right).
\end{equation*}
\end{enumerate}
\end{lemma}

\begin{proof}
For part~\textup{(i)}, we apply Pisier's inequality with the convex
function $z\mapsto\|z\|_2^p$ and an independent $G'\sim N(0,I_M)$:
\begin{equation*}
 \|F(G)-\E[G\sim N(0,I_M)]{F(G)}\|_{L^p(G)}
 \leq\frac\pi2\left\|D_GF(G)G'\right\|_{L^p(G,G')}
 \leq C_p\left\|D_GF(G)\right\|_{L^p(G)}.
\end{equation*}
For the last inequality, we condition on $G$ and apply Minkowski in
$L^{p/2}(G')$ to the sum of squared coordinates; each coordinate is
Gaussian with variance equal to the squared norm of the corresponding
Jacobian row.  At $p=2$, summing the scalar Gaussian Poincar\'e
inequalities over the coordinates of $F$ yields $C_2=1$.
Part~\textup{(ii)} is Meyer's inequality for the Gaussian divergence.
\end{proof}

For the cavity interpolation of \Cref{def:tap-cavity-pinning}, we fix $i$, set $M=n-1$, and write $a=a_i$, $G=G_i$, and $U=X_iX_{-i}$.
Given $(\mathcal F^{\setminus i},U)$, we have $a=\ell U+\sqrt\ell\,G$ with $G\sim N(0,I_M)$.
The cavity distribution $\mu_0=\mu^{\setminus i}$ and covariance $\Sigma_0=\Sigma^{\setminus i}$ still depend only on the surviving observations; they are not recomputed after revealing $U$.
The next two lemmas use this conditioning and average $G$ alone in $L^p(G)$; to bound $e_i$, we subsequently average over $U$ given $\mathcal F^{\setminus i}$.

\begin{lemma}
\label{lem:tap-weak-tilt}
For the cavity interpolation of \Cref{def:tap-cavity-pinning} with $M=n-1\geq1$, given $(\mathcal F^{\setminus i},U)$, the following bound holds for every $p\ge2$:
\begin{equation*}
 \sup_{|\gamma|\leq1}\norm{\Sigma_\gamma}_{L^p(G)}
 \leq\mathrm e^{40\ell M/p}
 \left\langle\mathcal D_{1234}^{p/2}\right\rangle_{\mu_0^{\otimes4}}^{1/p}.
\end{equation*}
\end{lemma}
\begin{proof}
For four replicas, the Radon--Nikodym derivative is
\[
 \frac{\mathrm d\mu_\gamma^{\otimes4}}
      {\mathrm d\mu_0^{\otimes4}}(\sigma^1,\ldots,\sigma^4)
 =\prod_{r=1}^4\frac{\mathrm e^{\gamma a\cdot\sigma^r}}{Z_\gamma}
 =Z_\gamma^{-4}\mathrm e^{\gamma a\cdot(\sigma^1+\cdots+\sigma^4)}.
\]
We set $S=\sigma^1+\cdots+\sigma^4$.  Jensen's inequality gives $Z_\gamma=\langle \mathrm e^{\gamma a\cdot\sigma}\rangle_0
 \ge \mathrm e^{\gamma a\cdot m_0}$, so
\[
 Z_\gamma^{-4}\mathrm e^{\gamma a\cdot S}
 \le \mathrm e^{\gamma a\cdot(S-4m_0)}.
\]
We fix the replicas and set $W=S-4m_0$.  Since every coordinate of $W$ lies
in $[-8,8]$, we have
\[
 |U\cdot W|\le8M,\qquad \norm{W}_2^2\le64M.
\]
We use $a=\ell U+\sqrt\ell G$ and evaluate the Gaussian MGF:
\begin{align*}
 \E[G\sim N(0,I_M)]{\mathrm e^{\gamma a\cdot W}}
 =&~\exp\left(\gamma\ell U\cdot W
              +\frac{\gamma^2\ell}{2}\norm{W}_2^2\right)
 \le \mathrm e^{8\ell M+32\ell M}=\mathrm e^{40\ell M}.
\end{align*}
For any nonnegative function $F$ of the four replicas, with no explicit dependence on $G$, we multiply by $F$ and interchange the finite replica sum and the Gaussian expectation:
\begin{equation*}
 \E[G\sim N(0,I_M)]{\langle F\rangle_{\mu_\gamma^{\otimes4}}}
 =\left\langle F\E[G\sim N(0,I_M)]{
       Z_\gamma^{-4}\mathrm e^{\gamma a\cdot S}}\right\rangle_{\mu_0^{\otimes4}}
 \leq\mathrm e^{40\ell M}\langle F\rangle_{\mu_0^{\otimes4}}.
\end{equation*}
To bound the tilted covariance, we apply
\cref{eq:tap-covariance-rectangle} under $\mu_\gamma$ and obtain
\[
 \operatorname{Tr}\Sigma_\gamma^2
 =\langle\mathcal D_{1234}\rangle_{\mu_\gamma^{\otimes4}}.
\]
For each fixed Gaussian vector $G$ and $p\ge2$, we apply Jensen's inequality
under $\mu_\gamma^{\otimes4}$ and obtain
\[
 \norm{\Sigma_\gamma}_{\mathrm F}^{p}
 =(\operatorname{Tr}\Sigma_\gamma^2)^{p/2}
 \le\langle\mathcal D_{1234}^{p/2}\rangle_{\mu_\gamma^{\otimes4}}.
\]
We apply the nonnegative comparison just proved with $F=\mathcal D_{1234}^{p/2}$ and obtain, for every $|\gamma|\leq1$,
\begin{equation*}
 \E[G\sim N(0,I_M)]{\norm{\Sigma_\gamma}_{\mathrm F}^{p}}
 \leq\mathrm e^{40\ell M}\left\langle\mathcal D_{1234}^{p/2}\right\rangle_{\mu_0^{\otimes4}}.
\end{equation*}
Taking $p$th roots and then the supremum over $|\gamma|\leq1$ proves the lemma.
\end{proof}

\begin{lemma}
\label{lem:tap-tilted-mean-displacement}
For the cavity interpolation of \Cref{def:tap-cavity-pinning}, assume $M=n-1\geq1$ and $0\leq\ell M\leq1$.
Given $(\mathcal F^{\setminus i},U)$, for every fixed $p\geq2$ and $|\gamma|\leq1$, we have
\begin{equation*}
 \norm{m_\gamma-m_0}_{L^p(G)}\leq C_p|\gamma|\sqrt\ell\sup_{|r|\leq1}\norm{\Sigma_r}_{L^p(G)}.
\end{equation*}
\end{lemma}

\begin{proof}
We first collect the derivatives with respect to $\gamma$.
With $(\mathcal F^{\setminus i},U,G)$ fixed, we have $\mu_\gamma(\sigma)=\mu_0(\sigma)\mathrm e^{\gamma a\cdot\sigma}/Z_\gamma$.
Differentiating this identity, we obtain
\[
 \frac{\mathrm d}{\mathrm d\gamma}\log Z_\gamma=a\cdot m_\gamma,\qquad
 \frac{\mathrm d}{\mathrm d\gamma} \mu_\gamma(\sigma)
 =\mu_\gamma(\sigma)a\cdot(\sigma-m_\gamma).
\]
Multiplying the second identity by $\sigma_j$ and summing gives
\[
 \frac{\mathrm d}{\mathrm d\gamma}(m_\gamma)_j
 =\left\langle\sigma_j\,a\cdot(\sigma-m_\gamma)\right\rangle_\gamma
 =\sum_k a_k
   \left\langle\sigma_j(\sigma_k-(m_\gamma)_k)\right\rangle_\gamma
 =\sum_k a_k(\Sigma_\gamma)_{jk}=(\Sigma_\gamma a)_j.
\]

We next collect the derivatives with respect to $G$.  Since
$\partial_{G_k}a_j=\sqrt\ell\,\1_{\{j=k\}}$, we have
\[
 \partial_{G_k}\log Z_\gamma
 =\gamma\sqrt\ell\,(m_\gamma)_k,\qquad
 \partial_{G_k}\mu_\gamma(\sigma)
 =\gamma\sqrt\ell\,\mu_\gamma(\sigma)
   (\sigma_k-(m_\gamma)_k).
\]
Consequently,
\[
 \partial_{G_k}(m_\gamma)_j
 =\gamma\sqrt\ell\,
   \left\langle\sigma_j(\sigma_k-(m_\gamma)_k)\right\rangle_\gamma
 =\gamma\sqrt\ell\,(\Sigma_\gamma)_{jk}.
\]
For a function $F=F(G,\sigma)$ differentiable in $G$, differentiating
both the Gibbs weight and $F$ gives
\begin{align*}
 \partial_{G_k}\langle F\rangle_\gamma
 =&~\sum_\sigma \partial_{G_k}\mu_\gamma(\sigma)F(G,\sigma)
   +\sum_\sigma\mu_\gamma(\sigma)\partial_{G_k}F(G,\sigma)\\
 =&~\gamma\sqrt\ell\,
   \left\langle F(\sigma_k-(m_\gamma)_k)\right\rangle_\gamma
   +\left\langle\partial_{G_k}F\right\rangle_\gamma.
\end{align*}
For
$F=(\sigma_j-(m_\gamma)_j)(\sigma_k-(m_\gamma)_k)$, we have
\[
 \partial_{G_k}F
 =-(\partial_{G_k}(m_\gamma)_j)(\sigma_k-(m_\gamma)_k)
  -(\partial_{G_k}(m_\gamma)_k)(\sigma_j-(m_\gamma)_j),
\]
and hence $\langle\partial_{G_k}F\rangle_\gamma=0$.  Using
$\sigma_k^2=1$, we obtain
\[
\begin{aligned}
 \partial_{G_k}(\Sigma_\gamma)_{jk}
 =&~\gamma\sqrt\ell
   \left\langle(\sigma_j-(m_\gamma)_j)
       (\sigma_k-(m_\gamma)_k)^2\right\rangle_\gamma\\
 =&~\gamma\sqrt\ell\bigl(1+(m_\gamma)_k^2\bigr)
       \left\langle\sigma_j-(m_\gamma)_j\right\rangle_\gamma
   -2\gamma\sqrt\ell\,(m_\gamma)_k
       \left\langle(\sigma_j-(m_\gamma)_j)\sigma_k\right\rangle_\gamma\\
 =&~-2\gamma\sqrt\ell\,(m_\gamma)_k(\Sigma_\gamma)_{jk}.
\end{aligned}
\]
Every entry of $\Sigma_\gamma$ is bounded, and the derivatives just
computed are bounded by constants depending on $M,\ell$.  The
$\gamma$-derivative is dominated by a constant times $1+\norm{G}_2$,
so differentiation under the Gaussian expectation and Gaussian
integration by parts are valid.

We now combine the two groups of derivatives by Gaussian integration by
parts.  Writing
$a=\ell U+\sqrt\ell G$ and integrating the Gaussian term by parts, we
obtain, component by component,
\begin{align*}
 \frac{\mathrm d}{\mathrm d\gamma}
 \E[G\sim N(0,I_M)]{(m_\gamma)_j}
 =&~\ell\E[G\sim N(0,I_M)]{(\Sigma_\gamma U)_j}
   +\sqrt\ell\sum_k
      \E[G\sim N(0,I_M)]{G_k(\Sigma_\gamma)_{jk}}\\
 =&~\ell\E[G\sim N(0,I_M)]{(\Sigma_\gamma U)_j}
   +\sqrt\ell\sum_k
      \E[G\sim N(0,I_M)]{
        \partial_{G_k}(\Sigma_\gamma)_{jk}}\\
 =&~\ell\E[G\sim N(0,I_M)]{
      (\Sigma_\gamma(U-2\gamma m_\gamma))_j}.
\end{align*}
Thus, in vector form,
\begin{equation}
 \frac{\mathrm d}{\mathrm d\gamma}\E[G\sim N(0,I_M)]{m_\gamma}
 =\ell\E[G\sim N(0,I_M)]{
      \Sigma_\gamma(U-2\gamma m_\gamma)}.
 \label{eq:tap-mean-derivative}
\end{equation}

Since $m_0$ does not depend on $G$, integrating
\cref{eq:tap-mean-derivative}, and applying the triangle and Jensen
inequalities, gives
\begin{align*}
 \norm{\E[G\sim N(0,I_M)]{m_\gamma}-m_0}_2
 \leq&~\ell\int_{\min(0,\gamma)}^{\max(0,\gamma)}
       \E[G\sim N(0,I_M)]{
          \norm{\Sigma_r(U-2rm_r)}_2}\mathrm d r\\
 \leq&~3\ell\sqrt M
       \int_{\min(0,\gamma)}^{\max(0,\gamma)}
       \norm{\Sigma_r}_{L^p(G)}\mathrm d r\\
 \leq&~3|\gamma|\ell\sqrt M
       \sup_{|r|\leq1}\norm{\Sigma_r}_{L^p(G)}.
\end{align*}
For the centered fluctuation, \Cref{lem:tap-gaussian-poincare-p}-\textup{(i)} yields
\begin{equation*}
 \norm{m_\gamma-\E[G\sim N(0,I_M)]{m_\gamma}}_{L^p(G)}
 \leq C_p|\gamma|\sqrt\ell\,
       \norm{\Sigma_\gamma}_{L^p(G)}.
\end{equation*}
Combining the two preceding inequalities with the triangle inequality,
and using $\ell M\leq1$, we obtain
\begin{align*}
 \norm{m_\gamma-m_0}_{L^p(G)}
 \leq&~\norm{m_\gamma-
       \E[G\sim N(0,I_M)]{m_\gamma}}_{L^p(G)}
       +\norm{\E[G\sim N(0,I_M)]{m_\gamma}-m_0}_2\\
 \leq&~C_p|\gamma|\sqrt\ell
       \sup_{|r|\leq1}\norm{\Sigma_r}_{L^p(G)}.
\end{align*}
This proves the lemma.
\end{proof}

\renewcommand{\restatedtheoremnumber}{\ref*{lem:tap-gaussian-cancellation}}
\begin{restatedlemma}[Restatement]
For the cavity interpolation of \Cref{def:tap-cavity-pinning}, assume $M=n-1\geq1$ and $0\leq\ell M\leq1$.
We condition on $(\mathcal F^{\setminus i},U)$ and let $x$ be any $[-1,1]$-valued random variable that may depend on $G$ and additional randomness.
Under this conditioning, $\E[\mathrm{row}]{\cdot}$ and $L^p(\mathrm{row})$ average over $G$ and the additional randomness used to construct $x$, while $\E[G\sim N(0,I_M)]{\cdot}$ and $L^p(G)$ average $G$ alone.
We define
\[
 \overline m_x=\tfrac{1+x}{2}m_{+1}+\tfrac{1-x}{2}m_{-1},
\]
\[
r_{\mathrm{row}}
 =\tfrac12(\log Z_{+1}-\log Z_{-1})-a\cdot \overline m_x
       +\ell(M-\norm{\overline m_x}_2^2)x.
\]
Then, for every $p\geq2$, there is a constant $C_p$, depending only on $p$, such that
\begin{equation*}
 \|r_{\mathrm{row}}\|_{L^p(\mathrm{row})}\leq C_p\ell\sup_{|\gamma|\leq1}\norm{\Sigma_\gamma}_{L^p(G)}.
\end{equation*}
\end{restatedlemma}
\begin{proof}
We fix $p\geq2$ and condition on $(\mathcal F^{\setminus i},U)$.
Under this conditioning, $m_0$ does not depend on $G$, and we have $a=\ell U+\sqrt\ell\,G$ and $\norm U_2=\sqrt M$.
Throughout the proof, $C_p$ denotes a constant depending only on $p$ that may change from line to line.
Using the definition of $\overline m_x$, we decompose the residual as
\[
 r_{\mathrm{row}}=A_i+xB_i,
\]
where
\[
 A_i \defeq\frac12(\log Z_{+1}-\log Z_{-1})
              -\frac12a\cdot(m_{+1}+m_{-1}),
\]
\[
 B_i \defeq\ell(M-\norm{\overline m_x}_2^2)
              -\frac12a\cdot(m_{+1}-m_{-1}).
\]
We estimate these two errors separately by subtracting the fixed cavity trace from the row response.

For $|\gamma|\leq1$, we set $F_\gamma\defeq m_\gamma-m_0$ and define
\[
 \Delta_\gamma\defeq a\cdot F_\gamma-\gamma\ell\operatorname{Tr}\Sigma_0.
\]
Since $D_GF_\gamma=\gamma\sqrt\ell\,\Sigma_\gamma$, we expand the Gaussian divergence as
\[
 \delta_GF_\gamma
 =G\cdot F_\gamma-\gamma\sqrt\ell\operatorname{Tr}\Sigma_\gamma.
\]
Substituting $a=\ell U+\sqrt\ell\,G$ into the definition of $\Delta_\gamma$, we obtain the exact decomposition
\begin{equation}
 \Delta_\gamma=\ell U\cdot F_\gamma+\sqrt\ell\,\delta_GF_\gamma
       +\gamma\ell(\operatorname{Tr}\Sigma_\gamma-\operatorname{Tr}\Sigma_0).
 \label{eq:tap-divergence-cancellation}
\end{equation}
For Ising spins, we have $\operatorname{Tr}\Sigma_\gamma=M-\norm{m_\gamma}_2^2$.
Using $\norm{m_\gamma}_2,\norm{m_0}_2\leq\sqrt M$, we obtain
\[
 |\operatorname{Tr}\Sigma_\gamma-\operatorname{Tr}\Sigma_0|
 =|\norm{m_\gamma}_2^2-\norm{m_0}_2^2|
 \leq2\sqrt M\,\norm{F_\gamma}_2.
\]
We apply the Gaussian divergence estimate in \Cref{lem:tap-gaussian-poincare-p}-\textup{(ii)} to $F_\gamma$ in \cref{eq:tap-divergence-cancellation}.
Together with $|\gamma|\leq1$ and the trace-difference bound, we obtain
\begin{align*}
 \norm{\Delta_\gamma}_{L^p(G)}
 \leq&~\ell\sqrt M\,\norm{F_\gamma}_{L^p(G)}
       +C_p\sqrt\ell\left(\norm{F_\gamma}_{L^p(G)}
              +|\gamma|\sqrt\ell\,\norm{\Sigma_\gamma}_{L^p(G)}\right)
 +2|\gamma|\ell\sqrt M\,\norm{F_\gamma}_{L^p(G)}\\
 \leq&~C_p(\sqrt\ell+\ell\sqrt M)\norm{F_\gamma}_{L^p(G)}
       +C_p|\gamma|\ell\norm{\Sigma_\gamma}_{L^p(G)}.
\end{align*}
By \Cref{lem:tap-tilted-mean-displacement}, we can upper bound the mean displacement by
\[
 \norm{F_\gamma}_{L^p(G)}
 \leq C_p|\gamma|\sqrt\ell
       \sup_{|r|\leq1}\norm{\Sigma_r}_{L^p(G)}.
\]
Substituting this bound into the preceding estimate and using $M\ell\leq1$, we obtain
\begin{equation}
\begin{aligned}
 \norm{\Delta_\gamma}_{L^p(G)}
 \leq C_p|\gamma|\ell(1+\sqrt{M\ell})
       \sup_{|r|\leq1}\norm{\Sigma_r}_{L^p(G)}
 \leq C_p|\gamma|\ell
       \sup_{|r|\leq1}\norm{\Sigma_r}_{L^p(G)}.
\end{aligned}
 \label{eq:tap-Delta-p}
\end{equation}

We now express $A_i$ and $B_i$ in terms of $\Delta_\gamma$.
By the definition of $\Delta_\gamma$, we have
\[
 a\cdot m_\gamma=a\cdot m_0+\gamma\ell\operatorname{Tr}\Sigma_0+\Delta_\gamma.
\]
We integrate $\frac{\mathrm d}{\mathrm d\gamma}\log Z_\gamma=a\cdot m_\gamma$ over $[-1,1]$ and use $\int_{-1}^1\gamma\,\mathrm d\gamma=0$ to obtain
\[
 \frac12(\log Z_{+1}-\log Z_{-1})
 =a\cdot m_0+\frac12\int_{-1}^1\Delta_\gamma\,\mathrm d\gamma.
\]
Taking the average of the two endpoint identities, we obtain
\[
 \frac12a\cdot(m_{+1}+m_{-1})
 =a\cdot m_0+\frac12(\Delta_1+\Delta_{-1}).
\]
Taking half their difference, we obtain
\[
 \frac12a\cdot(m_{+1}-m_{-1})
 =\ell\operatorname{Tr}\Sigma_0+\frac12(\Delta_1-\Delta_{-1}).
\]
Substituting these identities and $\operatorname{Tr}\Sigma_0=M-\norm{m_0}_2^2$ into the definitions of the two errors, we obtain
\[
A_i =\frac12\int_{-1}^1\Delta_\gamma\,\mathrm d\gamma
       -\frac12(\Delta_1+\Delta_{-1}),
\]
\[
B_i =\ell(\norm{m_0}_2^2-\norm{\overline m_x}_2^2)
       -\frac12(\Delta_1-\Delta_{-1}).
\]
Since $A_i$ depends only on $G$ under our conditioning, we apply Minkowski's inequality and \cref{eq:tap-Delta-p} to obtain
\begin{align*}
 \norm{A_i}_{L^p(G)}
\leq \frac12\int_{-1}^1\norm{\Delta_\gamma}_{L^p(G)}\,\mathrm d\gamma
       +\frac12\left(\norm{\Delta_1}_{L^p(G)}
                     +\norm{\Delta_{-1}}_{L^p(G)}\right)\leq C_p\ell\sup_{|\gamma|\leq1}\norm{\Sigma_\gamma}_{L^p(G)}.
\end{align*}
To estimate $B_i$, we use $\norm{\overline m_x}_2,\norm{m_0}_2\leq\sqrt M$ to obtain the pointwise bound
\[
 |\norm{\overline m_x}_2^2-\norm{m_0}_2^2|
 \leq2\sqrt M\,\norm{\overline m_x-m_0}_2.
\]
By convexity of the endpoint mixture, we also have
\begin{align*}
 \norm{\overline m_x-m_0}_2
 \leq \frac{1+x}{2}\norm{F_1}_2
       +\frac{1-x}{2}\norm{F_{-1}}_2
 \leq \norm{F_1}_2+\norm{F_{-1}}_2. 
\end{align*}
These pointwise inequalities use only $|x|\leq1$, so they remain valid when $x$ depends on $G$ and additional randomness.
Since $F_{\pm1}$ and $\Delta_{\pm1}$ depend only on $G$ under our conditioning, their $L^p(\mathrm{row})$ norms equal their $L^p(G)$ norms.
Using the pointwise bounds, \Cref{lem:tap-tilted-mean-displacement}, and \cref{eq:tap-Delta-p}, we obtain
\begin{align*}
 \norm{B_i}_{L^p(\mathrm{row})}
 \leq&~2\ell\sqrt M\left(\norm{F_1}_{L^p(G)}
                        +\norm{F_{-1}}_{L^p(G)}\right)
 \mathbin{+}\frac12\left(\norm{\Delta_1}_{L^p(G)}
               +\norm{\Delta_{-1}}_{L^p(G)}\right)\\
 \leq&~C_p\ell(1+\sqrt{M\ell})
       \sup_{|\gamma|\leq1}\norm{\Sigma_\gamma}_{L^p(G)}\\
 \leq&~C_p\ell\sup_{|\gamma|\leq1}\norm{\Sigma_\gamma}_{L^p(G)},
\end{align*}
where the last inequality uses $M\ell\leq1$.
Finally, using $r_{\mathrm{row}}=A_i+xB_i$ and $|x|\leq1$, we combine the two error estimates to obtain
\begin{align*}
 \norm{r_{\mathrm{row}}}_{L^p(\mathrm{row})}
 \leq\norm{A_i}_{L^p(G)}+\norm{B_i}_{L^p(\mathrm{row})}
 \leq C_p\ell\sup_{|\gamma|\leq1}\norm{\Sigma_\gamma}_{L^p(G)}.
\end{align*}
This proves the lemma.
\end{proof}

\subsection{Stability under \texorpdfstring{{deletion of one coordinate}}{deletion of one coordinate}}
\label{subsec:tap-stability-proofs}

In this subsection, we prove \Cref{lem:tap-parent-child}.

\renewcommand{\restatedtheoremnumber}{\ref*{lem:tap-parent-child}}
\begin{restatedlemma}[Restatement]
Under the planted model of \Cref{def:tap-planted-model}, set $J=\beta A^{\mathrm{full}}$.
If $n\geq2$ and $n\ell\leq\beta_*^2$, let $\mathbf r$ and
$\mathbf r^{\setminus i}$ be the residual before and after
deleting coordinate $i$.  Each $\mathbf r^{\setminus i}$ is
$\mathcal F^{\setminus i}$-measurable.  For every fixed integer $p\geq1$,
there is a constant $C_p$, depending only on $p,\beta_*,T$, such that,
for every $t\in[0,T]$,
\begin{equation*}
 \max_i\bigl\|\,\|\mathbf r\|_2-\|\mathbf r^{\setminus i}\|_2\,\bigr\|_{L^{4p}}\le C_pn^{-1/4}.
\end{equation*}
\end{restatedlemma}

We fix $p\geq2$ and a deleted coordinate $i$, and set $M=n-1$.  Let
$\iota_i:\mathbb R^{n-1}\to\mathbb R^n$ insert a zero at coordinate $i$.
{We recall that $J=\beta A^{\mathrm{full}}$, $J^{\setminus i}=J_{-i,-i}$, and $\ell=\beta^2/n$.}
By the triangle inequality,
we have
\begin{equation*}
 \bigl\|\,\|\mathbf r\|_2-\|\mathbf r^{\setminus i}\|_2\,\bigr\|_{L^p}
 \leq\|\mathbf r-\iota_i\mathbf r^{\setminus i}\|_{L^p}
 \leq\|r_i\|_{L^p}+\|\mathbf r_{-i}-\mathbf r^{\setminus i}\|_{L^p}.
\end{equation*}
By \Cref{lem:tap-complete-row} we
can upper bound $\|r_i\|_{L^p}$ by $C_pn^{-1/2}$.
We then focus on the second part.
We use the cavity representations from
\Cref{def:tap-cavity-pinning} and set
\[
 m^{\mathrm{av}}=\frac{m_{+1}+m_{-1}}{2},\qquad
 m^{\mathrm{dif}}=\frac{m_{+1}-m_{-1}}{2}.
\]
We also have that
$m_{-i}=\sum_b\pi_b m_b
=m^{\mathrm{av}}+m_i m^{\mathrm{dif}}$.
In particular, the pinned residual is
\[
 \mathbf r_{\mathrm{pin}}^b
 = \operatorname{arctanh}(m_b)-y_{-i}-b J_{-i,i}
   -J^{\setminus i}m_b
   +\ell(M-\norm{m_b}_2^2)m_b.
\]
We also recall definitions for residuals before and after
deleting coordinate $i$:
\[\mathbf r_{-i}
 = \operatorname{arctanh}(m_{-i})-y_{-i}-m_iJ_{-i,i}
    -J^{\setminus i}m_{-i}+\ell(n-\norm m_2^2)m_{-i},\]
\[
 \mathbf r^{\setminus i}
 = \operatorname{arctanh}(m^{\setminus i})-y_{-i}
    -J^{\setminus i}m^{\setminus i}
    +\ell(M-\|m^{\setminus i}\|_2^2)m^{\setminus i}.
\]
We claim that
\begin{equation*}
 \mathbf r_{-i}
 =\sum_{b=\pm1}\pi_b\mathbf r_{\mathrm{pin}}^b
   +\Delta_{\mathrm{ent}}+\Delta_Q,
\end{equation*}
where
\[
 \Delta_{\mathrm{ent}}
 =\operatorname{arctanh}(m_{-i})
       -\sum_{b=\pm1}\pi_b
                         \operatorname{arctanh}(m_b),
                         \]
 \[
 \Delta_Q
 = \ell(1-m_i^2)\left(
     (1+\norm{m^{\mathrm{dif}}}_2^2)m_{-i}
       +2\inner{m^{\mathrm{av}}}{m^{\mathrm{dif}}}m^{\mathrm{dif}}
                  \right).\]
Since $\sum_{b=\pm1}\pi_b=1$, subtracting
$\mathbf r^{\setminus i}$ from this decomposition yields
\begin{equation}
 \mathbf r_{-i}-\mathbf r^{\setminus i}
 =\sum_{b=\pm1}\pi_b
       (\mathbf r_{\mathrm{pin}}^b-\mathbf r^{\setminus i})
   +\Delta_{\mathrm{ent}}+\Delta_Q.
 \label{eq:tap-recombination}
\end{equation}

We now prove the claim.  By
$\sum_b\pi_b=1$,
$\sum_bb\pi_b=m_i$, and
$\sum_b\pi_b m_b=m_{-i}$, we have
\begin{equation*}
 \sum_b\pi_b
       (y_{-i}+b J_{-i,i}+J^{\setminus i}m_b)
 =y_{-i}+m_iJ_{-i,i}+J^{\setminus i}m_{-i}.
\end{equation*}
Substituting the formulas for the residuals and the definition of
$\Delta_{\mathrm{ent}}$, we obtain
\begin{align*}
 \mathbf r_{-i}-\sum_b\pi_b\mathbf r_{\mathrm{pin}}^b
 =&~\Delta_{\mathrm{ent}}+\ell(n-\norm m_2^2)m_{-i}
   -\ell\sum_b\pi_b
       (M-\norm{m_b}_2^2)m_b\\
 =&~\Delta_{\mathrm{ent}}+\ell(1-m_i^2)m_{-i}
   +\ell\left(\sum_b\pi_b
       \norm{m_b}_2^2m_b
       -\norm{m_{-i}}_2^2m_{-i}\right).
\end{align*}
For the second equality, we used $n=M+1$ and
$\norm m_2^2=m_i^2+\norm{m_{-i}}_2^2$.
It remains to show that the two terms multiplied by $\ell$ sum to
$\Delta_Q$.  We expand $m_b
=m^{\mathrm{av}}+b m^{\mathrm{dif}}$ and use
$\sum_b\pi_bb=m_i$:
\begin{align*}
 \sum_b\pi_b\|m_b\|_2^2m_b
 =&~(\|m^{\mathrm{av}}\|_2^2+\|m^{\mathrm{dif}}\|_2^2)m_{-i}
 +2\inner{m^{\mathrm{av}}}{m^{\mathrm{dif}}}
         (m_i m^{\mathrm{av}}+m^{\mathrm{dif}}),
\end{align*}
\begin{align*}
 \|m_{-i}\|_2^2
 =&~\|m^{\mathrm{av}}\|_2^2
       +2m_i\inner{m^{\mathrm{av}}}{m^{\mathrm{dif}}}
       +m_i^2\|m^{\mathrm{dif}}\|_2^2.
\end{align*}
After subtracting the second expression times $m_{-i}$ from the first,
we use
$m_i m^{\mathrm{av}}+m^{\mathrm{dif}}-m_i m_{-i}
=(1-m_i^2)m^{\mathrm{dif}}$ and obtain
\begin{equation*}
 \sum_b\pi_b\norm{m_b}_2^2m_b
 -\norm{m_{-i}}_2^2m_{-i}
 =(1-m_i^2)\left(\norm{m^{\mathrm{dif}}}_2^2m_{-i}
       +2\inner{m^{\mathrm{av}}}{m^{\mathrm{dif}}}m^{\mathrm{dif}}\right).
\end{equation*}
Substituting this identity into the residual difference, we obtain
\begin{equation*}
 \mathbf r_{-i}-\sum_b\pi_b\mathbf r_{\mathrm{pin}}^b
 =\Delta_{\mathrm{ent}}+\ell(1-m_i^2)
 \left((1+\norm{m^{\mathrm{dif}}}_2^2)m_{-i}
       +2\inner{m^{\mathrm{av}}}{m^{\mathrm{dif}}}m^{\mathrm{dif}}\right)
 =\Delta_{\mathrm{ent}}+\Delta_Q.
\end{equation*}

Later, under $X_i=b$ we will bound
$\mathbf r_{\mathrm{pin}}^b-\mathbf r^{\setminus i}$ by $C_pn^{-1/4}$ in $L^p$.
Under the unconditional planted law, we will also prove
$\|\Delta_{\mathrm{ent}}\|_{L^p}+\|\Delta_Q\|_{L^p}\leq C_pn^{-1/2}$.
{We first construct a localization path started from $y_{-i}$ to compare the cavity and pinned residuals, and then estimate the two remainders $\Delta_{\mathrm{ent}}$ and $\Delta_Q$.}

\subsubsection{From the cavity law to the pinned law}
We fix $i$ and $b\in\{\pm1\}$.
By \Cref{def:tap-cavity-pinning}, the cavity and pinned distributions have the same interaction $J^{\setminus i}$ and external fields $y_{-i}$ and $y_{-i}+ba_i$, respectively.
We connect these fields by a localization path.
{We condition on $X_i=b$ and sample $\xi$ according to $\mathcal L(\xi\mid\mathcal F^{\setminus i},X_i=b)=\mu^{\setminus i}$.}
We sample a standard Brownian motion $(\widetilde B_s)_{0\leq s\leq\ell}$
in $\mathbb R^M$, starting from $\widetilde B_0=0$ and independent of
$(\mathcal F^{\setminus i},X_i,\xi)$, and set
\begin{equation*}
 a_s=s\xi+\widetilde B_s,\qquad 0\leq s\leq\ell.
\end{equation*}
Here $a_s$ is the additional external field, so the total external field at time
$s$ is $y_{-i}+a_s$.

We next check that, given $(\mathcal F^{\setminus i},X_i=b)$,
$a_\ell$ and $b J_{-i,i}$ have the same conditional distribution.
{Given $(X,\mathcal F^{\setminus i})$, \Cref{def:tap-planted-model} and $J=\beta A^{\mathrm{full}}$ imply that the signed entries $b J_{j,i}$, $j\ne i$, are independent $N(\ell X_j,\ell)$ variables.}  Also,
$X_{-i}\mid\mathcal F^{\setminus i},X_i=b$ has law
$\mu^{\setminus i}$, since $X_i$ is independent of
$(X_{-i},J^{\setminus i},y_{-i})$.
On the other hand, $a_\ell\mid(\mathcal F^{\setminus i},X_i=b,\xi)$
has law $N(\ell\xi,\ell I_M)$.  Integrating these Gaussian laws over
$X_{-i}$ and $\xi$, respectively, with the same distribution
$\mu^{\setminus i}$, we obtain the same conditional law for
$a_\ell$ and $b J_{-i,i}$
given $(\mathcal F^{\setminus i},X_i=b)$.
Thus, for every nonnegative measurable function $F$, we have
\begin{equation*}
 \E{F(J^{\setminus i},y_{-i},y_{-i}+bJ_{-i,i})\mid\mathcal F^{\setminus i},X_i=b}
 =\E{F(J^{\setminus i},y_{-i},y_{-i}+a_\ell)\mid\mathcal F^{\setminus i},X_i=b}.
\end{equation*}
\begin{equation}
 \E{F(J^{\setminus i},y_{-i},y_{-i}+bJ_{-i,i})\mid X_i=b}
 =\E{F(J^{\setminus i},y_{-i},y_{-i}+a_\ell)\mid X_i=b}.
 \label{eq:tap-matched-endpoint-law}
\end{equation}

{
We next study the residual at $y_{-i}+a_s$ by deriving an SDE.
}
Let $M$ be an integer.
For $v\in\mathbb R^M$, we use the shorthand $\mu_J(v)$ for the distribution denoted by $\mu_{J,y}$ in \Cref{sec:tap-residual-bounds}, with $y$ replaced by $v$:
\begin{equation*}
 \mu_J(v)(\sigma)=\frac{\exp(\inner{\sigma}{J\sigma}/2+\inner{v}{\sigma})}{\sum_{\tau\in\{\pm1\}^M}\exp(\inner{\tau}{J\tau}/2+\inner{v}{\tau})},\qquad \sigma\in\{\pm1\}^M.
\end{equation*}
When $J=\beta_MA^{\mathrm{full}}$ under the $M$-spin planted model of \Cref{def:tap-planted-model}, this is exactly the posterior distribution $\mu_{A,v}$ defined in \Cref{subsec:prelim-planted-model} because the diagonal of $J$ does not affect the distribution.
We set
\begin{equation*}
 m_J(v)=\E[\sigma\sim\mu_J(v)]{\sigma},\qquad \Sigma(v)=\E[\sigma\sim\mu_J(v)]{(\sigma-m_J(v))(\sigma-m_J(v))^\top},
\end{equation*}
and write $m_{J,j}(v)$ for coordinate $j$ of $m_J(v)$.
For this $M$-spin system, we define the residual function at external field $y$ by
\begin{equation}
 \mathbf r(y)=\operatorname{arctanh}\tp{{m_J}(y)}-y-J{m_J}(y)
                     +\ell(M-\norm{{m_J}(y)}_2^2){m_J}(y).
 \label{eq:tap-M-residual}
\end{equation}
{
Later, when applying \cref{eq:tap-M-residual} to the cavity path, we use the matrix $J^{\setminus i}$ and dimension $M=n-1$.
With this matrix, $\mathbf r(y_{-i})=\mathbf r^{\setminus i}$, and \cref{eq:tap-matched-endpoint-law} compares the joint law of this initial residual and $\mathbf r(y_{-i}+a_\ell)$ with that of $\mathbf r^{\setminus i}$ and $\mathbf r_{\mathrm{pin}}^b$ under $X_i=b$.
}
\begin{lemma}
\label{lem:tap-residual-sde}
Fix a symmetric matrix $J\in\mathbb R^{M\times M}$, $\ell\geq0$, and
an initial external field $y\in\mathbb R^M$.  Let $\mu_{J,y}$ be the Gibbs {distribution} with
interaction $J$ and external field $y$, with mean ${m_J}(y)$ and covariance $\Sigma(y)$,
and let $\mathbf r$ be defined by \cref{eq:tap-M-residual}.
Sample $\xi\sim\mu_{J,y}$ and independently a standard Brownian motion
$(\widetilde B_s)_{s\geq0}$ in $\mathbb R^M$, starting from zero.
Set $y_s=y+s\xi+\widetilde B_s$ and
$\mathcal F_s=\sigma((y_u)_{0\leq u\leq s})$.
Then $\widetilde B_s^{\mathrm{loc}}
=y_s-y-\int_0^s {m_J}(y_u)\mathrm d u$ is Brownian motion in this filtration, and
\begin{equation*}
 \mathrm d\mathbf r(y_s)=\left.\frac{\partial\mathbf r(v)}{\partial v}\right|_{v=y_s}\mathrm d\widetilde B_s^{\mathrm{loc}}
                    +\mathbf b(y_s)\mathrm d s.
\end{equation*}
For every external field $v\in\mathbb R^M$, with
$D(v)=\operatorname{diag}((1-{m_{J,j}}(v)^2)^{-1})$, the coefficients are
\begin{align*}
 \frac{\partial\mathbf r(v)}{\partial v}=&~(D(v)-J+\ell(M-\norm{{m_J}(v)}_2^2)I_M
                   -2\ell {m_J}(v){m_J}(v)^\top)\Sigma(v)-I_M,\\
 \mathbf b(v)=&~{m_J}(v)\mathbin{\odot}
   (\operatorname{diag}(D(v)\Sigma(v)^2D(v))-\1) -\ell\operatorname{Tr}(\Sigma(v)^2){m_J}(v)-2\ell \Sigma(v)^2{m_J}(v).
\end{align*}
Here $\odot$ denotes coordinatewise multiplication and
$\operatorname{diag}$ applied to a matrix denotes its diagonal vector.
\end{lemma}

\begin{proof}
We apply \Cref{lem:app-wpi-observation-law} in dimension $M$ with base distribution $\mu=\mu_{J,y}$ and additional observation $y_s-y=s\xi+\widetilde B_s$.
The tilt of this base distribution by a field $v$ is $\mu_{J,y+v}$, so its mean and covariance at $v=y_s-y$ are ${m_J}(y_s)$ and $\Sigma(y_s)$.
Thus the innovation process is precisely $\widetilde B^{\mathrm{loc}}$ from the statement.
We set $m_s={m_J}(y_s)$, $\Sigma_s=\Sigma(y_s)$, and $D_s=\operatorname{diag}((1-m_{s,j}^2)^{-1})$.
By the field SDE in \Cref{lem:app-wpi-observation-law}, we have
\begin{equation*}
 \mathrm d y_s=m_s\,\mathrm d s+\mathrm d\widetilde B_s^{\mathrm{loc}}.
\end{equation*}
By the mean SDE in the same lemma, we have
\begin{equation*}
 \mathrm d m_s=\Sigma_s\,\mathrm d\widetilde B_s^{\mathrm{loc}}.
\end{equation*}
Multiplying the Brownian coefficients, we obtain the quadratic covariation
\begin{equation*}
 \mathrm d[m_j,m_k]_s=(\Sigma_s^2)_{jk}\,\mathrm d s.
\end{equation*}
The mean ${m_J}(v)$ is a ratio of finite sums of exponentials with positive denominator, so it is smooth in $v$; full support of the Gibbs distribution gives ${m_J}(v)\in(-1,1)^M$ for every $v$.
We can therefore apply It\^o's formula to the three terms of the residual along the continuous field path.

We now apply these identities to the residual
\begin{equation*}
 \mathbf r(y_s)=\operatorname{arctanh}\tp{m_s}
                 -(y_s+Jm_s)+\ell(M-\norm{m_s}_2^2)m_s.
\end{equation*}
We use $(\operatorname{arctanh})'(v)=(1-v^2)^{-1}$ and
$(\operatorname{arctanh})''(v)=2v(1-v^2)^{-2}$ in It\^o's formula to obtain
\begin{equation*}
 \mathrm d \operatorname{arctanh}\tp{m_{s,j}}
 =\frac{\mathrm d m_{s,j}}{1-m_{s,j}^2}
   +\frac{m_{s,j}(\Sigma_s^2)_{jj}}{(1-m_{s,j}^2)^2}\mathrm d s.
\end{equation*}
We write the resulting vector SDE as
\begin{equation*}
 \mathrm d \operatorname{arctanh}\tp{m_s}
 =D_s\Sigma_s\mathrm d \widetilde B_s^{\mathrm{loc}}
   +m_s\mathbin{\odot}\operatorname{diag}(D_s\Sigma_s^2D_s)\mathrm d s.
\end{equation*}
For $-y_s-Jm_s$, we substitute the field and mean SDEs:
\begin{equation*}
 \mathrm d (-y_s-Jm_s)=-\mathrm d y_s-J\mathrm d m_s
 =-(I_M+J\Sigma_s)\mathrm d \widetilde B_s^{\mathrm{loc}}-m_s\mathrm d s.
\end{equation*}
For $\ell(M-\norm{m_s}_2^2)m_s$, we first compute
\begin{equation*}
 \mathrm d \norm{m_s}_2^2
 =2m_s^\top \Sigma_s\mathrm d \widetilde B_s^{\mathrm{loc}}
   +\operatorname{Tr}(\Sigma_s^2)\mathrm d s.
\end{equation*}
The Brownian coefficient of $\norm{m_s}_2^2$ is $2m_s^\top \Sigma_s$.
We multiply it by the Brownian coefficient of $m_{s,j}$ to obtain
\begin{equation*}
 \mathrm d[\norm m_2^2,m_j]_s
 =2\sum_k(\Sigma_sm_s)_k(\Sigma_s)_{jk}\mathrm d s
 =2(\Sigma_s^2m_s)_j\mathrm d s.
\end{equation*}
We substitute these two identities in It\^o's product formula:
\begin{align*}
 \mathrm d (\ell(M-\norm{m_s}_2^2)m_s)
 =&~\ell(M-\norm{m_s}_2^2)\mathrm d m_s
    -\ell m_s\mathrm d \norm{m_s}_2^2-\ell\mathrm d [\norm m_2^2,m]_s\\
 =&~(\ell(M-\norm{m_s}_2^2)I_M-2\ell m_sm_s^\top)\Sigma_s
            \mathrm d \widetilde B_s^{\mathrm{loc}} -\ell(\operatorname{Tr}(\Sigma_s^2)m_s+2\Sigma_s^2m_s)\mathrm d s.
\end{align*}
Adding the three drift coefficients, we obtain
\begin{equation*}
 m_s\mathbin{\odot}\operatorname{diag}(D_s\Sigma_s^2D_s)-m_s
 -\ell\operatorname{Tr}(\Sigma_s^2)m_s-2\ell \Sigma_s^2m_s=\mathbf b(y_s).
\end{equation*}
For the Brownian coefficient, we differentiate the residual in $y$
using $\partial {m_J}(y)/\partial y=\Sigma(y)$ and obtain
\begin{equation*}
 \frac{\partial\mathbf r(y)}{\partial y}
 =D(y)\Sigma(y)-I_M-J\Sigma(y)
   +(\ell(M-\norm{{m_J}(y)}_2^2)I_M-2\ell {m_J}(y){m_J}(y)^\top)\Sigma(y),
\end{equation*}
and $\mathrm d\mathbf r(y_s)=(\partial\mathbf r/\partial y)(y_s)
\mathrm d\widetilde B_s^{\mathrm{loc}}+\mathbf b(y_s)\mathrm d s$.
\end{proof}

{
We condition on $(\mathcal F^{\setminus i},X_i=b)$ and apply \Cref{lem:tap-residual-sde} with interaction matrix $J^{\setminus i}$ and initial field $y_{-i}$, so that the field path is $y_{-i}+a_s$.
}
We write $\mathbf b_s=\mathbf b(y_{-i}+a_s)$.
By \Cref{lem:tap-residual-sde} and
\cref{eq:tap-matched-endpoint-law}, under $X_i=b$ we have
\begin{equation*}
 \mathbf r_{\mathrm{pin}}^b-\mathbf r^{\setminus i}
 \overset{\mathrm d}{=}\int_0^\ell \left.\frac{\partial\mathbf r(v)}{\partial v}\right|_{v=y_{-i}+a_s}\d \widetilde B_s^{\mathrm{loc}}
   +\int_0^\ell\mathbf b_s\d s.
\end{equation*}
{Applying \Cref{lem:bdg-stochastic-integral} under the conditional law given $X_i=b$ with $T=\ell$ and $G_s=(\partial\mathbf r/\partial y)(y_{-i}+a_s)$, and using Minkowski's inequality for the drift term, we obtain}
\begin{equation}
 \|\mathbf r_{\mathrm{pin}}^b-\mathbf r^{\setminus i}\|_{L^p}
 \leq C_p\left\|\left(\int_0^\ell\norm{\left.\frac{\partial\mathbf r(v)}{\partial v}\right|_{v=y_{-i}+a_s}}_{\mathrm F}^2\d s\right)^{1/2}
                    \right\|_{L^p}
       +\int_0^\ell\|\mathbf b_s\|_{L^p}\d s.
 \label{eq:tap-pinned-bdg}
\end{equation}
Since $\ell=O(n^{-1})$ and $M=n-1$, it suffices to prove
\begin{equation*}
 \norm{\left.\frac{\partial\mathbf r(v)}{\partial v}\right|_{v=y_{-i}+a_s}}_{L^p}\leq C_pM^{1/4},
 \qquad \|\mathbf b_s\|_{L^p}\leq C_p\sqrt M.
\end{equation*}

We next provide some useful {lemmas}.
\begin{lemma}
\label{lem:tap-influence}
The following statements hold.
\begin{enumerate}[label=\textup{(\roman*)}]
\item {Let $\mu$ be a probability distribution on $\{\pm1\}^n$ with mean $m$ satisfying $|m_j|<1$ for every $j$, and define the covariance $\Sigma$, and $D=\operatorname{diag}((1-m_j^2)^{-1})$.}
For each $i$, define the influence $\delta^{(i)}=(\delta_j^{(i)})_{j\ne i}$ by
\begin{equation*}
 \delta_j^{(i)}
 =\Pr[\sigma\sim\mu]{\sigma_j=1\mid\sigma_i=1}
 -\Pr[\sigma\sim\mu]{\sigma_j=1\mid\sigma_i=-1}.
\end{equation*}
Then we have that $(D\Sigma^2D)_{ii}-1=\|\delta^{(i)}\|_2^2$ for every $i$.
\item Under the planted model of \Cref{def:tap-planted-model}, set $J=\beta A^{\mathrm{full}}$.
If $n\ell\leq\beta_*^2$ and $t\geq0$, then for every fixed $p\geq2$
there is a constant $C_p$, depending only on $p,\beta_*$, such that
\begin{equation*}
 \sup_i\|\delta^{(i)}\|_{L^p}\leq C_p,
 \qquad \sup_i\|\Sigma_{i,\cdot}\|_{L^p}\leq C_p.
\end{equation*}
\end{enumerate}
\end{lemma}
\begin{proof}
For part~\textup{(i)}, a function of a binary spin is affine in that
spin, so
\begin{equation*}
 \E[\sigma\sim\mu]{\sigma_j\mid\sigma_i}-m_j
 =\delta_j^{(i)}(\sigma_i-m_i)\qquad(j\ne i).
\end{equation*}
We multiply by $\sigma_i-m_i$ and take expectation to obtain
$\Sigma_{ij}=(1-m_i^2)\delta_j^{(i)}$.
Since $\Sigma_{ii}=1-m_i^2$ and $\Sigma$ is symmetric, we have
\begin{equation*}
 (D\Sigma^2D)_{ii}
 =\frac{\sum_j\Sigma_{ij}^2}{(1-m_i^2)^2}
 =1+\|\delta^{(i)}\|_2^2.
\end{equation*}

{For part~\textup{(ii)}, we have that $|m_j|<1$ for every $j$ and the coefficients in part~\textup{(i)} are defined.}
For $n=1$, the influence vector is empty and $|\Sigma_{11}|=1-m_1^2\leq1$, so both bounds hold.
We now assume $n\geq2$.
For the pinned means of \Cref{def:tap-cavity-pinning}, we have $\delta^{(i)}=(m_{+1}-m_{-1})/2$ and $m_0=m^{\setminus i}$.
By the triangle inequality, we obtain
\[\|\delta^{(i)}\|_{L^p}
 \leq\tfrac12\bigl(\|m_{+1}-m^{\setminus i}\|_{L^p}+\|m_{-1}-m^{\setminus i}\|_{L^p}\bigr).\]
We use the cavity interpolation to bound both mean differences.
{
We fix $i$ and set $M=n-1$, $U=X_iX_{-i}$, and $G=G_i$.
We condition on $(\mathcal F^{\setminus i},U)$.
Then $m_0$ is fixed, and each $m_b$ depends on the remaining randomness only through $G$.
Here $L^p(G)$ averages only over $G$ with $(\mathcal F^{\setminus i},U)$ fixed, whereas the unsubscripted $L^p$ norm below averages over the full planted law.
Since $M\ell\leq n\ell\leq\beta_*^2<1$, we have that }
\begingroup

\begin{align*}
 \left(\E{\|m_b-m^{\setminus i}\|_2^p\mid\mathcal F^{\setminus i},U}\right)^{1/p}
 =&~\|m_b-m_0\|_{L^p(G)}\\
 (\text{by \Cref{lem:tap-tilted-mean-displacement}})\quad\leq&~C_p\sqrt\ell\sup_{|r|\leq1}\|\Sigma_r\|_{L^p(G)}\\
 (\text{by \Cref{lem:tap-weak-tilt}})\quad\leq&~C_p\sqrt\ell\,\mathrm e^{40\ell M/p}
 \left\langle\mathcal D_{1234}^{p/2}\right\rangle_{(\mu^{\setminus i})^{\otimes4}}^{1/p}.
\end{align*}
\endgroup
\begingroup

Next, we take expectations over $\+F^{\setminus i}$ and $U$,
\begin{align*}
 \|m_b-m^{\setminus i}\|_{L^p}^p
 =&~\E{\E{\|m_b-m^{\setminus i}\|_2^p\mid\mathcal F^{\setminus i},U}}\\
 \leq&~C_p^p\ell^{p/2}\mathrm e^{40\ell M}
 \E{\left\langle\mathcal D_{1234}^{p/2}\right\rangle_{(\mu^{\setminus i})^{\otimes4}}}\\
 (\text{by \Cref{cor:tap-covariance}-\textup{(i)}})\quad\leq&~C_p^p\ell^{p/2}\mathrm e^{40\ell M}(K\sqrt{Mp})^p
 \leq\bigl(C_p\sqrt{\ell M}\bigr)^p.
\end{align*}
\endgroup
{In the last inequality, we absorb $\mathrm e^{40\ell M}$ and the fixed $p$-dependent factors into $C_p^p$.  Taking $p$th roots and substituting the resulting bound for $b=1,-1$ into the triangle inequality above, we obtain that}
\begingroup

\begin{equation*}
 \|\delta^{(i)}\|_{L^p}\leq C_p\sqrt{\ell M}\leq C_p.
\end{equation*}
\endgroup
Finally, the covariance identity from part~\textup{(i)} gives
\begin{equation*}
 \|\Sigma_{i,\cdot}\|_2
 =(1-m_i^2)\sqrt{1+\|\delta^{(i)}\|_2^2}
 \leq1+\|\delta^{(i)}\|_2.\qedhere
\end{equation*}
\end{proof}
By the above lemma, the next lemma provides upper {bounds} for the residual, covariance matrix and the drift through the localization path.
\begin{lemma}
\label{lem:tap-residual-dynamics}
{Consider the $M$-spin planted model of \Cref{def:tap-planted-model} at inverse temperature $\beta_M=\sqrt{M\ell}$, with $J=\beta_MA^{\mathrm{full}}$ and $y=Y_t$.}
{Given $(J,y)$, draw $\xi\sim\mu_{J,y}$.  Independently of $(J,y,\xi)$, draw a standard Brownian motion $(\widetilde B_s)_{s\geq0}$ in $\mathbb R^M$ starting from zero, and set} ${y_s=y+s\xi+\widetilde B_s}$.

{We use $\mathbf r$, $\mathbf b$, and $\Sigma$ from \Cref{lem:tap-residual-sde} with dimension $M$, interaction matrix $J$, and the same $\ell$.}
{Suppose that $M\ell\leq\beta_*^2$.}
There is a constant $K>0$, depending only on $\beta_*$, such that, for every fixed $p\geq2$, there is a constant $C_p$, depending only on $p,\beta_*$, such that, for all $t,s\geq0$,
\begin{equation*}
 \|\mathbf r(y_s)\|_{L^p}\le C_p,\qquad
 \norm{\Sigma(y_s)}_{L^p}\le K\sqrt{Mp},\qquad
 \|\mathbf b(y_s)\|_{L^p}\le C_p\sqrt M{.}
\end{equation*}
Here, $L^p$ norms involving this path average jointly over $(J,y,\xi,\widetilde B)$.
\end{lemma}
\begin{proof}
We fix $t,s\geq0$.
By the posterior identity, $(J,y,\xi)$ and $(J,y,X)$ have the same joint law.
In the representation below, $g\sim N(0,I_M)$ is independent of $X$, the interaction noises, and the Brownian motion $(\widetilde B_s)_{s\geq0}$.
Since $\widetilde B$ is independent, we have
\begin{equation}
 (J,y_s)\overset{\mathrm d}{=}
 (J,(t+s)X+\sqrt t\,g+\widetilde B_s).
 \label{eq:tap-path-planted-marginal}
\end{equation}
{Given $X$, the noise on the right has law $N(0,(t+s)I_M)$ and is independent of the interaction noises.}
Thus we can apply the planted-model estimates at time $t+s$.

By \Cref{lem:tap-complete-row} at dimension $M$ and time $t+s$,
we have $\norm{r_i(y_s)}_{L^p}\leq C_pM^{-1/2}$.
Since $p/2\geq1$, we apply Minkowski's inequality in $L^{p/2}$:
\begin{equation*}
 \norm{\mathbf r(y_s)}_{L^p}^2
 =\left\|\sum_{i=1}^M r_i(y_s)^2\right\|_{L^{p/2}}
 \leq\sum_{i=1}^M\norm{r_i(y_s)}_{L^p}^2
 \leq M(C_pM^{-1/2})^2=C_p^2.
\end{equation*}
For the covariance matrix $\Sigma(y_s)$, we use $\norm{\Sigma(y_s)}_{\mathrm F}^2=\operatorname{Tr}(\Sigma(y_s)^2)$ and apply \Cref{cor:tap-covariance}-\textup{(ii)} at moment order $p/2$.
After increasing $K$ if necessary, we obtain
\begin{equation*}
 \norm{\Sigma(y_s)}_{L^p}
 =\bigl\|\operatorname{Tr}(\Sigma(y_s)^2)\bigr\|_{L^{p/2}}^{1/2}
 \leq K\sqrt{Mp}.
\end{equation*}

For the drift $\mathbf b(y_s)$, we write $m={m_J}(y_s)$, $\Sigma=\Sigma(y_s)$,
and $D=\operatorname{diag}((1-m_i^2)^{-1})$.
{For each $i$, we use the coefficients $\delta^{(i)}=(m_{+1}-m_{-1})/2$ from \Cref{lem:tap-influence}, where $m_{\pm1}$ are the pinned means obtained by deleting coordinate $i$ from the model with interaction $J$ and external field $y_s$.}
Substituting $(D\Sigma^2D)_{ii}-1=\|\delta^{(i)}\|_2^2$
into the drift formula in \Cref{lem:tap-residual-sde}, we obtain
\begin{equation*}
 b_i(y_s)
 =m_i\|\delta^{(i)}\|_2^2
   -\ell\operatorname{Tr}(\Sigma^2)m_i-2\ell(\Sigma^2m)_i.
\end{equation*}
We next bound the contribution of $m_i\|\delta^{(i)}\|_2^2$.
Since $|m_i|\leq1$, by Minkowski's inequality in $L^{p/2}$ and
\Cref{lem:tap-influence}-\textup{(ii)}, we have
\begin{align*}
 \E{\left(\sum_{i=1}^M m_i^2\|\delta^{(i)}\|_2^4\right)^{p/2}}^{2/p}
 = \left\|\sum_{i=1}^M m_i^2\|\delta^{(i)}\|_2^4
                                      \right\|_{L^{p/2}}
 \leq \sum_{i=1}^M
        \bigl\|m_i^2\|\delta^{(i)}\|_2^4\bigr\|_{L^{p/2}}
 \leq \sum_{i=1}^M\|\delta^{(i)}\|_{L^{2p}}^4
 \leq M C_{2p}^4.
\end{align*}
Taking square roots bounds this part by $C_p\sqrt M$.

Then, we bound $\ell\operatorname{Tr}(\Sigma^2)m$ and $2\ell \Sigma^2m$.
Since $\Sigma$ is symmetric, $\Sigma^2\succeq0$ and
$\norm{\Sigma^2}_{\mathrm{op}}\leq\operatorname{Tr}(\Sigma^2)$.
By $\norm m_2\leq\sqrt M$ and
$\norm{\Sigma^2m}_2\leq\norm{\Sigma^2}_{\mathrm{op}}\norm m_2$, we have
\begin{align*}
 \ell\operatorname{Tr}(\Sigma^2)\norm m_2+2\ell\norm{\Sigma^2m}_2
 \leq&~\ell\operatorname{Tr}(\Sigma^2)\norm m_2
       +2\ell\norm{\Sigma^2}_{\mathrm{op}}\norm m_2
 \leq 3\ell\sqrt M\,\operatorname{Tr}(\Sigma^2).
\end{align*}
By \Cref{cor:tap-covariance}-\textup{(ii)}, we have
$\norm{\operatorname{Tr}(\Sigma^2)}_{L^p}\leq KMp$.
We apply the triangle inequality to the three drift terms and obtain
\begin{equation*}
 \norm{\mathbf b(y_s)}_{L^p}
 \leq C_p\sqrt M+3\ell\sqrt M\,\norm{\operatorname{Tr}(\Sigma^2)}_{L^p}
 \leq C_p\sqrt M+3Kp(\ell M)\sqrt M
 \leq C'_p\sqrt M,
\end{equation*}
where we used $\ell M\leq\beta_*^2<1$.
\end{proof}
Next, we upper bound the norm of {the residual's Jacobian}.
Second-moment bounds for the corresponding covariance discrepancy were proved for zero-field SK in \cite[Theorem~1.3]{ElAlaouiGaitonde24} and for the planted localization model at $\beta<1/2$ in \cite[Theorem~4.10]{DaviesLeeSandhuShi26}; here we need $L^p$ control, but weaker dependence on the dimension suffices.
\begin{lemma}
\label{lem:tap-susceptibility}
Let $(J,y)$ follow the $M$-spin planted model of \Cref{def:tap-planted-model} at inverse temperature $\sqrt{M\ell}$, with $J=\sqrt{M\ell}A^{\mathrm{full}}$ and $y=Y_t$.
Suppose that $M\ell\leq\beta_*^2$ and $t\geq0$.
{Here $n=M$ specifies the dimension in that definition, and $\mathbf r$ is the residual in \cref{eq:tap-M-residual} for the resulting $M\times M$ matrix $J$.}
For every fixed $p\geq2$, there is a constant $C_p$,
depending only on $p,\beta_*$, such that
\begin{equation*}
 \norm{\frac{\partial\mathbf r(y)}{\partial y}}_{L^p}\leq C_pM^{1/4}.
\end{equation*}
\end{lemma}
\begin{proof}
We use the path $y_s$ from \Cref{lem:tap-residual-dynamics}.
For a function $F$ of the external field and $J$, we set
$(P_sF)(y)=\E{F(y_s)\mid J,y}$.
Integrating the SDE in \Cref{lem:tap-residual-sde} to time $0<h\leq1$
and taking expectation given $(J,y)$, we obtain
\begin{equation}
 \mathbf r(y)=(P_h\mathbf r)(y)-\int_0^h(P_s\mathbf b)(y)\d s.
 \label{eq:tap-smoothing-residual-identity}
\end{equation}
We will differentiate these two averages to bound the residual Jacobian;
the required integrability and differentiation justifications are given below.

We claim that, if $F$ is measurable and bounded as a function of the external field
for each fixed $J$, and $\norm{F(y_s)}_{L^{2p}}$ is finite, then
\begin{equation}
 \norm{\frac{\partial(P_sF)(y)}{\partial y}}_{L^p}
 \leq s^{-1/2}\norm{F(y_s)}_{L^p}
    +\bigl\|\,\|\Sigma(y)\|_{\mathrm{op}}\,\bigr\|_{L^p}^{1/2}
       \norm{F(y_s)}_{L^{2p}}.
 \label{eq:tap-smoothing-gradient}
\end{equation}

Assuming \cref{eq:tap-smoothing-gradient}, we first bound the residual Jacobian.
For $0<s\leq h$, we apply \Cref{lem:tap-residual-dynamics} at time $t+s$.
Using $\|\Sigma(y)\|_{\mathrm{op}}\leq\|\Sigma(y)\|_{\mathrm F}$, we substitute its bounds at orders $p$ and $2p$ and obtain
\begin{align*}
 \norm{\frac{\partial(P_h\mathbf r)(y)}{\partial y}}_{L^p}
 \leq&~C_p h^{-1/2}
       +(K\sqrt{Mp})^{1/2}C_{2p}
 \leq C_p(h^{-1/2}+M^{1/4}),\\
 \norm{\frac{\partial(P_s\mathbf b)(y)}{\partial y}}_{L^p}
 \leq&~C_p\sqrt M\,s^{-1/2}
       +(K\sqrt{Mp})^{1/2}C_{2p}\sqrt M
 \leq C_p\sqrt M(s^{-1/2}+M^{1/4}).
\end{align*}
We differentiate \cref{eq:tap-smoothing-residual-identity}, take
$L^p$ norms, and apply Minkowski to the integral:
\begin{align*}
 \norm{\frac{\partial\mathbf r(y)}{\partial y}}_{L^p}
 \leq&~\norm{\frac{\partial(P_h\mathbf r)(y)}{\partial y}}_{L^p}
       +\int_0^h
         \norm{\frac{\partial(P_s\mathbf b)(y)}{\partial y}}_{L^p}\d s
 \leq C_p(h^{-1/2}+M^{1/4}
                +\sqrt M(2\sqrt h+hM^{1/4})).
\end{align*}
We choose $h=M^{-1/2}\leq1$.  Then
$h^{-1/2}=\sqrt M\sqrt h=\sqrt M hM^{1/4}=M^{1/4}$,
so the right side is at most $C_pM^{1/4}$.

It remains to prove \cref{eq:tap-smoothing-gradient} and justify its
application to $\mathbf r$ and $\mathbf b$.
We fix $J$ and $s>0$, and let $p_s(y,y')$ be the conditional
density of $y_s$ evaluated at $y'\in\mathbb R^M$.
We write the conditional average as
\begin{equation*}
 (P_sF)(y)=\E{F(y_s)\mid J,y}
          =\int_{\mathbb R^M}F(y')p_s(y,y')\d y'.
\end{equation*}
When differentiating in the initial field $y$, we hold $J,s,y'$ fixed,
so $F(y')$ does not change.  We compute the density and its derivative
before justifying the exchange with the integral.  We set
\begin{equation*}
 \mathcal Z_J(y)=\sum_{\sigma\in\{\pm1\}^M}
      \exp(\inner{\sigma}{J\sigma}/2+\inner y\sigma).
\end{equation*}
{Given $(J,y,\xi=\sigma)$, we have $y_s\sim N(y+s\sigma,sI_M)$.}  We average these Gaussian densities
with the Gibbs probabilities of $\sigma$ to obtain
\begin{equation} \label{eq:tap-smoothing-kernel}
\begin{aligned}
 p_s(y,y')
 =&~\sum_{\sigma\in\{\pm1\}^M}
 \frac{\exp(\inner{\sigma}{J\sigma}/2+\inner y\sigma)}{\mathcal Z_J(y)}
 (2\pi s)^{-M/2}
       \exp\left(-\frac{\norm{y'-y-s\sigma}_2^2}{2s}\right)\\
 =&~(2\pi s)^{-M/2}
       \exp\left(-\frac{\norm{y'-y}_2^2}{2s}-\frac{Ms}{2}\right)
       \frac{\mathcal Z_J(y')}{\mathcal Z_J(y)},
\end{aligned}
\end{equation}
where the second equality uses $\norm\sigma_2^2=M$ and $\inner y\sigma+\inner{y'-y}\sigma=\inner{y'}\sigma$.

We differentiate in $y$, using $\nabla_y\log\mathcal Z_J(y)={m_J}(y)$:
\begin{equation*}
 \nabla_y\log p_s(y,y')
 =-\nabla_y\frac{\|y'-y\|_2^2}{2s}-\nabla_y\log\mathcal Z_J(y)
 =\frac{y'-y}{s}-{m_J}(y).
\end{equation*}
Multiplying by $p_s(y,y')$, we obtain
\begin{equation*}
 \nabla_yp_s(y,y')=p_s(y,y')
                 \left(\frac{y'-y}{s}-{m_J}(y)\right).
\end{equation*}

Then, we have
\begin{align*}
 \frac{\partial(P_sF)(y)}{\partial y}
 =&~\int_{\mathbb R^M}F(y')\bigl(\nabla_yp_s(y,y')\bigr)^\top\d y'\\
 =&~\int_{\mathbb R^M}F(y')
        \left(\frac{y'-y}{s}-{m_J}(y)\right)^\top p_s(y,y')\d y'\\
 =&~\E{F(y_s)\left(\frac{y_s-y}{s}-{m_J}(y)\right)^\top\mid J,y}.
\end{align*}
{We will apply Cauchy--Schwarz with $(J,y)$ fixed, taking $U=F(y_s)$ and $V=(y_s-y)/s-{m_J}(y)$.}
For any square-integrable $U,V$, with expectations taken under
their joint law, its rowwise form is
\begin{align*}
 \norm{\E{U_jV}}_2^2
 =&~\sup_{\norm v_2=1}\abs{\E{U_j\inner vV}}^2
 \leq \E{U_j^2}\sup_{\norm v_2=1}\E{\inner vV^2}
 =\E{U_j^2}\norm{\E{VV^\top}}_{\mathrm{op}}.
\end{align*}
Summing over $j$, we obtain
\begin{equation*}
 \norm{\E{UV^\top}}_{\mathrm F}^2
 \leq\norm{\E{VV^\top}}_{\mathrm{op}}\,\E{\norm U_2^2}.
\end{equation*}
Under the conditional law given $(J,y)$, we have $V=\xi-{m_J}(y)+\widetilde B_s/s$.
By the definitions of ${m_J}(y)$ and $\Sigma(y)$, we have
\begin{equation*}
 \E{\xi-{m_J}(y)\mid J,y}=0,\qquad
 \E{(\xi-{m_J}(y))(\xi-{m_J}(y))^\top\mid J,y}=\Sigma(y).
\end{equation*}
Under the same fixed $(J,y)$, we use $\widetilde B_s\sim N(0,sI_M)$ to compute
\begin{equation*}
 \E{\widetilde B_s/s\mid J,y}=0,\qquad
 s^{-2}\E{\widetilde B_s\widetilde B_s^\top\mid J,y}
 =s^{-2}sI_M=s^{-1}I_M.
\end{equation*}
By independence of $\xi$ and $\widetilde B_s$ given $(J,y)$,
we compute the cross term as
\begin{equation*}
 \E{(\xi-{m_J}(y))(\widetilde B_s/s)^\top\mid J,y}
 =\E{\xi-{m_J}(y)\mid J,y}\E{\widetilde B_s/s\mid J,y}^\top=0.
\end{equation*}
The other cross term is its transpose and also vanishes.
Expanding the product, we obtain
\begin{align*}
 &~\E{\left(\frac{y_s-y}{s}-{m_J}(y)\right)
       \left(\frac{y_s-y}{s}-{m_J}(y)\right)^\top\mid J,y}\\
 =&~\E{(\xi-{m_J}(y))(\xi-{m_J}(y))^\top\mid J,y}
   +s^{-2}\E{\widetilde B_s\widetilde B_s^\top\mid J,y}
 =\Sigma(y)+s^{-1}I_M.
\end{align*}
Substituting into the Cauchy--Schwarz bound, we obtain
\begin{align*}
 \norm{\frac{\partial(P_sF)(y)}{\partial y}}_{\mathrm F}
 \leq&~\norm{\Sigma(y)+s^{-1}I_M}_{\mathrm{op}}^{1/2}
       \bigl(\E{\norm{F(y_s)}_2^2\mid J,y}\bigr)^{1/2}\\
 \leq&~(s^{-1/2}+\norm{\Sigma(y)}_{\mathrm{op}}^{1/2})
       \bigl(\E{\norm{F(y_s)}_2^2\mid J,y}\bigr)^{1/2}.
\end{align*}
For any $q\geq2$, we fix $(J,y)$ and apply Jensen's inequality with exponent $q/2$:
\begin{equation*}
 \E{\bigl(\E{\|F(y_s)\|_2^2\mid J,y}\bigr)^{q/2}}
 \leq\E{\E{\|F(y_s)\|_2^q\mid J,y}}
 =\E{\|F(y_s)\|_2^q}.
\end{equation*}
Taking $q$th roots, we obtain
\begin{equation*}
 \left\|\bigl(\E{\norm{F(y_s)}_2^2\mid J,y}\bigr)^{1/2}\right\|_{L^q}
 \leq\norm{F(y_s)}_{L^q}.
\end{equation*}
We now take $L^p$ norms over $(J,y)$.  We use this inequality with $q=p$
for the term with coefficient $s^{-1/2}$, and Hölder with exponents
$2,2$, followed by the same inequality with $q=2p$, for the other term.
We obtain
\begin{equation*}
 \norm{\frac{\partial(P_sF)(y)}{\partial y}}_{L^p}
 \leq s^{-1/2}\norm{F(y_s)}_{L^p}
    +\bigl\|\,\|\Sigma(y)\|_{\mathrm{op}}\,\bigr\|_{L^p}^{1/2}
       \norm{F(y_s)}_{L^{2p}},
\end{equation*}
which proves \cref{eq:tap-smoothing-gradient}.

We now justify the steps deferred above.
We first establish the estimates needed for these steps.
We fix $J,M$.  Each probability $\mu_J(y)(\sigma)$
is a ratio of finite sums of exponentials in $y$, with a strictly
positive denominator.  Thus ${m_{J,j}}(y)=\sum_\sigma\sigma_j\mu_J(y)(\sigma)$
and $\Sigma_{jk}(y)=\sum_\sigma\sigma_j\sigma_k\mu_J(y)(\sigma)-{m_{J,j}}(y){m_{J,k}}(y)$
are infinitely differentiable in $y$.  

We next bound $\operatorname{arctanh}\tp{{m_{J,j}}(y)}-y_j$.
We write $\mu_{J^{\setminus j}}(y_{-j})$ for the Gibbs distribution on
$\{\pm1\}^{M-1}$ with interaction $J^{\setminus j}$ and external field $y_{-j}$.
{Here we delete $j$ from the $M$-spin system of this lemma, so $J^{\setminus j}$ has size $(M-1)\times(M-1)$.}
By summing the Gibbs weights separately over $\sigma_j=1$ and
$\sigma_j=-1$, we have
\begin{equation*}
 \frac{1+{m_{J,j}}(y)}{1-{m_{J,j}}(y)}
 =\frac{\mu_J(y)\{\sigma_j=1\}}{\mu_J(y)\{\sigma_j=-1\}}
 =\mathrm e^{2y_j}
 \frac{\E[\sigma\sim\mu_{J^{\setminus j}}(y_{-j})]{\mathrm e^{\sum_{k\ne j}J_{jk}\sigma_k}}}
      {\E[\sigma\sim\mu_{J^{\setminus j}}(y_{-j})]{\mathrm e^{-\sum_{k\ne j}J_{jk}\sigma_k}}}.
\end{equation*}
Taking half the logarithm, we obtain
\begin{equation*}
 \operatorname{arctanh}\tp{{m_{J,j}}(y)}-y_j
 =\frac12\log
 \frac{\E[\sigma\sim\mu_{J^{\setminus j}}(y_{-j})]{\mathrm e^{\sum_{k\ne j}J_{jk}\sigma_k}}}
      {\E[\sigma\sim\mu_{J^{\setminus j}}(y_{-j})]{\mathrm e^{-\sum_{k\ne j}J_{jk}\sigma_k}}}.
\end{equation*}
Since $|\sum_{k\ne j}J_{jk}\sigma_k|\leq\sum_{k\ne j}|J_{jk}|$,
both expectations in this ratio lie between
$\mathrm e^{-\sum_{k\ne j}|J_{jk}|}$ and
$\mathrm e^{\sum_{k\ne j}|J_{jk}|}$.  Hence
$\abs{\operatorname{arctanh}\tp{{m_{J,j}}(y)}-y_j}\leq\sum_{k\ne j}|J_{jk}|$.
We recall that $D(y)=\operatorname{diag}((1-{m_{J,j}}(y)^2)^{-1})_{j=1}^M$.
Also, $|\Sigma_{jk}|\leq1$ by Cauchy--Schwarz and the unit variance bound
for spins.  As proved in \Cref{lem:tap-influence},
$\Sigma_{jk}=(1-m_j^2)\delta_k^{(j)}$ for $k\ne j$.
Since $\delta_k^{(j)}$ is a difference of two probabilities,
$|\delta_k^{(j)}|\leq1$.  Therefore
$|(D\Sigma)_{jk}|\leq1$, also on the diagonal where $(D\Sigma)_{jj}=1$.
We now bound $\mathbf r(y)$, its Jacobian, and $\mathbf b(y)$ independently
of $y$.  By the bound on $\operatorname{arctanh}\tp{{m_{J,j}}(y)}-y_j$
and Cauchy--Schwarz in each row of $J$, we have
$\|\operatorname{arctanh}({m_J}(y))-y\|_2^2
\leq\sum_j(\sum_{k\ne j}|J_{jk}|)^2\leq M\|J\|_{\mathrm F}^2$.
Also $\|{m_J}(y)\|_2\leq\sqrt M$, $\|\Sigma(y)\|_{\mathrm F}\leq M$,
and $\|D(y)\Sigma(y)\|_{\mathrm F}\leq M$.
By the formulas for $\mathbf r$, its Jacobian, and $\mathbf b$, we obtain
\begin{align*}
 \norm{\mathbf r(y)}_2
 \leq&~\sqrt M(\norm J_{\mathrm F}+\norm J_{\mathrm{op}}+\ell M),\\
 \norm{\frac{\partial\mathbf r(y)}{\partial y}}_{\mathrm F}
 \leq&~\|D(y)\Sigma(y)\|_{\mathrm F}+\|J\|_{\mathrm{op}}\|\Sigma(y)\|_{\mathrm F}
       +3\ell M\|\Sigma(y)\|_{\mathrm F}+\sqrt M\\
 \leq&~M(1+\norm J_{\mathrm{op}})+3\ell M^2+\sqrt M,\\
 \norm{\mathbf b(y)}_2
 \leq&~\sqrt M(M+1)+3\ell M^2\sqrt M.
\end{align*}
To verify the bound on $\mathbf b(y)$, we use
$(D(y)\Sigma(y)^2D(y))_{jj}=\sum_k(D(y)\Sigma(y))_{jk}^2\leq M$
and $\operatorname{Tr}(\Sigma(y)^2)=\sum_{j,k}\Sigma_{jk}(y)^2\leq M^2$.
The three terms in its formula in \Cref{lem:tap-residual-sde} satisfy
\begin{align*}
 \norm{\bigl({m_{J,j}}(y)((D(y)\Sigma(y)^2D(y))_{jj}-1)\bigr)_{j=1}^M}_2
 \leq&~\sqrt M(M+1),\\
 \norm{\ell\operatorname{Tr}(\Sigma(y)^2){m_J}(y)}_2
 \leq&~\ell M^2\sqrt M,\\
 \norm{2\ell \Sigma(y)^2{m_J}(y)}_2
 \leq&~2\ell\norm{\Sigma(y)}_{\mathrm{op}}^2\sqrt M
 \leq2\ell M^2\sqrt M.
\end{align*}
We now justify differentiation under the integral over $y'$.
We fix $J,s>0$ and restrict $\|y\|_2\leq R$.
Every Gaussian center in \cref{eq:tap-smoothing-kernel} has norm
at most $R+s\sqrt M$.  Using
$\|y'-z\|_2^2\geq\|y'\|_2^2/2-\|z\|_2^2$, we obtain
\begin{equation*}
 p_s(y,y')\leq A_{s,R}\mathrm e^{-\|y'\|_2^2/(4s)},\qquad
 \|\nabla_y p_s(y,y')\|_2
 \leq A_{s,R}\left(\frac{\|y'\|_2+R}{s}+\sqrt M\right)
                  \mathrm e^{-\|y'\|_2^2/(4s)},
\end{equation*}
where $A_{s,R}=(2\pi s)^{-M/2}
\exp((R+s\sqrt M)^2/(2s))$.
Both bounds are integrable in $y'$.  Multiplying by the fixed-$J$
bound on $F$ provides a common integrable bound for the integrand and
its derivatives.  Differentiation under the integral is therefore valid.
The bounds on $\mathbf r$ and $\mathbf b$ above verify this condition
for both functions used in the proof.

We next justify the three operations in
\cref{eq:tap-smoothing-residual-identity} and its derivative.
In the following conditional expectations, we fix $(J,y)$ and average
over $\xi\sim\mu_{J,y}$ and the independent Brownian motion generating $y_s$.
First, the Jacobian bound ensures
\begin{equation*}
 \E{\int_0^h
 \norm{\left.\frac{\partial\mathbf r(v)}{\partial v}\right|_{v=y_s}}_{\mathrm F}^2
 \d s\mid J,y}
 \leq h\sup_v\norm{\frac{\partial\mathbf r(v)}{\partial v}}_{\mathrm F}^2
 <\infty.
\end{equation*}
For fixed $(J,y)$, the matrix
$\left.\frac{\partial\mathbf r(v)}{\partial v}\right|_{v=y_s}$
is a deterministic function of the current external field $y_s$.  It is therefore
measurable with respect to $\sigma(y_u:0\leq u\leq s)$, the filtration
in which $\widetilde B^{\mathrm{loc}}$ is Brownian.
The integrand is continuous and adapted.  Together with square
integrability, this makes the It\^o integral a martingale, so
\begin{equation*}
 \E{\int_0^h
 \left.\frac{\partial\mathbf r(v)}{\partial v}\right|_{v=y_s}
 \d \widetilde B_s^{\mathrm{loc}}\mid J,y}=0.
\end{equation*}

Second, the bound on $\mathbf b$ ensures absolute integrability:
\begin{equation*}
 \E{\int_0^h\norm{\mathbf b(y_s)}_2\d s\mid J,y}
 \leq h\sup_v\norm{\mathbf b(v)}_2<\infty.
\end{equation*}
By Fubini's theorem, applied coordinatewise with $(J,y)$ fixed, we have
\begin{equation*}
 \E{\int_0^h\mathbf b(y_s)\d s\mid J,y}
 =\int_0^h\E{\mathbf b(y_s)\mid J,y}\d s.
\end{equation*}
Taking conditional expectation in the integrated SDE and using these
two identities proves \cref{eq:tap-smoothing-residual-identity}.

Third, we justify differentiating the time integral with respect to the
initial external field $y$.  We fix $J,M$.  For every $s>0$,
the preceding domination argument proves that $P_s\mathbf b$ is
differentiable.
We apply the previously computed derivative estimate
with $F=\mathbf b$.  Since
$\|\Sigma(y)\|_{\mathrm{op}}\leq\operatorname{Tr}\Sigma(y)\leq M$, we have
\begin{equation*}
 \sup_y\left\|\frac{\partial(P_s\mathbf b)(y)}{\partial y}\right\|_{\mathrm F}
 \leq(s^{-1/2}+\sqrt M)\sup_v\|\mathbf b(v)\|_2.
\end{equation*}
The right side is integrable in $s$, because
$\sup_v\|\mathbf b(v)\|_2<\infty$ and
$\int_0^h(s^{-1/2}+\sqrt M)\d s=2\sqrt h+h\sqrt M$.
Also $\|(P_s\mathbf b)(y)\|_2\leq\sup_v\|\mathbf b(v)\|_2$,
so the function itself is integrable in $s$.
The endpoints $0,h$ do not depend on $y$.
By the mean-value formula, every coordinate difference quotient has
this same integrable bound.  We apply dominated convergence and obtain
\begin{equation*}
 \frac{\partial}{\partial y}\int_0^h(P_s\mathbf b)(y)\d s
 =\int_0^h\frac{\partial(P_s\mathbf b)(y)}{\partial y}\d s.
\end{equation*}
The differentiation identity above holds for fixed $J,y$.
For the subsequent $L^p$ estimate over random $(J,y)$, we separately use
the moment bounds already obtained from
\Cref{lem:tap-residual-dynamics}:
\begin{equation*}
 \int_0^h
 \norm{\frac{\partial(P_s\mathbf b)(y)}{\partial y}}_{L^p}\d s
 \leq C_p\sqrt M(2\sqrt h+hM^{1/4})<\infty,
\end{equation*}
which justifies Minkowski's inequality for the time integral under the
joint law of $(J,y)$.
\end{proof}

We now apply these estimates to \cref{eq:tap-pinned-bdg}.
Since $X_i$ is independent of $(J^{\setminus i},y_{-i})$, conditioning on $X_i=b$ preserves their planted distribution.
After averaging over the cavity observations, we apply \cref{eq:tap-path-planted-marginal}: for each $s\in[0,\ell]$, the pair $(J^{\setminus i},y_{-i}+a_s)$ has the $M=n-1$ spin planted distribution at inverse temperature $\sqrt{M\ell}$ and time $t+s$.
{By \Cref{lem:tap-susceptibility} and Minkowski's inequality in $L^{p/2}$ under the conditional law given $X_i=b$, we have}
\begin{equation*}
 \left\|\left(\int_0^\ell\norm{\left.\frac{\partial\mathbf r(v)}{\partial v}\right|_{v=y_{-i}+a_s}}_{\mathrm F}^2\d s\right)^{1/2}
           \right\|_{L^p}
 \leq\left(\int_0^\ell
           \norm{\left.\frac{\partial\mathbf r(v)}{\partial v}\right|_{v=y_{-i}+a_s}}_{L^p}^2\d s\right)^{1/2}
 \leq C_p\sqrt\ell\,M^{1/4}.
\end{equation*}
By \Cref{lem:tap-residual-dynamics}, we have $\int_0^\ell\norm{\mathbf b_s}_{L^p}\d s\leq C_p\ell\sqrt M$.
Substituting these upper bounds into \cref{eq:tap-pinned-bdg}, we obtain, for either $b\in\{\pm1\}$ under $X_i=b$,
\begin{equation}
 \|\mathbf r_{\mathrm{pin}}^b-\mathbf r^{\setminus i}\|_{L^p}
 \leq C_p(\sqrt\ell\,M^{1/4}+\ell\sqrt M)
 \leq C_pn^{-1/4}.
 \label{eq:tap-matched-row-increment}
\end{equation}

We next bound the two nonlinear errors in \cref{eq:tap-recombination}.

\begin{lemma}
\label{lem:tap-recombination-errors}
Under the planted model of \Cref{def:tap-planted-model}, set $J=\beta A^{\mathrm{full}}$.
If $n\geq2$, $n\ell\leq\beta_*^2$, and $t\in[0,T]$, then for every fixed $p\geq2$
there is a constant $C_p$, depending only on $p,\beta_*,T$, such that
for every deleted coordinate $i$,
\begin{equation*}
 \|\Delta_{\mathrm{ent}}\|_{L^p}+\|\Delta_Q\|_{L^p}
 \leq C_pn^{-1/2}.
\end{equation*}
\end{lemma}

We fix the parameters of \Cref{lem:tap-recombination-errors}, a deleted coordinate $i$, and $p\geq2$, and set $M=n-1$.
To bound the posterior-weighted sums in \cref{eq:tap-recombination},
we use the following identity to transfer estimates proved for each pinned
quantity under the matching condition $X_i=b$.
We recall that $\mathcal F=\sigma(J,y)$ and $\pi_b=\mu(\sigma_i=b)$.
By the posterior identity \cref{eq:planted-posterior}, $\pi_b=\E[\mathbb P_{n,\beta}^{\mathrm{pl,aug}}]{\1_{\{X_i=b\}}\mid\mathcal F}$.
For any nonnegative $\mathcal F$-measurable random variable $F_b$, the conditional expectation identity is
\begin{equation*}
 \E[\mathbb P_{n,\beta}^{\mathrm{pl,aug}}]{\pi_b F_b}
 =\E[\mathbb P_{n,\beta}^{\mathrm{pl,aug}}]{
   F_b\E[\mathbb P_{n,\beta}^{\mathrm{pl,aug}}]{\1_{\{X_i=b\}}\mid\mathcal F}}
 =\E[\mathbb P_{n,\beta}^{\mathrm{pl,aug}}]{\1_{\{X_i=b\}}F_b}.
\end{equation*}
Since $\mathbb P_{n,\beta}^{\mathrm{pl,aug}}(X_i=b)=1/2$, we sum over the two signs and obtain
\begin{equation}
 \E[\mathbb P_{n,\beta}^{\mathrm{pl,aug}}]{\sum_{b=\pm1}\pi_b F_b}
 =\frac12\sum_{b=\pm1}
       \E[\mathbb P_{n,\beta}^{\mathrm{pl,aug}}]{F_b\mid X_i=b}.
 \label{eq:tap-posterior-weighted-matching}
\end{equation}
Here $\E[\mathbb P_{n,\beta}^{\mathrm{pl,aug}}]{\cdot}$ averages the joint law of $(X,J,y)$ under the planted model of \Cref{def:tap-planted-model}.

{We first rewrite $\Delta_{\mathrm{ent}}$ as a difference of Taylor remainders of $\operatorname{arctanh}$.}
For $f=\operatorname{arctanh}$ and $j\ne i$, its definition is
\begin{equation*}
 (\Delta_{\mathrm{ent}})_j=f(m_j)-\sum_b\pi_b f((m_b)_j).
\end{equation*}
We subtract the affine Taylor approximation at $m_j^{\setminus i}$
from each occurrence of $f$.  Since $\sum_b\pi_b=1$ and
$\sum_b\pi_b((m_b)_j-m_j^{\setminus i})=m_j-m_j^{\setminus i}$,
the subtracted terms cancel exactly, and we obtain
\begin{align*}
 (\Delta_{\mathrm{ent}})_j =&~\bigl[f(m_j)-f(m_j^{\setminus i})-f'(m_j^{\setminus i})(m_j-m_j^{\setminus i})\bigr]\\
 &-\sum_b\pi_b\bigl[f((m_b)_j)-f(m_j^{\setminus i})-f'(m_j^{\setminus i})((m_b)_j-m_j^{\setminus i})\bigr].
\end{align*}
We denote these remainders by
\begin{equation*}
 \rho_{\mathrm{ent}}(a,u)\defeq f(a+u)-f(a)-f'(a)u,
 \qquad a,a+u\in(-1,1),
\end{equation*}
and set
\begin{equation*}
 \rho_j^b=\rho_{\mathrm{ent}}(m_j^{\setminus i},(m_b)_j-m_j^{\setminus i}),
 \qquad
 \overline\rho_j=\rho_{\mathrm{ent}}(m_j^{\setminus i},m_j-m_j^{\setminus i}).
\end{equation*}
Thus, in vector form, we have
\begin{equation}
 \Delta_{\mathrm{ent}}=\overline\rho-\sum_b\pi_b\rho^b.
 \label{eq:tap-child-entropy-expansion}
\end{equation}

We now bound each remainder.  Taylor's formula gives
\begin{equation*}
 \rho_{\mathrm{ent}}(a,u)
 =u^2\int_0^1(1-v)f''(a+vu)\d v.
\end{equation*}
Since $|f''(z)|\leq2(1-z^2)^{-2}$ and
$|a+vu|\leq\max(|a|,|a+u|)$ for $v\in[0,1]$, we obtain
\begin{equation}
 |\rho_{\mathrm{ent}}(a,u)|
 \leq u^2\bigl((1-a^2)^{-2}+(1-(a+u)^2)^{-2}\bigr).
 \label{eq:tap-child-entropy-remainder}
\end{equation}

Using H\"older with exponents $2,2$ in
\cref{eq:tap-child-entropy-remainder}, we obtain
\begin{align}
 \E{|\rho_j^b|^p\mid X_i=b}^{1/p}
 \leq&~\E{|(m_b)_j-m_j^{\setminus i}|^{4p}\mid X_i=b}^{1/(2p)}
 \left(\E{(1-(m_j^{\setminus i})^2)^{-4p}\mid X_i=b}^{1/(2p)}\right.\notag\\
 &\left.\hspace{35mm}
       +\E{(1-(m_b)_j^2)^{-4p}\mid X_i=b}^{1/(2p)}\right).
 \label{eq:tap-child-entropy-holder}
\end{align}

{To bound the right-hand side, we use the $L^{4p}$ bound on $(m_b)_j-m_j^{\setminus i}$ at scale $n^{-1/2}$ and the moment bounds for $(1-(m_j^{\setminus i})^2)^{-1}$ and $(1-(m_b)_j^2)^{-1}$ at order $4p$.}  The corresponding estimate for $\overline\rho_j$ uses
the unconditional increment $m_j-m_j^{\setminus i}$.
We prove these estimates in the next two lemmas, then return to
\Cref{lem:tap-recombination-errors}.

\begin{lemma}
\label{lem:tap-mean-increments}
Under the planted model of \Cref{def:tap-planted-model}, set $J=\beta A^{\mathrm{full}}$.
If $n\geq2$, $n\ell\leq\beta_*^2$, and $t\geq0$, let $m_b$ be the pinned
mean and $m^{\setminus i}$ the cavity mean for coordinate $i$.
For every fixed $p\geq2$, there is a constant $C_p$, depending only
on $p,\beta_*$, such that for every $i$, $b\in\{\pm1\}$, and $j\ne i$,
\begin{equation*}
 \E{|(m_b)_j-m_j^{\setminus i}|^p\mid X_i=b}^{1/p}\leq C_pn^{-1/2},
 \qquad \|m_j-m_j^{\setminus i}\|_{L^p}\leq C_pn^{-1/2}.
\end{equation*}
\end{lemma}
\begin{proof}
We express the mean difference as an increment along the same path $y_{-i}+a_s$ used for the residual increment.  We first condition on $(\mathcal F^{\setminus i},X_i=b)$, so that the interaction $J^{\setminus i}$ and the initial field $y_{-i}$ are fixed.
{Here $M=n-1$, and $\Sigma(v)$ denotes the covariance for interaction matrix $J^{\setminus i}$ and external field $v$.}
With ${m_{J^{\setminus i}}}(v)$ denoting the mean at external field $v$ and interaction $J^{\setminus i}$, we have that
\begin{equation*}
 m_j^{\setminus i}={m_{J^{\setminus i},j}}(y_{-i}),\qquad
 (m_b)_j={m_{J^{\setminus i},j}}(y_{-i}+bJ_{-i,i}).
\end{equation*}
{By the endpoint-law comparison preceding \cref{eq:tap-matched-endpoint-law}, $bJ_{-i,i}$ and $a_\ell$ have the same law given $(\mathcal F^{\setminus i},X_i=b)$.}
{Subtracting the fixed initial mean, we obtain, under this law,}
\begin{equation*}
 (m_b)_j-m_j^{\setminus i}
 \overset{\mathrm d}{=}{m_{J^{\setminus i},j}}(y_{-i}+a_\ell)-{m_{J^{\setminus i},j}}(y_{-i}).
\end{equation*}
Along this path, we apply \Cref{lem:app-wpi-observation-law} in dimension $M$ with base distribution $\mu=\mu^{\setminus i}$ and additional observation $a_s$.
The tilt of the cavity distribution by $v$ has interaction matrix $J^{\setminus i}$ and total field $y_{-i}+v$.
The process $\widetilde B_s^{\mathrm{loc}}=a_s-\int_0^s {m_{J^{\setminus i}}}(y_{-i}+a_u)\mathrm d u$ is Brownian motion in the filtration generated by $(a_u)_{0\leq u\leq s}$ {under the law above}, and
\begin{equation*}
 \mathrm d {m_{J^{\setminus i},j}}(y_{-i}+a_s)
 =\Sigma(y_{-i}+a_s)_{j,\cdot}\mathrm d\widetilde B_s^{\mathrm{loc}}.
\end{equation*}
We integrate from $0$ to $\ell$ and use $a_0=0$ to obtain
\begin{equation*}
 {m_{J^{\setminus i},j}}(y_{-i}+a_\ell)-{m_{J^{\setminus i},j}}(y_{-i})
 =\int_0^\ell\Sigma(y_{-i}+a_s)_{j,\cdot}\mathrm d\widetilde B_s^{\mathrm{loc}}.
\end{equation*}
{Thus, after averaging over $\mathcal F^{\setminus i}$ under the conditional law given $X_i=b$, the mean difference has the law of this stochastic integral.}
{The following path norms are under the same conditional law and average jointly over $(J^{\setminus i},y_{-i})$, $\xi$, and $\widetilde B$.}
{For each $s\in[0,\ell]$, by \cref{eq:tap-path-planted-marginal} the pair $(J^{\setminus i},y_{-i}+a_s)$ has the planted distribution in dimension $M=n-1$, with the same parameter $\ell$ and time $t+s$.}
By \Cref{lem:tap-influence}-\textup{(ii)}, we have $\|\Sigma(y_{-i}+a_s)_{j,\cdot}\|_{L^p}\leq C_p$.
We apply BDG and then Minkowski in $L^{p/2}$ to obtain
\begin{equation*}
 \E{|(m_b)_j-m_j^{\setminus i}|^p\mid X_i=b}^{1/p}
 \leq C_p\left(\int_0^\ell
       \|\Sigma(y_{-i}+a_s)_{j,\cdot}\|_{L^p}^2\mathrm d s\right)^{1/2}
 \leq C_p\sqrt\ell\leq C_pn^{-1/2}.
\end{equation*}
Finally, $m_j-m_j^{\setminus i}=\sum_b\pi_b((m_b)_j-m_j^{\setminus i})$.
The means are $\mathcal F$-measurable, so Jensen's inequality and \cref{eq:tap-posterior-weighted-matching} give
\begin{equation*}
 \E{|m_j-m_j^{\setminus i}|^p}
 \leq\frac12\sum_b\E{|(m_b)_j-m_j^{\setminus i}|^p\mid X_i=b}
 \leq C_p^p n^{-p/2}.
\end{equation*}
Taking $p$th roots proves the unconditional bound.
\end{proof}

{The Taylor remainder $\rho_{\mathrm{ent}}$ involves derivatives of $\operatorname{arctanh}$.}
Before estimating it, we bound the inverse moments of $1-m_j^2$.

\begin{lemma}
\label{lem:tap-inverse-marginals}
Fix $S>0$.
Under the planted model of \Cref{def:tap-planted-model}, set $J=\beta A^{\mathrm{full}}$.
If $n\ell\leq\beta_*^2$ and $t\in[0,S]$, let
$m_j=\E[\sigma\sim\mu_{J,y}]{\sigma_j}$.  For every fixed $p\geq1$,
there is a constant $C_p$, depending only on $p,\beta_*,S$, such that
\begin{equation*}
 \max_{j\in[n]}\E{(1-m_j^2)^{-p}}\leq C_p.
\end{equation*}
\end{lemma}
\begin{proof}
We adapt \cite[Lemma~5.3]{DaviesLeeSandhuShi26}, which is stated for $\beta<1/2$.
We will bound $(1-m_j^2)^{-p}$ by exponential moments of the total external field and the effect of row $j$.
We use the cavity model of \Cref{def:tap-cavity-pinning} at coordinate $j$, with $M=n-1$ and $a=a_j$.
The two unnormalized weights of $\sigma_j=\pm1$ are $\mathrm e^{y_j}\langle \mathrm e^{a\cdot\sigma}\rangle_{\mu^{\setminus j}}$ and $\mathrm e^{-y_j}\langle \mathrm e^{-a\cdot\sigma}\rangle_{\mu^{\setminus j}}$ respectively.
We take half the logarithm of their ratio and obtain
\[
 \operatorname{arctanh}\tp{m_j}=y_j+\phi(a),\qquad
 \phi(a)=\tfrac12(\log\langle \mathrm e^{a\cdot\sigma}\rangle_{\mu^{\setminus j}}
              -\log\langle \mathrm e^{-a\cdot\sigma}\rangle_{\mu^{\setminus j}}).
\]
Since $(1-m_j^2)^{-p}=\cosh^{2p}(y_j+\phi(a))$, we have
\begin{align*}
 \E{(1-m_j^2)^{-p}}
 \leq \E{\mathrm e^{2p|y_j|+2p|\phi(a)|}}
 \leq \E{\mathrm e^{4p|y_j|}}^{1/2}
        \E{\mathrm e^{4p|\phi(a)|}}^{1/2}.
\end{align*}
Thus it suffices to control these two exponential moments.

With $\mathcal F^{\setminus j}$ fixed, we differentiate the two log partition functions and obtain
\[
 \partial_{a_k}\phi(a)=\tfrac12(({m_{+1}})_k+({m_{-1}})_k),
\]
so $|\partial_{a_k}\phi|\le1$ and
$\norm{\nabla_a\phi}_2\le\sqrt M$.  Also $\phi(-a)=-\phi(a)$.
Since $G$ and $-G$ have the same law,
we obtain
\[
 \E[G\sim N(0,I_M)]{\phi(\sqrt\ell G)}
 =-\E[G\sim N(0,I_M)]{\phi(\sqrt\ell G)}=0.
\]
The function $G\mapsto\phi(\sqrt\ell G)$ has the Lipschitz property
with constant at most $\sqrt{\ell M}\leq\beta_*$.
We apply the Gaussian exponential-moment bound
\cite[Theorem~5.5]{BoucheronLugosiMassart13} and obtain, for $s\ge0$,
\[
 \E[G\sim N(0,I_M)]{\mathrm e^{s|\phi(\sqrt\ell G)|}}
 \le \E[G\sim N(0,I_M)]{\mathrm e^{s\phi(\sqrt\ell G)}}
     +\E[G\sim N(0,I_M)]{\mathrm e^{-s\phi(\sqrt\ell G)}}
 \le2\mathrm e^{s^2\beta_*^2/2}.
\]

The actual row is $a=\ell U+\sqrt\ell G$, where $U=X_jX_{-j}$
on the remaining coordinates.  Since
$\norm{\ell U}_2=\ell\sqrt M$ and $\phi$ has the Lipschitz property
with constant $\sqrt M$, we have
$|\phi(\ell U+\sqrt\ell G)-\phi(\sqrt\ell G)| \le\ell M\le \beta_*^2$.

{We then average the preceding row bound over $(\mathcal F^{\setminus j},X)$ and obtain $\E{\mathrm e^{s|\phi(a)|}}\leq2\exp(s\beta_*^2+s^2\beta_*^2/2)$.}
{For the external field, given $X$, we have} $y_j\sim N(tX_j,t)$.

{Given $X$, the absolute value of the mean and the variance of $y_j$ are both $t\leq S$.}
For every $s\geq0$, we use the Gaussian MGF and
$\mathrm e^{s|z|}\leq\mathrm e^{sz}+\mathrm e^{-sz}$ to obtain $\E{\mathrm e^{s|y_j|}}\leq2\exp(sS+s^2S/2)$.
Substituting these bounds at $s=4p$ into the initial Cauchy--Schwarz
estimate proves $\E{(1-m_j^2)^{-p}}
\leq\exp(C\cdot p^2+C\cdot p)$ after increasing $C$ for $p\geq1$.
\end{proof}

\begin{proof}[Proof of \Cref{lem:tap-recombination-errors}]
We fix $p\geq2$ and a deleted coordinate $i$, and set $M=n-1$.
{We first bound the Taylor remainders in \cref{eq:tap-child-entropy-expansion} using \cref{eq:tap-child-entropy-holder}.}
Under $X_i=b$, the pair $(J^{\setminus i},y_{-i})$ has the planted distribution
at dimension $M=n-1$ and time $t$, and the pair
$(J^{\setminus i},y_{-i}+bJ_{-i,i})$ has that distribution at time
$t+\ell$.  {Both use the same parameter $\ell$ and inverse temperature $\sqrt{M\ell}\leq\beta_*$.}
Since $t\leq T$ and $\ell<1$, both times lie in $[0,T+1]$.
We apply \Cref{lem:tap-inverse-marginals} with $S=T+1$ at order $4p$ to both means and
\Cref{lem:tap-mean-increments} at order $4p$ to their difference.
Substituting into \cref{eq:tap-child-entropy-holder}, we obtain
\begin{equation*}
 \E{|\rho_j^b|^p\mid X_i=b}^{1/p}\leq C_pn^{-1}.
\end{equation*}
Minkowski in $L^{p/2}$ under the same conditional law gives
\begin{equation*}
 \E{\|\rho^b\|_2^p\mid X_i=b}^{2/p}
 \leq\sum_{j\ne i}\E{|\rho_j^b|^p\mid X_i=b}^{2/p}
 \leq C_p^2Mn^{-2}\leq C_p^2n^{-1}.
\end{equation*}
Since $\rho^b$ is $\mathcal F$-measurable, Jensen's inequality and
\cref{eq:tap-posterior-weighted-matching} give
\begin{equation*}
 \left\|\sum_b\pi_b\rho^b\right\|_{L^p}
 \leq\left(\frac12\sum_b\E{\|\rho^b\|_2^p\mid X_i=b}\right)^{1/p}
 \leq C_pn^{-1/2}.
\end{equation*}

For $\overline\rho_j$, we use the full $n$-spin planted law.
{The endpoints are $m_j^{\setminus i}$ and $m_j$, so we apply \Cref{lem:tap-inverse-marginals} with the same choice $S=T+1$ to their inverse spin variances in dimensions $M$ and $n$, both at time $t$.}  H\"older and the bound on $m_j-m_j^{\setminus i}$ at order $4p$ give
\begin{equation*}
 \|\overline\rho_j\|_{L^p}
 \leq\|m_j-m_j^{\setminus i}\|_{L^{4p}}^2
 \left(\E{(1-(m_j^{\setminus i})^2)^{-4p}}^{1/(2p)}
       +\E{(1-m_j^2)^{-4p}}^{1/(2p)}\right)
 \leq C_pn^{-1}.
\end{equation*}
Again, Minkowski in $L^{p/2}$ gives
$\|\overline\rho\|_{L^p}^2\leq\sum_{j\ne i}
\|\overline\rho_j\|_{L^p}^2\leq C_p^2n^{-1}$.
We substitute these bounds in \cref{eq:tap-child-entropy-expansion}
to conclude $\|\Delta_{\mathrm{ent}}\|_{L^p}\leq C_pn^{-1/2}$.

{For $\Delta_Q$, we use the algebra preceding \cref{eq:tap-recombination} to write}
\begin{equation*}
 \Delta_Q=\ell(1-m_i^2)m_{-i}
       +\ell\left(\sum_b\pi_b\|m_b\|_2^2m_b-\|m_{-i}\|_2^2m_{-i}\right).
\end{equation*}
To separate the terms that cancel in this weighted difference, we first
write $\|a+u\|_2^2=\|a\|_2^2+2\inner au+\|u\|_2^2$ and multiply by $a+u$:
\begin{align*}
 \|a+u\|_2^2(a+u)
 =&~\|a\|_2^2a+\|a\|_2^2u+2\inner au\,a
       +2\inner au\,u+\|u\|_2^2a+\|u\|_2^2u.
\end{align*}
With $a$ fixed, the term $\|a\|_2^2a$ is constant in $u$, and
$\|a\|_2^2u+2\inner au\,a$ is linear in $u$.
We subtract these terms and denote the remaining quadratic and cubic
terms by
\begin{equation*}
 \rho_Q(a,u)
 \defeq \|a+u\|_2^2(a+u)-\|a\|_2^2a
       -(\|a\|_2^2u+2\inner au\,a)
 =2\inner au\,u+\|u\|_2^2a+\|u\|_2^2u.
\end{equation*}
If $a,a+u\in[-1,1]^M$, then $\|a\|_2\leq\sqrt M$ and
$\|u\|_2\leq2\sqrt M$.  Thus
$\|\rho_Q(a,u)\|_2\leq5\sqrt M\|u\|_2^2$.
Expanding both cubic terms about $m^{\setminus i}$, we obtain
\begin{align*}
 \|m_b\|_2^2m_b
 =\|m^{\setminus i}\|_2^2m^{\setminus i}
       +\|m^{\setminus i}\|_2^2(m_b-m^{\setminus i})
 +2\inner{m^{\setminus i}}{m_b-m^{\setminus i}}m^{\setminus i}
       +\rho_Q(m^{\setminus i},m_b-m^{\setminus i}),
\end{align*}
\[
 \|m_{-i}\|_2^2m_{-i}
 =\|m^{\setminus i}\|_2^2m^{\setminus i}
       +\|m^{\setminus i}\|_2^2(m_{-i}-m^{\setminus i})+2\inner{m^{\setminus i}}{m_{-i}-m^{\setminus i}}m^{\setminus i}
       +\rho_Q(m^{\setminus i},m_{-i}-m^{\setminus i}).
\]
We multiply the first identity by $\pi_b$, sum over $b$, and subtract
the second identity.  The constant terms cancel because
$\sum_b\pi_b=1$.  For the linear terms, we use
\begin{align*}
 \sum_b\pi_b(m_b-m^{\setminus i})=&~m_{-i}-m^{\setminus i},\\
 \sum_b\pi_b\inner{m^{\setminus i}}{m_b-m^{\setminus i}}
 =&~\inner{m^{\setminus i}}{m_{-i}-m^{\setminus i}}.
\end{align*}
Thus only the remainders remain, and we obtain
\begin{equation*}
 \sum_b\pi_b\|m_b\|_2^2m_b-\|m_{-i}\|_2^2m_{-i}
 =\sum_b\pi_b\rho_Q(m^{\setminus i},m_b-m^{\setminus i})
       -\rho_Q(m^{\setminus i},m_{-i}-m^{\setminus i}).
\end{equation*}
For the squared vector increments, we apply Minkowski in $L^p$
and \Cref{lem:tap-mean-increments} at order $2p$:
\begin{align*}
 \E{\|m_b-m^{\setminus i}\|_2^{2p}\mid X_i=b}^{1/p}
 \leq&~\sum_{j\ne i}\E{|(m_b)_j-m_j^{\setminus i}|^{2p}\mid X_i=b}^{1/p}
 \leq C_p,\\
 \|m_{-i}-m^{\setminus i}\|_{L^{2p}}^2
 \leq&~\sum_{j\ne i}\|m_j-m_j^{\setminus i}\|_{L^{2p}}^2\leq C_p.
\end{align*}
We substitute into $\|\rho_Q(a,u)\|_2\leq5\sqrt M\|u\|_2^2$ and obtain
\begin{align*}
 \E{\|\rho_Q(m^{\setminus i},m_b-m^{\setminus i})\|_2^p\mid X_i=b}^{1/p}
 \leq&~C_p\sqrt M,\\
 \|\rho_Q(m^{\setminus i},m_{-i}-m^{\setminus i})\|_{L^p}
 \leq&~C_p\sqrt M.
\end{align*}
We apply Jensen's inequality and \cref{eq:tap-posterior-weighted-matching}
to the weighted sum of $\rho_Q(m^{\setminus i},m_b-m^{\setminus i})$,
which are $\mathcal F$-measurable, and use
$\|m_{-i}\|_2\leq\sqrt M$ for the remaining term.  Therefore
\begin{equation*}
 \|\Delta_Q\|_{L^p}\leq\ell\sqrt M+C_p\ell\sqrt M
 \leq C_pn^{-1/2}.
\end{equation*}
{Combining the bounds for $\Delta_{\mathrm{ent}}$ and $\Delta_Q$ proves the lemma.}
\end{proof}

\begin{proof}[Proof of \Cref{lem:tap-parent-child}]
We fix a deleted coordinate $i$ and a moment order $p\geq2$.
We recall that $\iota_i:\mathbb R^{n-1}\to\mathbb R^n$ inserts a zero
at coordinate $i$, so
$\norm{\iota_i\mathbf r^{\setminus i}}_2
=\|\mathbf r^{\setminus i}\|_2$.
It suffices to bound $\norm{\mathbf r-\iota_i\mathbf r^{\setminus i}}_{L^p}$,
because the reverse triangle inequality bounds the difference of the
two Euclidean norms by this vector difference.

Both $\mathbf r_{\mathrm{pin}}^b$ and $\mathbf r^{\setminus i}$
are functions of $(J,y)$.  We apply Jensen's inequality and
\cref{eq:tap-posterior-weighted-matching} to their difference, then use
\cref{eq:tap-matched-row-increment} to obtain
\begin{equation*}
 \left\|\sum_b\pi_b(\mathbf r_{\mathrm{pin}}^b-\mathbf r^{\setminus i})\right\|_{L^p}
 \leq\left(\frac12\sum_b
 \E{\|\mathbf r_{\mathrm{pin}}^b-\mathbf r^{\setminus i}\|_2^p\mid X_i=b}\right)^{1/p}
 \leq C_pn^{-1/4}.
\end{equation*}

We now use \cref{eq:tap-recombination} on the coordinates different
from $i$.  Coordinate $i$ of
$\mathbf r-\iota_i\mathbf r^{\setminus i}$ equals $r_i$.
We apply the triangle inequality in vector $L^p$ to obtain
\begin{align*}
 \norm{\mathbf r-\iota_i\mathbf r^{\setminus i}}_{L^p}
 \leq&~\norm{r_i}_{L^p}
       +\left\|\sum_{b=\pm1}\pi_b
          (\mathbf r_{\mathrm{pin}}^b-\mathbf r^{\setminus i})
         \right\|_{L^p}
       +\norm{\Delta_{\mathrm{ent}}}_{L^p}+\norm{\Delta_Q}_{L^p}\\
 \leq&~C_pn^{-1/2}+C_pn^{-1/4}+C_pn^{-1/2}
 \leq C'_pn^{-1/4}.
\end{align*}
Here we used \Cref{lem:tap-complete-row} at order $p$ for $r_i$,
the weighted-sum estimate above, and
\Cref{lem:tap-recombination-errors} at order $p$ for the two errors.
Their constants do not depend on $i$.

For every fixed integer $p\geq1$, we apply this vector estimate at order
$4p\geq2$ and use the reverse triangle inequality:
\begin{equation*}
 \max_i\bigl\|\,\|\mathbf r\|_2-\|\mathbf r^{\setminus i}\|_2\,\bigr\|_{L^{4p}}
 \leq\max_i\norm{\mathbf r-\iota_i\mathbf r^{\setminus i}}_{L^{4p}}
 \leq C'_{4p}n^{-1/4}.
\end{equation*}
Renaming the constant proves the claimed bound.
Finally,
$\mathbf r^{\setminus i}
=R_{n-1}^{(Q)}(m^{\setminus i};J^{\setminus i},y_{-i})$,
and $m^{\setminus i}$ is a ratio of finite Gibbs sums determined by
$(J^{\setminus i},y_{-i})$.
Therefore $\mathbf r^{\setminus i}$ is
$\mathcal F^{\setminus i}$-measurable, as asserted.
\end{proof}

\endgroup

\section{Probability boosting for AMP good events}
\label{sec:amp-good-events-proof}
In this section, we prove \Cref{lem:amp-good-events} by considering two cases.
For small times, we use a matrix-norm bound.
{For large times, we combine AMP stability with the Gaussian concentration result in \Cref{lem:gaussian-enlargement-boosting} to prove the small-residual and local strong-convexity conditions with failure probability at most $C\e^{-cn}$.}
{We then use overlap concentration and fixed-depth AMP estimates to prove posterior-mean proximity with failure probability at most $C_Hn^{-H}$ for every fixed $H>0$.}

\subsection{Useful lemmas}

We first introduce some standard and useful lemmas.

\begin{lemma}[{\cite[Lemma~5.13, Eq.~(5.64), and Lemma~5.6(c)]{ElAlaouiMontanariSellke25}, \cite[Theorem~2(i)]{Celentano24}}]
\label{prop:fixed-time-amp-tap-input}
Fix $0<\beta<1$ and $s>0$.
There are constants $r,\kappa>0$, depending only on $(\beta,s)$, such that, for every $\epsilon>0$, there is an integer $K_0=K_0(\beta,s,\epsilon)$ with the following properties for every fixed $K\geq K_0$.
\begin{enumerate}[label=\textup{(\roman*)}]
\item Under $\mathbb P_{n,\beta}^{\mathrm{pl,aug}}$, with probability tending to one as $n\to\infty$, the pair $(A^{\mathrm{full}},Y_s)$ satisfies $\epsilon$-small TAP residual and $(r,\kappa)$-local strong convexity at time $s$ and AMP depth $K$, as defined in \Cref{def:amp-tap-geometry}.
\item The posterior-mean approximation satisfies
\begin{equation*}
 \limsup_{n\to\infty}\E[\mathbb P_{n,\beta}^{\mathrm{pl,aug}}]{
 \frac{\|m(A,Y_s)-\widehat m_K(A^{\mathrm{full}},Y_s)\|_2^2}{n}}
 \leq\epsilon^2.
\end{equation*}
\end{enumerate}
\end{lemma}

We briefly explain the adjustments needed to apply the cited results.
We specialize \cite{ElAlaouiMontanariSellke25} to $\xi(q)=q^2/2$ and take $\lambda=\beta$ and $\gamma_0=s$ in \cite{Celentano24}.
The limiting overlaps in both papers equal $q_\beta(s)$, and the AMS iterate $m^{K-1}$ in \cite{Celentano24} corresponds to our $\widehat m_K$.
Compared with our residual, the modified TAP gradient in \cite[Eq.~(5.84)]{ElAlaouiMontanariSellke25} has an additional term proportional to $(\|u\|_2^2/n-q)u$.
At the AMP iterate, overlap convergence ensures that the norm of this term divided by $\sqrt n$ tends to zero as $n\to\infty$ and then $K\to\infty$, so \cite[Lemma~5.13]{ElAlaouiMontanariSellke25} yields the small-residual condition.
The AMS free energy in \cite[Theorem~2(i)]{Celentano24} equals $F_{\mathrm{TAP}}/n$ plus a constant, so its Hessian lower bound $\kappa/n$ becomes $\kappa$ in our normalization, on the same ball of radius $r\sqrt n$.
We choose $r,\kappa$ first and then take $K_0\geq1$ large enough for part~\textup{(i)}.
Finally, by \cite[Eq.~(5.64) and Lemma~5.6(c)]{ElAlaouiMontanariSellke25}, the normalized mean-square error tends to zero in the same order of limits, so increasing $K_0$ establishes part~\textup{(ii)} as well.

We also need an operator-norm tail bound under the planted law.
\begin{lemma}
\label{lem:inverse-poly-planted-norm-tail}
Fix $0\leq\beta<1$.  For every fixed $\epsilon>0$, there are $c_{\beta,\epsilon},
C_{\beta,\epsilon}>0$ such that
\begin{equation*}
  \mathbb P_{n,\beta}^{\mathrm{pl,aug}}\left\{
    \norm{A^{\mathrm{full}}}_{\mathrm{op}}>2+\epsilon\right\}
  \leq C_{\beta,\epsilon}\e^{-c_{\beta,\epsilon}n}.
\end{equation*}
\end{lemma}
\begin{proof}
{Under the original augmented law,}
$A^{\mathrm{full}}=G+\tp{\beta/n}I_n$, where $G$ is a centered GOE matrix.
By \cite[Section~3.3]{BandeiraElAlaouiRodder26}, for every $u\geq0$, we have
$\Pr[G\sim\mathrm{GOE}(n)]{\norm{G}_{\mathrm{op}}\geq2+u}
\leq2\e^{-nu^2/4}$.
For $n\geq2\beta/\epsilon$, we have
$\norm{(\beta/n)I_n}_{\mathrm{op}}=\beta/n\leq\epsilon/2$. Hence,
\begin{align*}
  \mathbb P_{0,n,\beta}^{\mathrm{aug}}\left\{
    \norm{A^{\mathrm{full}}}_{\mathrm{op}}>2+\epsilon\right\}
  \leq\Pr[G\sim\mathrm{GOE}(n)]{
    \norm{G}_{\mathrm{op}}>2+\epsilon/2}
  \leq2\e^{-n\epsilon^2/16}.
\end{align*}
We set $E=\{\norm{A^{\mathrm{full}}}_{\mathrm{op}}>2+\epsilon\}$ and apply the forward probability comparison in \Cref{lem:planting-likelihood} to obtain
\begin{equation*}
 \mathbb P_{n,\beta}^{\mathrm{pl,aug}}(E)
 \leq(1-\beta^2)^{-1/4}\mathbb P_{0,n,\beta}^{\mathrm{aug}}(E)^{1/2}.
\end{equation*}
Combining this with the tail bound under $\mathbb P_{0,n,\beta}^{\mathrm{aug}}$, for $n\geq2\beta/\epsilon$ we obtain
\begin{equation*}
 \mathbb P_{n,\beta}^{\mathrm{pl,aug}}(E)
 \leq(1-\beta^2)^{-1/4}\left(2\e^{-n\epsilon^2/16}\right)^{1/2}
 =\sqrt2(1-\beta^2)^{-1/4}\e^{-n\epsilon^2/32}.
\end{equation*}
This finishes the proof.
\end{proof}
Next, we establish the Lipschitz properties.
{To cover the inverse-temperature interval in \Cref{lem:app-wpi-branches}, we also control perturbations in temperature.}
{In the next lemma, we include the inverse temperature as a subscript on $\widehat m_{K,\beta}$, $R_{\mathrm{TAP},\beta}$, and $F_{\mathrm{TAP},\beta}$.}

\begin{lemma}
\label{lem:inverse-poly-fixed-depth-stability}
Fix $0<\beta_0<\beta_1<1$ and an integer $K\geq0$.
{There is a constant $C=C_{\beta_0,\beta_1,K}$ such that the following bounds hold for every $\beta,\theta\in[\beta_0,\beta_1]$, $s,t\geq0$, symmetric matrices $A^{\mathrm{full}},\widetilde A^{\mathrm{full}}$ with operator norms at most $4$, and fields $y,\widetilde y\in\R^n$.}
\begin{enumerate}[label=\textup{(\roman*)}]
\item The scalar fixed points satisfy
\begin{equation*}
 \abs{q_\beta(t)-q_\theta(s)}
 \leq\frac{\abs{t-s}+2\beta_1\abs{\beta-\theta}}{1-\beta_1^2}.
\end{equation*}
\item The AMP centers and their TAP residuals satisfy
\begin{align*}
 \norm{\widehat m_{K,\beta}\tp{A^{\mathrm{full}},y}
       -\widehat m_{K,\theta}\tp{\widetilde A^{\mathrm{full}},\widetilde y}}_2
 \leq C\left(\sqrt n
   \norm{A^{\mathrm{full}}-\widetilde A^{\mathrm{full}}}_{\mathrm{op}}
   +\norm{y-\widetilde y}_2+\sqrt n\abs{\beta-\theta}\right),
\end{align*}
\begin{align*}
 &~\left\|R_{\mathrm{TAP},\beta}\tp{
     \widehat m_{K,\beta}\tp{A^{\mathrm{full}},y};
     A^{\mathrm{full}},y,q_\beta(t)}
   -R_{\mathrm{TAP},\theta}\tp{
     \widehat m_{K,\theta}\tp{\widetilde A^{\mathrm{full}},\widetilde y};
     \widetilde A^{\mathrm{full}},\widetilde y,q_\theta(s)}
   \right\|_2\\
 \leq&~C\left(\sqrt n
   \norm{A^{\mathrm{full}}-\widetilde A^{\mathrm{full}}}_{\mathrm{op}}
   +\norm{y-\widetilde y}_2
   +\sqrt n\tp{\abs{\beta-\theta}+\abs{t-s}}\right).
\end{align*}
\item For every $u\in(-1,1)^n$, the TAP Hessians satisfy
\begin{align*}
 &~\left\|\nabla_u^2F_{\mathrm{TAP},\beta}
       \tp{u;A^{\mathrm{full}},y,q_\beta(t)}
   -\nabla_u^2F_{\mathrm{TAP},\theta}
       \tp{u;\widetilde A^{\mathrm{full}},\widetilde y,q_\theta(s)}
       \right\|_{\mathrm{op}}\\
 \leq&~C\left(
   \norm{A^{\mathrm{full}}-\widetilde A^{\mathrm{full}}}_{\mathrm{op}}
   +\abs{\beta-\theta}+\abs{t-s}\right).
\end{align*}
\end{enumerate}
\end{lemma}

\begin{proof}
For part~\textup{(i)}, recall the scalar channel function
$\psi(\gamma)=\E[G\sim N(0,1)]{\tanh^2\tp{\gamma+\sqrt\gamma G}}$,
which is one-Lipschitz by \cite[Lemma~5.6(a),(b)]{ElAlaouiMontanariSellke25}.
The fixed-point equations are
$q_\beta(t)=\psi\tp{t+\beta^2q_\beta(t)}$ and
$q_\theta(s)=\psi\tp{s+\theta^2q_\theta(s)}$.  Hence
\begin{align*}
 \abs{q_\beta(t)-q_\theta(s)}
 =&~\abs{\psi\tp{t+\beta^2q_\beta(t)}
       -\psi\tp{s+\theta^2q_\theta(s)}}\\
 \leq&~\abs{t-s+\beta^2q_\beta(t)-\theta^2q_\theta(s)}\\
 =&~\abs{t-s+\beta^2\tp{q_\beta(t)-q_\theta(s)}
       +(\beta^2-\theta^2)q_\theta(s)}\\
 \leq&~\abs{t-s}+\beta^2\abs{q_\beta(t)-q_\theta(s)}
       +\abs{\beta^2-\theta^2}q_\theta(s).
\end{align*}
Since $0\leq q_\theta(s)\leq1$, $\beta^2\leq\beta_1^2$, and
$\abs{\beta^2-\theta^2}=(\beta+\theta)\abs{\beta-\theta}
\leq2\beta_1\abs{\beta-\theta}$, we obtain
\begin{align*}
 \abs{q_\beta(t)-q_\theta(s)}
 \leq&~\abs{t-s}+\beta_1^2\abs{q_\beta(t)-q_\theta(s)}
       +2\beta_1\abs{\beta-\theta}.
\end{align*}
Moving the term $\beta_1^2\abs{q_\beta(t)-q_\theta(s)}$ to the left
and dividing by $1-\beta_1^2>0$ proves part~\textup{(i)}.

{For part~\textup{(ii)}, we extend the recursion argument of \cite[Lemma~5.17]{ElAlaouiMontanariSellke25} to matrix, field, and inverse-temperature perturbations.}  Write $z_k,m_k$ and
$\widetilde z_k,\widetilde m_k$ for the AMP iterates at the two inputs,
with $m_k=\tanh(z_k)$ and $\widetilde m_k=\tanh(\widetilde z_k)$.
Set $Q_k=\norm{m_k}_2^2/n$ and
$\widetilde Q_k=\norm{\widetilde m_k}_2^2/n$.
Since the coordinates of both vectors lie in $[-1,1]$,
$\abs{Q_k-\widetilde Q_k}\sqrt n\leq2\norm{m_k-\widetilde m_k}_2$.
{The matrix and inverse-temperature coefficients satisfy}
\begin{align*}
 \norm{\beta A^{\mathrm{full}}-\theta\widetilde A^{\mathrm{full}}}_{\mathrm{op}}
 \leq&~\beta_1
   \norm{A^{\mathrm{full}}-\widetilde A^{\mathrm{full}}}_{\mathrm{op}}
   +4\abs{\beta-\theta},\\
 \abs{\beta^2-\theta^2}\leq&~2\beta_1\abs{\beta-\theta}.
\end{align*}
By definition, we have
\begin{align*}
 z_{k+1}=&~\beta A^{\mathrm{full}}m_k+y-\beta^2(1-Q_k)m_{k-1},\\
 \widetilde z_{k+1}=&~\theta\widetilde A^{\mathrm{full}}\widetilde m_k
       +\widetilde y-\theta^2(1-\widetilde Q_k)\widetilde m_{k-1}.
\end{align*}
Subtracting and separating the changes in the iterates from those in the coefficients, {we obtain}
\begin{align*}
 z_{k+1}-\widetilde z_{k+1}
 =&~\beta A^{\mathrm{full}}(m_k-\widetilde m_k)
    +(\beta A^{\mathrm{full}}-\theta\widetilde A^{\mathrm{full}})
       \widetilde m_k+(y-\widetilde y)
 -\beta^2(1-Q_k)(m_{k-1}-\widetilde m_{k-1})\\
 &~-(\beta^2(1-Q_k)-\theta^2(1-\widetilde Q_k))
       \widetilde m_{k-1}.
\end{align*}
For the matrix terms, $\norm{A^{\mathrm{full}}}_{\mathrm{op}}\leq4$
and $\norm{\widetilde m_k}_2\leq\sqrt n$ imply
\begin{align*}
 \norm{\beta A^{\mathrm{full}}(m_k-\widetilde m_k)}_2
 \leq&~4\beta_1\norm{m_k-\widetilde m_k}_2,\\
 \norm{(\beta A^{\mathrm{full}}-\theta\widetilde A^{\mathrm{full}})
       \widetilde m_k}_2
 \leq&~\beta_1\sqrt n
       \norm{A^{\mathrm{full}}-\widetilde A^{\mathrm{full}}}_{\mathrm{op}}
       +4\sqrt n\abs{\beta-\theta}.
\end{align*}
For the Onsager terms, $0\leq Q_k,\widetilde Q_k\leq1$, so the first
has norm at most $\beta_1^2\norm{m_{k-1}-\widetilde m_{k-1}}_2$.
The coefficient in the last term satisfies
\begin{align*}
 \abs{\beta^2(1-Q_k)-\theta^2(1-\widetilde Q_k)}
 =&~\abs{(\beta^2-\theta^2)(1-Q_k)
             +\theta^2(\widetilde Q_k-Q_k)}
 \leq 2\beta_1\abs{\beta-\theta}
             +\beta_1^2\abs{Q_k-\widetilde Q_k}.
\end{align*}
Since $\norm{\widetilde m_{k-1}}_2\leq\sqrt n$, the bound on
$\abs{Q_k-\widetilde Q_k}$ above bounds the last term by
\begin{equation*}
 2\beta_1\sqrt n\abs{\beta-\theta}
 +2\beta_1^2\norm{m_k-\widetilde m_k}_2.
\end{equation*}
The triangle inequality therefore gives
\begin{align*}
 \norm{z_{k+1}-\widetilde z_{k+1}}_2
 \leq&~(4\beta_1+2\beta_1^2)\norm{m_k-\widetilde m_k}_2
       +\beta_1^2\norm{m_{k-1}-\widetilde m_{k-1}}_2\\
 &~+\beta_1\sqrt n
       \norm{A^{\mathrm{full}}-\widetilde A^{\mathrm{full}}}_{\mathrm{op}}
       +(4+2\beta_1)\sqrt n\abs{\beta-\theta}
       +\norm{y-\widetilde y}_2.
\end{align*}
Absorbing the fixed coefficients into $C$, we obtain
\begin{align*}
 \norm{z_{k+1}-\widetilde z_{k+1}}_2
 \leq&~C\left(\norm{m_k-\widetilde m_k}_2
       +\norm{m_{k-1}-\widetilde m_{k-1}}_2
       +\sqrt n\norm{A^{\mathrm{full}}
                         -\widetilde A^{\mathrm{full}}}_{\mathrm{op}}
       +\sqrt n\abs{\beta-\theta}\right)
       +\norm{y-\widetilde y}_2.
\end{align*}
Take $C\geq1$ in this recursion.  To make the induction explicit, set
$c_{-1}=c_0=0$ and $c_{k+1}=C(c_k+c_{k-1}+1)$.
We claim that, for every $k\geq0$,
\begin{equation*}
 \norm{z_k-\widetilde z_k}_2
 \leq c_k\left(\sqrt n
 \norm{A^{\mathrm{full}}-\widetilde A^{\mathrm{full}}}_{\mathrm{op}}
 +\norm{y-\widetilde y}_2+\sqrt n\abs{\beta-\theta}\right).
\end{equation*}
At $k=0$, both vectors are zero.  Also $m_{-1}=\widetilde m_{-1}=0$.
Since $\tanh$ is one-Lipschitz,
$\norm{m_k-\widetilde m_k}_2\leq\norm{z_k-\widetilde z_k}_2$.
If the above bound holds at steps $k$ and $k-1$, substituting
these bounds into the recursion gives
\begin{align*}
 \norm{z_{k+1}-\widetilde z_{k+1}}_2
 \leq&~C(c_k+c_{k-1}+1)
 \left(\sqrt n
 \norm{A^{\mathrm{full}}-\widetilde A^{\mathrm{full}}}_{\mathrm{op}}
 +\norm{y-\widetilde y}_2+\sqrt n\abs{\beta-\theta}\right).
\end{align*}
For $k=0$, the same calculation uses the zero initial $m_{-1}$ vectors.
The coefficient is precisely $c_{k+1}$, completing the induction.
{At the fixed depth $K$, $c_K$ depends only on the inverse-temperature interval and $K$.}  The bound for $m_K-\widetilde m_K$ follows by
applying the same one-Lipschitz property of $\tanh$.

Recall that
$R_{\mathrm{TAP},\beta}(u;A^{\mathrm{full}},y,q)
=\operatorname{arctanh}(u)-y-\beta A^{\mathrm{full}}u
+\beta^2(1-q)u$.
To control the residuals, subtract their defining expressions:
\begin{align*}
 &~R_{\mathrm{TAP},\beta}\tp{m_K;A^{\mathrm{full}},y,q_\beta(t)}
   -R_{\mathrm{TAP},\theta}\tp{\widetilde m_K;
        \widetilde A^{\mathrm{full}},\widetilde y,q_\theta(s)}\\
 =&~z_K-\widetilde z_K-(y-\widetilde y)
       -\beta A^{\mathrm{full}}(m_K-\widetilde m_K)
       -(\beta A^{\mathrm{full}}-\theta\widetilde A^{\mathrm{full}})
           \widetilde m_K\\
 &~+\beta^2(1-q_\beta(t))(m_K-\widetilde m_K)
       +(\beta^2(1-q_\beta(t))-\theta^2(1-q_\theta(s)))
           \widetilde m_K.
\end{align*}
{For the last term, we first separate the inverse-temperature and fixed-point perturbations, then apply part~\textup{(i)}:}
\begin{align*}
 \abs{\beta^2(1-q_\beta(t))-\theta^2(1-q_\theta(s))}
 =&~\abs{(\beta^2-\theta^2)(1-q_\beta(t))
            +\theta^2(q_\theta(s)-q_\beta(t))}\\
 \leq&~2\beta_1\abs{\beta-\theta}
             +\beta_1^2\abs{q_\theta(s)-q_\beta(t)}\\
 \leq&~2\beta_1\abs{\beta-\theta}
       +\frac{\beta_1^2}{1-\beta_1^2}
          \tp{\abs{t-s}+2\beta_1\abs{\beta-\theta}}\\
 =&~\frac{2\beta_1}{1-\beta_1^2}\abs{\beta-\theta}
       +\frac{\beta_1^2}{1-\beta_1^2}\abs{t-s}.
\end{align*}

Taking norms in the residual decomposition and using
$\norm{\widetilde m_K}_2\leq\sqrt n$, we obtain
\begin{align*}
 &~\left\|R_{\mathrm{TAP},\beta}\tp{m_K;A^{\mathrm{full}},y,q_\beta(t)}
   -R_{\mathrm{TAP},\theta}\tp{\widetilde m_K;
        \widetilde A^{\mathrm{full}},\widetilde y,q_\theta(s)}\right\|_2\\
 \leq&~\norm{z_K-\widetilde z_K}_2+\norm{y-\widetilde y}_2
       +(4\beta_1+\beta_1^2)\norm{m_K-\widetilde m_K}_2\\
 &~+\sqrt n\left(
       \beta_1\norm{A^{\mathrm{full}}-\widetilde A^{\mathrm{full}}}_{\mathrm{op}}
       +4\abs{\beta-\theta}
       +\frac{2\beta_1\abs{\beta-\theta}+\beta_1^2\abs{t-s}}
                    {1-\beta_1^2}\right)\\
 \leq&~((1+4\beta_1+\beta_1^2)c_K+1)
       \left(\sqrt n
       \norm{A^{\mathrm{full}}-\widetilde A^{\mathrm{full}}}_{\mathrm{op}}
       +\norm{y-\widetilde y}_2+\sqrt n\abs{\beta-\theta}\right)\\
 &~+\sqrt n\left(
       \beta_1\norm{A^{\mathrm{full}}-\widetilde A^{\mathrm{full}}}_{\mathrm{op}}
       +4\abs{\beta-\theta}
       +\frac{2\beta_1\abs{\beta-\theta}+\beta_1^2\abs{t-s}}
                    {1-\beta_1^2}\right),
\end{align*}
where the last inequality substitutes the two iterate bounds just proved.
Enlarging the constant $C$ gives the residual bound in part~\textup{(ii)}.

For part~\textup{(iii)}, the Hessian is
$\nabla_u^2F_{\mathrm{TAP},\beta}
=\operatorname{diag}\tp{(1-u_i^2)^{-1}}
-\beta A^{\mathrm{full}}+\beta^2(1-q_\beta(t))I_n$.
{At the same $u$, the diagonal terms $\operatorname{diag}((1-u_i^2)^{-1})$ cancel each other, thus}
\begin{align*}
 &~\nabla_u^2F_{\mathrm{TAP},\beta}
      \tp{u;A^{\mathrm{full}},y,q_\beta(t)}
   -\nabla_u^2F_{\mathrm{TAP},\theta}
      \tp{u;\widetilde A^{\mathrm{full}},\widetilde y,q_\theta(s)}\\
 =&~-\beta A^{\mathrm{full}}+\theta\widetilde A^{\mathrm{full}}
    +(\beta^2(1-q_\beta(t))-\theta^2(1-q_\theta(s)))I_n.
\end{align*}
The operator norm of a scalar multiple of $I_n$ is the absolute value
of that scalar.  Hence
\begin{align*}
 &~\left\|\nabla_u^2F_{\mathrm{TAP},\beta}
      \tp{u;A^{\mathrm{full}},y,q_\beta(t)}
   -\nabla_u^2F_{\mathrm{TAP},\theta}
      \tp{u;\widetilde A^{\mathrm{full}},\widetilde y,q_\theta(s)}
   \right\|_{\mathrm{op}}\\
 \leq&~\norm{\beta A^{\mathrm{full}}
                 -\theta\widetilde A^{\mathrm{full}}}_{\mathrm{op}}
       +\abs{\beta^2(1-q_\beta(t))-\theta^2(1-q_\theta(s))}\\
 \leq&~\beta_1
       \norm{A^{\mathrm{full}}-\widetilde A^{\mathrm{full}}}_{\mathrm{op}}
       +\left(4+\frac{2\beta_1}{1-\beta_1^2}\right)\abs{\beta-\theta}
       +\frac{\beta_1^2}{1-\beta_1^2}\abs{t-s}.
\end{align*}
Taking $C$ at least as large as these three coefficients proves
part~\textup{(iii)}.
\end{proof}

\subsection{Convexity for small times}
We first establish part~\textup{(i)} of \Cref{lem:amp-good-events}.
{At small times, we show that the TAP free energy is strongly convex.}
\begin{lemma}
\label{lem:inverse-poly-initial-geometry}
Fix $0<\beta_1<1$ and let $t_0 = (1-\beta_1)^2 (1 - \beta_1^2)/(4 \beta_1^2)$.
There are $c,C>0$, depending only on $\beta_1$,
such that, for every $0<\beta\leq\beta_1$, with
$\mathbb P_{n,\beta}^{\mathrm{pl,aug}}$-probability at least $1-C\e^{-cn}$,
$\norm{A^{\mathrm{full}}}_{\mathrm{op}}\leq4$ and
\begin{equation*}
 \forall t\in[0,t_0],\ y\in\R^n,\ u\in(-1,1)^n, \quad \nabla_u^2F_{\mathrm{TAP}}
   (u;A^{\mathrm{full}},y,q_\beta(t))\succeq \frac{(1 - \beta_1)^2}{2} I_n.
\end{equation*}
\end{lemma}
\begin{proof}
Set
$\epsilon_{\mathrm{op}}=\min(1/2,(1-\beta_1)^2/(4\beta_1))$.
By \Cref{lem:inverse-poly-planted-norm-tail}, with probability at least $1-C\e^{-cn}$, we have  $\norm{A^{\mathrm{full}}}_{\mathrm{op}}\leq2+\epsilon_{\mathrm{op}}$.
On this event, we have
\begin{align*}
 \nabla_u^2F_{\mathrm{TAP}}
 =&~\operatorname{diag}((1-u_i^2)^{-1})
       -\beta A^{\mathrm{full}}+\beta^2(1-q_\beta(t))I_n\\
 \succeq&~((1-\beta)^2-\beta\epsilon_{\mathrm{op}}
                      -\beta^2q_\beta(t))I_n\\
 \succeq&~\left(\frac34(1-\beta_1)^2
             -\frac{\beta_1^2t}{1-\beta_1^2}\right)I_n.
\end{align*}
{For the last inequality, we used \Cref{lem:q-beta-regularity} and the definition of $\epsilon_{\mathrm{op}}$.}
{For $t\le t_0$, we have $\beta_1^2t/(1-\beta_1^2)\leq(1-\beta_1)^2/4$, so the Hessian is at least $(1-\beta_1)^2I_n/2$.}
{Also $\|A^{\mathrm{full}}\|_{\mathrm{op}}\le2+\epsilon_{\mathrm{op}}\le5/2<4$, completing the proof.}
\end{proof}

\subsection{Handling large times}
Next, we strengthen the probability that the AMP center has a small
TAP residual and a positive Hessian lower bound in its neighborhood, proving the second part of \Cref{lem:amp-good-events}.
{We first fix a parameter pair $(\beta_1,s)$.}
{By \Cref{lem:inverse-poly-fixed-depth-stability}, even when the standard Gaussian vector generating $(A^{\mathrm{full}},Y_s)$ is perturbed by at most $\delta\sqrt n$, these properties still hold with slightly weaker bounds.}
{We apply \Cref{lem:gaussian-enlargement-boosting} to boost the probability.}
{Finally, we use the parameter estimates in \Cref{lem:inverse-poly-fixed-depth-stability} to extend the conclusion to nearby inverse temperatures and times $(\beta,t)$.}

\begin{lemma}
\label{lem:inverse-poly-positive-geometry}
Fix $0<\beta_1<1$ and $s>0$. There are constants $r,\kappa>0$ and an integer $K_0\geq0$ such that, for every fixed $K\geq K_0$, there are constants $c,C>0$ and an open rectangle $U\subseteq(0,1)\times(s/2,2s)$ containing $(\beta_1,s)$, such that the following holds.

{For every $n$ and every $(\beta,t)\in U$, with $\mathbb P_{n,\beta}^{\mathrm{pl,aug}}$-probability at least $1-C\e^{-cn}$, the following hold simultaneously:}
\begin{enumerate}[label=\textup{(\roman*)}]
\item $(A^{\mathrm{full}},Y_t)$ satisfies $(\kappa r/16)$-small TAP
residual and $(4r,4\kappa)$-local strong convexity at inverse temperature
$\beta$, time $t$, and AMP depth $K$, as defined in
\Cref{def:amp-tap-geometry}.
\item
\begin{equation*}
  \norm{A^{\mathrm{full}}}_{\mathrm{op}}\leq4,\qquad
  \norm{Y_t}_2\leq C\sqrt n,\qquad
  \norm{Y_t}_\infty\leq n^2.
\end{equation*}
\end{enumerate}
\end{lemma}
{To prove the lemma, we define a set $E_x$ for each spin $x\in\set{\pm1}^n$ at the fixed parameters $(\beta_1,s)$.}
In the next lemma, we show that this good event happens with probability at least $1/2$.
{Later, we apply \Cref{lem:gaussian-enlargement-boosting} to the distance from this event.}

For fixed $0<\beta_1<1$, {$s>0$}, $r,\kappa>0$,
and an integer $K\geq0$, define $E_x$ for each $x\in\set{\pm1}^n$
as follows.  {For $(A^{\mathrm{full}},y)$, write $\widehat m=\widehat m_{K,\beta_1}(A^{\mathrm{full}},y)$.}
Let $E_x\subseteq\R_{\mathrm{sym}}^{n\times n}\times\R^n$ be the set
of pairs satisfying
$\norm{A^{\mathrm{full}}}_{\mathrm{op}}\leq3$,
$\norm{y-sx}_2\leq2\sqrt{sn}$, and
$\norm{R_{\mathrm{TAP},\beta_1}
(\widehat m;A^{\mathrm{full}},y,q_{\beta_1}(s))}_2
\leq\kappa r\sqrt n/64$, with
\begin{equation*}
 \nabla_u^2F_{\mathrm{TAP},\beta_1}
   \tp{u;A^{\mathrm{full}},y,q_{\beta_1}(s)}
 \succeq8\kappa I_n
 \qquad\text{for every }u\in
 B(\widehat m,8r\sqrt n)\cap(-1,1)^n.
\end{equation*}
\begin{lemma}
\label{lem:inverse-poly-conditional-good-event}
Fix $0<\beta_1<1$ and {$s>0$}.
There exist $r,\kappa>0$ and an integer $K_0\geq0$, such that, for every fixed $K\geq K_0$, sufficiently large $n$ and every
$x\in\set{\pm1}^n$, under the planted construction at inverse
temperature $\beta_1$, we have
\begin{equation*}
 \Pr{(A^{\mathrm{full}},Y_s)\in E_x\mid X=x}\geq\frac12.
\end{equation*}
\end{lemma}
\begin{proof}
We first show that $\Pr{(A^{\mathrm{full}},Y_s)\notin E_X}=o(1)$. We verify each constraint and use a union bound.
Apply \Cref{prop:fixed-time-amp-tap-input} at $(\beta_1,s)$ and
choose $r,\kappa>0$ small enough, set $\epsilon=\kappa r/64$ and let $K_0$ be the depth
threshold in \Cref{prop:fixed-time-amp-tap-input}.
{Fix any $K\geq K_0$. With probability tending to one, we have}
\begin{equation*}
 \nabla_u^2F_{\mathrm{TAP},\beta_1}
   \tp{u;A^{\mathrm{full}},Y_s,q_{\beta_1}(s)}
 \succeq8\kappa I_n
 \quad\text{for every }u\in
 B\tp{\widehat m_{K,\beta_1}(A^{\mathrm{full}},Y_s),8r\sqrt n}
       \cap(-1,1)^n,
\end{equation*}
and the TAP residual at this AMP center is at most
$\epsilon\sqrt n\leq\kappa r\sqrt n/64$.
The matrix bound $\norm{A^{\mathrm{full}}}_{\mathrm{op}}\leq3$ follows from \Cref{lem:inverse-poly-planted-norm-tail},
and the noise bound $\norm{y-sx}_2\leq2\sqrt{sn}$ follows from $Y_s-sX\sim N(0,sI_n)$.
A union bound therefore gives
$\Pr{(A^{\mathrm{full}},Y_s)\notin E_X}=o(1)$.

Next, we claim that for any $x, x'\in \set{\pm 1}^n$, we have 
\[
\Pr{(A^{\mathrm{full}},Y_s)\in E_x\mid X=x} = \Pr{(A^{\mathrm{full}},Y_s)\in E_{x'}\mid X=x'}.
\]
Then, this lemma follows directly.

We now prove this claim. Fix $x,x'\in\set{\pm1}^n$.
Let $D=\operatorname{diag}(x_1x'_1,\ldots,x_nx'_n)$, so that
$D^\top=D$, $D^2=I_n$, and $Dx=x'$.
{We condition on $X=x$ and write $A^{\mathrm{full}}=W+\beta_1xx^\top/n$ and $Y_s=sx+\sqrt sG$, where $W\sim\mathrm{GOE}(n)$ and $G\sim N(0,I_n)$ are independent.}  Then
\begin{equation*}
 D^\top A^{\mathrm{full}}D=D^\top WD+\frac{\beta_1}{n}x'(x')^\top,
 \qquad DY_s=sx'+\sqrt sDG.
\end{equation*}
Since $(D^\top WD,DG)$ has the same joint law as $(W,G)$, we have
\begin{equation*}
 \mathcal L(D^\top A^{\mathrm{full}}D,DY_s\mid X=x)
 =\mathcal L(A^{\mathrm{full}},Y_s\mid X=x'),
\end{equation*}
where $\mathcal L$ denotes the distribution under the planted construction.
{It suffices to verify that, for every $(A^{\mathrm{full}},y)$,}
\begin{equation*}
 (A^{\mathrm{full}},y)\in E_x
 \quad\Longleftrightarrow\quad
 (D^\top A^{\mathrm{full}}D,Dy)\in E_{x'}.
\end{equation*}

First, the matrix and noise bounds are preserved because
$\norm{D^\top A^{\mathrm{full}}D}_{\mathrm{op}}
=\norm{A^{\mathrm{full}}}_{\mathrm{op}}$ and
$\norm{Dy-sx'}_2=\norm{D(y-sx)}_2=\norm{y-sx}_2$.
Next, the AMP recursion satisfies
\begin{equation*}
 \widehat m_{K,\beta_1}(D^\top A^{\mathrm{full}}D,Dy)
 =D\widehat m_{K,\beta_1}(A^{\mathrm{full}},y).
\end{equation*}
Indeed, its initial vectors are zero.  If its iterates at the transformed
input are $Dz_k,Dm_k,Dm_{k-1}$, the next update is
\begin{align*}
 \beta_1D^\top A^{\mathrm{full}}D(Dm_k)+Dy
       -\beta_1^2\left(1-\frac{\norm{Dm_k}_2^2}{n}\right)Dm_{k-1}
 =&~ D\left(\beta_1A^{\mathrm{full}}m_k+y
       -\beta_1^2\left(1-\frac{\norm{m_k}_2^2}{n}\right)m_{k-1}\right)\\
 =&~ Dz_{k+1}.
\end{align*}
Also $\tanh(Dz_{k+1})=D\tanh(z_{k+1})$, completing the induction.

For the residual, its definition and
$\operatorname{arctanh}(Du)=D\operatorname{arctanh}(u)$ give
\begin{align*}
 R_{\mathrm{TAP},\beta_1}(Du;D^\top A^{\mathrm{full}}D,Dy,q)
 =&~D\operatorname{arctanh}(u)-Dy
       -\beta_1D^\top A^{\mathrm{full}}D(Du)
       +\beta_1^2(1-q)Du\\
 =&~D R_{\mathrm{TAP},\beta_1}(u;A^{\mathrm{full}},y,q).
\end{align*}
Taking $u=\widehat m_{K,\beta_1}(A^{\mathrm{full}},y)$ and
$q=q_{\beta_1}(s)$ shows that the residual norm is unchanged.
Finally, writing the Hessian as
$\operatorname{diag}((1-u_i^2)^{-1})
-\beta_1A^{\mathrm{full}}+\beta_1^2(1-q)I_n$, we obtain
\begin{equation*}
 \nabla^2F_{\mathrm{TAP},\beta_1}
       (Du;D^\top A^{\mathrm{full}}D,Dy,q)
 =D^\top\bigl[\nabla^2F_{\mathrm{TAP},\beta_1}
       (u;A^{\mathrm{full}},y,q)\bigr]D.
\end{equation*}
The function $u\mapsto Du$ sends the ball around the original AMP center
onto the ball of the same radius around the transformed center,
preserves $(-1,1)^n$, and preserves the lower bound $8\kappa I_n$.
Thus every condition defining $E_x$ is equivalent to the corresponding
condition defining $E_{x'}$.
{Combining this equivalence with the preceding identity in distribution, {we obtain}}
\begin{align*}
 \Pr{(A^{\mathrm{full}},Y_s)\in E_x\mid X=x}
 =&~\Pr{(D^\top A^{\mathrm{full}}D,DY_s)\in E_{x'}\mid X=x}\\
 =&~\Pr{(A^{\mathrm{full}},Y_s)\in E_{x'}\mid X=x'}.
\end{align*}
These probabilities are therefore equal for all $x$ and equal to their
average over $X$.  Hence
\begin{equation*}
 \Pr{(A^{\mathrm{full}},Y_s)\in E_x\mid X=x}
 =\Pr{(A^{\mathrm{full}},Y_s)\in E_X}=1-o(1)\geq1/2
\end{equation*}
for all sufficiently large $n$, with a threshold independent of $x$.
\end{proof}
{Given $X=x$, the pair $(A^{\mathrm{full}},Y_s)$ is generated by a standard Gaussian vector.}
{We apply \Cref{lem:gaussian-enlargement-boosting} to the distance from the Gaussian input to $E_x$ to improve this probability bound.}
\begin{lemma}
\label{lem:inverse-poly-anchor-geometry}
Take $r,\kappa,K_0$ from \Cref{lem:inverse-poly-conditional-good-event}
and fix $K\geq K_0$.
There exist $c,C>0$, depending only on $(\beta_1,s,K)$, such that for
every $n$, with probability at least $1-C\e^{-cn}$ under the planted
construction at inverse temperature $\beta_1$, the following hold
simultaneously:
\begin{enumerate}[label=\textup{(\roman*)}]
\item $(A^{\mathrm{full}},Y_s)$ satisfies $(\kappa r/32)$-small TAP
residual at inverse temperature $\beta_1$, time $s$, and AMP depth $K$,
as defined in \Cref{def:amp-tap-geometry}.
\item $(A^{\mathrm{full}},Y_s)$ satisfies $(7r,7\kappa)$-local strong
convexity at the same inverse temperature, time, and AMP depth.
\item $\norm{A^{\mathrm{full}}}_{\mathrm{op}}\leq7/2$ and
$\norm{Y_s-sX}_2\leq3\sqrt{sn}$.
\end{enumerate}
\end{lemma}

\begin{proof}
Fix $x\in\set{\pm1}^n$.
By \Cref{lem:inverse-poly-conditional-good-event},
\begin{equation*}
 \Pr{(A^{\mathrm{full}},Y_s)\in E_x\mid X=x}\geq1/2
\end{equation*}
for all sufficiently large $n$.  {We condition on $X=x$.}
We claim that there is $\delta>0$, depending only on $(\beta_1,s,K)$,
such that whenever $(A_0^{\mathrm{full}},y_0)\in E_x$ and
\begin{equation*}
 \frac n2\norm{A^{\mathrm{full}}-A_0^{\mathrm{full}}}_{\mathrm F}^2
 +\frac1s\norm{y-y_0}_2^2\leq\delta^2n,
\end{equation*}
then the following holds.
Let 
$\widehat m=\widehat m_{K,\beta_1}(A^{\mathrm{full}},y)$ and
$\widehat m_0=\widehat m_{K,\beta_1}(A_0^{\mathrm{full}},y_0)$.
We have $\norm{\widehat m-\widehat m_0}_2\leq r\sqrt n$ and 
\begin{align*}
 \norm{R_{\mathrm{TAP},\beta_1}(\widehat m;A^{\mathrm{full}},y,q_{\beta_1}(s))
 -R_{\mathrm{TAP},\beta_1}(\widehat m_0;A_0^{\mathrm{full}},y_0,q_{\beta_1}(s))}_2
 \leq \kappa r\sqrt n/64.
\end{align*}
For every $u\in(-1,1)^n$,
\begin{equation*}
 \norm{\nabla_u^2F_{\mathrm{TAP},\beta_1}(u;A^{\mathrm{full}},y,q_{\beta_1}(s))
 -\nabla_u^2F_{\mathrm{TAP},\beta_1}(u;A_0^{\mathrm{full}},y_0,q_{\beta_1}(s))}_{\mathrm{op}}
 \leq\kappa.
\end{equation*}
Also $\norm{A^{\mathrm{full}}}_{\mathrm{op}}\leq7/2$ and
$\norm{y-sx}_2\leq3\sqrt{sn}$.

{To prove the claim, let $C$ be the constant in \Cref{lem:inverse-poly-fixed-depth-stability} for depth $K$ and inverse temperatures at most $\beta_1$.
We choose $\delta>0$ so that}
\begin{equation*}
 \delta\leq\min\left(1,\frac1{2\sqrt2},
 \frac{r}{C(\sqrt2+\sqrt s)},
 \frac{\kappa r}{64C(\sqrt2+\sqrt s)},
 \frac{\kappa}{\sqrt2C}\right).
\end{equation*}
The assumed distance bound implies
$\norm{A^{\mathrm{full}}-A_0^{\mathrm{full}}}_{\mathrm{op}}
\leq\sqrt2\delta$ and $\norm{y-y_0}_2\leq\delta\sqrt{sn}$.
{Since $(A_0^{\mathrm{full}},y_0)\in E_x$, the triangle inequality gives}
\begin{equation*}
 \norm{A^{\mathrm{full}}}_{\mathrm{op}}\leq3+\sqrt2\delta\leq7/2,
 \qquad \norm{y-sx}_2\leq(2+\delta)\sqrt{sn}\leq3\sqrt{sn}.
\end{equation*}
Both matrices therefore satisfy the operator-norm bound required by
\Cref{lem:inverse-poly-fixed-depth-stability}.
{We apply its part~\textup{(ii)} at the same inverse temperature $\beta_1$ and time $s$.}
Each of the center and residual differences is at most
$C(\sqrt2+\sqrt s)\delta\sqrt n$.
Our choice of $\delta$ bounds this quantity by both $r\sqrt n$ and
$\kappa r\sqrt n/64$.
Part~\textup{(iii)} bounds the Hessian difference at each fixed $u$ by
$\sqrt2C\delta\leq\kappa$, proving the claim.

{Consequently, whenever $(A_0^{\mathrm{full}},y_0)\in E_x$ and the preceding distance bound holds, $(A^{\mathrm{full}},y)$ also satisfies parts~\textup{(i)--(iii)} of this lemma.}
Indeed, the residual norm is at most
$\kappa r\sqrt n/64+\kappa r\sqrt n/64=\kappa r\sqrt n/32$.
For every $u\in B(\widehat m,7r\sqrt n)$,
$\norm{u-\widehat m_0}_2\leq\norm{u-\widehat m}_2
+\norm{\widehat m-\widehat m_0}_2<8r\sqrt n$.
On this ball intersected with $(-1,1)^n$, the original Hessian is at least
$8\kappa I_n$, so the perturbed Hessian is at least $7\kappa I_n$.

\begingroup

It remains to apply \Cref{lem:gaussian-enlargement-boosting} to a distance function on the Gaussian input generating the planted model given $X=x$.
Write $N_G=n(n+1)/2+n$ and let $\xi=(g,h)\sim N(0,I_{N_G})$ be independent of $X$, where $g=(g_{ij})_{i\leq j}$ and $h=(h_i)_{i\leq n}$.
Define the symmetric matrix $A_x(\xi)$ and field $y_x(\xi)$ by
\begin{equation*}
 (A_x(\xi))_{ij}=\frac{\beta_1x_ix_j}{n}+\frac{g_{ij}}{\sqrt n} (i<j),
 \quad (A_x(\xi))_{ii}=\frac{\beta_1}{n}+\sqrt{\frac2n}g_{ii},
 \quad (y_x(\xi))_i=sx_i+\sqrt s\,h_i.
\end{equation*}
For every $x\in\set{\pm1}^n$, this pair has the law of $(A^{\mathrm{full}},Y_s)$ given $X=x$.
Define
\begin{equation*}
 \mathcal G_x\defeq
 \left\{\zeta\in\R^{N_G}:(A_x(\zeta),y_x(\zeta))\in E_x\right\},
\end{equation*}
and define
\begin{equation*}
 \mathcal H_x\defeq
 \left\{\zeta\in\R^{N_G}:(A_x(\zeta),y_x(\zeta))
 \text{ satisfies parts~\textup{(i)--(iii)} of this lemma}\right\}.
\end{equation*}
These sets are Borel because the finite-depth AMP output, residual, and Hessian are continuous in the Gaussian input, and failure of the Hessian condition is an open condition in the data and the trial point.
By \Cref{lem:inverse-poly-conditional-good-event}, for all sufficiently large $n$ and every $x\in\set{\pm1}^n$,
$\Pr{\xi\in\mathcal G_x}\geq\frac12$.

For every $x\in\set{\pm1}^n$, define the $1$-Lipschitz function
$F_x(z)\defeq\operatorname{dist}(z,\mathcal G_x)$ for any $z\in\R^{N_G}$.
Since $F_x\geq0$ and $\Pr{F_x(\xi)=0}\geq\Pr{\xi\in\mathcal G_x}\geq1/2$, zero is a median of $F_x(\xi)$.

We next verify the stability inclusion.
For $\xi=(g,h)$ and $\zeta=(g',h')$, the coordinate scalings give
\begin{align*}
 &~\frac n2\norm{A_x(\xi)-A_x(\zeta)}_{\mathrm F}^2
       +\frac1s\norm{y_x(\xi)-y_x(\zeta)}_2^2\\
 =&~\frac n2\left(\frac2n\sum_{i<j}(g_{ij}-g'_{ij})^2
                   +\frac2n\sum_i(g_{ii}-g'_{ii})^2\right)
       +\sum_i(h_i-h'_i)^2
 =\norm{\xi-\zeta}_2^2.
\end{align*}
If $F_x(\xi)\leq\delta\sqrt n/2$, then some $\zeta\in\mathcal G_x$ satisfies $\norm{\xi-\zeta}_2<\delta\sqrt n$.
The preceding identity and the perturbation claim show that $\xi\in\mathcal H_x$.
Thus, for every $x\in\set{\pm1}^n$,
$\left\{\xi\in\R^{N_G}:F_x(\xi)\leq\frac{\delta\sqrt n}{2}\right\}
 \subseteq\mathcal H_x$.
For each $x$, we apply \Cref{lem:gaussian-enlargement-boosting} to $F_x$ with $L=1$, $q=0$, and $u=\delta\sqrt n/2$, and then average this bound over $X$ to obtain
\begin{equation*}
 \Pr{\xi\notin\mathcal H_X}
 \leq\Pr{F_X(\xi)>\delta\sqrt n/2}
 \leq 2\Phi(-\delta\sqrt n/2)
 \leq 2\exp(-\delta^2n/8).
\end{equation*}
Since $(A_X(\xi),y_X(\xi))$ has the planted law at inverse temperature $\beta_1$ and time $s$, this is the required failure bound for parts~\textup{(i)--(iii)}.
Finally, increasing $C$ covers the finitely many smaller dimensions and finishes the proof.
\endgroup
\end{proof}
{We now prove \Cref{lem:inverse-poly-positive-geometry} using the preceding probability bounds and the perturbation estimates in \Cref{lem:inverse-poly-fixed-depth-stability}.}
\begin{proof}[Proof of \Cref{lem:inverse-poly-positive-geometry}]
{We take $r,\kappa,K_0$ from \Cref{lem:inverse-poly-conditional-good-event}, fix any $K\geq K_0$, and apply \Cref{lem:inverse-poly-anchor-geometry} at $(\beta_1,s)$.}
{It remains to transfer these bounds to nearby inverse temperatures and times through \Cref{lem:inverse-poly-fixed-depth-stability}.}
Draw $X\sim\operatorname{Unif}(\set{\pm1}^n)$,
$W\sim\mathrm{GOE}(n)$, and $G\sim N(0,I_n)$ independently.
For two points $(\beta_1,s)$ and $(\beta,t)$, we set
\begin{equation*}
 A_{\beta_1}^{\mathrm{full}}=W+\beta_1 XX^\top/n,\quad
 A_{\beta}^{\mathrm{full}}=W+\beta XX^\top/n,\quad
 Y_s=sX+\sqrt sG,\quad Y_t=tX+\sqrt tG.
\end{equation*}
On the event in
\Cref{lem:inverse-poly-anchor-geometry}, we have $\norm G_2\leq3\sqrt n$,
and for $t\in(s/2,2s)$,
\begin{equation*}
 \norm{A_{\beta}^{\mathrm{full}}-A_{\beta_1}^{\mathrm{full}}}_{\mathrm{op}}
 =|\beta-\beta_1|,\qquad
 \norm{Y_t-Y_s}_2
 \leq( |t-s|+3|\sqrt t-\sqrt s|)\sqrt n
 \leq C_s|t-s|\sqrt n.
\end{equation*}
We apply \Cref{lem:inverse-poly-fixed-depth-stability} on the temperature interval $[\beta_1/2,(1+\beta_1)/2]$, and choose $U\subseteq(\beta_1/2,(1+\beta_1)/2)\times(s/2,2s)$.
Parts~\textup{(ii)--(iii)} of
\Cref{lem:inverse-poly-fixed-depth-stability} bound the center and
residual changes by $C(|\beta-\beta_1|+|t-s|)\sqrt n$, and the
Hessian change by $C(|\beta-\beta_1|+|t-s|)$.
Choose an open rectangle $U$ around $(\beta_1,s)$ so small that these
changes are at most $r\sqrt n$, $\kappa r\sqrt n/32$, and $\kappa$,
respectively.  Also require $|\beta-\beta_1|\leq1/2$ throughout $U$.
{By \Cref{lem:inverse-poly-anchor-geometry}-\textup{(iii)} and $\|A_\beta^{\mathrm{full}}-A_{\beta_1}^{\mathrm{full}}\|_{\mathrm{op}}=|\beta-\beta_1|$, we obtain $\|A_\beta^{\mathrm{full}}\|_{\mathrm{op}}\le7/2+1/2=4$.}

The $4r\sqrt n$ ball around the AMP center at $(\beta,t)$ lies in the
$7r\sqrt n$ ball around the AMP center at $(\beta_1,s)$.
{The Hessian at inverse temperature $\beta$ and time $t$ is at least $6\kappa I_n$}
throughout the smaller ball intersected with $(-1,1)^n$,
and its center residual is at most $\kappa r\sqrt n/16$.
Also $\norm{Y_t}_2\leq(t+3\sqrt t)\sqrt n$ and
$\norm{Y_t}_\infty\leq n^2$ for large $n$.
For each parameter pair in $U$, failure is contained in the complement
of the event in \Cref{lem:inverse-poly-anchor-geometry}, whose
probability is at most $C\e^{-cn}$.
\end{proof}
{We next combine the mean-square estimate in \Cref{prop:fixed-time-amp-tap-input}-\textup{(ii)} with concentration to prove posterior-mean proximity with inverse-polynomial failure probability.}
\begin{lemma}
\label{lem:inverse-poly-positive-posterior}
Choose the fixed depth $K$ in \Cref{lem:inverse-poly-positive-geometry}
also above the threshold in \Cref{prop:fixed-time-amp-tap-input}
{at $(\beta_1,s)$ with $\epsilon=r/(4\sqrt2)$.}
For these parameters, the rectangle $U$ in
\Cref{lem:inverse-poly-positive-geometry} can be chosen so that, for every
fixed integer $p\geq1$ and every $(\beta,t)\in U$, with
$\mathbb P_{n,\beta}^{\mathrm{pl,aug}}$-probability at least
$1-C_{\beta_1,s,p}n^{-p}$, the pair $(A^{\mathrm{full}},Y_t)$ satisfies
$r$-posterior-mean proximity at inverse temperature $\beta$, time $t$,
and AMP depth $K$, as defined in \Cref{def:amp-tap-geometry}.
\end{lemma}

{To prove this lemma, we expand the squared distance:}
\begin{equation*}
 \frac{\norm{m\tp{A,Y_t}-\widehat m_K\tp{Y_t}}_2^2}{n}
 =\frac{\norm{m\tp{A,Y_t}}_2^2}{n}
  +\frac{\norm{\widehat m_K\tp{Y_t}}_2^2}{n}
  -\frac{2\inner{m\tp{A,Y_t}}{\widehat m_K\tp{Y_t}}}{n}.
\end{equation*}
We control these three terms separately.

{For the first term, $\|m(A,Y_t)\|_2^2/n$ is the conditional expectation of $\langle\sigma^1,\sigma^2\rangle/n$ for two independent samples from $\mu_{A,Y_t}$, so we apply \Cref{lem:tap-overlap-mgf}.}
{For the second term, \Cref{lem:inverse-poly-fixed-depth-stability} and Gaussian concentration bound the fluctuations of $\|\widehat m_K(Y_t)\|_2^2/n$.}
{For the cross term, we use Gaussian concentration for $n^{-1}\langle X,\widehat m_K(Y_t)\rangle$ and then condition on $(A^{\mathrm{full}},Y_t)$ and apply Jensen's inequality.}

We first use the next lemma to handle the last two terms.

\begin{lemma}
\label{lem:inverse-poly-fixed-depth-concentration}
Fix an integer $K\geq0$ and a compact rectangle {$I\subseteq(0,1)\times(0,\infty)$}.
{At each $(\beta,t)\in I$, let $\widehat m_K(Y_t)$ be the AMP center at inverse temperature $\beta$, and set}
\begin{equation*}
 S_2(\beta,t)=\frac{\norm{\widehat m_K(Y_t)}_2^2}{n},\qquad
 S_3(\beta,t)=\frac{\inner X{\widehat m_K(Y_t)}}{n}.
\end{equation*}
For $j\in\{2,3\}$ and every fixed integer $p\geq1$,
\begin{equation*}
 \E[\mathbb P_{n,\beta}^{\mathrm{pl,aug}}]{
  \left|S_j(\beta,t)
   -\E[\mathbb P_{n,\beta}^{\mathrm{pl,aug}}]{S_j(\beta,t)}\right|^{2p}}
 \leq C_{I,K,p}n^{-p}.
\end{equation*}
For every $(\beta,t),(\beta_1,s)\in I$,
\begin{align*}
 \left|\E[\mathbb P_{n,\beta}^{\mathrm{pl,aug}}]{S_j(\beta,t)}
       -\E[\mathbb P_{n,\beta_1}^{\mathrm{pl,aug}}]{S_j(\beta_1,s)}\right|
 \leq C_{I,K}(|\beta-\beta_1|+|t-s|)+C_{I,K}\e^{-c_I n}.
\end{align*}
\end{lemma}

\begin{proof}
We first prove the concentration.
{We fix $(\beta,t)\in I$ and condition on $X=x$.}
Take independent standard Gaussian variables $(g_{ij})_{i\leq j}$ and
$(h_i)_{i\leq n}$, and write
\begin{equation*}
 A^{\mathrm{full}}_{ij}=\frac{\beta x_ix_j}{n}+\frac{g_{ij}}{\sqrt n}
 \ (i<j),\qquad
 A^{\mathrm{full}}_{ii}=\frac{\beta}{n}+\sqrt{\frac2n}g_{ii},
 \qquad Y_{t,i}=tx_i+\sqrt t\,h_i.
\end{equation*}
Collect these variables into $\xi=(g,h)$, where $g=(g_{ij})_{i\leq j}$
and $h=(h_i)_{i\leq n}$.
For two such standard Gaussian inputs $\xi,\widetilde\xi$, the scaling gives
\begin{equation*}
 \norm{\xi-\widetilde\xi}_2^2
 =\frac n2\norm{A^{\mathrm{full}}-\widetilde A^{\mathrm{full}}}_{\mathrm F}^2
  +\frac1t\norm{y-\widetilde y}_2^2.
\end{equation*}
Consequently,
$\sqrt n\norm{A^{\mathrm{full}}-\widetilde A^{\mathrm{full}}}_{\mathrm{op}}
\leq\sqrt2\norm{\xi-\widetilde\xi}_2$ and
$\norm{y-\widetilde y}_2\leq\sqrt t\norm{\xi-\widetilde\xi}_2$.
When both matrix norms are at most $4$,
\Cref{lem:inverse-poly-fixed-depth-stability}-\textup{(ii)} bounds the
change of the AMP output by $C_{I,K}\norm{\xi-\widetilde\xi}_2$.
Let $u,v$ be the AMP outputs corresponding to $\xi,\widetilde\xi$,
respectively.  Both belong to $[-1,1]^n$, so
$\norm{u+v}_2\leq2\sqrt n$.
Writing $S_j(\xi)$ for the scalar evaluated at the data generated by
$\xi$, with $(\beta,t,x)$ fixed, we obtain
\begin{align*}
 |S_2(\xi)-S_2(\widetilde\xi)|
 =&~\frac{|\norm u_2^2-\norm v_2^2|}{n}
 =\frac{|\inner{u+v}{u-v}|}{n}\\
 \leq&~\frac{\norm{u+v}_2}{n}\norm{u-v}_2
 \leq\frac{2}{\sqrt n}\norm{u-v}_2
 \leq\frac{2C_{I,K}}{\sqrt n}\norm{\xi-\widetilde\xi}_2.
\end{align*}
For $S_3$, the planted spin is the same fixed $x$ at both inputs and
$\norm x_2=\sqrt n$.  Hence
\begin{align*}
 |S_3(\xi)-S_3(\widetilde\xi)|
 =\frac{|\inner x{u-v}|}{n}
 \leq\frac{\norm x_2}{n}\norm{u-v}_2
 =\frac1{\sqrt n}\norm{u-v}_2
 \leq\frac{C_{I,K}}{\sqrt n}\norm{\xi-\widetilde\xi}_2.
\end{align*}
Absorbing the factor two into $C_{I,K}$, each $S_j$, $j\in\{2,3\}$,
is therefore $C_{I,K}/\sqrt n$-Lipschitz in the standard Gaussian input
on the set where the matrix norm is at most $4$.

\begingroup

{We have proved this Lipschitz property on the set $\norm{A^{\mathrm{full}}}_{\mathrm{op}}\le4$.
For each fixed $x$, we use McShane's extension theorem and clip the extension to $[-1,1]$, preserving its Lipschitz constant; denote the resulting function by $\overline S_{j,x}$, which agrees with $S_j$ on the matrix-norm cutoff.
Doing this for every $x$ defines the family $(\overline S_{j,x})_{x\in\set{\pm1}^n}$, and under the planted construction we set $\overline S_j=\overline S_{j,X}(\xi)$, where $X$ and $\xi$ are independent.}
{For each $x\in\set{\pm1}^n$, we apply the moment conclusion of \Cref{lem:gaussian-enlargement-boosting} to $\overline S_{j,x}$ with $L=C_{I,K}/\sqrt n$.}
Since $\xi$ is independent of $X$, for every integer $p\geq1$ and every $x\in\set{\pm1}^n$, we obtain
\begin{equation*}
 \E{\left|\overline S_j-\E{\overline S_j\mid X=x}\right|^{2p}\mid X=x}
 \leq C_{I,K,p}n^{-p}.
\end{equation*}
Next, we compare $S_j$ with its extension $\overline S_j$.
{Given $X=x$, we have $A^{\mathrm{full}}=W+\beta xx^\top/n$.
Conjugating by $\operatorname{diag}(x)$ shows that the distribution of $\|A^{\mathrm{full}}\|_{\mathrm{op}}$ does not depend on $x$.
Since the inverse temperatures in $I$ are bounded above by a constant smaller than one, the proof of \Cref{lem:inverse-poly-planted-norm-tail} gives constants $C_I,c_I>0$ such that}
\begin{align*}
 \Pr{\norm{A^{\mathrm{full}}}_{\mathrm{op}}>4\mid X=x}
 =\Pr{\norm{A^{\mathrm{full}}}_{\mathrm{op}}>4} \leq C_I\e^{-c_In}{.}
\end{align*}
Since $|S_j|$ and $|\overline S_j|$ are bounded by one,
\begin{equation*}
 \E{|S_j-\overline S_j|^{2p}\mid X=x}\leq2^{2p}C_I\e^{-c_In}.
\end{equation*}
Jensen's inequality also gives
\begin{equation*}
 \left|\E{\overline S_j-S_j\mid X=x}\right|^{2p}
 \leq\E{|\overline S_j-S_j|^{2p}\mid X=x}
 \leq2^{2p}C_I\e^{-c_In}.
\end{equation*}
We decompose
\begin{equation*}
 S_j-\E{S_j\mid X=x}
 =(S_j-\overline S_j)
 +(\overline S_j-\E{\overline S_j\mid X=x})
 +\E{\overline S_j-S_j\mid X=x}.
\end{equation*}
Using $|a+b+c|^{2p}\leq3^{2p-1}(|a|^{2p}+|b|^{2p}+|c|^{2p})$,
we obtain
\begin{align*}
 &~\E{|S_j-\E{S_j\mid X=x}|^{2p}\mid X=x}\\
 \leq&~3^{2p-1}\left(
 \E{|S_j-\overline S_j|^{2p}\mid X=x}
 +\E{|\overline S_j-\E{\overline S_j\mid X=x}|^{2p}\mid X=x}
 +\left|\E{\overline S_j-S_j\mid X=x}\right|^{2p}\right).
\end{align*}
Substituting the preceding bounds, we obtain
\begin{equation*}
 \E{|S_j-\E{S_j\mid X=x}|^{2p}\mid X=x}
 \leq3^{2p-1}\left(2^{2p+1}C_I\e^{-c_In}
 +C_{I,K,p}n^{-p}\right).
\end{equation*}
For each fixed $p$, $n^p\e^{-c_In}\leq(p/(\e c_I))^p$.
After increasing the constant depending only on $(I,K,p)$, we therefore obtain
\begin{equation*}
 \E{|S_j-\E{S_j\mid X=x}|^{2p}\mid X=x}
 \leq C_{I,K,p}n^{-p}.
\end{equation*}
\endgroup

By the same sign-change argument as in the proof of
\Cref{lem:inverse-poly-conditional-good-event}, the conditional law of
$S_j$ given $X=x$ does not depend on $x$: both the squared AMP norm and
its overlap with the planted spin are unchanged by the sign change.
Hence $\E{S_j\mid X=x}=\E[\mathbb P_{n,\beta}^{\mathrm{pl,aug}}]{S_j}$.
Averaging the conditional moment bound proves the first part of this lemma.

Next, we prove the stability.
We couple two quantities $\E[\mathbb P_{n,\beta}^{\mathrm{pl,aug}}]{S_j(\beta,t)}$ and $\E[\mathbb P_{n,\beta_1}^{\mathrm{pl,aug}}]{S_j(\beta_1,s)}$ by using a common uniform spin $X$, full GOE matrix $W$, and independent $G\sim N(0,I_n)$.
Hence we have $A_\beta^{\mathrm{full}}=W+\beta XX^\mathsf T/n$,
$A_{\beta_1}^{\mathrm{full}}=W+\beta_1 XX^\mathsf T/n$,
$Y_t=tX+\sqrt tG$, and $Y_s=sX+\sqrt sG$.
Then
\begin{equation*}
 \norm{A_\beta^{\mathrm{full}}-A_{\beta_1}^{\mathrm{full}}}_{\mathrm{op}}
 =|\beta-\beta_1|,\qquad
 Y_t-Y_s=(t-s)X+(\sqrt t-\sqrt s)G.
\end{equation*}
On the event
$\norm{A_\beta^{\mathrm{full}}}_{\mathrm{op}}\leq4$ and $\norm{A_{\beta_1}^{\mathrm{full}}}_{\mathrm{op}}\leq4$,
part~\textup{(ii)} of
\Cref{lem:inverse-poly-fixed-depth-stability} implies
\begin{equation*}
 |S_j(\beta,t)-S_j(\beta_1,s)|
 \leq C_{I,K}\left(|\beta-\beta_1|+|t-s|
       +|\sqrt t-\sqrt s|\frac{\norm G_2}{\sqrt n}\right).
\end{equation*}
The time coordinates in $I$ are bounded away from zero, so
$|\sqrt t-\sqrt s|\leq C_I|t-s|$, and
$\E[G\sim N(0,I_n)]{\norm G_2}\leq\sqrt n$.
The event
$\norm{A_\beta^{\mathrm{full}}}_{\mathrm{op}}>4$ or
$\norm{A_{\beta_1}^{\mathrm{full}}}_{\mathrm{op}}>4$
has probability at most $C_I\e^{-c_In}$, and
$|S_j(\beta,t)-S_j(\beta_1,s)|\leq2$.
Taking expectations gives the second part of this lemma.
\end{proof}

Now, we are ready to prove \Cref{lem:inverse-poly-positive-posterior}.
\begin{proof}[Proof of \Cref{lem:inverse-poly-positive-posterior}]
Recall from \Cref{lem:inverse-poly-fixed-depth-concentration} that
$S_2(\beta,t)=\norm{\widehat m_K(Y_t)}_2^2/n$ and
$S_3(\beta,t)=\inner X{\widehat m_K(Y_t)}/n$.
Set $S_1(\beta,t)=\norm{m(A,Y_t)}_2^2/n$.
Since $\E[\mathbb P_{n,\beta}^{\mathrm{pl,aug}}]{X\mid A^{\mathrm{full}},Y_t}=m(A,Y_t)$, expanding the squared distance gives
\begin{equation}
  \frac{\norm{m\tp{A,Y_t}-\widehat m_{K}\tp{Y_t}}_2^2}{n}
  =S_1(\beta,t)+S_2(\beta,t)-2\E[\mathbb P_{n,\beta}^{\mathrm{pl,aug}}]{S_3(\beta,t)\mid A^{\mathrm{full}},Y_t}.
  \label{eq:inverse-poly-three-scalars}
\end{equation}
We now prove the concentration of the above quantity.
For the first term, by \Cref{lem:tap-overlap-mgf}, we use $q_\beta(t)$ to approximate it.
For the other two terms, we use their expectations. 
This gives the next approximation
\begin{equation*}
 a_n(\beta,t)\defeq q_\beta(t)+\E[\mathbb P_{n,\beta}^{\mathrm{pl,aug}}]{S_2(\beta,t)}
                    -2\E[\mathbb P_{n,\beta}^{\mathrm{pl,aug}}]{S_3(\beta,t)}.
\end{equation*}
Let $r$ be the radius constant from
\Cref{lem:inverse-poly-positive-geometry}.
We claim that $U$ can be chosen independently of $p$ so that, for every
fixed integer $p\geq1$ and every $(\beta,t)\in U$,
for all sufficiently large $n$, we have
\begin{equation*}
 a_n(\beta,t)\leq r^2/4,\qquad
  \E[\mathbb P_{n,\beta}^{\mathrm{pl,aug}}]{\abs{
    \frac{\norm{m\tp{A,Y_t}-\widehat m_{K}\tp{Y_t}}_2^2}{n}
    -a_n(\beta,t)}^{2p}}\leq C_{\beta_1,s,p}n^{-p}.
\end{equation*}

Assume the claim.  If
$\norm{m(A,Y_t)-\widehat m_K(Y_t)}_2>r\sqrt n$, then
\begin{equation*}
 \frac{\norm{m(A,Y_t)-\widehat m_K(Y_t)}_2^2}{n}-a_n(\beta,t)
 >r^2-a_n(\beta,t)\geq\frac{3r^2}{4}.
\end{equation*}
Applying Markov's inequality to the $2p$-th power
of the absolute difference, we obtain
\begin{align*}
 &~\mathbb P_{n,\beta}^{\mathrm{pl,aug}}
 \left\{\norm{m(A,Y_t)-\widehat m_K(Y_t)}_2>r\sqrt n\right\}\\
 \leq&~\mathbb P_{n,\beta}^{\mathrm{pl,aug}}\left\{
 \left|\frac{\norm{m(A,Y_t)-\widehat m_K(Y_t)}_2^2}{n}
 -a_n(\beta,t)\right|>\frac{3r^2}{4}\right\}\\
 \leq&~\left(\frac4{3r^2}\right)^{2p}
 \E[\mathbb P_{n,\beta}^{\mathrm{pl,aug}}]{\left|
 \frac{\norm{m(A,Y_t)-\widehat m_K(Y_t)}_2^2}{n}
 -a_n(\beta,t)\right|^{2p}}
 \leq C_{\beta_1,s,p}n^{-p}.
\end{align*}
Increasing the constant to cover the finitely many smaller $n$ proves
the lemma.

{It remains to prove the claim.
We first bound $a_n(\beta_1,s)$ and then use the parameter stability estimates to bound $a_n(\beta,t)$ throughout $U$.}
\begingroup
By the choice of $K$ and \Cref{prop:fixed-time-amp-tap-input}-\textup{(ii)} at $(\beta_1,s)$ with $\epsilon=r/(4\sqrt2)$, we have
\begin{equation*}
 \E[\mathbb P_{n,\beta_1}^{\mathrm{pl,aug}}]{
 \frac{\|m(A,Y_s)-\widehat m_{K,\beta_1}(A^{\mathrm{full}},Y_s)\|_2^2}{n}}
 \leq\frac{r^2}{32}+o(1).
\end{equation*}
\endgroup
Taking expectations in
\cref{eq:inverse-poly-three-scalars} under the planted construction
at inverse temperature $\beta_1$ yields
\begin{align*}
  a_n(\beta_1,s)
  =\E[\mathbb P_{n,\beta_1}^{\mathrm{pl,aug}}]{\frac{\norm{m\tp{A,Y_s}-\widehat m_{K}\tp{Y_s}}_2^2}{n}}
    +q_{\beta_1}\tp s-\E[\mathbb P_{n,\beta_1}^{\mathrm{pl,aug}}]{S_1(\beta_1,s)}.
\end{align*}
{For the second term, we condition on $(A,Y_t)$, draw $\sigma^1,\sigma^2$ independently from $\mu_{A,Y_t}$, and set $R_{12}=\inner{\sigma^1}{\sigma^2}/n$.}
Then we have that $S_1(\beta,t)=\E[(\sigma^1,\sigma^2)\sim\mu_{A,Y_t}^{\otimes2}]{R_{12}}$.
{By \Cref{lem:tap-overlap-mgf} and Jensen's inequality, we have}
\begin{align*}
 \E[\mathbb P_{n,\beta}^{\mathrm{pl,aug}}]{|S_1(\beta,t)-q_\beta(t)|^{2p}}
 \leq&~\E[\mathbb P_{n,\beta}^{\mathrm{pl,aug}}]{\E[(\sigma^1,\sigma^2)\sim\mu_{A,Y_t}^{\otimes2}]{
 |R_{12}-q_\beta(t)|^{2p}}}
 \leq C_{\beta_1,s,p}n^{-p}.
\end{align*}
Taking $p=1$ and applying Cauchy--Schwarz at $(\beta_1,s)$ gives
\begin{equation*}
 |q_{\beta_1}(s)-\E[\mathbb P_{n,\beta_1}^{\mathrm{pl,aug}}]{S_1(\beta_1,s)}|
 \leq\E[\mathbb P_{n,\beta_1}^{\mathrm{pl,aug}}]{|S_1(\beta_1,s)-q_{\beta_1}(s)|^2}^{1/2}
 \leq C_{\beta_1,s}n^{-1/2}.
\end{equation*}
{Combining the two estimates, we obtain $a_n(\beta_1,s)\leq r^2/32+o(1)+C_{\beta_1,s}n^{-1/2}\leq r^2/8$ for all sufficiently large $n$.}

Next, we extend this upper bound to any other $(\beta, t) \in U$.
Parts~\textup{(i)--(ii)} of \Cref{lem:inverse-poly-fixed-depth-stability}
and \Cref{lem:inverse-poly-fixed-depth-concentration} imply
\begin{equation*}
 a_n(\beta,t)\leq a_n(\beta_1,s)
       +C(|\beta-\beta_1|+|t-s|)+C\e^{-cn}.
\end{equation*}
Shrink $U$ around $(\beta_1,s)$, keeping its closure inside the original
rectangle, so that
$a_n(\beta,t)\leq r^2/4$ for $(\beta,t)\in U$ for all sufficiently large $n$.

We now prove the moment bound in the claim.  Subtracting the definition
of $a_n$ from \cref{eq:inverse-poly-three-scalars}, we have
\begin{align*}
 \frac{\norm{m(A,Y_t)-\widehat m_K(Y_t)}_2^2}{n}-a_n(\beta,t)
 =&~S_1(\beta,t)-q_\beta(t)
    +S_2(\beta,t)-\E[\mathbb P_{n,\beta}^{\mathrm{pl,aug}}]{S_2(\beta,t)}\\
 &~-2\left(\E[\mathbb P_{n,\beta}^{\mathrm{pl,aug}}]{S_3(\beta,t)\mid A^{\mathrm{full}},Y_t}
             -\E[\mathbb P_{n,\beta}^{\mathrm{pl,aug}}]{S_3(\beta,t)}\right).
\end{align*}
The first term has the overlap moment bound just established.
For the second term, \Cref{lem:inverse-poly-fixed-depth-concentration}
gives a $2p$-moment at most $C_{\beta_1,s,p}n^{-p}$.
For the conditional expectation in the last term, Jensen's inequality gives
\begin{align*}
 \E[\mathbb P_{n,\beta}^{\mathrm{pl,aug}}]{\left|\E[\mathbb P_{n,\beta}^{\mathrm{pl,aug}}]{S_3(\beta,t)\mid A^{\mathrm{full}},Y_t}
             -\E[\mathbb P_{n,\beta}^{\mathrm{pl,aug}}]{S_3(\beta,t)}\right|^{2p}}
 \leq&~\E[\mathbb P_{n,\beta}^{\mathrm{pl,aug}}]{|S_3(\beta,t)-\E[\mathbb P_{n,\beta}^{\mathrm{pl,aug}}]{S_3(\beta,t)}|^{2p}}\\
 \leq&~C_{\beta_1,s,p}n^{-p},
\end{align*}
where the last inequality is again
\Cref{lem:inverse-poly-fixed-depth-concentration}.
Applying $|a+b-2c|^{2p}\leq3^{2p-1}
(|a|^{2p}+|b|^{2p}+2^{2p}|c|^{2p})$ to the decomposition proves
the moment bound, and hence the claim.
\end{proof}

\begin{proof}[Proof of \Cref{lem:amp-good-events}]
We first prove part~\textup{(i)}. We choose
\begin{equation*}
 t_0=\min\left(T,\frac{(1-\beta_1)^2(1-\beta_1^2)}{4\beta_1^2}\right),
 \qquad \kappa_0=\frac{(1-\beta_1)^2}{8}.
\end{equation*}
At $u=0$, the Hessian is
\begin{equation*}
 \nabla_u^2F_{\mathrm{TAP}}(0;A^{\mathrm{full}},y,q_\beta(t))
 =I_n-\beta A^{\mathrm{full}}+\beta^2(1-q_\beta(t))I_n.
\end{equation*}
For every $H>0$ and $\beta\in[\beta_0,\beta_1]$, we apply \Cref{lem:inverse-poly-initial-geometry} at $u=0$ and use $4\kappa_0=(1-\beta_1)^2/2$ to obtain
\begin{equation*}
 \mathbb P_{n,\beta}^{\mathrm{pl,aug}}\left\{
 \norm{A^{\mathrm{full}}}_{\mathrm{op}}>4\ \text{or}\
 \exists t\in[0,t_0]:
 \lambda_{\min}\bigl(I_n-\beta A^{\mathrm{full}}+\beta^2(1-q_\beta(t))I_n\bigr)<4\kappa_0
 \right\}\le C\e^{-cn}\le C_Hn^{-H}
\end{equation*}
for all sufficiently large $n$.

We next prove part~\textup{(ii)}. We fix a reference pair $(\bar\beta,s)$ in $[\beta_0,\beta_1]\times[t_0,T]$ and use \Cref{lem:inverse-poly-positive-geometry,lem:inverse-poly-positive-posterior} with the same fixed AMP depth. {Their proofs use stability, \Cref{lem:gaussian-enlargement-boosting}, and overlap estimates on the compact time interval $[t_0/2,2T]\subseteq(0,\infty)$.}
We obtain an open rectangle $U_{\bar\beta,s}$ around $(\bar\beta,s)$ and constants $r,\kappa>0$ and $K\in\mathbb Z_{\geq0}$.

For every fixed $(\beta,t)\in U_{\bar\beta,s}\cap([\beta_0,\beta_1]\times[t_0,T])$, write $\widehat m=\widehat m_K(Y_t)$.
By \Cref{lem:inverse-poly-positive-geometry}, on one event of $\mathbb P_{n,\beta}^{\mathrm{pl,aug}}$-probability at least $1-C\e^{-cn}$, the pair $(A^{\mathrm{full}},Y_t)$ satisfies $(\kappa r/16)$-small TAP residual and $(4r,4\kappa)$-local strong convexity at depth $K$, together with $\norm{A^{\mathrm{full}}}_{\mathrm{op}}\leq4$ and $\norm{Y_t}_\infty\leq n^2$.
For every fixed integer $p\geq1$, the posterior-mean estimate at the same $K$ and $(\beta,t)$ is
\begin{equation*}
 \mathbb P_{n,\beta}^{\mathrm{pl,aug}}\left\{\norm{m(A,Y_t)-\widehat m}_2>r\sqrt n\right\}
 \leq C_pn^{-p}.
\end{equation*}
For each reference pair $(\bar\beta,s)$, we choose a closed rectangle contained in $U_{\bar\beta,s}$ with $(\bar\beta,s)$ strictly inside it.
By compactness, we retain finitely many such rectangles $I_1,\ldots,I_{\ell_{\max}}$, with their corresponding $K_\ell,r_\ell,\kappa_\ell$, so that every pair in $[\beta_0,\beta_1]\times[t_0,T]$ lies strictly inside at least one $I_\ell$.
We now fix $H>0$ and set $p=\max(1,\lceil H\rceil)$. For fixed $\ell$ and $(\beta,t)\in I_\ell\cap([\beta_0,\beta_1]\times[t_0,T])$, let $G_{\ell,\beta,t}$ be the event in part~\textup{(ii)}. Since $\kappa_\ell r_\ell/16\leq\kappa_\ell r_\ell/8$, by the above two probability bounds and a union bound, {for suitable constants $C_H,n_H$ and $n\geq n_H$, we have}
\begin{equation*}
 \mathbb P_{n,\beta}^{\mathrm{pl,aug}}(G_{\ell,\beta,t}^c)
 \leq C_\ell\e^{-c_\ell n}+C_{\ell,p}n^{-p}
 \leq C_Hn^{-H}.\qedhere
\end{equation*}
\end{proof}

\section{WPI under a strong enough field}
\label{app:terminal-posterior-variance}

In this section, we prove \Cref{lem:app-wpi-two-path-terminal}.
For a symmetric $J\in\R^{m\times m}$ and $h\in\R^m$, we define the distribution $\pi_{J,h}$ on $\{\pm1\}^m$ by
\begin{equation*}
 \pi_{J,h}(\sigma)\propto\exp(\inner{\sigma}{J\sigma}/2+\inner{h}{\sigma}).
\end{equation*}
We define its correlation matrix by
\begin{equation*}
 \operatorname{Cor}_{\pi_{J,h}}(\sigma)\defeq D^{-1/2}\operatorname{Cov}_{\pi_{J,h}}(\sigma)D^{-1/2},
\end{equation*}
where $D$ is the diagonal matrix of coordinate variances.
The following consequence of the strong-field estimates bounds this correlation matrix even when the constant field has a perturbation of order $\sqrt m$.

\begin{lemma}[{\cite[Lemma~2.2 and Theorem~1.1]{BandeiraElAlaouiRodder26}}]
\label{lem:strong-field-correlation}
Fix $C_1\geq0$ and $\delta\in(0,1)$.
There is a constant $C=C(C_1,\delta)>0$ such that, for every fixed $C_2\geq0$ and all sufficiently large $m$, the following holds.
If $J\in\R^{m\times m}$ is symmetric with $\norm{J}_{\mathrm{op}}\leq C_1$ and $\max_{|S|\leq\delta m}\norm{2J_{S,S}}_{\mathrm{op}}\leq2/5$, then, for every $v\in\R^m$ with $\norm{v}_2\leq C_2\sqrt m$ and every $\theta\geq C(1+C_2)$,
\begin{equation*}
 \norm{\operatorname{Cor}_{\pi_{J,\theta\ones_m+v}}(\sigma)}_{\mathrm{op}}\leq6.
\end{equation*}
\end{lemma}

{
We briefly explain the proof.
After pinning $\sigma_i=b$, the field perturbation on the remaining coordinates is $v_{-i}+bJ_{-i,i}$, with norm at most $C_2\sqrt m+C_1$.
We use $\delta/2$ as the parameter in \cite[Lemma~2.3 and Corollary~2.4]{BandeiraElAlaouiRodder26}.
By the field threshold in \cite[Eq.~(7)]{BandeiraElAlaouiRodder26}, we can choose $C=C(C_1,\delta)$ so that these estimates hold for every $\theta\geq C(1+C_2)$.
Thus, given any single-coordinate pinning, a sample from $\pi_{J,\theta\ones_m+v}$ has fewer than $\delta m/4$ negative spins with probability at least $1-\e^{-cm}$, for a constant $c>0$.
For two independent samples, given that they differ at coordinate $i$, we apply this bound to the two opposite pinnings and use a union bound to obtain probability at most $2\e^{-cm}$ of disagreement at at least $\delta m$ coordinates.
Finally, by \cite[Lemma~2.2]{BandeiraElAlaouiRodder26} and Eq.~(11) in the proof of \cite[Theorem~1.1]{BandeiraElAlaouiRodder26}, we bound the covariance contribution by $5$ using the small-submatrix assumption, and hence the correlation norm by $5+2m\e^{-cm}\leq6$ for all sufficiently large $m$.
}

{The next lemma adapts the conditioning argument in the proof of \cite[Lemma~3.1]{BandeiraElAlaouiRodder26} to the Poincar\'e inequality in our normalization.}
\begin{lemma}
\label{lem:pinned-glauber-mixing}
Fix $\delta\in(0,1)$, and let $n$ be sufficiently large depending only on $\delta$.
Let $J\in\R^{n\times n}$ be symmetric, let $h\in\R^n$, and set $\pi=\pi_{J,h}$.
Suppose that for every pinning $\sigma_{I^c}=x_{I^c}$ with $I\ne\varnothing$:
\begin{enumerate}
 \item If $|I|\geq\delta n$, the conditional distribution on $I$ has a correlation matrix of operator norm at most $6$.
 \item If $|I|<\delta n$, then $\lambda_{\max}(J_{I,I})-\lambda_{\min}(J_{I,I})\leq2/5$.
\end{enumerate}
Then, for every $f:\{\pm1\}^n\to\R$,
\begin{equation*}
 \Var[\pi]{f}\leq\frac53 n^7\Diri{\pi}{f,f}.
\end{equation*}
\end{lemma}

{
We briefly explain the adjustment.
We take $m=\lceil\delta n\rceil-1\geq12$ and reveal uniformly chosen distinct coordinates of a sample from $\pi$ until $m$ coordinates remain.
Writing $\nu_r$ for the conditional distribution with $r$ coordinates remaining, we apply the variance iteration in the proof of \cite[Lemma~3.1]{BandeiraElAlaouiRodder26} with the correlation bound in hypothesis~{\textup{(1)}} to obtain
\begin{equation*}
 \E{\Var[\nu_m]{f}}
 \geq\prod_{r=m+1}^n(1-6/r)\Var[\pi]{f}
 \geq n^{-6}\Var[\pi]{f}.
\end{equation*}
Here the expectations are over the revealed coordinates and their spin values, and every correlation bound is used with $r\geq m+1\geq\delta n$.
Since $m<\delta n$, by hypothesis~{\textup{(2)}}, a diagonal shift, and \cite[Theorem~1]{EldanKoehlerZeitouni22}, we obtain $\Var[\nu_m]{f}\leq(5m/3)\Diri{\nu_m}{f,f}$.
The remaining coordinates form a uniform $m$-subset, and averaging over the revealed spin values restores $\pi$, so $\E{\Diri{\nu_m}{f,f}}=\Diri{\pi}{f,f}$.
We combine these bounds and $m\leq n$ to obtain the claimed bound.
}

\subsection{Application to the SK model}

We fix $0<\beta_0<\beta_1<1$ and choose $\delta=\delta_{\beta_1}\in(0,1/2)$ such that $\delta\log(\e/\delta)\leq1/(800\beta_1^2)$.
To apply the preceding lemmas, we use the matrix event
\begin{equation*}
 \Omega_n^{\mathrm{end}}
 \defeq\left\{A:\norm{A}_{\mathrm{op}}\leq3,\
   \max_{|I|\leq\delta n}\norm{\beta_1 A_{I,I}}_{\mathrm{op}}\leq\frac15\right\}.
\end{equation*}
{We adapt the concentration and union bound in \cite[Section~3.3]{BandeiraElAlaouiRodder26} to the hollow matrix used here.}
\begin{lemma}
\label{lem:terminal-disorder-tail}
For all sufficiently large $n$,
\begin{equation*}
 \mathbb P_{0,n}((\Omega_n^{\mathrm{end}})^c)
 \leq2\e^{-n/4}+2\e^{-n/(800\beta_1^2)}.
\end{equation*}
\end{lemma}

{
We sketch its proof as follows.
For $I\subseteq[n]$ with $|I|=k\geq1$, we add an independent centered Gaussian diagonal with entry variance $2/n$ to $A_{I,I}$ to obtain a GOE matrix.
By Jensen's inequality for $\lambda_{\max}$ and the GOE expectation bound in \cite[Section~3.3]{BandeiraElAlaouiRodder26}, we obtain $\E{\lambda_{\max}(A_{I,I})}\leq2\sqrt{k/n}$.
The largest eigenvalue is $\sqrt{2/n}$-Lipschitz in the independent standard Gaussian entries above the diagonal, so Gaussian concentration, also applied to $-A_{I,I}$, gives
\begin{equation*}
 \mathbb P_{0,n}\{\|A_{I,I}\|_{\mathrm{op}}>2\sqrt{k/n}+u\}
 \leq2\e^{-nu^2/4},\qquad u\geq0.
\end{equation*}
We take $I=[n]$ and $u=1$ to obtain the first term in the failure bound.
For $k\leq\delta n$, our choice of $\delta$ ensures $2\sqrt\delta\leq1/(10\beta_1)$, so taking $u=1/(10\beta_1)$ bounds each small-submatrix failure by $2\e^{-n/(400\beta_1^2)}$.
We take a union bound over at most $\exp(n\delta\log(\e/\delta))$ subsets and use $\delta\log(\e/\delta)\leq1/(800\beta_1^2)$ to obtain the second term.
}

On this matrix event, we apply \Cref{lem:strong-field-correlation,lem:pinned-glauber-mixing} to obtain a Poincar\'e inequality at the terminal fields needed below.
\begin{lemma}
\label{lem:terminal-field-poincare}
There are constants $T_0,C,N>0$, depending only on $(\beta_0,\beta_1)$, such that, for every $n\geq N$, $A\in\Omega_n^{\mathrm{end}}$, $\beta\in[\beta_0,\beta_1]$, and $y\in\R^n$ with $\||y|-T_0\ones_n\|_2\leq2\sqrt{T_0n}$,
\begin{equation*}
 \Var[\mu_{A,y}]{f}\leq Cn^7\Diri{\mu_{A,y}}{f,f}
\end{equation*}
for every $f:\{\pm1\}^n\to\R$, where $\mu_{A,y}$ is the Gibbs distribution at inverse temperature $\beta$ and external field $y$.
\end{lemma}

\begin{proof}
We first choose $T_0$ large enough that
\begin{equation*}
 T_0\geq C(3\beta_1,\delta)\bigl(1+\delta^{-1/2}(2\sqrt{T_0}+3\beta_1)\bigr),
\end{equation*}
where $C(3\beta_1,\delta)$ is the constant in \Cref{lem:strong-field-correlation}.
Such a fixed $T_0$ exists because the right side grows as $O(\sqrt{T_0})$.
We now fix $A\in\Omega_n^{\mathrm{end}}$, $\beta\in[\beta_0,\beta_1]$, and $y$ with $\||y|-T_0\ones_n\|_2\leq2\sqrt{T_0n}$.

To make the external field nonnegative, we set $s_i\defeq\operatorname{sgn}(y_i)$, with $\operatorname{sgn}(0)=1$, and $S\defeq\operatorname{diag}(s_1,\ldots,s_n)$.
For $\widetilde x\defeq Sx$ and $\widetilde A\defeq SAS$, we have $S^2=I_n$ and
$\frac\beta2\inner{x}{Ax}+\inner{y}{x}
 =\frac\beta2\inner{\widetilde x}{\widetilde A\widetilde x}
   +\inner{|y|}{\widetilde x}$.

Thus, if $X\sim\mu_{A,y}$, then $\widetilde X\defeq SX\sim\mu_{\widetilde A,|y|}$.
Since $S$ is orthogonal, $\norm{\widetilde A}_{\mathrm{op}}=\norm{A}_{\mathrm{op}}$.
For every $I\subseteq[n]$, writing $S_I$ for the restriction of $S$ to $I$, we also have $\norm{\widetilde A_{I,I}}_{\mathrm{op}}
 =\norm{S_I A_{I,I}S_I}_{\mathrm{op}}
 =\norm{A_{I,I}}_{\mathrm{op}}$.
Hence $\widetilde A\in\Omega_n^{\mathrm{end}}$.
For any $f$, we set $\widetilde f(u)\defeq f(Su)$, so that $\Var[\mu_{A,y}]{f}=\Var[\mu_{\widetilde A,|y|}]{\widetilde f}$.

The change of variables only exchanges the two spin values in each coordinate, so the one-coordinate conditional variances, and hence their average, satisfy
$\Diri{\mu_{A,y}}{f,f}
 =\Diri{\mu_{\widetilde A,|y|}}{\widetilde f,\widetilde f}$.
It therefore suffices to prove the Poincar\'e inequality for $\mu_{\widetilde A,|y|}$.

We now check the two pinning conditions in \Cref{lem:pinned-glauber-mixing}.
For a nonempty $I\subseteq[n]$, we fix the spins outside $I$ to $\tau\in\{\pm1\}^{I^c}$.
Expanding the quadratic term in the Gibbs weight, we obtain, for $u\in\{\pm1\}^I$,
\begin{equation*}
 \mu_{\widetilde A,|y|}(\widetilde X_I=u\mid\widetilde X_{I^c}=\tau)
 \propto\exp\left(\frac\beta2\inner{u}{\widetilde A_{I,I}u}
  +\inner{|y|_I+\beta\widetilde A_{I,I^c}\tau}{u}\right),
\end{equation*}
where terms independent of $u$ are absorbed into the normalizing constant.
Thus the interaction on $I$ is $\beta\widetilde A_{I,I}$ and its field is $T_0\ones_I+v_I$, where
$v_I\defeq |y|_I-T_0\ones_I+\beta\widetilde A_{I,I^c}\tau$.

If $|I|\geq\delta n$, we use $\norm{\widetilde A}_{\mathrm{op}}\leq3$ and $\norm{\tau}_2\leq\sqrt n$ to obtain
\begin{align*}
 \norm{v_I}_2
 \leq \||y|-T_0\ones_n\|_2
       +\beta\norm{\widetilde A_{I,I^c}}_{\mathrm{op}}\norm{\tau}_2
 \leq (2\sqrt{T_0}+3\beta_1)\sqrt n
 \leq\delta^{-1/2}(2\sqrt{T_0}+3\beta_1)\sqrt{|I|}.
\end{align*}
Also, $\norm{\beta\widetilde A_{I,I}}_{\mathrm{op}}\leq3\beta_1$, and for every $J\subseteq I$ with $|J|\leq\delta|I|$, we have $\norm{2\beta\widetilde A_{J,J}}_{\mathrm{op}}\leq2/5$.
By \Cref{lem:strong-field-correlation} with $m=|I|$, $C_1=3\beta_1$, $C_2=\delta^{-1/2}(2\sqrt{T_0}+3\beta_1)$, and $\theta=T_0$, the correlation matrix of this pinned distribution has operator norm at most $6$ for all sufficiently large $n$.
If $|I|<\delta n$, we instead use $\widetilde A\in\Omega_n^{\mathrm{end}}$ to obtain
\begin{equation*}
 \lambda_{\max}(\beta\widetilde A_{I,I})-\lambda_{\min}(\beta\widetilde A_{I,I})
 \leq2\norm{\beta\widetilde A_{I,I}}_{\mathrm{op}}\leq\frac25.
\end{equation*}
We have verified both hypotheses of \Cref{lem:pinned-glauber-mixing}, so we can apply it to $\mu_{\widetilde A,|y|}$ and $\widetilde f$.
Using the variance and Dirichlet-form identities above, we obtain
\begin{equation*}
 \Var[\mu_{A,y}]{f}
 =\Var[\mu_{\widetilde A,|y|}]{\widetilde f}
 \leq Cn^7\Diri{\mu_{\widetilde A,|y|}}{\widetilde f,\widetilde f}
 =Cn^7\Diri{\mu_{A,y}}{f,f}.
\end{equation*}
\end{proof}

\subsection{Averaged terminal posterior variance}

\begin{proof}[Proof of \Cref{lem:app-wpi-two-path-terminal}]
We take $T_0,C,N$ from \Cref{lem:terminal-field-poincare}.
By \Cref{lem:terminal-disorder-tail}, the event $\Omega_n^{\mathrm{end}}$ has $\mathbb P_{0,n}$-probability at least $1-C\exp(-cn)$, after adjusting $C,c$.
We fix $A\in\Omega_n^{\mathrm{end}}$ and $\beta\in[\beta_0,\beta_1]$, and let $X,B,Y$ be as in \Cref{lem:app-wpi-two-path-terminal}.
We split the expectation according to whether the terminal field satisfies the condition in \Cref{lem:terminal-field-poincare}.
For this, we define the event
\begin{equation*}
 \mathcal G\defeq\left\{\||Y_{T_0}|-T_0\ones_n\|_2\leq2\sqrt{T_0n}\right\}.
\end{equation*}
We first bound the probability of $\mathcal G^c$.
We write $Y_{T_0}=T_0X+\sqrt{T_0}G$, where $G\sim N(0,I_n)$ is independent of $X$, and use the coordinatewise inequality $\bigl||Y_{T_0,i}|-T_0\bigr|\leq\sqrt{T_0}|G_i|$.
By Markov's inequality and the Gaussian moment-generating function, we obtain
\begin{align*}
 \Pr[X\sim\mu_{\beta,A},\,B]{\mathcal G^c}
 \leq&~\Pr[G\sim N(0,I_n)]{\|G\|_2>2\sqrt n}\\
 \leq&~\exp(-3n/2)\E[G\sim N(0,I_n)]{\exp(3\|G\|_2^2/8)}\\
 =&~\exp(-(3/2-\log2)n).
\end{align*}
By the law of total variance, we have that
\begin{align*}
 \E[X\sim\mu_{\beta,A},\,B]{\Diri{\mu_{A,Y_{T_0}}}{f,f}}
 =&~\frac1n\sum_{i=1}^n\E[X\sim\mu_{\beta,A},\,B]{\Var{f(X)\mid X_{-i},Y_{T_0}}}\\
 \leq&~\frac1n\sum_{i=1}^n\E[X\sim\mu_{\beta,A}]{\Var{f(X)\mid X_{-i}}}
 =\Diri{\mu_{\beta,A}}{f,f}.
\end{align*}
On $\mathcal G$, we apply \Cref{lem:terminal-field-poincare} to bound the posterior variance by $Cn^7\Diri{\mu_{A,Y_{T_0}}}{f,f}$.
On $\mathcal G^c$, we use $\Var[\mu_{A,Y_{T_0}}]{f}\leq\operatorname{osc}(f)^2/4$, which holds because $f$ takes values in $[\inf f,\sup f]$.
Splitting the expectation over these two events, we obtain
\begin{align*}
 \E[X\sim\mu_{\beta,A},\,B]{\Var[\mu_{A,Y_{T_0}}]{f}}
 =&~\E[X\sim\mu_{\beta,A},\,B]{\1_{\mathcal G}\Var[\mu_{A,Y_{T_0}}]{f}}
   +\E[X\sim\mu_{\beta,A},\,B]{\1_{\mathcal G^c}\Var[\mu_{A,Y_{T_0}}]{f}}\\
 \leq&~Cn^7\E[X\sim\mu_{\beta,A},\,B]{\1_{\mathcal G}\Diri{\mu_{A,Y_{T_0}}}{f,f}}
   +\frac{\operatorname{osc}(f)^2}{4}\Pr[X\sim\mu_{\beta,A},\,B]{\mathcal G^c}\\
 \leq&~Cn^7\E[X\sim\mu_{\beta,A},\,B]{\Diri{\mu_{A,Y_{T_0}}}{f,f}}
   +\tfrac14\exp(-cn)\operatorname{osc}(f)^2\\
 \leq&~Cn^7\Diri{\mu_{\beta,A}}{f,f}
   +\tfrac14\exp(-cn)\operatorname{osc}(f)^2,
\end{align*}
where we use the bound on $\Pr[X\sim\mu_{\beta,A},\,B]{\mathcal G^c}$ and the averaged Dirichlet-form inequality proved above, and decrease $c$ if necessary.
\end{proof}

\end{document}